\documentclass[runningheads]{llncs}
\usepackage[T1]{fontenc}
\usepackage{xcolor}
\usepackage[hyphens]{url}
\usepackage{graphicx}
\usepackage{booktabs}
\usepackage{multirow}
\usepackage{tabularx}
\usepackage{subcaption}   %

\usepackage{float}
\usepackage{listings}
\newfloat{listing}{tb}{lst}
\floatname{listing}{Listing}

\usepackage{algpseudocodex}
\usepackage{stmaryrd}
\usepackage{tikz}
\usepackage{tikzsymbols}
\usepackage{pgfplots}
\usepackage{amssymb}
\usepackage{amsmath}
\usepackage{proof}
\usepackage{bussproofs}

\usepackage[]{cancel}
\usepackage{stackengine}
\usepackage[inline]{enumitem}

\usepackage{latexsym}
\usepackage[Euler]{upgreek}
\usepackage{utfsym}
\usepackage{mathtools}

\usepackage{bm}

\usepackage[disable]{todonotes}

\usetikzlibrary{arrows,arrows.meta,calc,decorations.markings,automata,patterns,positioning}
\tikzset{
    varnode/.style={rectangle,outer sep=0mm},
    varnodenoperi/.style={rectangle,outer sep=-1mm},
    ourarrow/.style={>=stealth}, %
ourarc/.style={>=stealth,thick,arc}}

\usepgflibrary{plotmarks}

\presetkeys%
{todonotes}%
{inline,backgroundcolor=yellow,bordercolor=black,linecolor=cyan}{}

\newcommand{\tododavid}[1]
{}
\newcommand{\commentdavid}[1]
{}
\newcommand{\todomohammad}[1]
{}
\newcommand{\donedavid}[1]{}
\newcommand{\donemohammad}[1]{}

\tikzset{inlinenotestyle/.style={notestyle,text width=\linewidth,inner sep=1em,outer sep=2pt,align=justify}}

\spnewtheorem{encoding}{Encoding}{\bfseries}{\rmfamily}
\spnewtheorem{declaration}{Declaration}{\bfseries}{\itshape}
\spnewtheorem{mycases}{Cases}{\bfseries}{\itshape}

\newenvironment{paperdefinition}[1]{\begin{definition}[#1]}{\end{definition}}
\newenvironment{paperremark}{\par}{\par}

\setlist[enumerate]{label=(\roman*),mode=unboxed}
\input{enumitem-unboxed-join}  %
\lstdefinestyle{inline}{%
basicstyle=\ttfamily\small%
}

\usepackage[most]{tcolorbox}
\tcbuselibrary{breakable,hooks,skins,theorems}
\newtcbtheorem{IsabelleTheorem}{Listing}{
enhanced jigsaw
, breakable
, float=h
, top=0pt
, right=0pt
, bottom=0pt
, left=0pt
, boxrule=0pt
, toprule=1pt
, bottomrule=1pt
, titlerule=1pt
, colframe=black
, sharp corners
, colbacktitle=white
, coltitle=black
, fonttitle=\footnotesize
, colback=white
, fontupper=\ttfamily\footnotesize
, before upper={\parindent0em}
}{isa}

\lstdefinelanguage{bnf}{%
    keywords={:=,|},
    keywordstyle=\bfseries,
    basicstyle=\ttfamily\footnotesize,
    moredelim=[is][\bfseries]{\$}{\$}
}

\lstdefinelanguage{pddl}{
sensitive=false,    %
morecomment=[l]{;}, %
alsoletter={:,-},   %
morekeywords={
    define,domain,problem,not,and,or,when,forall,exists,either,
    :domain,:requirements,:types,:objects,:constants,
    :predicates,:action,:parameters,:precondition,:effect,
    :fluents,:primary-effect,:side-effect,:init,:goal,
    :strips,:adl,:equality,:typing,:conditional-effects,
    :negative-preconditions,:disjunctive-preconditions,
    :existential-preconditions,:universal-preconditions,:quantified-preconditions,
    :functions,assign,increase,decrease,scale-up,scale-down,
    :metric,minimize,maximize,
    :durative-actions,:duration-inequalities,:continuous-effects,
    :durative-action,:duration,:condition
},
numberstyle=\ttfamily\footnotesize,
basicstyle=\ttfamily\footnotesize
}

\lstdefinestyle{inline}{%
    basicstyle=\ttfamily%
}

\usepackage{stmaryrd}

\usepackage[most]{tcolorbox}

\lstdefinelanguage{isabelle}{
    morekeywords={record,type_synonym,definition,partial,fun,primrec,function,where,lemma,theorem,unfolding,by,shows,assumes,and,datatype,using,abbreviation
,moreover,have,hence,thus,qed,proof,ultimately,show,next,locale,fixes,tailrec,domintros,defines,
for,context,inductive,begin,end,locale,fixes,record,type_synonym,definition,fun,function,primrec,where,lemma,theorem,unfolding,by,
shows,assumes,and,datatype,using,abbreviation,moreover,have,hence,thus,qed,proof,ultimately,show,next, thm,corollary,\ in\ ,if,case,\ of,let,else,then,
global_interpretation,interpretation}
, sensitive=true
    , showstringspaces=false
    , framerule=0pt
    , xleftmargin=2em
    , numbers=left
    , numberstyle=\ttfamily\small
    , firstnumber=1
    , stepnumber=2
    , basicstyle=\ttfamily\small
    , backgroundcolor = \color{white}
    , keywordstyle = \textbf
    , breaklines=true
    , showspaces=false
    , morecomment=[s]{(*}{*)}
    , commentstyle=\color{gray}
    , morestring=[b]"    
    , literate={\\<times>}{{ $\times$ }}{1} {\\<equiv>}{{ $\equiv$ }}{1} {≡}{{ $\equiv$ }}{1} {\\<forall>}{{$\forall$}}{1} 
    {)))}{{\upshape )))}}{1} {))}{{\upshape ))}}{1} {\ )}{{\upshape )}}{1} {)}{{\upshape )}}{1} {\}}{{\upshape $\rbrace$}}{1}
    {\{}{{\upshape $\lbrace$}}{1}
    {\\<^sup>+}{{$^+$}}{1}  {\\<^sup>-}{{$^-$}}{1} {\\<Delta>}{{$\Delta$}}{1}  {\\<^sub>F}{{$_F$}}{1} {\\<^sub>A}{{$_A$}}{1} {\\<^sub>u}{{$_u$}}{1} {\\<^sub>f}{{$_f$}}{1}
    {\\<^bsub>f\\<^esub>}{{$_\mathtt{f}$}}{1}  {\\<Sum>}{{$\sum$}}{1} {\\<uu>}{{$\mathfrak{u}$}}{1}
      {∀}{{$\forall$}}{1} {\\<exists>}{{$\exists$}}{1} {\\<and>}{{ $\land$}}{1} {∧}{{ $\land$}}{1}
        {\\<in>}{{$\in$}}{1} {\\<Rightarrow>}{{$\Rightarrow$}}{1} {\\<lambda>}{{$\lambda$}}{1} {::}{{$::$}}{1}
        {\\<subset>}{{$\subset$}}{1}  {\\<supset>}{{$\supset$}}{1} {\\<subseteq>}{{ $\subseteq$ }}{1} {\\<^sub>m}{{$_m$}}{1} {\\<^sub>C}{{$_C$}}{1} {\\<^sub>B}{{$_B$}}{1} {\\<longleftrightarrow>}{{ $\longleftrightarrow$ }}{3} {⟷}{{ $\longleftrightarrow$ }}{3}
        {\\<pi>}{{ $\pi$ }}{1} {\\<delta>}{{$\delta$}}{1} {\\<lbrakk>}{{$\llbracket$}}{1} {\\<rbrakk>}{{$\rrbracket$}}{1}
        {\\<Longrightarrow>}{{ $\Longrightarrow$ }}{3} {⟶}{{ $\longrightarrow$ }}{3} 
        {\\<not>}{{$\lnot$}}{1} {\\<le>}{ \upshape {$\le$}}{1} {\\<rightharpoonup>}{{$\rightharpoonup$}}{2}
        {\\<^sub>\\<V>}{{$_{\mathcal V}$}}{1} {\\<lparr>}{{$\llparenthesis$}}{1} {\\<rparr>}{{$\rrparenthesis$}}{1}
        {\\<leftarrow>}{{$\leftarrow$}}{1} {\\<^sub>\\<O>}{{$_{\mathcal O}$}}{1} {\\<^sub>I}{{$_{\texttt{I}}$}}{1}
        {\\<^sub>G}{{$_{\texttt{G}}$}}{2} {\\<phi>}{{$\varphi$}}{1} {\\<Phi>}{{$\Phi$}}{1} {\\<psi>}{{$\psi$}}{1} {\\<Psi>}{{$\Psi$}}{1}
        {\\<^sub>S}{{$_{\texttt S}$}}{1} {\\<inverse>}{{$^{-1}$}}{1} {\\<^sub>O}{{$_O$}}{1} {\\<^bold>\\<And>}{{$\bm\bigwedge$}}{1}
        {\\<^bold>\\<or>}{{$\bm\lor$}}{1} {\\<^sub>G}{{$_{\texttt G}$}}{1} {\\<Pi>}{{$\Pi$}}{1} {\\<^sub>I}{{$_{\texttt I}$}}{1} {\\<noteq>}{{ $\neq$ }}{1} {≠}{{ $\neq$ }}{1}
        {\\<bottom>}{{$\bot$}}{1} {\\<^sub>+}{{$_\texttt +$}}{1} {\\<^bold>\\<and>}{{$\bm\land$}}{1} {\\<^bold>\\<not>}{{$\bm\lnot$}}{1}
        {\\<^sub>1}{{$_1$}}{1} {\\<^sub>2}{{$_2$}}{1} {\\<bar>}{{ $|$ }}{1} {\\<A>}{{$\mathcal A$}}{1} {\\<Turnstile>}{{$\models$}}{2} {\\<^sub>\\<forall>}{{$_\forall$}}{1}
        {\\<^sub>0}{{$_0$}}{1} {\\<tau>}{{$\tau$}}{1}  {\\<^sub>\\<Omega>}{{$_\Omega$}}{1} {\\<^sub>V}{{$_V$}}{1} {\\<^bold>\\<Or>}{{$\bm\bigvee$}}{1}
        {\\<^sub>P}{{$_\texttt P$}}{1} {\\<^sub>X}{{$_\texttt X$}}{1} {\\<longrightarrow>}{{$\longrightarrow$}}{2} {\\<or>}{{$\lor$}}{1} {∨}{{$\lor$}}{1} {\\<^sub>\\<pi>}{{$_\pi$}}{1}
        {\\<^sub>s}{{$_s$}}{1} {\\<^sub>t}{{$_t$}}{1} {\\<^sub>a}{{$_a$}}{1} {\\<^sub>r}{{$_r$}}{1} {\\<^sub>t}{{$_t$}}{1} {\\<^sub>e}{{$_e$}}{1} {\\<^sub>n}{{$_n$}}{1} {\\<^sub>d}{{$_d$}}{1} {\\<^sub>i}{{$_i$}}{1} {\\<^sub>v}{{$_v$}}{1} {\\<^sub>j}{{$_j$}}{1} {\\<^sub>b}{{$_b$}}{1} {\\<inter>}{{$\cap$}}{1} {\\<union>}{{ $\cup$ }}{1} {\\<Union>}{{ $\bigcup$ }}{1} {\\<^sup>c\\<TTurnstile>\\<^sub>=}{{${}^c\models_=$}}{1}
        {\\<open>}{{<}}{1} {\\<close>}{{>}}{1} {\\<langle>}{{$\langle$}}{1} {\\<rangle>}{{$\rangle$}}{1} {\\<ge>}{{\upshape $\ge$ }}{1} {\\<^sup>-\\<^sup>1\\<^sub>C}{{$^{\texttt{-1}}_\texttt{C}$}}{2} {\\<^sup>+\\<^sub>C}{{$^\texttt{+}_\texttt{C}$}}{1} {\\<circ>\\<^sub>C}{{$\circ^\texttt C$}}{1} {\\<top>\\<^sub>C}{{$\top_\texttt{C}$}}{2} {\\<bottom>\\<^sub>C}{{$\bot_\texttt{C}$}}{2} {\\<not>\\<^sub>C}{{$\neg_\texttt{C}$}}{2} {\\<circ>}{{$\circ$}}{1}
        {\\<squnion>\\<^sub>C}{{$\sqcup_\texttt{C}$}}{2} {\\<sqinter>\\<^sub>C}{{$\sqcap_\texttt{C}$}}{2} {\\<exists>\\<^sub>C}{{$\exists_\texttt{C}$}}{2}  {\\<forall>\\<^sub>C}{{$\forall_\texttt{C}$}}{2} {=\\<^sub>C}{{$=_\texttt{C}$}}{2} {\\<sigma>}{{$\sigma$}}{2} {\\<notin>}{{$\notin$}}{1} {\\<oplus>}{{$\oplus$}}{1} {\\<nexists>}{{$\nexists$}}{1} {\\<setminus>}{{$\setminus$}}{1} {_}{{$\text{-}$}}{1} {"}{{}}{1}
    {\ (}{{\upshape \ (}}{1}    {(}{{\upshape (}}{1} {=}{{\upshape =}}{1} {:=}{{\upshape := }}{1}  {\\if}{{if}}{1} {\\else}{{else}}{1} {\\then}{{then}}{1} {\\case}{{case}}{1} {\\of}{{of}}{1}   
        {\\<V>}{{$\mathcal{V}$}}{1} {\\<F>}{{$\mathcal{F}$}}{1} {\\<X>}{{$\mathcal{X}$}}{1} {\\<FF>}{{$\mathfrak{F}$}}{1} {\\<E>}{{$\mathcal{E}$}}{1} {\\<C>}{{$\mathcal{C}$}}{1} {\\<epsilon>}{{\upshape $\epsilon$ }}{1} 
        {\\<gamma>}{{$\gamma$\ }}{1} {'}{{'}}{1} {\\let}{{let}}{2} {\\in}{{in}}{2}
        {\\<b>}{{\upshape b}}{1} {\\<u>}{{\upshape u}}{1} {\\<c>}{{\upshape c}}{1} {\\<And>}{{ $\bigwedge$ }}{1}
        {\\<k>}{{k}}{1} {\\<lceil>}{{$\lceil$}}{1} {\\<rceil>}{{$\rceil$}}{1} {\\<l>}{{$\ell$}}{1}
        {+}{{\upshape +}}{1} {|}{{\upshape |}}{1} {\\<infinity>}{{$\infty$}}{1} {\\<cc>}{{$\mathfrak{c}$}}{1} {\\<EE>}{{$\mathfrak{E}$}}{1} {\\<CC>}{{$\mathfrak{C}$}}{1}
        {\\<space>}{{\ }}{1}
}

\lstnewenvironment{plainisabelle}{
\lstset{language=isabelle, numbers=left, columns=fullflexible, basicstyle=\rmfamily\footnotesize\ttfamily, keywordstyle=\bfseries\upshape    }
}
{}

\newtcblisting{plainisabelle2}{
\lstset{language=isabelle, numbers=left, columns=fullflexible, basicstyle=\rmfamily\footnotesize\ttfamily, keywordstyle=\bfseries\upshape    }

}

\lstdefinestyle{isainline}{
  language=isabelle,
  basicstyle=%
    \ttfamily\small
}

\newcommand{\timelinescale}{0.9}

\makeatletter
\newcommand{\citenamedef}[3]{%
  \@namedef{citename@#1}{#2}\@namedef{citeyear@#1}{#3}}
\newcommand{\citename@lookup}[2]{%
  \@ifundefined{#1@#2}{\textbf{??}\PackageWarning{citenames}{no #1 for `#2'}}{\@nameuse{#1@#2}}}
\newcommand{\citename@map}[2]{%
  \begingroup\let\citename@sep\@empty
  \edef\@tempa{\zap@space#2 \@empty}%
  \@for\@tempb:=\@tempa\do{\citename@sep\citename@lookup{#1}{\@tempb}\def\citename@sep{, }}%
  \endgroup}
\protected\def\citeauthor#1{\citename@map{citename}{#1}}   %
\protected\def\citeyear#1{\citename@map{citeyear}{#1}}
\protected\def\citet#1{\citeauthor{#1}~\cite{#1}}
\newcommand{\citep}[1]{\cite{#1}}
\makeatother
\citenamedef{10.1137/0605024}{Korte and Lov{\'a}sz}{1984}
\citenamedef{2023InternationalPlanning}{??}{????}
\citenamedef{6773228}{Prim}{1957}
\citenamedef{aaai2019}{Abdulaziz}{2019}
\citenamedef{aaai2022}{Abdulaziz and Koller}{2022}
\citenamedef{abboud2016approximation}{Abboud et~al.}{2016}
\citenamedef{abdulaziz2017formally}{Abdulaziz}{2017}
\citenamedef{abdulazizComputingPlanlengthBounds2020}{Abdulaziz and
  Berger}{2020}
\citenamedef{abdulazizExploitingSymmetriesPlanning2015}{Abdulaziz et~al.}{2015}
\citenamedef{abdulazizIsabelleHOLFormalisation2016}{Abdulaziz and
  Paulson}{2016}
\citenamedef{abdulazizMechanisingTheoreticalUpper2014}{Abdulaziz et~al.}{2014}
\citenamedef{abdulazizStateSpaceAcyclicity2017}{Abdulaziz et~al.}{2017}
\citenamedef{abdulazizVerifiedCompositionalAlgorithm2019}{Abdulaziz
  et~al.}{2019}
\citenamedef{abdulazizVerifiedOverApproximationDiameter2015}{Abdulaziz
  et~al.}{2015}
\citenamedef{abelson-et-al:scheme}{Abelson et~al.}{1985}
\citenamedef{AdWords2007}{Mehta et~al.}{2007}
\citenamedef{AFP}{??}{????}
\citenamedef{ahn2015minimum}{Ahn et~al.}{2015}
\citenamedef{ahujaFindingMinimumcostFlows1992}{Ahuja et~al.}{1992}
\citenamedef{aingworth1999fast}{Aingworth et~al.}{1999}
\citenamedef{ajtaiReduction}{Ajtai}{1998}
\citenamedef{alexandriaFinished}{??}{????}
\citenamedef{almostLinearTimeFlows}{Chen et~al.}{2022}
\citenamedef{alonColorcoding1995}{Alon et~al.}{1995}
\citenamedef{alonExponentAllPairs1997}{Alon et~al.}{1997}
\citenamedef{alurPartialorderReductionSymbolic1997}{Alur et~al.}{1997}
\citenamedef{amirFactoredPlanning2003}{Amir and Engelhardt}{2003}
\citenamedef{AmortizedComplexity}{Nipkow}{2015}
\citenamedef{anonymous}{{Hidden}}{2014}
\citenamedef{AoCP_4_2}{Knuth}{2022}
\citenamedef{appelCompilingContinuationsCorr2006}{Appel}{2006}
\citenamedef{approxAsymmTSP}{Svensson et~al.}{2020}
\citenamedef{ApproxPolIterIsabelle}{Sch{\"a}ffeler and Abdulaziz}{2025}
\citenamedef{ARIAProjectReview}{??}{????}
\citenamedef{aroraBarakComplexity}{Arora and Barak}{2009}
\citenamedef{AroraPCP}{Arora and Safra}{1998}
\citenamedef{arthurSmoothedAnalysisKMeans2011}{Arthur et~al.}{2011}
\citenamedef{artOfProgrammingCombAlgos}{Knuth}{2011}
\citenamedef{asai:2019}{Asai}{2019}
\citenamedef{asai:and:fukunaga:2018}{Asai and Fukunaga}{2018}
\citenamedef{asghariansardroudApproximationAlgorithmLongest2016}{Asgharian~Sardroud
  and Bagheri}{2016}
\citenamedef{AutoCorres}{Greenaway et~al.}{2012}
\citenamedef{autocorresII}{??}{????}
\citenamedef{autocorrode}{??}{2025}
\citenamedef{automationProbLogic}{Kaminski}{2019}
\citenamedef{avigadFormalMaths}{Avigad}{2023}
\citenamedef{AvigadHarissonWhyFormal}{Avigad and Harrison}{2014}
\citenamedef{avigadSATEncodings}{Codel et~al.}{2023}
\citenamedef{AwardsBritishAcademy2023}{??}{2023}
\citenamedef{aziz2012machine}{Abdulaziz et~al.}{2012}
\citenamedef{baaderIntroductionDescriptionLogic2017}{Baader et~al.}{2017}
\citenamedef{babaiCanonicalLabelingGraphs1983}{Babai and Luks}{1983}
\citenamedef{backstromComplexityResultsSAS1995}{B{\"a}ckstr{\"o}m and
  Nebel}{1995}
\citenamedef{backstromPlanningPolynomialTime1991}{B{\"a}ckstr{\"o}m and
  Klein}{1991}
\citenamedef{ballCompressionPowerSymbolic}{Ball and Holte}{????}
\citenamedef{balyo2013relaxing}{Balyo}{2013}
\citenamedef{Bancerek2008IntroductionTM}{Bancerek and Shidama}{2008}
\citenamedef{banihashemiAbstractionSituationCalculus2017a}{Banihashemi
  et~al.}{2017}
\citenamedef{barrettFormalMethodsApplied1989}{Barrett}{1989}
\citenamedef{bartheFoundationsProbabilisticProgramming2020}{??}{2020}
\citenamedef{bauer2013tableaux}{Bauer et~al.}{2013}
\citenamedef{baumgartner2002property}{Baumgartner et~al.}{2002}
\citenamedef{baumgartnerEnhancedDiameterBounding2004}{Baumgartner and
  Kuehlmann}{2004}
\citenamedef{bayerMathematicalProofGenerations2024}{Bayer et~al.}{2024}
\citenamedef{Bellman57}{Bellman}{1957}
\citenamedef{bellmanFordI}{Bellman}{1958}
\citenamedef{beniaminiBipartitePerfectMatching2021}{Beniamini and Nisan}{2021}
\citenamedef{bensalem2014verification}{Bensalem et~al.}{2014}
\citenamedef{bentkampVerifiedReductionsOptimization2023}{Bentkamp et~al.}{2023}
\citenamedef{berezin1998compositional}{Berezin et~al.}{1998}
\citenamedef{bergeGraphsHypergraphs1973}{Berge}{1973}
\citenamedef{BergesLemma}{Berge}{1957}
\citenamedef{bernsteinComplexityDecentralizedControl2013}{Bernstein
  et~al.}{2013}
\citenamedef{bertot2004interactive}{Bertot and Cast{\'e}ran}{2004}
\citenamedef{betterTspApproxKarlin}{Karlin et~al.}{2021}
\citenamedef{beynierSolvingEfficientlyDecentralized2011}{Beynier and
  Mouaddib}{2011}
\citenamedef{beyondWorstCase}{Roughgarden}{2021}
\citenamedef{bezemMechanizationProofHessenberg2008}{Bezem and Hendriks}{2008}
\citenamedef{BiereCCZ99}{Biere et~al.}{1999}
\citenamedef{BiereFazekasFleuryHeisinger-SAT-Competition-2020-solvers}{Biere
  et~al.}{2020}
\citenamedef{biernackaStrongCallValue2021}{Biernacka et~al.}{2021}
\citenamedef{bishopPatternRecognitionMachine2006}{Bishop}{2006}
\citenamedef{bjorklundApproximatingLongestDirected2004}{Bj{\"o}rklund
  et~al.}{2004}
\citenamedef{bjorklundFindingPathSuperlogarithmic2003}{Bj{\"o}rklund and
  Husfeldt}{2003}
\citenamedef{blackBlackBoxRaising1998}{Black and Wiliam}{1998}
\citenamedef{blanchette2016verified}{Blanchette et~al.}{2016}
\citenamedef{blanchetteNitpickCounterexampleGenerator2010}{Blanchette and
  Nipkow}{2010}
\citenamedef{bloomTaxonomyEducationalObjectives1956}{Bloom}{1956}
\citenamedef{BlossomIsabelle}{Abdulaziz et~al.}{2019}
\citenamedef{BlossomIsabelleJAR}{Abdulaziz and Mehlhorn}{2025}
\citenamedef{blossomMNInvAckermann}{Gabow}{1976}
\citenamedef{blumFastPlanningPlanning1997}{Blum and Furst}{1997}
\citenamedef{blumMatching}{Blum}{1990}
\citenamedef{blumMaximumMatchingGeneral2015}{Blum}{2015}
\citenamedef{BMC}{Biere et~al.}{2003}
\citenamedef{boasMachineModels}{Boas}{1990}
\citenamedef{boasReasonableModels}{Slot and Boas}{1984}
\citenamedef{boldo2006high}{Boldo and Munoz}{2006}
\citenamedef{bondyGraphTheory}{Bondy and Murty}{2008}
\citenamedef{bonetPlanningHeuristicSearch2001}{Bonet and Geffner}{2001}
\citenamedef{bonwellActiveLearningCreating1991}{Bonwell and Eison}{1991}
\citenamedef{boudgoustHardnessModuleLearning2022}{Boudgoust et~al.}{2022}
\citenamedef{bourbakisettheory}{Bourbaki}{2004}
\citenamedef{boutilier:etal:2000}{Boutilier et~al.}{2000}
\citenamedef{brachman-schmolze:kl-one}{Brachman and Schmolze}{1985}
\citenamedef{brafman2006factored}{Brafman and Domshlak}{2006}
\citenamedef{brafmanQualitativePlanningPartial2013}{Brafman et~al.}{2013}
\citenamedef{brafmanRegularDecisionProcesses2024}{Brafman and De~Giacomo}{2024}
\citenamedef{brafmanStructureComplexityPlanning2003}{Brafman and
  Domshlak}{2003}
\citenamedef{brakerskiClassicalHardnessLearning2013}{Brakerski et~al.}{2013}
\citenamedef{Braun2024}{Braun et~al.}{2024}
\citenamedef{brown2013loneliness}{Brown}{2013}
\citenamedef{brownBacktrackSearchingPresence1996}{Brown et~al.}{1996}
\citenamedef{bruennerLTL}{Brunner et~al.}{2019}
\citenamedef{brunnerFormalVerificationExecutable2016}{Brunner and
  Lammich}{2016}
\citenamedef{bryantCertifiedKnowledgeCompilation2023}{Bryant et~al.}{2023}
\citenamedef{buchbinderOnlineKTaxiDouble2021}{Buchbinder et~al.}{2021}
\citenamedef{BuchbinderOnlineMatchingVtxWeightsPrimalDual}{Buchbinder
  et~al.}{2007}
\citenamedef{buchbinderPrimalDualOnline}{Buchbinder}{2008}
\citenamedef{bultermanComputingLongestPath2002}{Bulterman et~al.}{2002}
\citenamedef{bundala2012magnitude}{Bundala et~al.}{2012}
\citenamedef{bundyProofPlanning}{Bundy}{2002}
\citenamedef{bylanderComputationalComplexityPropositional1994}{Bylander}{1994}
\citenamedef{caseEnhancedVerificationTemporal2009}{Case et~al.}{2009}
\citenamedef{cashmoreOpportunisticPlanningAutonomous2018}{Cashmore
  et~al.}{2018}
\citenamedef{CauchyResidue}{Li and Paulson}{2016}
\citenamedef{certBinarisedNNs}{Yang et~al.}{2025}
\citenamedef{certificationTemporalPlanningWshop}{Wang and Abdulaziz}{2025}
\citenamedef{chan2010more}{Chan}{2010}
\citenamedef{chatterjeeFixedPointCertificates2025}{Chatterjee et~al.}{2025}
\citenamedef{chechik2014better}{Chechik et~al.}{2014}
\citenamedef{chen2009completeness}{Chen and Yao}{2009}
\citenamedef{childreBackwardDesignTargeting2009}{Childre et~al.}{2009}
\citenamedef{ChrisTspApprox}{Christofides}{1972}
\citenamedef{chuchufanStochContCerts}{So et~al.}{2023}
\citenamedef{churchSTT}{Church}{1940}
\citenamedef{cimattiPlanningModelChecking1997}{Cimatti et~al.}{1997}
\citenamedef{cizeljProbabilisticallySafeVehicle2011}{Cizelj et~al.}{2011}
\citenamedef{clarke2009turing}{Clarke et~al.}{2009}
\citenamedef{clarkeCompletenessComplexityBounded2004}{Clarke et~al.}{2004}
\citenamedef{clarkeCounterexampleguidedAbstractionRefinement2000}{Clarke
  et~al.}{2000}
\citenamedef{clarkeExploitingSymmetryTemporal1996}{Clarke et~al.}{1996}
\citenamedef{ClarkeMCbook}{Clarke et~al.}{1999}
\citenamedef{clarkeModelCheckingAbstraction1994}{Clarke et~al.}{1994}
\citenamedef{clarkeProgressStateExplosion2001}{Clarke et~al.}{2001}
\citenamedef{clarkeSymmetryReductionsModel1998}{Clarke et~al.}{1998}
\citenamedef{classenExploringBoundariesDecidable}{Cla{\ss}en et~al.}{????}
\citenamedef{classenSymbolicVerificationGolog}{Classen}{????}
\citenamedef{ClassicalHardnessLearning}{??}{????}
\citenamedef{classicalPlanningSemAFP}{Abdulaziz and Lammich}{2020}
\citenamedef{CLRS}{Cormen et~al.}{2009}
\citenamedef{codellama}{??}{2025}
\citenamedef{collectionsFramework}{Lammich and Lochbihler}{2010}
\citenamedef{CompilingContinuationsContinued}{??}{????}
\citenamedef{completenessThresholdAbstractInterpretation}{Feldman et~al.}{2022}
\citenamedef{completenessThresholdsMemory}{Reinhard et~al.}{2023}
\citenamedef{complexityBDDsVardi}{Feigenbaum et~al.}{1999}
\citenamedef{ComputersAndIntractability}{Garey and Johnson}{1990}
\citenamedef{ConcreteSemantics}{Nipkow and Klein}{2014}
\citenamedef{concurrencymonad}{Claessen}{1999}
\citenamedef{constable2000constructively}{Constable et~al.}{2000}
\citenamedef{continuitymonad}{Neves et~al.}{2016}
\citenamedef{contStochasticLogic}{Baier et~al.}{2003}
\citenamedef{cookLevin}{Cook}{1971}
\citenamedef{cookLevinCoq}{G{\"a}her and Kunze}{2021}
\citenamedef{cookLevinFramework}{Kappelmann et~al.}{2025}
\citenamedef{cookLevinGit}{Abdulaziz}{????}
\citenamedef{CookLevinTMIsabelle}{Balbach}{2023}
\citenamedef{cooperWeightedCSPApproach2011}{Cooper et~al.}{2011}
\citenamedef{Coqprover}{Bertot}{2008}
\citenamedef{coquandCalculusConstructions1988}{Coquand and Huet}{1988}
\citenamedef{crawford1996symmetry}{Crawford et~al.}{1996}
\citenamedef{cryptHOL}{Basin et~al.}{2020}
\citenamedef{cryspen}{??}{????}
\citenamedef{crystalsKyber}{Bos et~al.}{2017}
\citenamedef{culbersonSearchingPatternDatabases1996}{Culberson and
  Schaeffer}{1996}
\citenamedef{CVC3prover}{Barrett and Tinelli}{2007}
\citenamedef{dadushSmooth}{Dadush and Huiberts}{2019}
\citenamedef{dadushSmoothedAnalysisSimplex2020}{Dadush and Huiberts}{2020}
\citenamedef{daniSublinearTimeAlgorithmNearlyPerfect2025}{Dani and Hayes}{2025}
\citenamedef{dankelmann2010diameter}{Dankelmann and Volkmann}{2010}
\citenamedef{dankelmannDiameterDirectedGraphs2005}{Dankelmann}{2005}
\citenamedef{dankelmannDiameterMaximumDegree2016}{Dankelmann and
  Dorfling}{2016}
\citenamedef{datatypesIsabelle}{Traytel et~al.}{2012}
\citenamedef{daviesAdvancingMathematicsGuiding2021}{Davies et~al.}{2021}
\citenamedef{DBLP:books/lib/Knuth98}{Knuth}{1998}
\citenamedef{DBLP:books/lib/Knuth98a}{Knuth}{1998}
\citenamedef{DBLP:conf/aaai/CimattiMR17}{Cimatti et~al.}{2017}
\citenamedef{DBLP:conf/aaai/ClassenHL07}{Cla{\ss}en et~al.}{2007}
\citenamedef{DBLP:conf/aaai/GiunchigliaM07}{Giunchiglia and Maratea}{2007}
\citenamedef{DBLP:conf/aips/0001SUW18}{Katz et~al.}{2018}
\citenamedef{DBLP:conf/aips/ButtnerR05}{B{\"u}ttner and Rintanen}{2005}
\citenamedef{DBLP:conf/aips/Helmert04}{Helmert}{2004}
\citenamedef{DBLP:conf/aips/PommereningH12}{Pommerening and Helmert}{2012}
\citenamedef{DBLP:conf/aips/ScalaRHT16}{Scala et~al.}{2016}
\citenamedef{DBLP:conf/atal/0001GSW19}{Aziz et~al.}{2019}
\citenamedef{DBLP:conf/cav/KanavL014}{Kanav et~al.}{2014}
\citenamedef{DBLP:conf/ijcai/BelleL15a}{Belle and Levesque}{2015}
\citenamedef{DBLP:conf/ijcai/HoffmannGSK07}{Hoffmann et~al.}{2007}
\citenamedef{DBLP:conf/ijcai/LeofanteGAT20}{Leofante et~al.}{2020}
\citenamedef{DBLP:conf/kr/KlassenML20}{Klassen et~al.}{2020}
\citenamedef{DBLP:conf/kr/LakemeyerL04}{Lakemeyer and Levesque}{2004}
\citenamedef{DBLP:conf/nips/ChristianoLBMLA17}{Christiano et~al.}{2017}
\citenamedef{DBLP:conf/popl/KumarMNO14}{Kumar et~al.}{2014}
\citenamedef{DBLP:conf/sat/Devriendt0BD16}{Devriendt et~al.}{2016}
\citenamedef{DBLP:conf/sat/EenS03}{E{\'e}n and S{\"o}rensson}{2003}
\citenamedef{DBLP:conf/sat/WetzlerHH14}{Wetzler et~al.}{2014}
\citenamedef{DBLP:conf/tacas/ClarkeKL04}{Clarke et~al.}{2004}
\citenamedef{DBLP:journals/ai/BelleL18}{Belle and Levesque}{2018}
\citenamedef{DBLP:journals/ai/GereviniHLSD09}{Gerevini et~al.}{2009}
\citenamedef{DBLP:journals/aicom/GereviniSV15}{Gerevini et~al.}{2015}
\citenamedef{DBLP:journals/aim/ColesCOJLSY12}{Coles et~al.}{2012}
\citenamedef{DBLP:journals/amai/GiacomoTR20}{Giacomo et~al.}{2020}
\citenamedef{DBLP:journals/corr/abs-2112-05996}{Chevallier and Fleuriot}{2021}
\citenamedef{DBLP:journals/jair/MuiseBM16}{Muise et~al.}{2016}
\citenamedef{DBLP:journals/pacmpl/0002JKD18}{Jung et~al.}{2018}
\citenamedef{DBLP:phd/basesearch/Winterer24}{Winterer}{2024}
\citenamedef{DBLP:phd/dnb/Classen13}{Cla{\ss}en}{2013}
\citenamedef{de2006assisted}{De~Dinechin et~al.}{2006}
\citenamedef{de2010certifying}{De~Dinechin et~al.}{2010}
\citenamedef{deepMindIMO}{Trinh et~al.}{2024}
\citenamedef{deepseek-aiDeepSeekR1IncentivizingReasoning2025}{{DeepSeek-AI}
  et~al.}{2025}
\citenamedef{deepseekProver}{Xin et~al.}{2024}
\citenamedef{degiacomoNondeterministicSituationCalculus2021}{De~Giacomo and
  Lesp{\'e}rance}{2021}
\citenamedef{degiacomoSituationCalculusController2022}{De~Giacomo et~al.}{2022}
\citenamedef{delormeImprovedInstanceGeneration2022}{Delorme et~al.}{2022}
\citenamedef{derakhshan2026improved}{Derakhshan et~al.}{2026}
\citenamedef{deslauriersMeasuringActualLearning2019}{Deslauriers et~al.}{2019}
\citenamedef{DevanurOnlineMatchingPrimalDual}{Devanur et~al.}{2013}
\citenamedef{devlin2012fundamentals}{Devlin}{1979}
\citenamedef{devlinZermeloFraenkelAxioms1979}{Devlin}{1979}
\citenamedef{diestelGraphTheory4th2012}{Diestel}{2012}
\citenamedef{dijkstra1959note}{Dijkstra}{1959}
\citenamedef{dilithiumEasyCrypt}{Barbosa et~al.}{2023}
\citenamedef{doczkalConstructiveTheoryRegular2013}{Doczkal et~al.}{2013}
\citenamedef{domshlak2012enhanced}{Domshlak et~al.}{2012}
\citenamedef{DomshlakKS13}{Domshlak et~al.}{2013}
\citenamedef{dragerDirectedModelChecking2009}{Dr{\"a}ger et~al.}{2009}
\citenamedef{DS05SmoothAnalysisSimplex}{Deshpande and Spielman}{2005}
\citenamedef{DuanScalingMatching}{Duan et~al.}{2017}
\citenamedef{dufourdIntuitionisticProofDiscrete2009}{Dufourd}{2009}
\citenamedef{dutertreYices222014}{Dutertre}{2014}
\citenamedef{EasyCrypt}{Barthe et~al.}{2014}
\citenamedef{ebenHowSetStudent2010}{Eben et~al.}{2010}
\citenamedef{edelkampBDDsStrikeBack2015}{Edelkamp et~al.}{2015}
\citenamedef{edelkampLexPartitioningNewOption2012}{Edelkamp et~al.}{2012}
\citenamedef{edelkampLimitsPossibilitiesBDDs2008}{Edelkamp and Kissmann}{2008}
\citenamedef{edelkampOBDDsHeuristicSearch1998}{Edelkamp and Reffel}{1998}
\citenamedef{edelkampOptimalSymbolicPlanning}{Edelkamp and Kissmann}{????}
\citenamedef{edmond1965blossom}{Edmonds}{1965/ed}
\citenamedef{edmondsHypergraphs}{Edmonds and Paulson}{2024}
\citenamedef{EdmondsKarpScalingFlows}{Edmonds and Karp}{1972}
\citenamedef{EdmondsMatroids1971}{Edmonds}{1971}
\citenamedef{EdmondsPaulsonCombDesign}{Edmonds and Paulson}{2021}
\citenamedef{EdmondsWeightedMatching}{Edmonds}{1965}
\citenamedef{emersonSymmetryModelChecking1996}{Emerson and Sistla}{1996}
\citenamedef{englertWorstCaseProbabilistic2014}{Englert et~al.}{2014}
\citenamedef{Eprover}{Schulz}{1999}
\citenamedef{ercoli1970software}{Ercoli and Bauer}{1970}
\citenamedef{erdosRadiusDiameterMinimum1989}{Erd{\H o}s et~al.}{1989}
\citenamedef{erickson2014models}{Erickson}{2014}
\citenamedef{eriksonSATPlanCert}{Eriksson and Helmert}{2020}
\citenamedef{eriksonStateSpaceCert}{Eriksson et~al.}{2017}
\citenamedef{ErikssonOptimalCert}{Mugdan et~al.}{2023}
\citenamedef{erkTheo}{Erk and Priese}{2000}
\citenamedef{ernest2004relevance}{Ernest}{2004}
\citenamedef{esparza2013fully}{Esparza et~al.}{2013}
\citenamedef{eugenePermutationGroupsPolynomialtime1993}{Eugene}{1993}
\citenamedef{EUKidneyMatchingReview}{Bir{\'o} et~al.}{2021}
\citenamedef{EulerPolyhedronIsabelle}{Paulson}{2023}
\citenamedef{experimentalMathsFormalisation}{Kontorovich}{2022}
\citenamedef{eyerichUsingContextenhancedAdditive2009}{Eyerich et~al.}{2009}
\citenamedef{eyerichUsingContextEnhancedAdditive2012}{Eyerich et~al.}{2012}
\citenamedef{fastdownward}{Helmert}{2006}
\citenamedef{fc0df122-3305-33a8-bac5-7a6fc3666dfb}{Kruskal}{1956}
\citenamedef{fdsBook}{??}{2025}
\citenamedef{feigenbaumComplexityProblemsGraphs1998}{Feigenbaum et~al.}{1998}
\citenamedef{fengControllerSynthesisAutonomous2015}{Feng et~al.}{2015}
\citenamedef{ferreinLogicbasedRobotControl2008}{Ferrein and Lakemeyer}{2008}
\citenamedef{fikes1971strips}{Fikes and Nilsson}{1971}
\citenamedef{filiotAntichainsCompositionalAlgorithms2011}{Filiot et~al.}{2011}
\citenamedef{finkCreatingSignificantLearning2013}{Fink}{2013}
\citenamedef{fiorellaRelativeBenefitsLearning2013}{Fiorella and Mayer}{2013}
\citenamedef{fishburnAnotherOptimizationAlphabeta1983}{Fishburn}{1983}
\citenamedef{floatingPointHOLLight}{Harrison}{1999}
\citenamedef{floatingPointIntel}{Harrison}{2003}
\citenamedef{flocq}{Boldo and Melquiond}{2011}
\citenamedef{flowsRailways1955}{Harris and Ross}{1955}
\citenamedef{floydWarshallIsabelle}{Wimmer and Lammich}{2017}
\citenamedef{FordFulkerson}{Ford and Fulkerson}{1956}
\citenamedef{FordFulkersonMizar}{Lee}{2005}
\citenamedef{formalisingUndefinedTerms}{Coen and Zoli}{????}
\citenamedef{formalMathsNytimes}{??}{????}
\citenamedef{formalMathsQuanta}{??}{????}
\citenamedef{formosa}{{team}}{????}
\citenamedef{formplan}{??}{????}
\citenamedef{forsterLambdaTM}{Forster et~al.}{2019}
\citenamedef{forsterThesis}{Forster}{2021}
\citenamedef{forsterTuringMachines}{Forster et~al.}{2020}
\citenamedef{fox2002pddl+}{Fox and Long}{2002}
\citenamedef{fox2003pddl2}{Fox and Long}{2003}
\citenamedef{fox2005validating}{Fox et~al.}{2005}
\citenamedef{fox2006exploration}{Fox et~al.}{2006}
\citenamedef{foxDetectionExploitationSymmetry1999}{Fox and Long}{1999}
\citenamedef{FoxL02}{Fox and Long}{2002}
\citenamedef{foxLongExplainable}{Fox et~al.}{2017}
\citenamedef{francesGeneralizedPotentialHeuristics2019}{Franc{\`e}s
  et~al.}{2019}
\citenamedef{FRANK1981328}{Frank}{1981}
\citenamedef{fredmanFibonacciHeapsTheir1987}{Fredman and Tarjan}{1987}
\citenamedef{fredmanNewBoundsComplexity1976}{Fredman}{1976}
\citenamedef{freek100}{??}{????}
\citenamedef{FricPapco80}{Fric and Papco}{2010}
\citenamedef{fullyOnlineMatchingI}{Huang et~al.}{2020}
\citenamedef{fullyOnlineMatchingII}{Huang et~al.}{2020}
\citenamedef{functionPackageIsabelle}{Krauss}{2009}
\citenamedef{GabowMatchingScaling}{Gabow}{2017}
\citenamedef{GabowTarjanMatchingScaling}{Gabow and Tarjan}{1989}
\citenamedef{GabowTarjanScalingMatching}{Gabow and Tarjan}{1991}
\citenamedef{galeshapley}{Gale and Shapley}{1962}
\citenamedef{galeshapleyIsabelle}{Nipkow}{2024}
\citenamedef{galperinSuccinctRepresentationsGraphs1983}{Galperin and
  Wigderson}{1983}
\citenamedef{gamboaFormalizationPowerlistAlgebra2009}{Gamboa}{2009}
\citenamedef{gamboaMechanicalProofCookLevin2004}{Gamboa and Cowles}{2004}
\citenamedef{gamlathOnlineMatchingGeneral2019}{Gamlath et~al.}{2019}
\citenamedef{ganaiCompletenessSMTbasedBMC2008}{Ganai and Gupta}{2008}
\citenamedef{gargOneSidedMatchingMarkets2022}{Garg et~al.}{2022}
\citenamedef{gargOnesidedMatchingMarkets2024}{Garg et~al.}{2024}
\citenamedef{GenPlanIsabelleHSDIP}{Abdulaziz et~al.}{2022}
\citenamedef{genPlanning}{Bonet and Geffner}{????}
\citenamedef{gereviniInferringStateConstraints1998}{Gerevini and
  Schubert}{1998}
\citenamedef{gerthPartialOrderApproach1995}{Gerth et~al.}{1995}
\citenamedef{ghallabAutomatedPlanningTheory2004}{Ghallab et~al.}{2004}
\citenamedef{giganteDecidabilityComplexity2022}{Gigante et~al.}{2022}
\citenamedef{giganteDecidabilityComplexityActionBased2020}{Gigante
  et~al.}{2020}
\citenamedef{gimenezInfluenceKdependenceComplexity2012}{Gim{\'e}nez and
  Jonsson}{2012}
\citenamedef{Giry80}{Giry}{1982}
\citenamedef{giryMonadIsabelle}{Eberl et~al.}{2015}
\citenamedef{gls:hypertrees}{Gottlob et~al.}{2002}
\citenamedef{glucosesolver}{Audemard and Simon}{2018}
\citenamedef{godefroidUsingPartialOrders1990}{Godefroid}{1990}
\citenamedef{goedelUeberFormalUnentscheidbare1931}{G{\"o}del}{1931}
\citenamedef{goelOnlineBudgetedMatching2008}{Goel and Mehta}{2008}
\citenamedef{GoldbergScalingShortestPath}{Goldberg}{1995}
\citenamedef{goldman2012using}{Goldman et~al.}{2012}
\citenamedef{goldreichPrimerPseudorandomGenerators2010}{Goldreich}{2010}
\citenamedef{goldwasserProbabilisticEncryption1984}{Goldwasser and
  Micali}{1984}
\citenamedef{gonthier2008formal}{Gonthier}{2008}
\citenamedef{gonthierMachineCheckedProofOdd2013}{Gonthier et~al.}{2013}
\citenamedef{gordon1993introduction}{Gordon and Melham}{1993}
\citenamedef{gordonHOLTypes}{Gordon}{1985}
\citenamedef{gottlob:nonmon}{Gottlob}{1992}
\citenamedef{graphAlgosBook}{Abdulaziz}{2025}
\citenamedef{graphAlgosinCoq}{Mohan et~al.}{2021}
\citenamedef{graphLibRepo}{??}{2024}
\citenamedef{gratLammich}{Lammich}{2017}
\citenamedef{green1969application}{Green}{1969}
\citenamedef{GreensIsabelle}{Abdulaziz and Paulson}{2019}
\citenamedef{guaranteedsafeAI}{Dalrymple et~al.}{2024}
\citenamedef{gueneauMechanizedVerificationCorrectness2019}{Gu{\'e}neau}{2019}
\citenamedef{guereOneActionEnough2001}{Guere and Alami}{2001}
\citenamedef{guestrinEfficientSolutionAlgorithms2003}{Guestrin et~al.}{2003}
\citenamedef{gutinFindingLongestPath1993}{Gutin}{1993}
\citenamedef{haddadIntervalIterationAlgorithm2018}{Haddad and Monmege}{2018}
\citenamedef{HaftmannW06}{Haftmann and Wenzel}{2006}
\citenamedef{hahnelGOLEXBridgingGap1998}{H{\"a}hnel et~al.}{1998}
\citenamedef{hales2011revision}{Hales et~al.}{2011}
\citenamedef{hanCounterexamplesProbabilisticModel2007}{Han and Katoen}{2007}
\citenamedef{handbookTheoryA}{??}{1990}
\citenamedef{handbookTheoryB}{??}{1990}
\citenamedef{hanheide:etal:2017}{Hanheide et~al.}{2017}
\citenamedef{harperSchoolComputerScience2019}{Harper}{2019}
\citenamedef{harrison2007formalizing}{Harrison}{2007}
\citenamedef{hartFormalBasisHeuristic1968}{Hart et~al.}{1968}
\citenamedef{hashPDef}{Valiant}{1979}
\citenamedef{haslum:07}{Haslum}{2007}
\citenamedef{haslum2003domain}{Haslum and Scholz}{2003}
\citenamedef{haslumAdmissibleHeuristicsOptimal2000}{Haslum and Geffner}{2000}
\citenamedef{haslumComputingGenomeEdit2011}{Haslum and {al.}}{2011}
\citenamedef{haslumDomainindependentConstructionPattern2007}{Haslum
  et~al.}{2007}
\citenamedef{hattie2014using}{Hattie and Yates}{2014}
\citenamedef{havelundAutomatedTestingPlanning2008}{Havelund et~al.}{2008}
\citenamedef{helmertConciseFinitedomainRepresentations2009}{Helmert}{2009}
\citenamedef{helmertMergeandshrinkAbstractionMethod2014}{Helmert et~al.}{2014}
\citenamedef{hemaspaandra2010complexity}{Hemaspaandra et~al.}{2010}
\citenamedef{hermannsProbabilisticCEGAR2008}{Hermanns et~al.}{2008}
\citenamedef{hitzlerMathematicalAspectsLogic}{Hitzler}{????}
\citenamedef{HoareGolog}{Liu}{2002}
\citenamedef{hoareRefinement}{Hoare}{1972}
\citenamedef{hodges1993model}{Hodges}{1993}
\citenamedef{HoffmanKruskalIntegral}{Hoffman and Kruskal}{2010}
\citenamedef{hoffmann2001ff}{Hoffmann and Nebel}{2001}
\citenamedef{HOLlight}{Harrison}{2009}
\citenamedef{HOLsystem}{Slind and Norrish}{2008}
\citenamedef{Holzl13}{H{\"o}lzl}{2013}
\citenamedef{holzl2017markov}{H{\"o}lzl}{2017}
\citenamedef{holzlThreeChaptersMeasure2011}{H{\"o}lzl and Heller}{2011}
\citenamedef{holzlTypeClassesFilters2013}{H{\"o}lzl et~al.}{2013}
\citenamedef{holzmannModelCheckerSPIN1997}{Holzmann}{1997}
\citenamedef{HopcroftKarp1973}{Hopcroft and Karp}{1973}
\citenamedef{HopcroftTarjanSPQR}{Hopcroft and Tarjan}{1973}
\citenamedef{hosseiniNashBargainingBasedModelsMatching2022}{Hosseini and
  Vazirani}{2022}
\citenamedef{houPracticalGraphIsomorphism2014}{Hou et~al.}{2014}
\citenamedef{hpcpl}{{van Doorn} et~al.}{2023}
\citenamedef{HttpsArxivorgPdfa}{??}{????}
\citenamedef{HttpsCiteseerxIst}{??}{????}
\citenamedef{HttpsDropsdagstuhldeStorage}{??}{????}
\citenamedef{HttpsFmmizarorg200513}{??}{????}
\citenamedef{HttpsTheorycsprincetoneduComplexity}{??}{????}
\citenamedef{HttpsWwwAria}{??}{????}
\citenamedef{huangOnlinePrimalDual2020}{Huang and Zhang}{2020}
\citenamedef{huffmanTransfer}{Huffman and Kuncar}{2013}
\citenamedef{hughesUseSelfreferentialIpsative2014}{Hughes et~al.}{2014}
\citenamedef{HungarianMethodAssignment}{Kuhn}{1955}
\citenamedef{hybridautomata}{Henzinger}{1996}
\citenamedef{hyperGraphs}{{Koutsoukou-Argyraki} et~al.}{2023}
\citenamedef{ictai2018}{Abdulaziz and Lammich}{2018}
\citenamedef{immlerODESolver}{Immler}{2018}
\citenamedef{IMP2}{Lammich and Wimmer}{????}
\citenamedef{improvedTASPApproxHalfIntegral}{Karlin et~al.}{2020}
\citenamedef{institutfurtheoretischeinformatiktudresdenDecidabilityVerifyingLTL2013}{{Institut
  f\"ur Theoretische Informatik TU Dresden} et~al.}{2013}
\citenamedef{inval}{Haslum}{2024}
\citenamedef{IPC1998reference}{McDermott}{2000}
\citenamedef{IPC2022}{Taitler et~al.}{2024}
\citenamedef{ipc98}{Long et~al.}{2000}
\citenamedef{IPeqPSPACE}{Shamir}{1992}
\citenamedef{ipVerifyingSystemsReplicated1996}{Ip and Dill}{1996}
\citenamedef{ipVerifyingSystemsReplicated1999}{Ip and Dill}{1999}
\citenamedef{Iris}{Jung et~al.}{2015}
\citenamedef{isabelleGo}{St{\"u}binger and Hupel}{2025}
\citenamedef{isabelleGreensThm}{Abdulaziz et~al.}{2018}
\citenamedef{IsabelleHOLRef}{Nipkow et~al.}{2002}
\citenamedef{isabelleKYBER}{Kreuzer}{2024}
\citenamedef{isabelleLLVM}{Lammich}{2019}
\citenamedef{IsabelleSimplex}{Spasi{\'c} and Mari{\'c}}{2012}
\citenamedef{IsabelleSimplexAFP}{Maric et~al.}{2018}
\citenamedef{jacksonSoftwareAbstractionsLogic2006}{Jackson}{2006}
\citenamedef{jani}{Budde et~al.}{2017}
\citenamedef{jaumeFullFormalizationSLDresolution1999}{Jaume}{1999}
\citenamedef{jewellMinCostSSP}{Jewell}{1958}
\citenamedef{jonsson1998state}{Jonsson and B{\"a}ckstr{\"o}m}{1998}
\citenamedef{jonssonTractablePlanExistence1998}{Jonsson and
  B{\"a}ckstr{\"o}m}{1998}
\citenamedef{joslinExploitingSymmetryLifted1997}{Joslin and Roy}{1997}
\citenamedef{kaelbling1998planning}{Kaelbling et~al.}{1998}
\citenamedef{kamareddineFormalizingStrongNormalization2003}{Kamareddine and
  Qiao}{2003}
\citenamedef{karayelExpanderGraphs2023}{Karayel}{2023}
\citenamedef{karayelListUpdate}{Karayel}{2022}
\citenamedef{katharinaKyber}{Kreuzer}{2023}
\citenamedef{katharinNPHard}{Kreuzer and Nipkow}{2023}
\citenamedef{katzReluplexEfficientSMT2017}{Katz et~al.}{2017}
\citenamedef{kautz:selman:92}{Kautz and Selman}{1992}
\citenamedef{kautz:selman:96}{Kautz and Selman}{1996}
\citenamedef{kelareva2007factored}{Kelareva et~al.}{2007}
\citenamedef{KeplerConjecture}{Hales et~al.}{2015}
\citenamedef{KeplerConjectureInformal}{Hales}{2005}
\citenamedef{klarlundNlognAlgorithmOnline1999}{Klarlund}{1999}
\citenamedef{klein2009sel4}{Klein et~al.}{2009}
\citenamedef{knightPeerInstruction2018}{Knight and Brame}{2018}
\citenamedef{knoblock:94}{Knoblock}{1994}
\citenamedef{knoblock1991characterizing}{Knoblock et~al.}{1991}
\citenamedef{knuth2015art}{Knuth}{2015}
\citenamedef{knuthAnalysisAlphaBetaPriming1975a}{Knuth and Moore}{1975}
\citenamedef{knyazevDiametersPseudosymmetricGraphs1987}{Knyazev}{1987}
\citenamedef{kolbExperientialLearningExperience2014}{Kolb}{2014}
\citenamedef{kolmogorovCoq}{Forster et~al.}{2022}
\citenamedef{korfFindingOptimalSolutions1997}{Korf}{1997}
\citenamedef{Korte1978AnAO}{Korte and Hausmann}{1978}
\citenamedef{KorteVygenOptimisation}{Korte and Vygen}{2012}
\citenamedef{kovacscomplete}{Kovacs}{2011}
\citenamedef{krebbersStandardFormalizedCoq}{Krebbers}{????}
\citenamedef{kroening2011linear}{Kroening et~al.}{2011}
\citenamedef{kroeningComputingOverapproximationsBounded2006}{Kroening}{2006}
\citenamedef{KroeningS03}{Kroening and Strichman}{2003}
\citenamedef{kumar2014hol}{Kumar et~al.}{2014}
\citenamedef{kuncarTypesToSets}{Kuncar and Popescu}{2016}
\citenamedef{kurutach:etal:2018}{Kurutach et~al.}{2018}
\citenamedef{KVV90}{Karp et~al.}{1990}
\citenamedef{kyberEasyCryptIV}{Almeida et~al.}{2023}
\citenamedef{kyberEasyCryptV}{Almeida et~al.}{2024}
\citenamedef{kyberVerificationIsabelle}{Kreuzer}{2024}
\citenamedef{lambdaCalcHOL4}{Norrish}{2006}
\citenamedef{lambdaCalcTuringCoq}{Forster et~al.}{2021}
\citenamedef{lammichContainers}{Lammich and Lochbihler}{2019}
\citenamedef{lammichEfficientVerifiedSAT2017}{Lammich}{2017}
\citenamedef{LammichFlows}{Lammich and Sefidgar}{2019}
\citenamedef{LammichLLVM}{Lammich}{2022}
\citenamedef{lammichMonads}{Lammich and Tuerk}{2012}
\citenamedef{lammichMonadsII}{Lammich}{2012}
\citenamedef{LammichRefinement}{Lammich}{2019}
\citenamedef{lammichSepRef}{Lammich and Meis}{2012}
\citenamedef{lamport2018if}{Lamport and others}{2018}
\citenamedef{LaplacianMethodGraphsLP}{Spielman and Teng}{2004}
\citenamedef{larryGoedelNominalIsabelle}{Paulson}{2015}
\citenamedef{larrySTTIsabelle}{Paulson}{1988}
\citenamedef{Lawler1975}{Lawler}{1975}
\citenamedef{lawlerCombinatorialOptimizationNetworks2001}{Lawler}{2001}
\citenamedef{Leanprover}{de~Moura and Ullrich}{2021}
\citenamedef{LEDAbook}{Mehlhorn and N{\"a}her}{1999}
\citenamedef{leroy2009formal}{Leroy}{2009}
\citenamedef{levesque:belief}{Levesque}{1984}
\citenamedef{levesque:functional-foundations}{Levesque}{1984}
\citenamedef{levesqueGOLOGLogicProgramming1997}{Levesque et~al.}{1997}
\citenamedef{liAutomaticVerificationPartial2015}{Li and Liu}{2015}
\citenamedef{libsodium}{??}{2025}
\citenamedef{linearOptimisationBertsimas}{Bertsimas and Tsitsiklis}{1997}
\citenamedef{LiquidTensorChallenge}{Scholze}{2022}
\citenamedef{liquidtensorexperimentnature}{Castelvecchi}{2021}
\citenamedef{listUpdateIsabelle}{Haslbeck and Nipkow}{2016}
\citenamedef{liuFasterEnergyMaximization2020}{Liu and Sidford}{2020}
\citenamedef{liuVerificationBeliefPrograms}{Liu and Lakemeyer}{????}
\citenamedef{LLMPDDLModeling}{Tantakoun et~al.}{2025}
\citenamedef{LocalesBallarin}{Ballarin}{2014}
\citenamedef{LocalesKammueler}{Kamm{\"u}ller et~al.}{1999}
\citenamedef{lochbihler2013light}{Lochbihler}{2013}
\citenamedef{lochbihler2018lazy}{Lochbihler and Stoop}{2018}
\citenamedef{loewenberg}{L{\"o}wenberg and Haslbeck}{????}
\citenamedef{long2009planning}{Long et~al.}{2009}
\citenamedef{louiOptimalOnLineSimulations1992}{Loui and Luginbuhl}{1992}
\citenamedef{lozanoComplexityGraphProblems1989}{Lozano and Balc{\'a}zar}{1989}
\citenamedef{LPDualityIsabelle}{Thiemann}{2022}
\citenamedef{LPsCertIsabelle}{Obua and Nipkow}{2009}
\citenamedef{lukasOrdersProc}{Stevens and Nipkow}{2021}
\citenamedef{lukasSetsProc}{Stevens}{2023}
\citenamedef{lunaNewPolynomialClasses2007}{Luna et~al.}{2007}
\citenamedef{macaulayMatroid}{Grayson and Stillman}{????}
\citenamedef{macaulayMatroidLink}{??}{????}
\citenamedef{maccarthyTheoryOfComputation}{McCarthy}{1963}
\citenamedef{MAGAUD2012406}{Magaud et~al.}{2012}
\citenamedef{makoelle2014pedagogy}{Makoelle}{2014}
\citenamedef{manoliosOrdinalArithmeticAlgorithms2005}{Manolios and Vroon}{2005}
\citenamedef{manuelLimitsISSAC}{Eberl}{2019}
\citenamedef{maricFormalizationImplementationModern2009}{Mari{\'c}}{2009}
\citenamedef{maricIsomorphism}{Mari{\'c}}{2020}
\citenamedef{MassotWhyFormal}{Massot}{N/A}
\citenamedef{matchingCertIsabelle}{Alkassar et~al.}{2014}
\citenamedef{matchingCombGeometry}{Aigner and Dowling}{1971}
\citenamedef{matchingGeneralGraphsQuasiNC}{Svensson and Tarnawski}{2017}
\citenamedef{MatchingHospitalDoctor}{Kamada and Kojima}{2017}
\citenamedef{MatchingSkewSymmetric}{Goldberg and Karzanov}{2004}
\citenamedef{mathlib}{??}{2024}
\citenamedef{matichukEisbachProofMethod2016}{Matichuk et~al.}{2016}
\citenamedef{matroidIntersectionNC}{Gurjar and Thierauf}{2017}
\citenamedef{MatroidIsabelle}{Abdulaziz et~al.}{2025}
\citenamedef{matroidParityPrimalDual}{Iwata and Kobayashi}{2022}
\citenamedef{mccarthy1981some}{McCarthy and Hayes}{1969}
\citenamedef{mccarthy1981someBook}{McCarthy and Hayes}{1969}
\citenamedef{mccarthy1985formalization}{McCarthy}{1985}
\citenamedef{mccarthyProgramsCommonSense}{McCarthy}{????}
\citenamedef{mcconnell1996active}{McConnell}{1996}
\citenamedef{mcdermott1996heuristic}{McDermott}{1996}
\citenamedef{mcdermottFormalSemanticsProcesses}{McDermott}{????}
\citenamedef{MckayNautyII}{McKay and Piperno}{2014}
\citenamedef{mckinnaLambdaCalculusType1999}{McKinna and Pollack}{1999}
\citenamedef{mehlhorn98mfcs}{Mehlhorn and N{\"a}her}{1998}
\citenamedef{mehlhornemail}{Mehlhorn}{2020}
\citenamedef{mehtaBayesianMechanismDesign2013}{Mehta}{2013}
\citenamedef{meuleauSolvingPOMDPsSearching2013}{Meuleau et~al.}{2013}
\citenamedef{MicaliVazirani}{Micali and Vazirani}{1980}
\citenamedef{MichaelisN-TYPES17}{Michaelis and Nipkow}{2018}
\citenamedef{michaelisPropositionalProofSystems2017}{Michaelis and
  Nipkow}{2017}
\citenamedef{michaelWhereEvidenceThat2006}{Michael}{2006}
\citenamedef{miguel2001symmetry}{Miguel}{2001}
\citenamedef{milner1972logic}{Milner}{1972}
\citenamedef{milner1997definition}{Milner}{1997}
\citenamedef{miner1995defining}{Miner}{1995}
\citenamedef{minsky1967computation}{Minsky}{1967}
\citenamedef{mochiba}{Sickert and K{\v r}et{\'i}nsk{\'y}}{2016}
\citenamedef{modelcountingCert}{Tan et~al.}{2024}
\citenamedef{moonDiameterGraph1965}{Moon and others}{1965}
\citenamedef{moscatoFloatingPoint2017}{Moscato et~al.}{2017}
\citenamedef{moscatoFP2018}{Titolo et~al.}{2018}
\citenamedef{moscatoFP2020}{Titolo et~al.}{2020}
\citenamedef{moscatoFP2023}{Ferreira et~al.}{2023}
\citenamedef{moscatoFP2024}{Faissole et~al.}{2024}
\citenamedef{moscatoFP2025}{Titolo et~al.}{2025}
\citenamedef{mulmuleyVazVaz}{Mulmuley et~al.}{1987}
\citenamedef{MVProofCurrent}{Vazirani}{2020}
\citenamedef{MVproofone}{Vazirani}{1990}
\citenamedef{MVproofthree}{Vazirani}{2024}
\citenamedef{MVprooftwo}{Vazirani}{1994}
\citenamedef{nakamuraFashodaMeetTheorem}{Nakamura and Trybulec}{????}
\citenamedef{NakhostCompilingResources}{Nakhost et~al.}{2012}
\citenamedef{nebel:jair-2000}{Nebel}{2000}
\citenamedef{NFH}{Abdulaziz}{2024}
\citenamedef{Nipkow-ICTAC06}{Nipkow}{2006}
\citenamedef{nipkowAlphaBeta}{Nipkow}{2024}
\citenamedef{nipkowApprox}{E{\ss}mann et~al.}{2022}
\citenamedef{nipkowProofPearlDefining2005}{Nipkow and Paulson}{2005}
\citenamedef{nipkowRunningTime}{Nipkow}{2025}
\citenamedef{nipkowVerifiedTextbookAlgorithms2020}{Nipkow et~al.}{2020}
\citenamedef{nipkowWhileCombinator}{Berghofer and Nipkow}{2002}
\citenamedef{nishidaConversionTailRecursion2014}{Nishida and Vidal}{2014}
\citenamedef{NNVerificationFarkas}{Isac et~al.}{2022}
\citenamedef{norell2007towards}{Norell}{2007}
\citenamedef{norrishC}{Norrish}{1999}
\citenamedef{norrishComputabilityThy}{Norrish}{2011}
\citenamedef{noschinskiGraphLib}{Noschinski}{2015}
\citenamedef{nuXmv}{Cavada et~al.}{2014}
\citenamedef{odedGoldrecihCompelxityBook}{Goldreich}{2008}
\citenamedef{onlineAlgosBook}{Borodin and {El-Yaniv}}{1998}
\citenamedef{onlinebmatching}{Kalyanasundaram and Pruhs}{2000}
\citenamedef{onlineLPPackingPrimalDual}{Kesselheim et~al.}{2014}
\citenamedef{onlineMatchingEcon}{Eden et~al.}{2021}
\citenamedef{onlineMatchingSimple2008}{Birnbaum and Mathieu}{2008}
\citenamedef{onlineVertexWeightedMatching}{Aggarwal et~al.}{2011}
\citenamedef{optCritMinCostFlows}{Klein}{1967}
\citenamedef{OPTICplanner}{Benton et~al.}{2012}
\citenamedef{OptimalSATPlan}{Abdulaziz}{2021}
\citenamedef{OrlinScalingAssignment}{Orlin and Ahuja}{1992}
\citenamedef{OrlinScalingFlow}{Orlin}{1993}
\citenamedef{papadimitriouNoteSuccinctRepresentations1986}{Papadimitriou and
  Yannakakis}{1986}
\citenamedef{parallelPerfectMatching}{Anari and Vazirani}{2020}
\citenamedef{parallelPerfectMatchingPlanar}{Anari and Vazirani}{2020}
\citenamedef{pardalosNoteComplexityLongest2004}{Pardalos and Migdalas}{2004}
\citenamedef{parsonsKnowledgeActionLogical2005}{Parsons}{2005}
\citenamedef{pathTSPtoTSPReduction}{Traub et~al.}{2022}
\citenamedef{paulson1994isabelle}{Paulson}{1994}
\citenamedef{paulson2010three}{Paulson and Blanchette}{2010}
\citenamedef{paulsonFormalisationFiniteAutomata2015}{Paulson}{2015}
\citenamedef{pctlIsabelle}{H{\"o}lzl and Nipkow}{2012}
\citenamedef{PDDL}{McDermott et~al.}{1998}
\citenamedef{pddl3}{Gerevini and Long}{????}
\citenamedef{pednault1989adl}{Pednault}{1989}
\citenamedef{peikertPublickeyCryptosystemsWorstcase2009}{Peikert}{2009}
\citenamedef{penix1998using}{Penix et~al.}{1998}
\citenamedef{petersenAugPaths}{Petersen}{1891}
\citenamedef{petersonGeneralMaximumMatching1988}{Peterson and Loui}{1988}
\citenamedef{petterssonImprovingSolutionTimes2021}{Pettersson et~al.}{2021}
\citenamedef{pnueliSmallModelProperty2002}{Pnueli et~al.}{2002}
\citenamedef{pochter2011exploiting}{Pochter et~al.}{2011}
\citenamedef{poincarebendixson}{Immler and Tan}{2020}
\citenamedef{porco2013automatic}{Porco et~al.}{2013}
\citenamedef{porterHandbookKnowledgeRepresentation2008}{??}{2008}
\citenamedef{PQVerificationFujisaki}{Unruh}{2020}
\citenamedef{premiseSel}{Alama et~al.}{2014}
\citenamedef{primaldualConvexOnline}{Molinaro}{2021}
\citenamedef{primalDualOnlineReview}{Buchbinder and Naor}{2009}
\citenamedef{primalDualResourceAllocOnline}{Ekbatani et~al.}{2023}
\citenamedef{primalDualRoutingPackingOnline}{Buchbinder and Naor}{2006}
\citenamedef{primsLammichNipkow}{Lammich and Nipkow}{2019}
\citenamedef{princeDoesActiveLearning2004}{Prince}{2004}
\citenamedef{principlesMC}{Baier and Katoen}{2008}
\citenamedef{PRISM2006}{Hinton et~al.}{2006}
\citenamedef{probContPlanning}{Younes and Simmons}{????}
\citenamedef{probContPlanningWilliams}{Ono et~al.}{2013}
\citenamedef{probMCRobots}{Lahijanian et~al.}{2012}
\citenamedef{probModelChecking}{Baier et~al.}{1997}
\citenamedef{probNondetMonad}{Mio et~al.}{2021}
\citenamedef{probProgrammingCerts}{Chatterjee et~al.}{2022}
\citenamedef{probTemporalPlanning}{Little et~al.}{????}
\citenamedef{projectEverestCrypto}{{team}}{????}
\citenamedef{protzenkoVerifiedLowlevelProgramming2017}{Protzenko et~al.}{2017}
\citenamedef{provingRecursivelyLLMs}{Wang et~al.}{2024}
\citenamedef{PupilSchoolMatching}{Abdulkadiroglu}{2005}
\citenamedef{Puterman94}{Puterman}{1994}
\citenamedef{qrhlUnruh}{Unruh}{2019}
\citenamedef{qrhlUnruhGithub}{Unruh}{2025}
\citenamedef{qsoptex}{Applegate et~al.}{2007}
\citenamedef{quantitativeProgramAnalysisIsabelle}{Haslbeck}{2021}
\citenamedef{Rado1957}{Rado}{1957}
\citenamedef{raghvendraOptimalAnalysisOnline2018a}{Raghvendra}{2018}
\citenamedef{randomAlgosCoq}{Audebaud and {Paulin-Mohring}}{2009}
\citenamedef{randomTreesIsabelle}{Eberl et~al.}{2020}
\citenamedef{rankingHighProbability}{Mihail and Tr{\"o}bst}{2021}
\citenamedef{RankingIsabelle}{Abdulaziz and Madlener}{2023}
\citenamedef{rankoohITSATEfficientSATBased2015}{Rankooh and
  {Ghassem-Sani}}{2015}
\citenamedef{RauGeometricAlgos}{Rau and Nipkow}{2020}
\citenamedef{RBTIsabelle}{Nipkow}{2025}
\citenamedef{RD}{Abdulaziz and Berger}{2021}
\citenamedef{RDDLref}{Sanner}{????}
\citenamedef{reductionsLWE}{Langlois and Stehl{\'e}}{2015}
\citenamedef{reductionsLWEstoc}{Brakerski et~al.}{2013}
\citenamedef{reductionsLWEstocB}{Peikert}{2009}
\citenamedef{refinementWirth}{Wirth}{1971}
\citenamedef{regin1994filtering}{R{\'e}gin}{1994}
\citenamedef{reiterSitCalc}{Reiter}{2001}
\citenamedef{richter2010lama}{Richter and Westphal}{2010}
\citenamedef{riddle:etal:16}{Riddle et~al.}{2016}
\citenamedef{RingerProofRepair}{First et~al.}{2023}
\citenamedef{RingerReviewProgramLogics}{Ringer et~al.}{2019}
\citenamedef{rintanen:06}{Rintanen et~al.}{2006}
\citenamedef{rintanen:08}{Rintanen}{2008}
\citenamedef{rintanen:12}{Rintanen}{2012}
\citenamedef{Rintanen:Gretton:2013}{Rintanen and Gretton}{2013}
\citenamedef{rintanen2003symmetry}{Rintanen}{2003}
\citenamedef{rintanenComplexityPlanningPartial}{Rintanen}{????}
\citenamedef{rintanenComplexityProbabilisticPlanning}{Rintanen}{????}
\citenamedef{rintanenEvaluationStrategiesPlanning2004}{Rintanen}{2004}
\citenamedef{RLFactoredMDPs}{Osband and Roy}{????}
\citenamedef{RLFactoredMDPsII}{Deng et~al.}{2022}
\citenamedef{RLforProving}{Gauthier et~al.}{2021}
\citenamedef{robinson09}{Robinson et~al.}{2009}
\citenamedef{robinson2010partial}{Robinson et~al.}{2010}
\citenamedef{rodittyFastApproximationAlgorithms2013}{Roditty and
  Vassilevska~Williams}{2013}
\citenamedef{roehl2013flipped}{Roehl et~al.}{2013}
\citenamedef{romTheory}{Mittelbach and Fischlin}{2021}
\citenamedef{Russell1910-RUSPMV}{Russell and Whitehead}{1910}
\citenamedef{russellArtificialIntelligenceModern2010}{Russell et~al.}{2010}
\citenamedef{russellArtificialIntelligenceModern2020}{Russell and Norvig}{2020}
\citenamedef{russellProvablyBeneficialArtificial}{Russell}{????}
\citenamedef{ruttenMathematicalTechniquesAnalyzing2004}{Rutten}{2004}
\citenamedef{saboVygenTSPApprox}{Seb{\"o} and Vygen}{2014}
\citenamedef{sacerdotiPlanningHierarchyAbstraction1974}{Sacerdoti}{1974}
\citenamedef{SafeguardedAITA11}{??}{????}
\citenamedef{sageMatroid}{Developers}{????}
\citenamedef{sangHeuristicsFastExact2005}{Sang et~al.}{2005}
\citenamedef{sankarSmoothedAnalysisCondition2006}{Sankar et~al.}{2006}
\citenamedef{sannerFactoredMDPs}{Delgado et~al.}{2011}
\citenamedef{sastrySolvabilitybasedComparisonFailure2014}{Sastry and
  Widder}{2014}
\citenamedef{savagemodelcomputing}{Savage}{1998}
\citenamedef{savitchThm}{Savitch}{1970}
\citenamedef{say:and:sanner:2018}{Say and Sanner}{2018}
\citenamedef{ScalingIsabelle}{Abdulaziz and Ammer}{2024}
\citenamedef{SCCsMDPs}{Hartmanns et~al.}{2023}
\citenamedef{schimpfConstructionBuchiAutomata2009}{Schimpf et~al.}{2009}
\citenamedef{schrijverBook}{Schrijver}{2003}
\citenamedef{schueleSymbolicModelChecking}{Schuele and Schneider}{????}
\citenamedef{seboEightFifthApproximationPath2013}{Seb{\"o}}{2013}
\citenamedef{seippfast}{Seipp et~al.}{2014}
\citenamedef{Selman94}{Selman}{1994}
\citenamedef{SelsamNNVerification}{Selsam et~al.}{2017}
\citenamedef{sepLogicParallelRandom}{Tassarotti and Harper}{2019}
\citenamedef{seshiaVerifiedArtificialIntelligence2022}{Seshia et~al.}{2022}
\citenamedef{shadowVertex}{Gass and Saaty}{1955}
\citenamedef{sheeranCheckingSafetyProperties2000}{Sheeran et~al.}{2000}
\citenamedef{shleyfman2015heuristics}{Shleyfman et~al.}{2015}
\citenamedef{sievers:2015}{Sievers et~al.}{2015}
\citenamedef{sieversEfficientImplementationPattern2012}{Sievers et~al.}{2012}
\citenamedef{silver2010monte}{Silver and Veness}{2010}
\citenamedef{siminiceanuExperimentalEvaluationPlanning2008}{Siminiceanu
  et~al.}{2008}
\citenamedef{simplexDantzig}{Dantzig}{1951}
\citenamedef{simplIsabelle}{Schirmer}{2005}
\citenamedef{sistlaComplexityPropositionalLinear1985}{Sistla and Clarke}{1985}
\citenamedef{sitCalcPDDL+}{Batusov and Soutchanski}{2019}
\citenamedef{skipListsIsabelle}{Haslbeck and Eberl}{2020}
\citenamedef{smallGoodQuestionsGreat2012}{Small}{2012}
\citenamedef{smith2000challenges}{Smith et~al.}{2000}
\citenamedef{smith2008anml}{Smith et~al.}{2008}
\citenamedef{smithModelCheckingAutonomous2005}{Smith et~al.}{2005}
\citenamedef{smithValidationVerificationRemote1999}{Smith et~al.}{1999}
\citenamedef{smoothedAnalysis}{Spielman and Teng}{2004}
\citenamedef{SMVmodelchecker}{McMillan}{1993}
\citenamedef{soaresMaximumDiameterRegular1992}{Soares}{1992}
\citenamedef{SoftwareFoundations}{??}{????}
\citenamedef{sohrabiImplementingFlippedClassroom2016}{Sohrabi and Iraj}{2016}
\citenamedef{SolovyevBBJRG19}{Solovyev et~al.}{2018}
\citenamedef{soplex}{Bestuzheva et~al.}{2023}
\citenamedef{soundSMC}{Budde et~al.}{2025}
\citenamedef{spielmanSmoothedAnalysisTermination2003}{Spielman and Teng}{2003}
\citenamedef{SPQRMutzel}{Gutwenger and Mutzel}{2000}
\citenamedef{sprengerVerifiedModelChecker1998}{Sprenger}{1998}
\citenamedef{steepestEdgeSimplex}{Goldfarb}{1976}
\citenamedef{stormModelChecker}{Dehnert et~al.}{2017}
\citenamedef{streeterUsingDecisionProcedures2007}{Streeter and Smith}{2007}
\citenamedef{successiveShortestPath}{Minty}{1960}
\citenamedef{sutton1998introduction}{Sutton and Barto}{1998}
\citenamedef{SvenssonAsymmTSPRelax}{Svensson}{2015}
\citenamedef{SVPComplexity}{Ajtai}{1996}
\citenamedef{szegedyPremiseSelGNNs}{Paliwal et~al.}{2020}
\citenamedef{TAMERPlanner}{Valentini et~al.}{2020}
\citenamedef{tarjan1972depth}{Tarjan}{1972}
\citenamedef{taylorFixedPointProperty1991}{Taylor}{1991}
\citenamedef{tegmarkProvablySafeSystems2023}{Tegmark and Omohundro}{2023}
\citenamedef{temporalFD}{Eyerich et~al.}{2009}
\citenamedef{temporalLogic}{Pnueli}{1977}
\citenamedef{temporalPlanningRefinementMC}{Heinz et~al.}{2019}
\citenamedef{temporalPlanningTimedAutomata2002}{Dierks et~al.}{2002}
\citenamedef{theoremsFreeWadler}{Wadler}{1989}
\citenamedef{thiebaux2005defense}{Thi{\'e}baux et~al.}{2005}
\citenamedef{thiemannJordanNormalForms}{Thiemann and Yamada}{2016}
\citenamedef{ThiemannLPs}{Bottesch et~al.}{2019}
\citenamedef{ThomposonHaskell}{Thompson}{2011}
\citenamedef{thorupIntegerPriorityQueues2003}{Thorup}{2003}
\citenamedef{timedautomata}{Alur and Dill}{1994}
\citenamedef{TjoinsusingMatching}{Padberg and Rao}{1982}
\citenamedef{toyer:etal:2018}{Toyer et~al.}{2018}
\citenamedef{traubImprovedApproximationAlgorithm2020}{Traub and Vygen}{2020}
\citenamedef{trobstRealPolynomialBipartite2020}{Tr{\"o}bst and Vazirani}{2020}
\citenamedef{trobstRealPolynomialBipartite2023}{Tr{\"o}bst and Vazirani}{2023}
\citenamedef{troelstraBasicProofTheory2000}{Troelstra and Schwichtenberg}{2000}
\citenamedef{TrustworthyAICommunications2021}{??}{2021}
\citenamedef{TspHeuristicLinKerr}{Lin and Kernighan}{1973}
\citenamedef{turingMachinesMatita}{Asperti and Ricciotti}{2012}
\citenamedef{tutteSkewFlowsGraph}{Tutte}{1967}
\citenamedef{ueharaEfficientAlgorithmsLongest2004}{Uehara and Uno}{2004}
\citenamedef{unsolvableTemporalLMs}{Marzal et~al.}{2008}
\citenamedef{unsolvableTemporalPlanning}{Panjkovic et~al.}{2022}
\citenamedef{urbanTuringMachines}{Xu et~al.}{2013}
\citenamedef{vajjhaFormalizationStochasticApproximation2022}{Vajjha
  et~al.}{2022}
\citenamedef{VajjhaSTPF21}{Vajjha et~al.}{2021}
\citenamedef{VAL}{Howey et~al.}{2004}
\citenamedef{valiantComplexityEnumerationReliability1979}{Valiant}{1979}
\citenamedef{ValIterIsabelle}{Sch{\"a}ffeler and Abdulaziz}{2023}
\citenamedef{vallati20152014}{Vallati et~al.}{2015}
\citenamedef{valmeekamLLMsStillCant2024}{Valmeekam et~al.}{2024}
\citenamedef{vampire}{Riazanov and Voronkov}{1999}
\citenamedef{VaziraniApprox}{Vazirani}{????}
\citenamedef{vaziraniComputationalComplexityHyllandZeckhauser2021a}{Vazirani
  and Yannakakis}{2021}
\citenamedef{VaziraniMVDraft}{Vazirani}{2014}
\citenamedef{vaziraniOnlineMatching2022}{Vazirani}{2022}
\citenamedef{vaziraniOnlineMatchingArxiv}{Vazirani}{2022}
\citenamedef{vaziraniOnlineMatchingBasedMarket}{Vazirani et~al.}{????}
\citenamedef{verifiedBMC}{Ishii and Fujii}{2020}
\citenamedef{verifiedCertPMC}{Chatterjee et~al.}{2025}
\citenamedef{verifiedFPArith}{Kohlen et~al.}{2025}
\citenamedef{verifiedIsabelleKernel}{Ro{\ss}kopf and Nipkow}{2022}
\citenamedef{VerifiedReductionsOptimization}{??}{????}
\citenamedef{verifiedSATPlan}{Abdulaziz and Kurz}{2023}
\citenamedef{verifiedSATSolver}{Blanchette et~al.}{2018}
\citenamedef{vershySmoothedAnalysisSimples}{Vershynin}{2009}
\citenamedef{vijayemail}{Vazirani}{2020}
\citenamedef{vygenemailproof}{Vygen}{2024}
\citenamedef{wahlReplicationAbstractionSymmetry2010}{Wahl and Donaldson}{2010}
\citenamedef{walsh2005rules}{Walsh}{2005}
\citenamedef{walshBreakingValueSymmetry2007}{Walsh}{2007}
\citenamedef{walshTrustAITrustworthyArtificial}{Walsh}{????}
\citenamedef{wendaLiSketchProveLLMs}{Jiang et~al.}{2023}
\citenamedef{west2001introduction}{West}{2001}
\citenamedef{WhileHOLfour}{Gordon}{1989}
\citenamedef{williamsReactivePlannerModelbased1997}{Williams and Nayak}{1997}
\citenamedef{wimmerCertTimedAutomata}{Wimmer and von Mutius}{2020}
\citenamedef{wimmerMonadDP}{Wimmer et~al.}{2018}
\citenamedef{wingTrustworthyAI2020}{Wing}{2020}
\citenamedef{winskelOpSem}{Winskel}{1993}
\citenamedef{wuFormalisationMyhillnerodeTheorem2011}{Wu et~al.}{2011}
\citenamedef{XenaLiquidTensor}{{xenaproject}}{2021}
\citenamedef{xu2011theory}{Xu et~al.}{2011}
\citenamedef{younes2004ppddlTR}{Younes and Littman}{2004}
\citenamedef{yuster2010computing}{Yuster}{2010}
\citenamedef{Z3solver}{de~Moura and Bj{\o}rner}{2008}
\citenamedef{zarriessDecidableVerificationGolog2015}{Zarrie{\ss} and
  Cla{\ss}en}{2015}
\citenamedef{zarriessDecidableVerificationKnowledgeBased}{Zarrie{\ss} and
  Cla{\ss}en}{????}
\citenamedef{zarriessVerificationKnowledgeBasedPrograms2015}{Zarrie{\ss} and
  Cla{\ss}en}{2015}
\citenamedef{zarriessVerifyingCTLProperties}{Zarrie{\ss} and Cla{\ss}en}{????}
\citenamedef{zhanAUTO2SaturationbasedHeuristic2016}{Zhan}{2016}
\citenamedef{zhangFormalizingReinforcementLearning2025}{Zhang}{2025}
\citenamedef{zhouDontTrustVerify2023}{Zhou et~al.}{2023}
\citenamedef{ZKPfull}{Goldreich et~al.}{1991}
\citenamedef{ZKPnumbertheory}{Goldwasser et~al.}{1989}
\citenamedef{zotero-4595}{??}{????}
\citenamedef{10.1007/978-3-030-29662-9_14}{Wimmer}{2019}
\citenamedef{10.1007/978-3-030-57628-8_4}{Wimmer et~al.}{2020}
\citenamedef{10.1007/978-3-642-39634-2_14}{Gonthier et~al.}{2013}
\citenamedef{10.1145/3290384}{Armstrong et~al.}{2019}
\citenamedef{10.1145/3372885.3373830}{Buzzard et~al.}{2020}
\citenamedef{10.1145/3676641.3715999}{Tan et~al.}{2025}
\citenamedef{10.1145/3839484}{G{\"a}her et~al.}{2026}
\citenamedef{10.1145/512950.512973}{Cousot and Cousot}{1977}
\citenamedef{113766}{Alur et~al.}{1990}
\citenamedef{abdulazizTrustworthyGraphAlgorithms2019a}{Abdulaziz et~al.}{2019}
\citenamedef{ALEXANDRIAproject}{Paulson}{2023}
\citenamedef{AnthropicsFermatslasttheorem2026}{??}{2026}
\citenamedef{bashBashGNUProjectn.d.}{{Bash}}{n.d.}
\citenamedef{behnkeWithhMaxDefinitely}{Behnke et~al.}{2026}
\citenamedef{bengtssonUppaalTool1996}{Bengtsson et~al.}{1996}
\citenamedef{Blanchet2002}{Blanchet et~al.}{2002}
\citenamedef{bogomolovPlanningModelChecking2014}{Bogomolov et~al.}{2014}
\citenamedef{Bouyer2004Forward}{Bouyer}{2004}
\citenamedef{bouyerConcisenessExtensionsTimed2005}{Bouyer and Chevalier}{2005}
\citenamedef{bouyerUpdatableTimedAutomata2004}{Bouyer et~al.}{2004}
\citenamedef{buzzard2025lean}{Buzzard and Taylor}{2024}
\citenamedef{Cashmore2015}{Cashmore et~al.}{2015}
\citenamedef{castelvecchiMathematiciansWelcomeComputerassisted2021a}{Castelvecchi}{2021}
\citenamedef{cimattiExtendingNuXmvTimed2019}{Cimatti et~al.}{2019}
\citenamedef{clarkeHandbookModelChecking2018}{Clarke et~al.}{2018}
\citenamedef{correaGroundingPlanningTasks2023}{Corr{\^e}a et~al.}{2023}
\citenamedef{DBLP:conf/aaai/HoffmannHKSKM20}{Hoffmann et~al.}{2020}
\citenamedef{DBLP:conf/aips/CorreaPHF20}{Corr{\^e}a et~al.}{2020}
\citenamedef{DBLP:conf/aips/CushingWKMT07}{Cushing et~al.}{2007}
\citenamedef{DBLP:conf/aips/ErikssonRH17}{Eriksson et~al.}{2017}
\citenamedef{DBLP:conf/aips/Helmert02}{Helmert}{2002}
\citenamedef{DBLP:conf/ascm/Gonthier07}{Gonthier}{2007}
\citenamedef{DBLP:conf/birthday/BouyerLMOW17}{Bouyer et~al.}{2017}
\citenamedef{DBLP:conf/colog/Paulson88}{Paulson}{1988}
\citenamedef{DBLP:conf/ecai/0002SB23}{Kuroiwa et~al.}{2023}
\citenamedef{DBLP:conf/ecai/MarzalSO08}{Marzal et~al.}{2008}
\citenamedef{DBLP:conf/ecai/ScalaHTR16}{Scala et~al.}{2016}
\citenamedef{DBLP:conf/flops/HaftmannN10}{Haftmann and Nipkow}{2010}
\citenamedef{DBLP:conf/fmcad/GriggioRT18}{Griggio et~al.}{2018}
\citenamedef{DBLP:conf/formats/HendriksBLNV03}{Hendriks et~al.}{2003}
\citenamedef{DBLP:conf/ftrtft/BehrmannBDLPY02}{Behrmann et~al.}{2002}
\citenamedef{DBLP:conf/ijcai/CorreaG24}{Corr{\^e}a and Giacomo}{2024}
\citenamedef{DBLP:conf/ijcai/GiganteS23}{Gigante and Scala}{2023}
\citenamedef{DBLP:conf/ml/Weeks06}{Weeks}{2006}
\citenamedef{DBLP:conf/socs/RiddleBFD15}{Riddle et~al.}{2015}
\citenamedef{DBLP:conf/stoc/BlanchardBDFFGH26}{Blanchard et~al.}{2026}
\citenamedef{DBLP:conf/vmcai/Bradley11}{Bradley}{2011}
\citenamedef{DBLP:journals/ai/GereviniHLSD09/corrected}{Gerevini et~al.}{2009}
\citenamedef{DBLP:journals/aicom/FuentetajaR15}{Fuentetaja and {de la
  Rosa}}{2016}
\citenamedef{DBLP:journals/corr/abs-2503-02975}{Kappelmann et~al.}{2025}
\citenamedef{DBLP:journals/fmsd/GriggioRT21}{Griggio et~al.}{2021}
\citenamedef{DBLP:journals/jair/ColesCFL13}{Coles et~al.}{2013}
\citenamedef{DBLP:journals/jair/LongF03}{Long and Fox}{2003}
\citenamedef{DBLP:journals/jsyml/Church40}{Church}{1940}
\citenamedef{DBLP:reference/mc/ClarkeHV18}{Clarke et~al.}{2018}
\citenamedef{equationalTheories2025}{Bolan et~al.}{2025}
\citenamedef{erikssonProofSystemUnsolvable2018}{Eriksson et~al.}{????}
\citenamedef{foxModellingMixedDiscreteContinuous2006}{Fox and Long}{2006}
\citenamedef{gereviniPlanConstraintsPreferences2005a}{Gerevini and Long}{2005}
\citenamedef{GewirthPGCProof-AFP}{Fuenmayor and Benzm{\"u}ller}{2018}
\citenamedef{giunchigliaPlanningModelChecking2000}{Giunchiglia and
  Traverso}{2000}
\citenamedef{godefroidPartialOrderMethodsVerification1996}{Godefroid
  et~al.}{1996}
\citenamedef{haftmannDataRefinementIsabelle2013}{Haftmann et~al.}{2013}
\citenamedef{halbwachsVerificationRealTimeSystems1997}{Halbwachs et~al.}{1997}
\citenamedef{halesFormalProofKepler}{Hales et~al.}{2017}
\citenamedef{hallFasterOptimalPlanning2013}{Hall et~al.}{2013}
\citenamedef{HERBRETEAU201667}{Herbreteau et~al.}{2016}
\citenamedef{hoffmannDistanceWhoCares2014}{Hoffmann et~al.}{2014}
\citenamedef{IEEE2024}{{IEEE}}{2024}
\citenamedef{kambhampatiComparativeAnalysisPartial1995}{Kambhampati}{1995}
\citenamedef{khatibMappingTemporalPlanning2001}{Khatib et~al.}{2001}
\citenamedef{KR2020-80}{Benzm{\"u}ller}{2020}
\citenamedef{lammichGeneratingVerifiedLLVM2019a}{Lammich}{2019}
\citenamedef{largouetTemporalPlanningExtended2016}{Largou{\"e}t et~al.}{2016}
\citenamedef{larsenUppaalNutshell1997}{Larsen et~al.}{1997}
\citenamedef{liModelCheckingApproach2014}{Li et~al.}{2014}
\citenamedef{liPlanningModelChecking2012}{Li et~al.}{2012}
\citenamedef{mcconnellCertifyingAlgorithms2011}{McConnell et~al.}{2011}
\citenamedef{MDP-Algorithms-AFP}{Sch{\"a}ffeler and Abdulaziz}{2021}
\citenamedef{MDP-Rewards-AFP}{Sch{\"a}ffeler and Abdulaziz}{2021}
\citenamedef{micheliTemporalPlanningTemporal2019}{Micheli and Scala}{2019}
\citenamedef{mine:hal-01657536}{Min{\'e}}{2017}
\citenamedef{muiseUnplannabilityIPCTrack2015}{Muise and Lipovetzky}{2015}
\citenamedef{mukherjeeReachabilityTimedAutomata2021}{Mukherjee}{2021}
\citenamedef{panjkovicUnsolvabilityArtifact2022}{Panjkovic et~al.}{2022}
\citenamedef{paulsonFoundationGenericTheorem1989}{Paulson}{1989}
\citenamedef{paulsonIsabelleNext7001990}{Paulson}{1990}
\citenamedef{paulsonNaturalDeductionHigherorder1986}{Paulson}{1986}
\citenamedef{peledTenYearsPartial1998}{Peled}{1998}
\citenamedef{Roeger2017}{R{\"o}ger}{2017}
\citenamedef{Scholze03042022}{Scholze}{2022}
\citenamedef{seippFastDownwardAidos2020}{Seipp et~al.}{2020}
\citenamedef{steinProvablyCompleteGeneralized2026}{Stein et~al.}{2026}
\citenamedef{stohrAutomatedCompositionTimed2012}{St{\"o}hr and Glesner}{2012}
\citenamedef{sudaPropertyDirectedReachability2014}{Suda}{2014}
\citenamedef{tanFormalisingCXLCache2025}{Tan et~al.}{2025}
\citenamedef{taoLeanCompanionAnalysis2025}{Tao}{2025}
\citenamedef{taoTerenceYoutube}{Tao}{????}
\citenamedef{TicktacprojectTchecker2023}{Herbreteau et~al.}{2023}
\citenamedef{timeGNUTimeGNUn.d.}{{Time}}{n.d.}
\citenamedef{Timeout1LinuxManual}{??}{????}
\citenamedef{Top1LinuxManualn.d.}{??}{n.d.}
\citenamedef{Ulimit3LinuxManual}{??}{????}
\citenamedef{vollathVerifiedGroundingPDDL2025}{Vollath}{2025}
\citenamedef{wehrlePartialOrderReduction2012}{Wehrle and Helmert}{2012}
\citenamedef{wenzelIsabelleIsarReference2001}{Wenzel}{2001}
\citenamedef{wenzelIsarGenericInterpretative1999}{Wenzel}{1999}
\citenamedef{wimmerTimedAutomataAFP2016}{Wimmer}{2016}
\citenamedef{zamaniSymbolicDynamicProgramming2012}{Zamani et~al.}{2012}
\citenamedef{zotero-item-3032}{??}{????}

\usepackage{hyperref}
\usepackage[capitalise]{cleveref}
\crefname{listing}{Listing}{Listings}
\Crefname{listing}{Listing}{Listings}
\crefname{example}{Example}{Examples}
\Crefname{example}{Example}{Examples}
\crefname{figure}{Figure}{Figures}

\newif\ifanonymous
\anonymousfalse

\title{Formally Verified Certification of Unsolvability of Temporal Planning Problems}
\titlerunning{Verified Certification of Unsolvability of Temporal Planning}
\ifanonymous
  \author{Anonymous Author(s)}
  \authorrunning{Anonymous Author(s)}
  \institute{Anonymous Institution(s)}
\else
  \author{David Wang\inst{1} \and Mohammad Abdulaziz\inst{1}}
  \authorrunning{D. Wang and M. Abdulaziz}
  \institute{King's College London, London, UK\\
    \email{\{david.wang,mohammad.abdulaziz\}@kcl.ac.uk}}
\fi

\begin{document}

\maketitle

\begin{abstract}
We present an approach to unsolvability certification of temporal planning.
Our approach is based on encoding the planning problem into a network of timed automata, and then using an efficient model checker on the network followed by a certificate checker to certify the output of the model checker.
Our approach prioritises trustworthiness of the certification: we formally verify our implementation of the encoding to timed automata using the theorem prover Isabelle/HOL and we use an existing certificate checker (also formally verified in Isabelle/HOL) to certify the model checking result.
\end{abstract}

\providecommand{\insts}{}
\renewcommand{\insts}{\ensuremath{\Delta}}
\providecommand{\inst}{\ensuremath{\tvsal}}
\newcommand{\act}{\ensuremath{\pi}}
\newcommand{\asarrow}[1]{\vec{#1}}

\makeatletter
\newcommand{\oset}[2]{%
  {\mathop{#2}\limits^{\vbox to -1\ex@{\kern-\tw@\ex@
   \hbox{\scriptsize #1}\vss}}}}
\makeatother

\renewcommand{\vec}[1]{\ensuremath{\oset{$\rightarrow$}{\ensuremath{#1}}}}
\newcommand{\as}{\ensuremath{\vec{{\act}}}}
\newcommand{\asb}{\ensuremath{\vec{{\act_2}}}}

\newcommand{\etc}{\textit{etc.}}
\newcommand{\versus}{\textit{vs.}}
\newcommand{\ie}{i.e.}
\newcommand{\Ie}{I.e.}
\newcommand{\eg}{e.g.}
\newcommand{\michael}[1]{\textcolor{blue}{M: #1}}
\newcommand{\abziz}[1]{\textcolor{brown}{#1}}
\newcommand{\sublist}[2]{ \ensuremath{#1} \preceq\!\!\!\raisebox{.4mm}{\ensuremath{\cdot}}\; \ensuremath{#2}}
\newcommand{\subscriptsublist}[2]{\ensuremath{#1}\preceq\!\raisebox{.05mm}{\ensuremath{\cdot}}\ensuremath{#2}}
\newcommand{\PLS}{\Pi^\preceq\!\raisebox{1mm}{\ensuremath{\cdot}}}
\newcommand{\PLScharles}{\Pi^d}
\newcommand{\execname}{\mathsf{ex}}
\newcommand{\IndHyp}{\mathsf{IH}}
\newcommand{\exec}[2]{#2(#1)}

\newcommand{\ancestorssymbol}{\textsf{\upshape ancestors}}
\newcommand{\ancestors}{\ancestorssymbol}
\newcommand{\satpreas}[2]{\ensuremath{sat_precond_as(s, \as)}}
\newcommand{\proj}[2]{\ensuremath{#1{\downharpoonright}_{#2}}}
\newcommand{\dep}[3]{\ensuremath{#2 {\rightarrow} #3}}
\newcommand{\deptc}[3]{\ensuremath{#2 {\rightarrow^+} #3}}
\newcommand{\negdep}[3]{\ensuremath{#2 \not\rightarrow #3}}
\newcommand{\leavessymbol}{\textsf{\upshape leaves}}
\newcommand{\leaves}{\leavessymbol}

\newcommand{\childrensymbol}{\textsf{\upshape children}}
\newcommand{\children}[2]{\mathcal{\childrensymbol}_{#2}(#1)}
\newcommand{\succsymbol}{\textsf{\upshape succ}}
\newcommand{\succstates}[2]{\succsymbol(#1, #2)}
\newcommand{\concat}{\#}
\newcommand{\RG}{\cite{Rintanen:Gretton:2013}\ }
\newcommand{\KG}{Kovacs' grammar}
\newcommand{\cupdot}{\charfusion[\mathbin]{\cup}{\cdot}}
\newcommand{\cuparrow}{\charfusion[\mathbin]{\cup}{{\raisebox{.5ex} {\smathcalebox{.4}{\ensuremath{\leftarrow}}}}}}
\newcommand{\bigcuparrow}{\charfusion[\mathop]{\bigcup}{\leftarrow}}
\newcommand{\finiteunion}{\cuparrow}
\newcommand{\finitemap}{\ensuremath{\sqsubseteq}}
\newcommand{\dgraph}{dependency graph}
\newcommand{\domain}[1]{{\sc #1}}
\newcommand{\solver}[1]{{\sc #1}}
\providecommand{\problem}[1]{\domain{#1}}
\renewcommand{\v}{\ensuremath{\mathit{v}}}
\providecommand{\vs}[1]{\domain{#1}}
\renewcommand{\vs}{\ensuremath{\mathit{vs}}}
\newcommand{\VS}{\ensuremath{\mathit{VS}}}
\newcommand{\Aut}{\ensuremath{\mathit{Aut}}}
\newcommand{\Inst}[2]{\ensuremath{\mathit{#2 \rightarrow_{#1} #1}}}
\newcommand{\Image}{\ensuremath{\mathit{Im}}}
\newcommand{\Img}[2]{\protect{#1 \llparenthesis #2 \rrparenthesis}}
\newcommand{\SND}{\ensuremath{\mathit{\pi_2}}}
\newcommand{\FST}{\ensuremath{\mathit{\pi_1}}}
\newcommand{\tvsal}{{\pitchfork}}
\newcommand{\nauty}{CGIP}%

\newcommand{\pwinter}{\ensuremath{\mathit{\bigcap_{pw}}}}

\newcommand{\dom}{\ensuremath{\mathit{\mathcal{D}}}}
\newcommand{\codom}{\ensuremath{\mathcal{R}}}

\newcommand{\map}{\ensuremath{\mathit{map}}}
\newcommand{\BIJEC}{\ensuremath{\mathit{bij}}}
\newcommand{\INJ}{\ensuremath{\mathit{inj}}}
\newcommand{\funion}{\ensuremath{\overset{\leftarrow}{\cup}}}

\newcommand{\ifnew}{\mbox{\upshape \textsf{if}}}
\newcommand{\thennew}{\mbox{\upshape \textsf{then}}}
\newcommand{\elsenew}{\mbox{\upshape \textsf{else}}}
\newcommand{\choice}{\ensuremath{\epsilon}}
\newcommand{\arbchoice}{\mbox{\upshape \textsf{arb}}}
\newcommand{\acycchoice}{\mbox{\upshape \textsf{ac}}}
\newcommand{\cycchoice}{\mbox{\upshape \textsf{cyc}}}
\newcommand{\filter}{\ensuremath{\mathit{FIL}}}
\newcommand{\probset}{\ensuremath{\boldsymbol \Pi}}
\newcommand{\probleq}{\ensuremath{\leq_\Pi}}
\newcommand{\CommVar}{\ensuremath{\bigcap_\v} }
\newcommand{\quotfun}{\ensuremath{ \mathcal{Q}}}

\newcommand{\apre}{\mbox{\upshape \textsf{pre}}}
\newcommand{\aeff}{\mbox{\upshape \textsf{eff}}}
\newcommand{\problist}{\ensuremath \probset}
\newcommand{\cat}{{\frown}}
\newcommand{\probproj}[2]{{#1}{\downharpoonright}^{#2}}
\newcommand{\preced}{\mathbin{\rotatebox[origin=c]{180}{\ensuremath{\rhd}}}}
\newcommand{\perm}{\ensuremath{\sigma}}
\newcommand{\invariant}[2]{\ensuremath{\mathit{inv({#1},{#2})}}}
\newcommand{\invstates}[1]{\ensuremath{\mathit{inv({#1})}}}
\newcommand{\probss}[1]{{\mathcal S}(#1)}
\newcommand{\parChildRel}[3]{\ensuremath{\negdep{#1}{#2}{#3}}}
\newcommand{\asessymbol}{\ensuremath{\mathbb{A}}}
\newcommand{\ases}[1]{{#1}^*}
\newcommand{\uniStates}{\ensuremath{\mathbb{U}}}
\newcommand{\recurrenceDiam}{\ensuremath{\mathit{rd}}}
\newcommand{\recurrenceAcycDiamfun}{\ensuremath{\mathit{{\mathfrak A}}}}
\newcommand{\recurrenceDiamfun}{\ensuremath{\mathit{\mathfrak R}}}
\newcommand{\traversalDiam}{\ensuremath{\mathit{td}}}
\newcommand{\traversalDiamfun}{\ensuremath{\mathit{\mathfrak T}}}
\newcommand{\isPrefix}[2]{\ensuremath{#1 \preceq #2}}
\providecommand{\path}{\ensuremath{\gamma}}
\newcommand{\aspath}{\ensuremath{\vec{\path}}}
\renewcommand{\path}{\ensuremath{\gamma}}
\newcommand{\n}{\textsf{\upshape n}}
\providecommand{\graph}{}
\providecommand{\cal}{}
\renewcommand{\cal}{}
\renewcommand{\graph}{{\cal G}}
\newcommand{\undirgraph}{{\cal G}}
\newcommand{\sset}{\ensuremath{\mbox{\upshape \textsf{ss}}}}
\newcommand{\slist}{\ensuremath{\vec{\mbox{\upshape \textsf{ss}}}}}
\newcommand{\sll}{\ensuremath{\vec{\state}}}
\newcommand{\listset}{\mbox{\upshape \textsf{set}}}
\newcommand{\asset}{\ensuremath{\mathit{K}}}
\newcommand{\aslist}{\ensuremath{\mathit{\overset{\rightarrow}{\gamma}}}}
\newcommand{\head}{\mbox{\upshape \textsf{hd}}}
\renewcommand{\max}{\textsf{\upshape max}}
\newcommand{\argmax}{\textsf{\upshape argmax}}
\renewcommand{\min}{\textsf{\upshape min}}
\newcommand{\bool}{\mbox{\upshape \textsf{bool}}}
\newcommand{\last}{\mbox{\upshape \textsf{last}}}
\newcommand{\front}{\mbox{\upshape \textsf{front}}}
\newcommand{\rot}{\mbox{\upshape \textsf{rot}}}
\newcommand{\stuff}{\mbox{\upshape \textsf{intlv}}}
\newcommand{\tail}{\mbox{\upshape \textsf{tail}}}
\newcommand{\ngrtoas}{\ensuremath{\mathit{\as_{\graph_\mathbb{N}}}}}
\newcommand{\vsfun}{\mbox{\upshape \textsf{vs}}}
\newcommand{\inits}{\mbox{\upshape \textsf{init}}}
\newcommand{\satprecondas}{\mbox{\upshape \textsf{sat-pre}}}
\newcommand{\remcondlessact}{\mbox{\upshape \textsf{rem-cless}}}
\providecommand{\state}{}
\renewcommand{\state}{x}
\newcommand{\statea}{\ensuremath{x_1}}
\newcommand{\stateb}{\ensuremath{x_2}}
\newcommand{\statec}{\ensuremath{x_3}}
\newcommand{\fals}{\mbox{\upshape \textsf{F}}}
\newcommand{\indices}{\ensuremath{V}}
\newcommand{\edges}{\ensuremath{E}}
\newcommand{\vertices}{\ensuremath{V}}
\newcommand{\listtype}{\mbox{\upshape \textsf{list}}}
\newcommand{\settype}{\mbox{\upshape \textsf{set}}}
\newcommand{\acttype}{\mbox{\upshape \textsf{action}}}
\newcommand{\graphtype}{\mbox{\upshape \textsf{graph}}}
\newcommand{\projfun}[2]{\ensuremath{\Delta_{#1}^{#2}}}
\newcommand{\snapfun}[2]{\ensuremath{\Sigma_{#1}^{#2}}}
\newcommand{\RDfun}[1]{\ensuremath{{\mathcal R}_{#1}}}
\newcommand{\elldbound}[1]{\ensuremath{{\mathcal LS}_{#1}}}
\newcommand{\distinct}{\textsf{\upshape distinct}}
\newcommand{\ddistinct}{\mbox{\upshape \textsf{ddistinct}}}
\newcommand{\simple}{\mbox{\upshape \textsf{simple}}}

\newcommand{\reachable}[3]{\ensuremath{{#1}\rightsquigarrow{#3}}}

\newcommand{\Omit}[1]{}

\newcommand{\charles}[1]{\textcolor{red}{#1}}
\newcommand{\mohammad}[1]{Mohammad: \textcolor{green}{#1}}
\newcommand{\lukas}[1]{Lukas: \textcolor{red}{#1}}

\newcommand{\negreachable}[3]{\ensuremath{{#2}\not\rightsquigarrow{#3}}}
\newcommand{\wdiam}[2]{{#1}^{#2}}
\newcommand{\dsnapshot}[2]{\Delta_{#1}}
\newcommand{\ellsnapshot}[2]{{\mathcal L}_{#1}}

\newcommand{\snapshotsymbol}{|\kern-.7ex\raise.08ex\hbox{\scalebox{0.7}{$\bullet$}}}
\newcommand{\snapshot}[2]{\ensuremath{\mathrel{#1\snapshotsymbol_{#2}}}}
\newcommand{\vstype}{\texttt{\upshape VS}}
\newcommand{\vtype}{{\scriptsize \ensuremath{\dom(\delta)}}}
\newcommand{\Balgo}{{\mbox{\textsc{Hyb}}}}
\newcommand{\ssgraph}[1]{\graph_\ss}
\newcommand{\agree}{\textsf{\upshape agree}}
\newcommand{\ck}{\ensuremath{\texttt{ck}}}
\newcommand{\lk}{\ensuremath{\texttt{lk}}}
\newcommand{\gr}{\ensuremath{\texttt{gr}}}
\newcommand{\gk}{\ensuremath{\texttt{gk}}}
\newcommand{\CK}{\ensuremath{\texttt{CK}}}
\newcommand{\LK}{\ensuremath{\texttt{LK}}}
\newcommand{\GR}{\ensuremath{\texttt{GR}}}
\newcommand{\GK}{\ensuremath{\texttt{GK}}}
\newcommand{\safe}{\ensuremath{\texttt{s}}}

\newcommand{\derivname}{\ensuremath{\partial}}
\newcommand{\deriv}[3]{\ensuremath{\derivname(#1,#2,#3)}}
\newcommand{\derivabbrev}[3]{\ensuremath{{\partial(#1,#2)}}}
\newcommand{\subsetoracle}{\ensuremath{ \Omega}}
\newcommand{\Aalgo}{{\mbox{\textsc{Pur}}}}
\newcommand{\Sname}{\textsf{\upshape S}}
\newcommand{\Sbrace}[1]{\Sname\langle#1\rangle}
\newcommand{\SalgoName}{\Sname_{\textsf{\upshape max}}}
\newcommand{\Salgo}[1]{\SalgoName\langle#1\rangle}

\newcommand{\WLPname}{{\mbox{\textsc{wlp}}}}
\newcommand{\WLPbrace}[1]{\WLPname\langle#1\rangle}
\newcommand{\WLPalgoName}{\WLPname_{\textsf{\upshape max}}}
\newcommand{\WLP}[1]{\WLPalgoName\langle#1\rangle}

\newcommand{\Nname}{\ensuremath{\textsf{\upshape N}}}
\newcommand{\Nbrace}[1]{\Nname\langle#1\rangle}
\newcommand{\NalgoName}{\Nname{_{\textsf{\upshape sum}}}}
\newcommand{\Nalgobrace}[1]{\NalgoName\langle#1\rangle}

\newcommand{\acycNname}{\widehat{\textsf{\upshape N}}}
\newcommand{\acycNbrace}[1]{\acycNname\langle#1\rangle}
\newcommand{\acycNalgoName}{\acycNname{_{\textsf{\upshape sum}}}}
\newcommand{\acycNalgobrace}[1]{\acycNalgoName\langle#1\rangle}

\newcommand{\Mname}{\ensuremath{\textsf{\upshape M}}}
\newcommand{\Mbrace}[1]{\Mname\langle#1\rangle}
\newcommand{\MalgoName}{\Mname{_{\textsf{\upshape sum}}}}
\newcommand{\Malgobrace}[1]{\MalgoName\langle#1\rangle}
\newcommand{\cardinality}[1]{{\ensuremath{|#1|}}}
\newcommand{\length}[1]{\cardinality{#1}}
\newcommand{\basecasefun}{\ensuremath{b}}
\newcommand{\Basecasefun}{\ensuremath{\mathcal B}}

\newcommand{\vertexgen}{\ensuremath{u}}
\newcommand{\vertexa}{{\ensuremath{\vertexgen_1}}}
\newcommand{\vertexb}{{\ensuremath{\vertexgen_2}}}
\newcommand{\vertexc}{{\ensuremath{\vertexgen_3}}}
\newcommand{\vertexd}{{\ensuremath{\vertexgen_4}}}
\newcommand{\vertexsetgen}{\ensuremath{\mathit{us}}}
\newcommand{\vertexseta}{\vertexsetgen_1}
\newcommand{\vertexsetb}{\vertexsetgen_2}
\newcommand{\labelsymbol}{\ensuremath{l}}
\newcommand{\labelfun}{\ensuremath{\mathcal{L}}}
\newcommand{\DAG}{\ensuremath{A}}
\newcommand{\NalgoNameN}{{\ensuremath{\NalgoName_{\mathbb{N}}}}}
\newcommand{\NnameN}{\ensuremath{\Nname_\mathbb{N}}}
\newcommand{\replaceprojsinglename}{\raisebox{-0.3mm} {\scalebox{0.7}{\textpmhg{H}}}}
\newcommand{\replaceprojsingle}[3] {{ #2} \underset {#1} {\raisebox{-0.3mm} {\scalebox{0.7}{\textpmhg{H}}}}  #3}
\newcommand{\HOLreplaceprojsingle}[1]{\underset {#1} {\raisebox{-0.3mm} {\scalebox{0.7}{\textpmhg{H}}}}}

\newcommand{\lotus}{{\scalebox{0.6}{\includegraphics{lotus.pdf}}}}
\newcommand{\invlotus}{\mathbin{\rotatebox[origin=c]{180}{$\lotus$}}}
\newcommand{\clique}{\ensuremath{K}}
\newcommand{\partition}{\ensuremath{\vs_{1..n}}}
\newcommand{\partitiontype}{\ensuremath{\vstype_{1..n}}}
\newcommand{\vtxpartition}{\ensuremath{P}}

\newcommand{\traversalDiamAlgo}{{\mbox{\textsc{TravDiam}}}}
\newcommand{\prefix}{\textsf{\upshape pfx}}
\newcommand{\powerset}{\mathbb{P}}
\newcommand{\postfix}{\textsf{\upshape sfx}}
\newcommand{\dfunproj}{\ensuremath{{\mathfrak D}}}
\newcommand{\dfunsnap}{\ensuremath{{\textgoth D}}}
\newcommand{\ellfunproj}{\ensuremath{\mathfrak L}}
\newcommand{\ellfunsnap}{\ensuremath{\textgoth L}}
\newcommand{\cycle}{\ensuremath{C}}
\newcommand{\petal}{\ensuremath{\eta}}
\renewcommand{\prod}{\ensuremath{{{{{\mathlarger{\mathlarger {{\mathlarger {\Pi}}}}}}}}}}
\newcommand{\sccset}{{\ensuremath{SCC}}}
\newcommand{\scc}{{\ensuremath{scc}}}
\newcommand{\negate}[1]{\overline{#1}}
\newcommand{\setofsets}{\ensuremath{S}}
\newcommand{\group}{\ensuremath{\cal \Gamma}}
\newcommand{\neededvars}{{\cal N}}
\newcommand{\sspace}{\mbox{\upshape \textsf{sspc}}}
\newcommand{\tip}{\ensuremath{t}}
\newcommand{\vara}{\ensuremath{\v_1}}
\newcommand{\varb}{\ensuremath{\v_2}}
\newcommand{\varc}{\ensuremath{\v_3}}
\newcommand{\vard}{\ensuremath{\v_4}}
\newcommand{\vare}{\ensuremath{\v_5}}
\newcommand{\varf}{\ensuremath{\v_6}}
\newcommand{\varg}{\ensuremath{\v_7}}
\newcommand{\varh}{\ensuremath{\v_8}}
\newcommand{\vari}{\ensuremath{\v_9}}
\newcommand{\acta}{{\ensuremath{\act^1}}}
\newcommand{\actb}{{\ensuremath{\act^2}}}
\newcommand{\actc}{\ensuremath{\act^3}}
\newcommand{\actd}{\ensuremath{\act^4}}
\newcommand{\acte}{\ensuremath{\act^5}}
\newcommand{\actf}{\ensuremath{\act^6}}
\newcommand{\actg}{\ensuremath{\act^7}}
\newcommand{\acth}{\ensuremath{\act^8}}
\newcommand{\acti}{\ensuremath{\act^9}}

\newcommand{\planningproblem}{\ensuremath{\Pi}}

\tikzset{dots/.style args={#1per #2}{line cap=round,dash pattern=on 0 off #2/#1}}
\providecommand{\moham}[1]{\fbox{{\bf \@Mohammad: }#1}}
\newcommand{\TDbound}{{\mbox{\textsc{Arb}}}}
\newcommand{\expbound}{{\mbox{\textsc{Exp}}}}
\newcommand{\sasdom}{\expbound}
\newcommand{\cardfun}{\ensuremath{\mathbb{C}}}
\newcommand{\AGNa}{AGN1}
\newcommand{\AGNb}{AGN2}
\newcommand{\reset}{{\ensuremath{reset}}}
\newcommand{\cost}{{\ensuremath{\mathcal{C}}}}
\newcommand{\goal}{{\ensuremath{\mathcal{G}}}}
\newcommand{\init}{{\ensuremath{\mathcal{I}}}}
\newcommand{\completenessthreshold}{{\ensuremath{\mathcal{CT}}}}
\newcommand{\subsetDiam}{\mathscr{S}}

\newcommand{\planningProb}{\planningproblem}
\newcommand{\initState}{\init}
\newcommand{\goalState}{\goal}
\newcommand{\planningState}{M}
\newcommand{\snapAction}{\ensuremath{\alpha}}
\newcommand{\snapActionPre}{\snapAction_\textit{pre}}
\newcommand{\snapActionAdd}{\snapAction_\textit{add}}
\newcommand{\snapActionDel}{\snapAction_\textit{del}}
\newcommand{\snapActiona}{\acta}
\newcommand{\snapActionPrea}{{\snapActiona_\textit{pre}}}
\newcommand{\snapActionAdda}{{\snapActiona_\textit{add}}}
\newcommand{\snapActionDela}{{\snapActiona_\textit{del}}}
\newcommand{\snapActionb}{\actb}
\newcommand{\snapActionPreb}{{\snapActionb_\textit{pre}}}
\newcommand{\snapActionAddb}{{\snapActionb_\textit{add}}}
\newcommand{\snapActionDelb}{{\snapActionb_\textit{del}}}
\newcommand{\durAction}{\act}
\newcommand{\durActionStart}{\durAction_\textit{start}}
\newcommand{\durActionEnd}{\durAction_\textit{end}}
\newcommand{\durActionInv}{\durAction_\textit{inv}}
\newcommand{\durActionInvParam}[1]{\operatorname{pre}_{inv}(#1)}
\newcommand{\durationLower}{\act_l}
\newcommand{\durationUpper}{\act_u}
\newcommand{\plan}{\ensuremath{\as}}
\newcommand{\happening}{h}
\newcommand{\happeningActs}{A}
\newcommand{\happeningTime}{r}
\newcommand{\happTimePoint}{t}
\newcommand{\happTimePoints}{\textit{htps}}
\newcommand{\recSetActs}{B}
\newcommand{\recSetInvs}{I}
\newcommand{\predAtm}{\texttt{pred}}
\newcommand{\floorType}{floor}
\newcommand{\elevatorType}{elevator}
\newcommand{\passengerType}{passenger}
\newcommand{\elevatorAtPred}{el-at}
\newcommand{\passengerAtPred}{p-at}
\newcommand{\inElevatorPred}{in-el}
\newcommand{\elevatorDoorOpenPred}{el-op}
\newcommand{\elevatorAtPredInline}[2]{\textit{\small(\elevatorAtPred~#1~#2)}}
\newcommand{\passengerAtPredInline}[2]{\textit{\small(\passengerAtPred~#1~#2)}}
\newcommand{\inElevatorPredInline}[2]{\textit{\small(\inElevatorPred~#1~#2)}}
\newcommand{\elevatorDoorOpenPredInline}[1]{\textit{\small(\elevatorDoorOpenPred~#1)}}
\newcommand{\elevatorDurationFunc}{el-dur}
\newcommand{\moveElevatorAct}{mv}
\newcommand{\openElevatorDoorAct}{op}
\newcommand{\closeElevatorDoorAct}{cl}
\newcommand{\enterElevatorAct}{en}
\newcommand{\exitElevatorAct}{ex}

\newcommand{\lstinlinemacro}[1]{\small\textit{#1}}
\newcommand{\insertActionAlgo}{\operatorname{insert-action}}
\newcommand{\simplifyActionAlgo}{\operatorname{simplify-action}}
\newcommand{\simplifyPlanAlgo}{\operatorname{simplify-plan}}
\newcommand{\validHapSeqAlgo}{\operatorname{valid-hap-seq}}
\newcommand{\checkPlanAlgo}{\operatorname{check-plan}}
\newcommand{\sortAlgo}{\operatorname{sort}}
\newcommand{\happSeq}{H}
\newcommand{\false}{\textit{False}}
\newcommand{\true}{\textit{True}}

\makeatletter
\newcommand*{\my@test@it}{it}
\newcommand*{\mathulifit}[1]{%
  \ifx\f@shape\my@test@it
    \underline{\mathit{#1}}
  \else
    \mathit{#1}
  \fi
}
\makeatother

\newcommand{\varul}[1]{\ensuremath{\mathulifit{#1}}}
\newcommand{\varnoul}[1]{\ensuremath{\mathit{#1}}}

\newcommand{\propositions}{\varnoul{P}}
\newcommand{\actions}{\varnoul{A}}
\newcommand{\initstate}{\initState}
\newcommand{\goalstate}{\goalState}

\newcommand{\propval}{\ensuremath{\mathcal{M}}}
\newcommand{\propstate}{\varnoul{M}}

\newcommand{\tss}[1]{H(#1)}

\newcommand{\durAct}[1]{\ensuremath{\durAction_{#1}}}
\newcommand{\snapAct}[1]{\ensuremath{\snapAction_{#1}}}

\newcommand{\plana}{\varnoul{\alpha}}

\newcommand{\real}{\ensuremath{\mathbb{R}}}

\newcommand{\clocks}{\ensuremath{\mathcal{X}}}
\newcommand{\vars}{\ensuremath{\mathcal{V}}}
\newcommand{\locs}{\ensuremath{\mathcal{L}}}
\newcommand{\auto}{\ensuremath{\mathfrak{T}}}  %

\newcommand{\clk}[1]{\ensuremath{x_{#1}}}

\newcommand{\Loc}{L}
\newcommand{\Var}{v}
\newcommand{\Clk}{c}

\newcommand{\Lvct}[3]{{\Loc}{#1}, {\Var}{#2}, {\Clk}{#3}}
\newcommand{\bLvct}[3]{\langle {\Loc}{#1}, {\Var}{#2}, {\Clk}{#3} \rangle}
\newcommand{\Lvc}[1]{\Lvct{#1}{#1}{#1}}
\newcommand{\bLvc}[1]{\langle \Lvct{#1}{#1}{#1} \rangle}

\newcommand{\cfgs}{\varnoul{cfgs}}

\newcommand{\pluseq}{\mathrel{{+}\mkern1mu{=}}}
\newcommand{\minuseq}{\mathrel{{-}\mkern1mu{=}}}
\newcommand{\constraints}[1]{\mathcal{C}(#1)}

\newcommand{\cEntails}[2]{#1 \vDash_c #2}
\newcommand{\bEntails}[2]{#1 \vDash_b #2}
\newcommand{\pEntails}[2]{#1 \vDash_b #2}

\newcommand{\bexps}[1]{\mathcal{B}(#1)}
\newcommand{\exps}[1]{\mathcal{E}(#1)}

\newcommand{\poset}[1]{2^{#1}}

\newcommand{\dcsatf}{\varnoul{dur\dsh c\dsh sat}}
\newcommand{\nsof}{\varnoul{no\dsh self\dsh overlap}}
\newcommand{\valf}{\varnoul{valid}}
\newcommand{\encbef}{\varnoul{encodes\dsh before}}
\newcommand{\encaft}{\varnoul{enc\dsh after}}
\newcommand{\encat}{\varnoul{encodes\dsh at}}
\newcommand{\enclbef}{\varnoul{encodes\dsh locs\dsh before}}
\newcommand{\encpbef}{\varnoul{encodes\dsh props\dsh before}}
\newcommand{\encplbef}{\varnoul{encodes\dsh prop\dsh locks\dsh before}}
\newcommand{\encaabef}{\varnoul{encodes\dsh all\dsh active\dsh before}}
\newcommand{\enccbef}{\varnoul{encodes\dsh clocks\dsh before}}

\newcommand{\runbef}{\varnoul{running\dsh before}}
\newcommand{\runaft}{\varnoul{running\dsh after}}
\newcommand{\runat}{\varnoul{runs\dsh at}}
\newcommand{\timesincebef}{\varnoul{time\dsh since\dsh before}}
\newcommand{\timesinceaft}{\varnoul{time\dsh since\dsh after}}
\newcommand{\timesinceat}{\varnoul{ran\dsh for\dsh at}}

\newcommand{\escond}{\varnoul{end\dsh start\dsh cond}}
\newcommand{\esinvs}{\varnoul{end\dsh start\dsh invs}}
\newcommand{\run}{\varnoul{run}}

\newcommand{\eval}{\varnoul{eval}}
\newcommand{\bval}{\varnoul{bval}}
\newcommand{\ccval}{\varnoul{ccval}}

\newcommand{\initedge}{\varnoul{e_{1M}}}
\newcommand{\goaledge}{\varnoul{e_{2M}}}
\newcommand{\loopedge}{\varnoul{e_{3M}}}  %

\newcommand{\initloc}{\varnoul{init_M}}
\newcommand{\planningloc}{\varnoul{plan_M}}
\newcommand{\goalloc}{\varnoul{goal_M}}

\newcommand{\initlocs}{\varnoul{ilc}}
\newcommand{\planninglocs}{\varnoul{plc}}
\newcommand{\goallocs}{\varnoul{glc}}

\newcommand{\sse}{\varnoul{se}}
\newcommand{\see}{\varnoul{se'}}
\newcommand{\inste}{\varnoul{ie}}
\newcommand{\ese}{\varnoul{ee}}
\newcommand{\eee}{\varnoul{ee'}}

\newcommand{\slc}{\varnoul{starting}}
\newcommand{\elc}{\varnoul{ending}}
\newcommand{\olc}{\varnoul{inactive}}
\newcommand{\rlc}{\varnoul{running}}

\newcommand{\slcs}{\varnoul{slc}}
\newcommand{\elcs}{\varnoul{elc}}
\newcommand{\olcs}{\varnoul{olc}}
\newcommand{\rlcs}{\varnoul{rlc}}

\newcommand{\pl}{\varnoul{ps}}
\newcommand{\vp}[1]{\varnoul{vp}_{#1}}
\newcommand{\lp}[1]{\varnoul{lp}_{#1}}
\newcommand{\actsactive}{\varnoul{aa}}

\newcommand{\numacts}{\ensuremath{N}}

\newcommand{\mvact}{\varnoul{move}}
\newcommand{\opact}{\varnoul{open\dsh door}}
\newcommand{\clact}{\varnoul{close\dsh door}}
\newcommand{\pdact}{\varnoul{put\dsh down}}
\newcommand{\puact}{\varnoul{pick\dsh up}}

\newcommand{\mvactsrt}{\varnoul{mv}}
\newcommand{\opactsrt}{\varnoul{od}}
\newcommand{\clactsrt}{\varnoul{cd}}
\newcommand{\pdactsrt}{\varnoul{pd}}
\newcommand{\puactsrt}{\varnoul{pu}}

\newcommand{\adds}{\varnoul{adds}}
\newcommand{\dels}{\varnoul{dels}}
\newcommand{\pres}{\varnoul{pres}}
\newcommand{\ova}{\varnoul{over\_all}}

\newcommand{\robv}{\varnoul{rb}}
\newcommand{\roomv}{\varnoul{rm}}
\newcommand{\doorv}{\varnoul{d}}
\newcommand{\blockv}{\varnoul{b}}

\newcommand{\openp}{\varnoul{open}}
\newcommand{\closedp}{\varnoul{closed}}
\newcommand{\holdingp}{\varnoul{holds}}
\newcommand{\idlep}{\varnoul{idle}}
\newcommand{\inp}{\varnoul{in}}
\newcommand{\connectp}{\varnoul{connects}}

\newcommand{\batteryf}{\varnoul{battery\dsh level}}
\newcommand{\mvcostf}{\varnoul{cost\dsh move}}
\newcommand{\pucostf}{\varnoul{cost\dsh pick\dsh up}}
\newcommand{\opcostf}{\varnoul{cost\dsh open}}

\newcommand{\autonet}{\ensuremath{\mathcal{T}}}

\newcommand{\nmvrai}{\varnoul{\durAct{2}}}
\newcommand{\nmvraii}{\varnoul{\durAct{3}}}
\newcommand{\nopa}{\varnoul{\durAct{1}}}
\newcommand{\ncla}{\varnoul{\durAct{4}}}

\newcommand{\nmvraicol}{blue}
\newcommand{\nmvraiicol}{red}
\newcommand{\nopacol}{orange}
\newcommand{\nclacol}{cyan}

\newcommand{\rat}{\ensuremath{\mathbb{Q}}}
\newcommand{\vardom}{\ensuremath{\mathbb{D}}}
\newcommand{\mutex}{\text{mutex}}

\newcommand{\arbop}{\ensuremath{\mathop{\circ}}}

\newcommand{\compop}{\ensuremath{\mathop{\bowtie}}}
\newcommand{\arithop}{\ensuremath{\mathop{\circledast}}}

\tikzset{
    robot1/.pic={
        \draw node at (0,0) {\includegraphics[height=2cm]{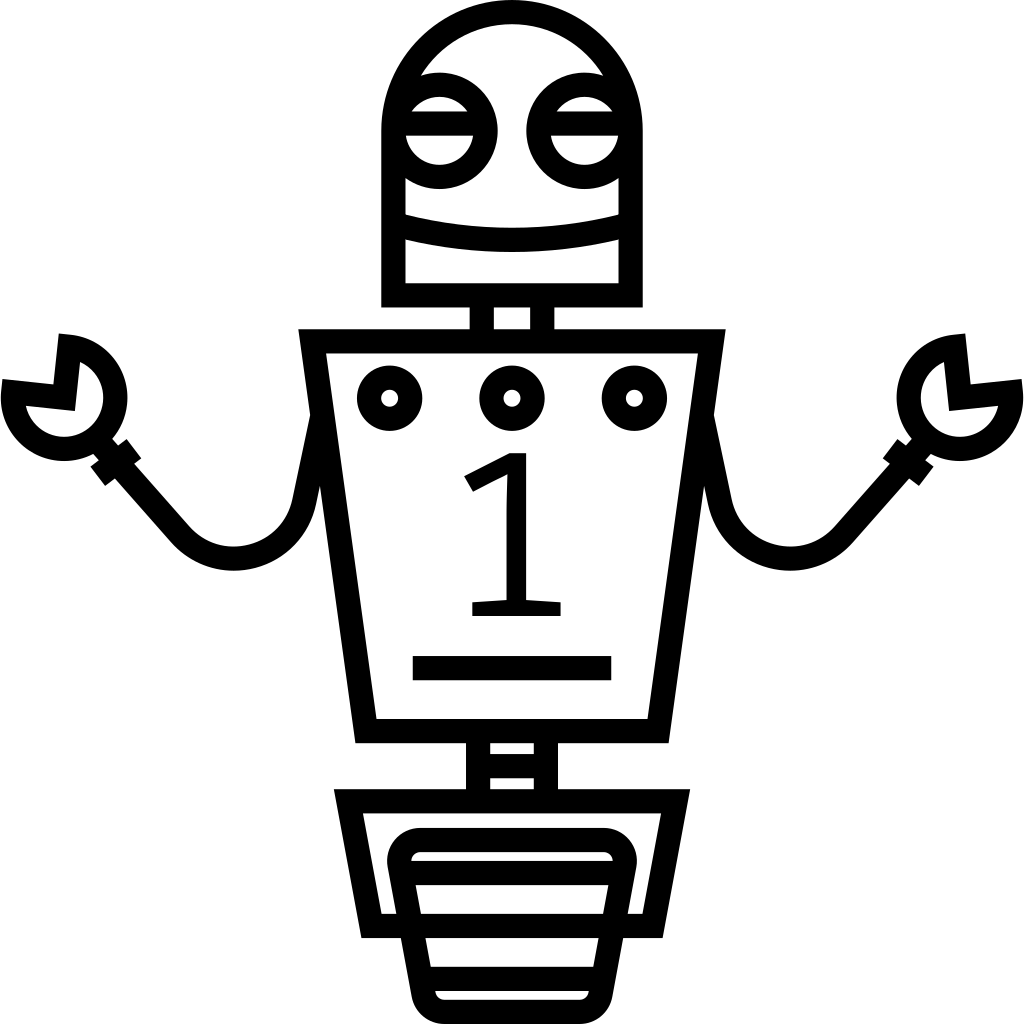}};
    };
}
\tikzset{
    robot2/.pic={
        \draw node at (0,0) {\includegraphics[height=2cm]{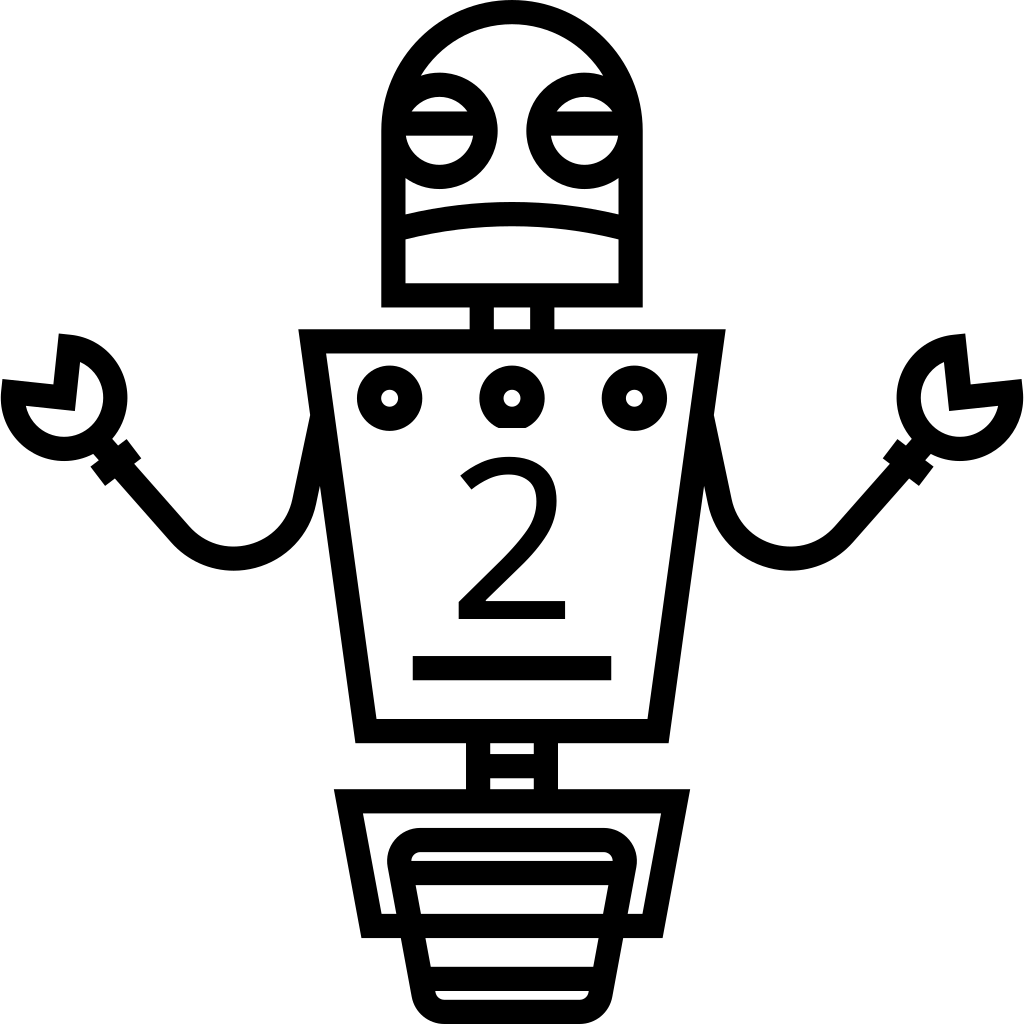}};
    };
}
\tikzset{
    block/.pic={
        \draw node at (0,0) {\includegraphics[height=1cm]{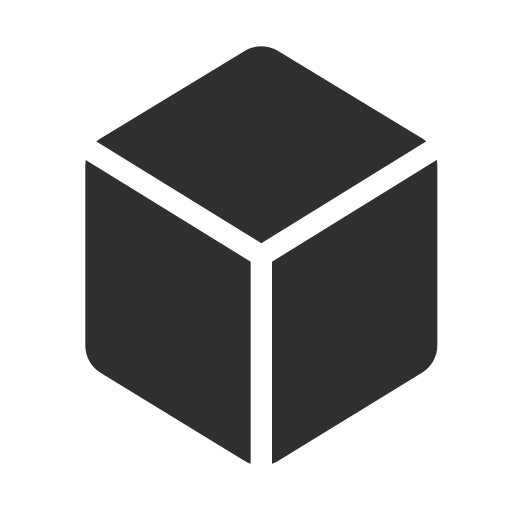}};
    };
}
\tikzset{
    blockinv/.pic={
        \draw node at (0,0) {\includegraphics[height=1cm]{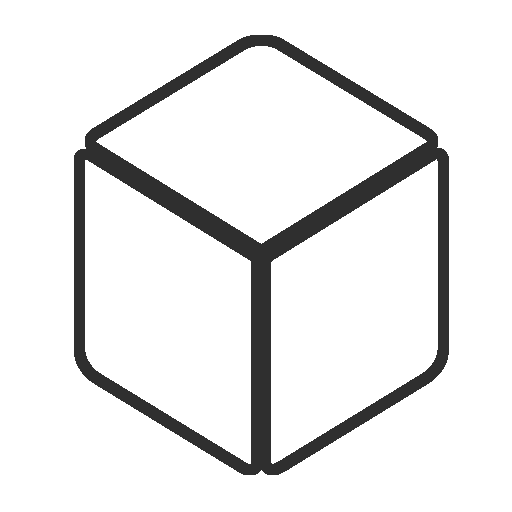}};
    };
}
\tikzset{
    r1b1l/.pic={
        \draw pic at (0,0) {robot1}
            pic at (-1.25,0.5) {block};
    },
    r1b1r/.pic={
        \draw pic at (0,0) {robot1}
            pic at (1.25,0.5) {block};
    },
    r2b1l/.pic={
        \draw pic at (0,0) {robot2}
            pic at (-1.25,0.5) {block};
    },
    r2b1r/.pic={
        \draw pic at (0,0) {robot2}
            pic at (1.25,0.5) {block};
    },
    r1b2l/.pic={
        \draw pic at (0,0) {robot1}
            pic at (-1.25,0.5) {blockinv};
    },
    r1b2r/.pic={
        \draw pic at (0,0) {robot1}
            pic at (1.25,0.5) {blockinv};
    },
    r2b2l/.pic={
        \draw pic at (0,0) {robot2}
            pic at (-1.25,0.5) {blockinv};
    },
    r2b2r/.pic={
        \draw pic at (0,0) {robot2}
            pic at (1.25,0.5) {blockinv};
    };
}
\tikzset{
    dooropen/.pic={
        \draw[color=gray] (0,0) -- ++(90:0.5) -- ++(300:0.5)
        (0,2) -- ++(270:0.5) -- ++(60:0.5);
    },
    doorclosed/.pic={
        \draw[color=gray] (0,0) -- (0,2);
    },
    doorhalfopen/.pic={
        \draw[color=gray] (0,0) ++(90:0.5) -- ++(15:0.5)
        (0,2) ++(270:0.5) -- ++(345:0.5);
        \draw[color=gray] (0,0) -- (0,2);
        \draw[color=gray] (0,0) ++(90:0.5) -- ++(300:0.5)
        (0,2) ++(270:0.5) -- ++(60:0.5);
    };
}
\tikzset{
    global scale/.style={
        scale=#1,
        every node/.style={scale=#1}
    }
}

\tikzset{>={Stealth}}

\tikzset{
    room width/.initial=3,
    room height/.initial=1.5,
    pics/room/.style 2 args={
        code={
            \draw[color=gray] (0,0) rectangle
                (\pgfkeysvalueof{/tikz/room width},\pgfkeysvalueof{/tikz/room height});
            \draw pic[global scale=0.5*\pgfkeysvalueof{/tikz/room height}]
                at (0.5*\pgfkeysvalueof{/tikz/room width},0) {#1};
            \draw node[anchor=north]
                at (0.5*\pgfkeysvalueof{/tikz/room width},0) {#2};
        }
    }
}

\tikzset{
    action box/.style={draw=#1!60,solid,fill=#1!20},
    invalid action box/.style={draw=none,pattern=north west lines,pattern color=#1!20}
}

\newcommand{\forbiddenmark}{\usym{1F6AB}}

\tikzset{
    to*/.style={
        shorten >=.25em,#1-to,
        to path={-- node[inner sep=0pt,at end,sloped] {${}^*$} (\tikztotarget) \tikztonodes}
    },
    to*/.default=
}

\makeatletter
\pgfarrowsdeclare{to*}{to*}
{
  \pgfutil@tempdima=-0.84pt%
  \advance\pgfutil@tempdima by-1.3\pgflinewidth%
  \pgfutil@tempdimb=0.21pt%
  \advance\pgfutil@tempdimb by.625\pgflinewidth%
  \advance\pgfutil@tempdimb by2.5pt%
  \pgfarrowsleftextend{+\pgfutil@tempdima}
 textual matterfarrowsrightextend{+\pgfutil@tempdimb}
}
{
  \pgfutil@tempdima=0.28pt%
  \advance\pgfutil@tempdima by.3\pgflinewidth%
  \pgfsetlinewidth{0.8\pgflinewidth}
  \pgfsetdash{}{+0pt}
  \pgfsetroundcap
  \pgfsetroundjoin
  \pgfpathmoveto{\pgfqpoint{-3\pgfutil@tempdima}{4\pgfutil@tempdima}}
  \pgfpathcurveto
  {\pgfqpoint{-2.75\pgfutil@tempdima}{2.5\pgfutil@tempdima}}
  {\pgfqpoint{0pt}{0.25\pgfutil@tempdima}}
  {\pgfqpoint{0.75\pgfutil@tempdima}{0pt}}
  \pgfpathcurveto
  {\pgfqpoint{0pt}{-0.25\pgfutil@tempdima}}
  {\pgfqpoint{-2.75\pgfutil@tempdima}{-2.5\pgfutil@tempdima}}
  {\pgfqpoint{-3\pgfutil@tempdima}{-4\pgfutil@tempdima}}
  \pgfusepathqstroke
  \begingroup
    \pgftransformxshift{2.5pt}
    \pgftransformyshift{2pt}
    \pgftransformscale{.7}
    \pgfuseplotmark{asterisk}
  \endgroup
}

\newcommand{\replaced}[2]{{\color{blue}#1} {\color{red}#2}}

\newcommand{\splitatcommas}[1]{%
  \begingroup
  \begingroup\lccode`~=`, \lowercase{\endgroup
    \edef~{\mathchar\the\mathcode`, \penalty0 \noexpand\hspace{0pt plus 1em}}%
  }\mathcode`,="8000 #1%
  \endgroup
}

\newcommand{\daStart}[1]{\ensuremath{{#1}_{start}}}
\newcommand{\daEnd}[1]{\ensuremath{{#1}_{end}}}
\newcommand{\daInv}[1]{\ensuremath{{#1}_{inv}}}
\newcommand{\daNumInv}[1]{\ensuremath{{#1}_{ninv}}}
\newcommand{\daLower}[1]{\ensuremath{{#1}_{l}}}
\newcommand{\daUpper}[1]{\ensuremath{{#1}_{u}}}

\newcommand{\saPre}[1]{\ensuremath{{#1}_{pre}}}
\newcommand{\saAdd}[1]{\ensuremath{{#1}_{add}}}
\newcommand{\saDel}[1]{\ensuremath{{#1}_{del}}}
\newcommand{\saNumPre}[1]{\ensuremath{{#1}_{nc}}}
\newcommand{\saNumEff}[1]{\ensuremath{{#1}_{upd}}}

\newcommand{\readsf}[1]{\varnoul{reads}(#1)}
\newcommand{\writesf}[1]{\varnoul{writes}(#1)}

\newcommand{\planAction}{\ensuremath{a}}
\newcommand{\atoms}{\operatorname{atoms}}

\newcommand{\integers}{\mathbb{Z}}

\newcommand{\goalform}{\Phi}

\newcommand{\lowerboundop}[1]{\triangleleft_{l#1}}
\newcommand{\upperboundop}[1]{\triangleleft_{u#1}}
\newcommand{\lowerboundval}[1]{l_{#1}}
\newcommand{\upperboundval}[1]{u_{#1}}
\renewcommand{\paragraph}[1]{\noindent{\bf #1}}

\renewcommand{\false}{\ensuremath{\bot}}
\renewcommand{\true}{\ensuremath{\top}}
\newcommand{\undefVal}{\ensuremath{\mathbf{u}}}
\newcommand{\propof}[1]{\ensuremath{{#1}_{P}}}
\newcommand{\numof}[1]{\ensuremath{{#1}_{F}}}
\newcommand{\propState}{\propstate}
\newcommand{\saApp}[2]{\ensuremath{{#2}[{#1}]}}
\newcommand{\snapActs}{\recSetActs}
\newcommand{\appSas}[2]{\ensuremath{{#2}[{#1}]}}
\newcommand{\propcondof}[1]{\ensuremath{{#1}_{P}}}
\newcommand{\numcondof}[1]{\ensuremath{{#1}_{F}}}
\newcommand{\numState}{N}
\newcommand{\snapActionNumPre}{\snapAction_\textit{nc}}
\newcommand{\snapActionNumEff}{\snapAction_\textit{upd}}

\newcommand{\fluents}{\varnoul{F}}

\newcommand{\goalCond}{\goal}

\makeatletter
\newtcolorbox{pseudocodebox}{enhanced, colback=black!4, frame hidden,
  boxrule=0pt, arc=1pt, left=3pt, right=3pt, top=3pt, bottom=3pt,
  before skip=6pt, after skip=6pt,
  before upper={\footnotesize\parindent\z@
    \begin{list}{}{\topsep\z@\partopsep\z@\parsep\z@\itemsep\z@
      \labelwidth\z@\labelsep\z@\itemindent\z@\listparindent\z@
      \leftmargin-1em\rightmargin\z@}\item[]\relax},
  after upper={\end{list}}}
\makeatother
\newcommand{\fn}[1]{\textsf{\upshape #1}}
\newcommand{\Mapt}[2]{\ensuremath{\mathsf{Map}\langle #1, #2 \rangle}}
\newcommand{\Sett}[1]{\ensuremath{\mathsf{Set}\langle #1 \rangle}}
\newcommand{\Listt}[1]{\ensuremath{\mathsf{List}\langle #1 \rangle}}

\newcommand{\tpp}{\ensuremath \planningproblem}
\newcommand{\tp}{\ensuremath{\plan}}
\newcommand{\tppa}{\ensuremath{\tpp a}}
\newcommand{\tpa}{\ensuremath{\tp a}}
\newcommand{\tppl}{\ensuremath{\tpp l}}
\newcommand{\tpl}{\ensuremath{\tp l}}
\newcommand{\tppal}{\ensuremath{\tppa l}}
\newcommand{\tpal}{\ensuremath{\tpa l}}
\newcommand{\tppli}{\ensuremath{\tppl i}}
\newcommand{\tpli}{\ensuremath{\tpl i}}
\newcommand{\tppali}{\ensuremath{\tppal i}}
\newcommand{\tpali}{\ensuremath{\tpal i}}
\newcommand{\tpplin}{\ensuremath{\tppli N}}
\newcommand{\tplin}{\ensuremath{\tpli N}}
\newcommand{\tppalin}{\ensuremath{\tppali N}}
\newcommand{\tpalin}{\ensuremath{\tpali N}}
\newcommand{\tpc}{\ensuremath{C \tp}}
\newcommand{\pdpp}{\ensuremath{\tpp S}}
\newcommand{\pdp}{\ensuremath{\tp S}}
\newcommand{\gpdpp}{\ensuremath{\pdpp g}}
\newcommand{\gpdp}{\ensuremath{\pdp g}}
\newcommand{\gpdppdefs}{\ensuremath{\gpdpp D}}
\newcommand{\actt}{\ensuremath{A}}
\newcommand{\actta}{\ensuremath{\actt a}}
\newcommand{\timaut}{\ensuremath{\autonet}}
\newcommand{\acttd}{\ensuremath{\actt d}}
\newcommand{\tppd}{\ensuremath{\tpp d}}
\newcommand{\tpd}{\ensuremath{\tp d}}
\newcommand{\actts}{\ensuremath{\actt S}}
\newcommand{\tpps}{\ensuremath{\tpp S}}
\newcommand{\tps}{\ensuremath{\tp S}}

\newcommand{\bnfdef}{\ensuremath{\Coloneqq}}

\newcommand{\kw}[1]{\ensuremath{\text{\normalfont\scshape #1}}}

\renewcommand{\initState}{\varnoul{I}}
\renewcommand{\goalCond}{\varnoul{G}}

\newcommand{\preenc}{\ensuremath{\mathit{pre}_{e}}}
\newcommand{\invenc}{\ensuremath{\mathit{inv}_{e}}}
\newcommand{\effenc}{\ensuremath{\mathit{eff}_{e}}}
\newcommand{\mutexenc}{\ensuremath{\mathit{mutex}_{e}}}
\newcommand{\durenc}{\ensuremath{\mathit{dur}_{e}}}

\makeatletter
\newcommand{\flagdavid}[2]{\if@todonotes@disabled#1\else{\color{orange}#1}\fi}
\newcommand{\removedavid}[1]{\if@todonotes@disabled\else{\color{orange}#1}\fi}
\newcommand{\movedavid}[1]{\if@todonotes@disabled\else{\color{orange}#1}\fi}
\makeatother

\newcommand{\contribref}[1]{Contribution~\ref{#1}}  %
\newcommand{\Contribref}[1]{Contribution~\ref{#1}}

\newcommand{\noselfoverlap}{\varnoul{no\dsh self\dsh overlap}}

\newcommand{\cnEntails}[2]{#1 \nvDash_c #2}

\newcommand{\dlo}[1]{\varbounds_{#1}^{\mathrm{lo}}}
\newcommand{\dhi}[1]{\varbounds_{#1}^{\mathrm{hi}}}

\newcommand{\varbounds}{\ensuremath{\mathcal{D}}}

\newcommand{\autolist}[1]{\ensuremath{\vec{#1}}}

\spnewtheorem*{proofsketch}{Proof sketch}{\itshape}{\rmfamily}

\newcommand{\initgap}{\kappa}

\newcommand{\infbox}{\mathit{box}}
\newcommand{\compops}{\mathit{CmpOps}}
\newcommand{\arithops}{\mathit{ArOps}}

\newcommand{\matchv}{\varnoul{m}}            %
\newcommand{\fusev}{\varnoul{f}}             %

\providecommand{\dsh}{\text{-}}
\newcommand{\lightact}{\varnoul{light\dsh match}}   %
\newcommand{\mendact}{\varnoul{mend\dsh fuse}}
\newcommand{\lightactsrt}{\varnoul{lm}}          %
\newcommand{\mendactsrt}{\varnoul{mf}}

\newcommand{\handfreep}{\varnoul{handfree}}   %
\newcommand{\unusedp}{\varnoul{unused}}
\newcommand{\lightp}{\varnoul{light}}
\newcommand{\mendedp}{\varnoul{mended}}

\newcommand{\durlitf}{\varnoul{dur\dsh lit}}
\newcommand{\durlitlbf}{\varnoul{dur\dsh lit\dsh lb}}   %
\newcommand{\durlitubf}{\varnoul{dur\dsh lit\dsh ub}}   %
\newcommand{\durmendf}{\varnoul{dur\dsh mend}}
\newcommand{\sparef}{\varnoul{spare\dsh fuses}}   %
\newcommand{\litf}{\varnoul{matches\dsh lit}}     %

\newcommand{\nlma}{\varnoul{\durAct{1}}}      %
\newcommand{\nmfa}{\varnoul{\durAct{2}}}      %
\newcommand{\nmfb}{\varnoul{\durAct{3}}}      %
\newcommand{\nmfc}{\varnoul{\durAct{4}}}      %

\newcommand{\nlmacol}{orange}
\newcommand{\nmfacol}{blue}
\newcommand{\nmfbcol}{red}
\newcommand{\nmfccol}{cyan}

\newcommand{\picheight}{2cm}
\newcommand{\counticon}[3]{\raisebox{\dimexpr 0.55ex-0.5\height+#3\height\relax}{\includegraphics[height=#2]{assets/#1.png}}}
\newcommand{\litlabel}[1]{#1\,\counticon{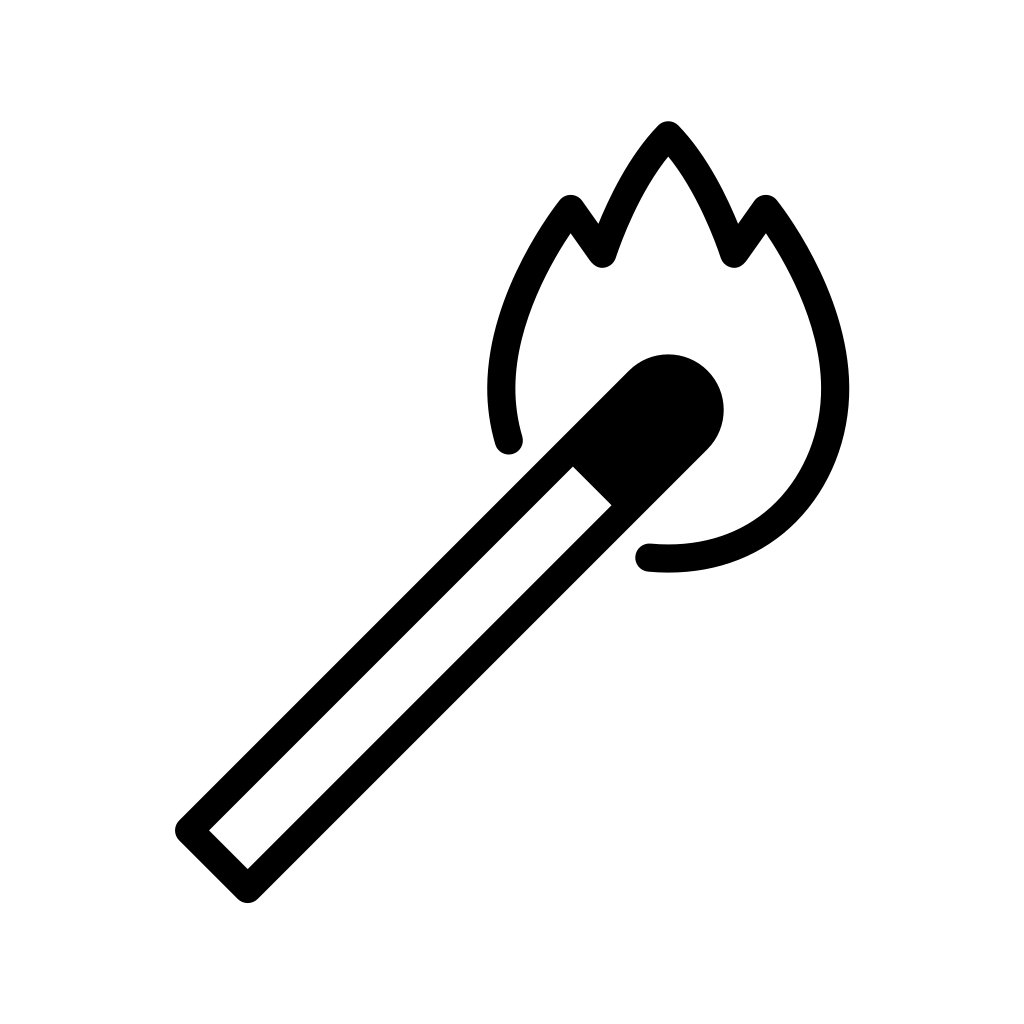}{\pgfkeysvalueof{/tikz/lit icon height}}{0.2}}
\newcommand{\spareslabel}[1]{#1\,\counticon{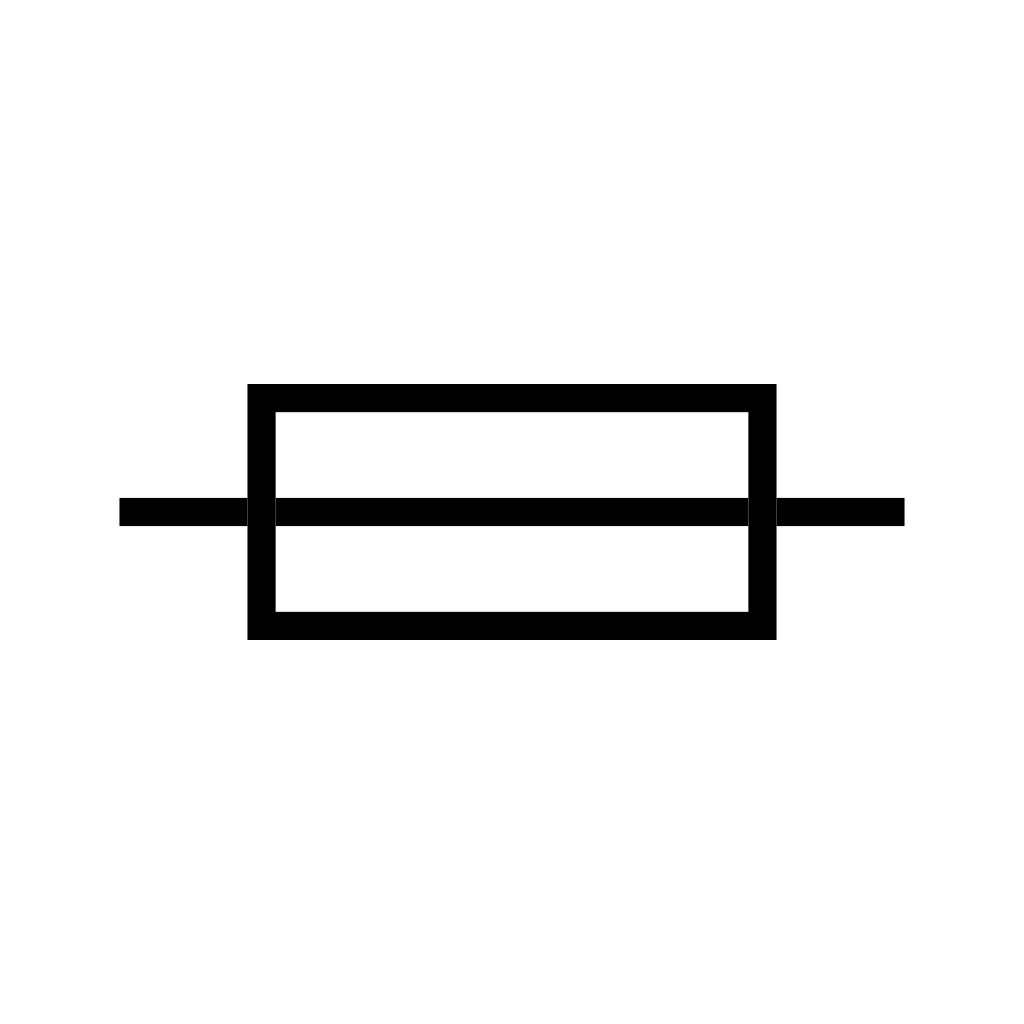}{\pgfkeysvalueof{/tikz/spares icon height}}{0}}
\newcommand{\counters}[2]{%
  \ifx#1\empty\else\spareslabel{#1}\fi
  \ifx#1\empty\else\ifx#2\empty\else\\[-7pt]\fi\fi
  \ifx#2\empty\else\litlabel{#2}\fi}
\tikzset{
    pic number/.style={font=\scriptsize\sffamily\bfseries, inner sep=0.4pt, circle, draw, fill=white, minimum size=0.9em},
    match angle/.store in=\matchangle, match angle=0,
    pics/match/.style 2 args={code={
        \begin{scope}[rotate=\matchangle]
            \node (-img) at (0.0670*\picheight,0.1199*\picheight)
                {\includegraphics[height=\picheight,angle=\matchangle]{assets/match_#1.png}};
            \ifx\relax#2\relax\else\node[pic number] at (-0.27*\picheight,-0.27*\picheight) {#2};\fi
        \end{scope}
    }},
    pics/fuse/.style 2 args={code={
        \node (-img) at (0,0) {\includegraphics[height=\picheight]{assets/fuse_#1.png}};
        \ifx\relax#2\relax\else\node[pic number] at (0,0.2*\picheight) {#2};\fi
    }},
    pics/hand/.style={code={
        \node (-img) at (-0.0267*\picheight,0.0667*\picheight) {\includegraphics[height=\picheight]{assets/hand_#1.png}};
    }},
    pics/sparebox/.style={code={
        \node (-img) at (0,0) {\includegraphics[height=\picheight]{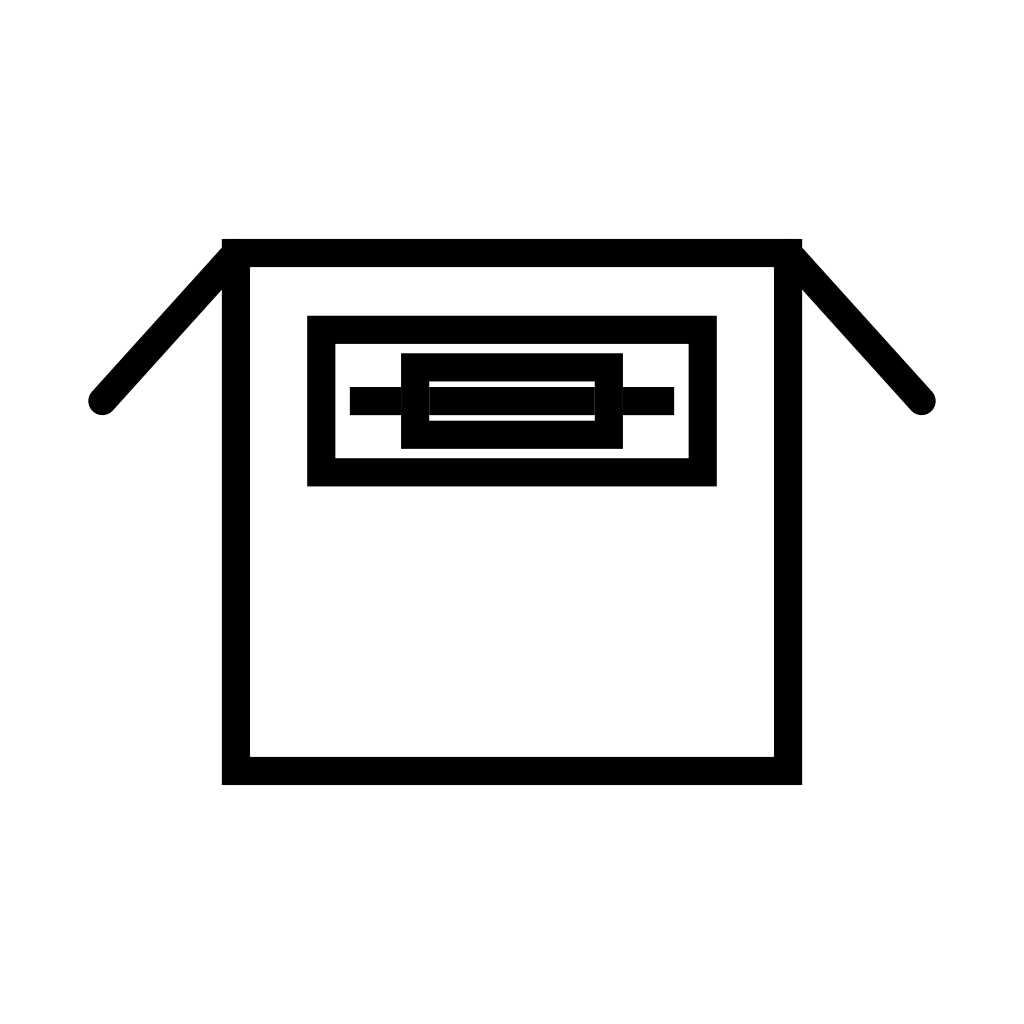}};
        \pgfmathsetmacro{\sparescale}{\picheight/2cm}
        \node[font=\sffamily\bfseries, inner sep=0, scale=\sparescale] at (0,-0.1050*\picheight) {\##1};
    }},
}
\tikzset{
    state width/.initial=3.1,  state height/.initial=2.5,
    match column/.initial=0.6,           %
    fuse column/.initial=0.7,            %
    column top/.initial=0.6,             %
    match sep/.initial=0.62, fuse sep/.initial=0.7,  %
    hand at/.initial={2.15,0.75},        %
    spares/.initial=,                    %
    lit/.initial=,                       %
    lit counter/.style={font=\sffamily\bfseries, inner sep=0, scale=0.7, align=left, anchor=south west},
    lit icon height/.initial=1.2em, spares icon height/.initial=1.9em,
    counters at/.initial={0.2,0.18},     %
    state frame/.style={draw=gray, line width=0.4pt},
    dark/.style={state frame/.append style={fill=black!12}},
    state scale/.initial=1,              %
    diagonal radius/.initial=,           %
    pics/state/.style n args={3}{code={
        \pgfmathsetmacro{\ss}{\pgfkeysvalueof{/tikz/state scale}}
        \pgfmathsetlengthmacro{\pich}{\ss*\picheight}\edef\picheight{\pich}
        \pgfmathsetmacro{\sw}{\ss*\pgfkeysvalueof{/tikz/state width}}
        \pgfmathsetmacro{\sh}{\ss*\pgfkeysvalueof{/tikz/state height}}
        \node[state frame, anchor=south west, minimum width=\sw cm, minimum height=\sh cm, inner sep=0] (-box) at (0,0) {};
        \begin{scope}[scale=\ss]   %
        \pgfmathsetmacro{\shu}{\pgfkeysvalueof{/tikz/state height}}
        \pgfmathsetmacro{\swu}{\pgfkeysvalueof{/tikz/state width}}
        \edef\sparesval{\pgfkeysvalueof{/tikz/spares}}
        \edef\litval{\pgfkeysvalueof{/tikz/lit}}
        \edef\diagr{\pgfkeysvalueof{/tikz/diagonal radius}}
        \ifx\diagr\empty
        \foreach \m [count=\i from 0] in {#1} {
            \pic at ({\pgfkeysvalueof{/tikz/match column}}, {\shu-\pgfkeysvalueof{/tikz/column top}-\i*\pgfkeysvalueof{/tikz/match sep}}) {match={\m}{\i}};
        }
        \foreach \f [count=\i from 0] in {#3} {
            \pic at ({\swu-\pgfkeysvalueof{/tikz/fuse column}}, {\shu-\pgfkeysvalueof{/tikz/column top}-\i*\pgfkeysvalueof{/tikz/fuse sep}}) {fuse={\f}{\i}};
        }
        \pic at (\pgfkeysvalueof{/tikz/hand at}) {hand={#2}};
        \node[lit counter] at (\pgfkeysvalueof{/tikz/counters at}) {\counters\sparesval\litval};
        \else
        \pgfmathsetmacro{\dd}{\diagr/sqrt(2)}
        \pic at ({\swu/2-\dd},{\shu/2+\dd}) {match={#1}{0}};
        \pic at ({\swu/2+\dd},{\shu/2+\dd-0.12}) {fuse={#3}{0}};    %
        \pic at ({\swu/2+\dd},{\shu/2-\dd+0.12}) {hand={#2}};
        \node[lit counter] at (\pgfkeysvalueof{/tikz/counters at}) {\counters\sparesval\litval};
        \fi
        \end{scope}
    }},
}

\tikzset{
    matchunlit/.pic={ \pic at (0,0) {match={unlit}{}}; }, matchlit/.pic={ \pic at (0,0) {match={lit}{}}; },
    matchburnt/.pic={ \pic at (0,0) {match={burnt}{}}; },
    fusebroken/.pic={ \pic at (0,0) {fuse={broken}{}}; }, fusemended/.pic={ \pic at (0,0) {fuse={mended}{}}; },
    handfree/.pic={ \pic at (0,0) {hand={free}}; }, handbusy/.pic={ \pic at (0,0) {hand={busy}}; },
    handlit/.pic={ \pic at (0,0) {hand={busy}}; }, handfuse/.pic={ \pic at (0,0) {hand={busy}}; },
}

\tikzset{
    cellar width/.initial=3,
    cellar height/.initial=1.5,
    pics/cellar/.style 2 args={
        code={
            \draw[color=gray] (0,0) rectangle
                (\pgfkeysvalueof{/tikz/cellar width},\pgfkeysvalueof{/tikz/cellar height});
            \draw pic[global scale=0.5*\pgfkeysvalueof{/tikz/cellar height}]
                at (0.5*\pgfkeysvalueof{/tikz/cellar width},0) {#1};
            \draw node[anchor=north]
                at (0.5*\pgfkeysvalueof{/tikz/cellar width},0) {#2};
        }
    }
}
\tikzset{
    cellardark/.pic={ \draw[color=gray] (0,0) -- (0,2); },
    cellarlit/.pic={ \draw[color=gray] (0,0) -- ++(90:0.5) -- ++(300:0.5) (0,2) -- ++(270:0.5) -- ++(60:0.5); };
}

\newcommand{\matchnums}{\ensuremath{N_{\matchv}}}          %
\newcommand{\fusenums}{\ensuremath{N_{\fusev}}}            %
\newcommand{\durlit}{\ensuremath{d_{\mathit{lit}}}}        %
\newcommand{\durlitlb}{\ensuremath{d^{\mathit{lb}}_{\mathit{lit}}}}  %
\newcommand{\durlitub}{\ensuremath{d^{\mathit{ub}}_{\mathit{lit}}}}  %
\newcommand{\durmend}{\ensuremath{d_{\mathit{mend}}}}      %
\newcommand{\numspare}{\ensuremath{n_{\mathit{spare}}}}    %
\section{Introduction}
\label{sec:intro}

In AI planning, given a model of the world describing its dynamics in terms of \emph{actions} and its \emph{state} using state variables, the aim is to find a sequence of actions that reaches a \emph{goal} from an \emph{initial state}.
Since research on this problem started in the 1960s~\cite{mccarthy1981some}, the performance and expressiveness of planning systems have improved substantially, as the planning competitions show~\cite{IPC1998reference,DBLP:journals/ai/GereviniHLSD09/corrected,DBLP:journals/aim/ColesCOJLSY12,vallati20152014,IPC2022}, and state-of-the-art planners now scale to real-world sequential decision problems.

Despite this progress, the trustworthiness of planners' outputs can still be improved, given the complexity of state-of-the-art planning algorithms at an abstract level as well as at an implementation optimisation level.
Much of the work in this direction develops \emph{certificate checkers}, i.e.\ (relatively) small, independent programs that check the output of a planner, such as plan validators~\cite{VAL,inval}, which confirm that a candidate solution solves a given planning problem.
Validators have also been formally verified, i.e.\ proved correct with respect to a semantics of the planning formalism, for classical planning~\cite{ictai2018} and for temporal planning~\cite{aaai2022}, which adds action durations on a dense timeline and concurrent execution of plan actions.

Certifying that a problem is unsolvable, or that a plan is optimal, is much harder than validating a candidate plan.
A plan is usually succinct and is itself a certificate, checked by executing its steps, while a certificate of unsolvability or optimality can, for most relevant planning formalisms, be exponentially long in the size of the planning task in the worst case (unless $\mathrm{NP} = \mathrm{PSPACE}$~\cite[Chapter 4]{aroraBarakComplexity}).
In practice, such certificates are exponentially long unless they are very carefully designed, like those of \citet{DBLP:conf/aips/ErikssonRH17}, which can be succinct for large classes of classical planning problems.

\begin{figure}[t]
  \centering
    \begin{tikzpicture}[%
      scale=5.0,
      every node/.style={font=\small},
      node/.style={draw,inner sep=0.2cm}, 
      decision/.style={draw,align=center,inner sep=0.5em}, 
      mylabel/.style={text width=3cm}]
    \node(problem)[varnode]{};
    \node(verified)[node, right=1.5cm of problem]{\scriptsize Verified Encoder};
    \node(unverified)[dotted,node,above right=-0.15cm and 0.6cm of verified]{\scriptsize Unverified Solver};
    \node(certifier)[node,below=0.29cm of unverified]{\scriptsize Certifier};

    \draw[->,transform canvas={yshift=1.4mm}] (problem.east) --
      node[midway,above] {\scriptsize Input} (verified.west);
    \draw[<-,transform canvas={yshift=-1.4mm}] (problem.east) --
      node[midway,below] {\scriptsize Output} (verified.west);
    \draw[->] (verified.north) |- (unverified.west);
    \draw[->] (unverified.south) -- (certifier.north);
    \draw[->] (certifier.west) -| (verified.south);
  \end{tikzpicture}
\caption{The planner architecture proposed by \citet{verifiedSATPlan}, combining verified (solid) and unverified (dotted) modules.
The input is a grounded planning problem, which is encoded into an automaton by a verified encoder, the automaton is checked by an efficient, unverified model checker, and the result is certified with a verified certificate checker.}\label{fig:plannerarch}
  \vspace{-3.5ex}
\end{figure}

We consider certifying unsolvability of temporal planning~\cite{fox2003pddl2} (see \Cref{ex:matchcellar_instance_1_1_2_56_5} and \Cref{def:matchcellar_template}).
Making state-of-the-art temporal planners produce succinct certificates would require substantial changes to their algorithms and optimised implementations.
We thus follow another approach (\Cref{fig:plannerarch}): we encode the planning problem into timed automata (TA) model checking, for which a practical certifying algorithm exists.
The certificate then implies unsolvability of the planning problem only if the encoding and its implementation are correct.
However, that encoding is complex (\Cref{ex:matchcellar_instance_1_3_5_2}) and can be a source of bugs/untrustworthiness.
Although temporal planning, like model checking, is defined on state transition systems, planning languages have richer constructs: durative invariants (i.e.\ conditions that must hold throughout an action's execution) and numeric fluents with linear comparison guards (\Cref{ex:matchcellar_instance_1_1_2_56_5} and \Cref{fig:example_plan}) cannot be modelled natively in common TA formalisms, e.g.\ TA networks, which makes an efficient encoding too complex to trust without verification.
We therefore formally verify, using an interactive theorem prover (ITP), the correctness of the encoding and its implementation.%

More specifically, our contributions are: \begin{enumerate*}[label=(Contribution~\roman*), ref=\roman*]
    \item\label{contrib:encoding} a practical encoding of grounded temporal planning problems into networks of timed automata as formalised by Wimmer~\cite{10.1007/978-3-030-29662-9_14,wimmerTimedAutomataAFP2016}, for use with TChecker~\cite{TicktacprojectTchecker2023}, which combines and extends the encodings of \citet{temporalPlanningRefinementMC} and \citet{giganteDecidabilityComplexity2022} (\Cref{sec:reduction}), in particular with numeric fluents, whose finite bounds we infer by interval analysis and widening with thresholds (\Cref{sec:bound_inference}),
    \item\label{contrib:validity} a more general notion of plan validity and unsolvability than earlier encodings, with dense duration intervals $[l, u]$ (like \citet{giganteDecidabilityComplexity2022}, unlike \citet{temporalPlanningRefinementMC}) and concurrent execution of non-interfering snap actions without a global lock (\Cref{fig:example_plan}),
    \item\label{contrib:proof} a detailed correctness proof of the encoding (\Cref{sec:proof}, \Cref{appendix:proofs}), which accounts for the order in which transitions are applied and shows that variable bounds are respected throughout runs, closing gaps in previous informal proofs,
    \item\label{contrib:formalisation} a formal verification in Isabelle/HOL of this proof, of an efficient executable implementation of the encoding (\Cref{sec:formalisation}), and of the soundness of the system that combines it with the verified certificate checker of \citet{wimmerCertTimedAutomata} and an unverified certifying TA model checker, and
    \item\label{contrib:experiments} an experimental evaluation showing that our system is practical (\Cref{sec:experiments}).
\end{enumerate*}

\section{Background and Definitions}

\begin{paperdefinition}{Boolean expression}\label{def:bexp}
    Sets of well-formed arithmetic \( \exps{\vars} \) and Boolean \( \bexps{\vars} \)
    expressions are constructed from a set of variables \( \vars \).
    \( e \) and \( b \) are in \( \exps{\vars} \) and \( \bexps{\vars} \) respectively, if
    \( e \bnfdef x\ |\ d\ |\ e_h \arithop e_i \)
    and \( b \bnfdef e_j \compop e_k\ |\ b_j \land b_k\ |\ \true \), %
    where \( \compop \in \compops \equiv \{ <, \le, =, \ge, > \} \),
    \( \arithop \in \arithops \equiv \{ +, -, \cdot, / \} \),
    \( x \in \vars \),
    \( d \in \vardom \),
    \( \{e_h,e_i\} \subseteq \exps{\vars} \), and
    \( \{b_j, b_k\} \subseteq \bexps{\vars} \) for some \( h \), \( i \), \( j \), and \( k \).

\end{paperdefinition}

\begin{paperdefinition}{Variable assignment}\label{def:var_asmt}
    Let \( \poset{S} \) denote the powerset of a set \( S \).
    A \emph{variable assignment} \( \Var: \vars \rightarrow \vardom \) is a mapping
    from a set of variables \( \vars \) to members of a numeric type \( \vardom \) (e.g.\ \( \rat \)).
    Expressions \( e, f \in \bexps{\vars} \cup \exps{\vars} \)
    are evaluated with respect to a variable assignment \( \Var \) according to the rules
    \( \eval(\Var,x) \equiv \Var(x) \), \(\eval(\Var,d) \equiv d\), \(\eval(\Var, \true) \equiv \true\)
    and \( \eval(\Var, e \arbop f) \equiv e' \arbop f' \),
    where \( x \in \vars \), \( d \in \vardom \), \( e' \equiv \eval(\Var, e) \), \( f' \equiv \eval(\Var, f) \),
    and \( (\{ e, f \}, \arbop) \in \poset{\exps{\vars}} \times (\arithops \cup \compops) \) or
    \( (\{ e, f \},\arbop) \in \poset{\bexps{\vars}} \times \{ \land \} \).
    \( \bEntails{\Var}{g} \) denotes the \emph{satisfaction} of a Boolean expression \( g \in \bexps{\vars} \)
    by a variable assignment \( \Var \),
    i.e.\ its evaluation to \( \true \).
    Note that a division by \( 0 \) or
    the presence of a variable \( y \notin \vars \) in a subexpression of \( g \)
    evaluates to an undefined value \( \undefVal \),
    which falsifies \( \bEntails{\Var}{g} \).

\end{paperdefinition}

\begin{paperdefinition}{Variable updates}\label{def:var_update}
    An \emph{update} \( u \) for a variable \( z \in \vars \) is a pair
    \( \langle z, e \rangle \in \vars \times \exps{\vars} \), also denoted as \( z \coloneq e \).
    \( \writesf{u} \) is the set of written
    variables \( \{ z \} \) and \( \readsf{u} \) the set of read variables in \( e \).
    Similarly, \( \writesf{U} \) and \( \readsf{U} \) are the sets of written and read variables
    in a set of updates \( U \).

\end{paperdefinition}

\paragraph{Timed Automata}

{\emergencystretch=1.5em
Now follow the concepts of timed automata~\cite{timedautomata}
as formalised by Wimmer~\cite{10.1007/978-3-030-29662-9_14,wimmerTimedAutomataAFP2016},
which we present in simplified form.
Their work builds on \citeauthor{larsenUppaalNutshell1997}'s work~\cite{larsenUppaalNutshell1997}.
A distinctive feature is that it has explicit semantics for variable assignments and concerns networks of timed automata, 
as opposed to the standard formulation,
which concerns individual automata~\cite{timedautomata}.\par}

\begin{paperdefinition}{Clocks}\label{def:clock_const}\label{def:clock_val}
    A \emph{clock constraint} \( \delta_x \) is defined with respect to a clock variable \( x \) in a set of clock variables \( \clocks \) 
    and a constant \( d \in \rat \) as \( \delta_x \equiv x \compop d \),
    where \( \compop \in \{ <, \le, =, \ge, > \} \).

    A \emph{clock valuation} \( \Clk \) is a mapping from \( \clocks \) to \( \rat_{\ge 0} \).
    A constraint \( \delta_x \) can be evaluated against a valuation \( \Clk \) using the rule:
    \( \ccval(\Clk, x \compop d) \) if and only if \( \Clk(x) \compop d \).
    \( \Clk \) satisfies a set of clock constraints \( g \)
    if and only if each constraint \( \delta_x \) in \( g \) is satisfied by \( \Clk \).
    We denote the \emph{satisfaction} (respectively \emph{falsification}) of a set of constraints as:
    \( \cEntails{\Clk}{g} \) (respectively \( \cnEntails{\Clk}{g} \)).

\end{paperdefinition}

\begin{paperdefinition}{Clock updates}\label{def:clock_update}
    A clock \( x \in \clocks \) can only be updated explicitly by resetting its value in a clock valuation \( \Clk \) to \( 0 \).
    We denote single resets as \( x \coloneq 0 \)
    and multiple resets using the set of clocks \( r \subseteq \clocks \) they apply to.

\end{paperdefinition}

\begin{paperdefinition}{Timed Automaton}\label{def:timed_automaton}\label{def:timed_transition_relation}
    A \emph{timed transition relation}
    \( \Delta \subseteq \locs \times \bexps{\vars} \times \constraints{\clocks} \times \poset{\vars \times \exps{\vars}} \times \poset{\clocks} \times \locs \) is a set of \emph{transitions} \( t \equiv \langle l, b, g, u, r, l' \rangle \)
    consisting of
    \begin{enumerate*}[itemjoin={;\ }, itemjoin*={; and\ }]
        \item an outgoing location \( l \)
        \item a condition \( b \)
        \item a guard \( g \)
        \item a set of updates \( u \)
        \item a set of resets \( r \)
        \item an incoming location \( l' \).
    \end{enumerate*}
    In diagrams (e.g.\ \Cref{fig:TA_example,def:main_automaton}),
    we omit locations
    (e.g.\ \( \langle l, b, g, u, r, l' \rangle \simeq \langle b, g, u, r \rangle \)),
    and as a shorthand, list conditions and guards, as well as updates and resets together.
    We write each condition or guard as a query \( (b)? \), 
    and each update or reset as \( z \coloneq e \).

    A \emph{timed automaton} \( \auto = \langle \locs, \Delta \rangle \) is a tuple consisting of:
    \begin{enumerate*}[itemjoin={;\ }, itemjoin*={; and\ }]
        \item a set of locations: \( \locs \)
        \item a transition relation: \( \Delta \).
    \end{enumerate*}
See \Cref{fig:TA_example} for an example of a timed automaton's transitions.

\end{paperdefinition}

\begin{paperdefinition}{Network of Timed Automata}\label{def:network_of_timed_automata}\label{def:urgency}
    A \emph{network of timed automata}
    \( \autonet \equiv \langle \autolist{\autonet}, \vars, \clocks, \varbounds, \langle urg_0, \ldots, urg_n \rangle \rangle \) is:
    \begin{enumerate*}[itemjoin={;\ }, itemjoin*={; and\ }]
        \item a list of timed automata \( \autolist{\autonet} = \langle \autonet_0, \ldots, \autonet_n \rangle \), where \( \autonet_i = \langle \locs_i, \Delta_i \rangle \) for \( 0 \le i \le n \), with pairwise disjoint sets of locations
        \item a set of variables \( \vars \)
        \item a set of clocks \( \clocks \)
        \item a map \( \varbounds \) from variables \( z \in \vars \) to finite interval domains \( \varbounds_z = [\dlo{z}, \dhi{z}] \subseteq \integers \)
        \item a subset of locations of each automaton \( \autonet_i \) declared to be \emph{urgent}: \( urg_i \subseteq \locs_i \)
    \end{enumerate*}.
    We mark urgent locations in diagrams with exclamation marks: e.g.\ ``loc\textsubscript{1}!''.

\end{paperdefinition}

\begin{paperdefinition}{Network of Timed Automata Transitions}\label{def:nta_semantics}\label{def:configuration}\label{def:urgency_sem}
    
    The semantics of a network of timed automata are defined with respect to \emph{configurations} (\( \langle \Loc, \Var, \Clk \rangle \)),
    where
    \begin{enumerate*}[itemjoin={;\ }, itemjoin*={; and\ }]
        \item \( \Loc = \langle l_0, \dots, l_n \rangle \) is a list of the same length as \( \autolist{\autonet} \), 
            recording with each \( l_i \) a location for \( \autonet_i \)
        \item \( \Var \) is a variable assignment for \( \vars \) with values in \( \integers = \vardom \) (see \Cref{def:var_asmt})
        \item \( \Clk \) is a clock valuation for \( \clocks \) (see \Cref{def:clock_val})
    \end{enumerate*}.
    We denote an \emph{initial configuration} as \( \langle \Loc_I, \Var_I, \Clk_I \rangle \),
    where \( \Var_I(z) \in \varbounds_z \) must hold for all \( z \in \vars \).

    The transition relation between \emph{configurations} is characterised by two kinds of transitions,
    namely
    \begin{enumerate*}[mode=boxed, itemjoin={;\ }, itemjoin*={; and\ }, label=(\arabic*)]
        \item \emph{Delay:} \( \langle \Loc, \Var, \Clk \rangle \longrightarrow_\delta \langle \Loc, \Var, \Clk' \rangle \), 
        which represents the passage of a time period \( \delta \) by updating \( \Clk \) to \( \Clk' \),
        with the following conditions and updates:
        \begin{enumerate*}[mode=boxed, itemjoin={;\ }, itemjoin*={; and\ }, label={(\arabic{enumi}.\alph*)}]
            \item Condition: No locations \( l_j \) in \( \Loc \) are \emph{urgent} (see \Cref{def:urgency}) and \( \delta \ge 0 \)
            \item Update: \( \Clk'(x) = \Clk(x) + \delta \) for all clocks \( x \in \clocks \), which we denote as \( \Clk' = \Clk \oplus \delta \)
        \end{enumerate*}
        \item \emph{Internal:} \( \langle \Loc, \Var, \Clk \rangle \longrightarrow_{t} \langle \Loc', \Var', \Clk' \rangle \),
        which applies an update to the location of a single automaton \( \autonet_i \) according to the following rules:
        \begin{enumerate*}[mode=boxed, itemjoin={;\ }, itemjoin*={; and\ }, label={(\arabic{enumi}.\alph*)}]
            \item Conditions: 
            \begin{enumerate*}[mode=boxed, itemjoin={;\ }, itemjoin*={; and\ }, label={(\arabic{enumi}.\alph{enumii}.\roman*)}]
                \item \( t \equiv \langle l, b, g, u, r, l' \rangle \) is in \( \autonet_i \)'s transition relation \( \Delta_i \)
                \item the location \( l_i \) of automaton \( \autonet_i \) matches \( l \)
                \item the condition on variables is satisfied (\( \bEntails{\Var}{b}\))
                \item the guard is satisfied (\( \cEntails{\Clk}{g} \))
                \item \( \eval(\Var, e) \in \varbounds_z \) for \( \langle z, e \rangle \in u \)
                    \footnote{In Isabelle/HOL and Munta, division is integer division, 
                    which rounds down, and falls back to 0 on division by zero~\cite{10.1007/978-3-030-29662-9_14}.
                    In our checker, we disallow division to avoid unintended behaviour.}
            \end{enumerate*}
            \item Updates:
            \begin{enumerate*}[mode=boxed, itemjoin={;\ }, itemjoin*={; and\ }, label={(\arabic{enumi}.\alph{enumii}.\roman*)}]
                \item \( \Loc' = \langle l_0, \dots, l', \dots, l_n \rangle \), with \( l' \) at position \( i \)
                \item if \( \langle z, e \rangle \in u \) then \( \Var'(z) = \eval(\Var, e) \) else \( \Var'(z) = \Var(z) \) (that is simultaneous updates)\footnote{Only one update per variable and transition is allowed.}
                \item if \( x \in r \) then \( \Clk'(x) = 0 \) else \( \Clk'(x) = \Clk(x) \)
            \end{enumerate*}
            for all \( z \in \vars \) and \( x \in \clocks \)
        \end{enumerate*}
    \end{enumerate*}.

    We write \( \longrightarrow \) for the transition relation with unspecified labels
    and \( \longrightarrow^* \) for its reflexive transitive closure 
    (i.e.\ any sequence of \( \longrightarrow_\delta \) and \( \longrightarrow_t \)).

\end{paperdefinition}
    
  Intuitively, when a \emph{configuration} contains an urgent (\Cref{def:urgency_sem}) location, 
  the network cannot make a transition that represents the passage of time.
 
\begin{paperdefinition}{Run}\label{def:run_path}
    A \emph{run} consists of alternating \( \longrightarrow_t \) and \( \longrightarrow_\delta \) transitions.
    A finite run is modelled by \citet{wimmerTimedAutomataAFP2016} as an inductive list of configurations \( \cfgs = \langle \langle \Loc_0, \Var_0, \Clk_0 \rangle, \dots, \langle \Loc_m, \Var_m, \Clk_m \rangle \rangle \).
    \( \cfgs \) satisfies the conditions to be a run (\(\mathit{run}(\cfgs)\)), if for each \( 0 \le k < m \) there exists a \( \Clk_k'\), such that \( \langle \Loc_k, \Var_k, \Clk_k \rangle \longrightarrow_\delta \langle \Loc_k, \Var_k, \Clk_k' \rangle \longrightarrow_t \langle \Loc_{k+1}, \Var_{k+1}, \Clk_{k+1} \rangle \).
    When not otherwise specified, we use the term \emph{run} to refer to finite runs.

\end{paperdefinition}

\begin{paperremark}
Model checkers for timed automata can check that properties expressed as [Timed] Computation Tree Logic ([T]CTL) formulas hold for a network of timed automata.
These are asserted on some or all runs beginning at the initial configuration~\cite{timedautomata}. 
For modern treatments of the subject, we refer to~\cite{DBLP:reference/mc/ClarkeHV18,DBLP:conf/birthday/BouyerLMOW17}.

One limitation is that model checkers need a finite discrete state space,
which can be impossible to find when arbitrary integer variables are involved~\cite{DBLP:conf/aips/Helmert02}.
Hence the need for a finite domain for each variable in \Cref{def:nta_semantics}.
We use techniques from abstract interpretation to
precompute over-approximations of the domains of variables
produced when encoding planning instances (see \Cref{sec:numeric_fluents}).

\end{paperremark}

\paragraph{Temporal Planning}

Temporal planning is a formalism that can logically represent problems with continuous time and 
discrete states and actions.
Planners explore the space of possible action sequences to achieve specified goals.
We illustrate the concepts using instances of a planning domain (\Cref{def:matchcellar_template}) as running examples throughout the paper.
In this domain, 
all broken fuses must be mended in a dark cellar
using the light provided by matches. 
The agent only has one free
hand to mend one fuse at a time and hence cannot perform multiple actions simultaneously.

We defer the full details to \Cref{sec:reduction},
but introduce the basic idea of encoding actions shown in \Cref{fig:example_plan,ex:matchcellar_instance_1_1_2_56_5},
like lighting a match
or mending a fuse,
that can change a propositional and numeric state of the world 
as timed automata in a network,
propositions, 
such as ``fuse 1 is broken'', 
as binary (integer) variables,
and quantities, 
like the number of matches used,
as integer variables.

{\emergencystretch=1.5em
Real-world planning software works on abstract syntax of languages
such as PDDL~\cite{pddl3} and ANML~\cite{smith2008anml},
which represent actions as templates (or \emph{schemas}),
propositions as predicates and numeric fluents as functions to be
instantiated with constants (or \emph{objects}).
PDDL also allows the representation of state-dependent duration constraints
and continuous change,
which are not supported by standard timed automata.
The semantics of PDDL are complex enough that
checking plans output by planning software by eye is hard; 
a significant amount of research has been dedicated to producing tools that 
perform the check automatically and improving their reliability~\cite{VAL,inval},
and even formally verifying such checkers~\cite{aaai2019,aaai2022}.\par}

We base our temporal planning formalism on a fragment of the formal semantics by \citet{aaai2022}, 
except that our preconditions here are limited to conjunctions.
This roughly corresponds to a subset of grounded PDDL.
We also refer to the work of \citet{giganteDecidabilityComplexity2022},
who assume that a planning problem has been preprocessed to use a set-theoretic formulation of propositional planning, 
with conditions being minimal sets of true propositions.

\newcommand{\truetruth}{1}
\newcommand{\falsetruth}{0}
\newcommand{\valuation}{\ensuremath{\mathcal{A}}}

\newcommand{\consts}{\mathcal{C}}

\begin{paperdefinition}{Valuation}\label{def:valuation}
A valuation \( \valuation \) %
is a mapping from a set of propositions \( \propositions \) to the set \( \{\falsetruth,\truetruth\} \).
A valuation \( \valuation \) is a model for a set of propositions \( \phi \), written \( \pEntails{\valuation}{\phi} \), if and only if
\( \valuation(p)=\truetruth \) for all propositions \( p \in \phi \).
A valuation \( \valuation:\propositions\rightarrow\{\falsetruth,\truetruth\} \) is interpreted as the set \( \{\v\mid \valuation(\v) = \truetruth\} \)
and a set of propositions \( Q \) is interpreted as a valuation which maps any \( p \in Q \) to \( \truetruth \), 
and everything else to \( \falsetruth \).
For example, a set \( Q_1 \) in \Cref{ex:matchcellar_instance_1_1_2_56_5} can include the proposition \( p \),
which indicates that ``the hand is busy''.
Also, in the rest of this paper a \emph{propositional state} \( \propState \) is synonymous with a valuation.

\end{paperdefinition}

\begin{paperremark}
As all propositional formulas in this paper are conjunctions of positive atoms (i.e.\ propositions),
we characterise them using sets of their propositions.
\end{paperremark}

\begin{paperdefinition}{Temporal planning problem}\label{def:planning_problem}\label{def:planning_state}
    A pair of a propositional valuation \( \propof{\planningState} \) and a variable assignment \( \numof{\planningState} \)
    with values in \( \rat = \vardom \) (see \Cref{def:valuation,def:var_asmt})
    forms a \emph{state} \( \planningState \equiv \langle \propof{\planningState},\numof{\planningState} \rangle \).

    A planning problem \( \planningproblem \) is a tuple \( \langle \propositions,\fluents,\actions,\initState,\goalCond\rangle \), where
  \begin{enumerate*}[itemjoin={;\ }, itemjoin*={; and\ }]
    \item \( \propositions \) is a set of state-characterising propositions
    \item \( \fluents \) is a set of state-characterising numeric fluents (i.e. numeric variables)
    \item \( \actions\equiv\{\durAction_i\mid 1 \leq i \leq \numacts\} \): an indexed set of \( N \) actions, 
        each of which is a tuple \( \langle\daStart{\durAction},\daEnd{\durAction},\daInv{\durAction},\daNumInv{\durAction},\daLower{\durAction},\daUpper{\durAction}\rangle \) where
  \begin{enumerate*}[itemjoin={;\ }, itemjoin*={; and\ }, label={(\arabic*)}]
    \item \( \daStart{\durAction},\daEnd{\durAction} \) are start and end \emph{snap actions}
    \item \( \daInv{\durAction} \) is a propositional invariant, 
        a propositional formula defined over the propositions \( \propositions \)
    \item \( \daNumInv{\durAction} \) is a numeric invariant,
         a Boolean expression defined over the fluents \( \fluents \)
    \item \( \daLower{\durAction} \equiv \langle \lowerboundval\durAction, \lowerboundop\durAction \rangle \) 
        and \( \daUpper{\durAction} \equiv \langle \upperboundval\durAction, \upperboundop\durAction \rangle \) 
        are lower and upper bounds on the duration any instance of \( \durAction \) can take,
        expressed as pairs in \( \rat \times \{<, \le\} \).
  \end{enumerate*}
   A snap action \( \snapAction \) is a tuple \( \langle\saPre{\snapAction},\saNumPre{\snapAction},\saAdd{\snapAction},\saDel{\snapAction},\saNumEff{\snapAction}\rangle \) where
  \begin{enumerate*}[itemjoin={;\ }, itemjoin*={; and\ }, label={(\alph*)}]
    \item \( \saPre{\snapAction} \) is its propositional precondition, 
        a formula using propositions \( \propositions \)
    \item \( \saNumPre{\snapAction} \) is its numeric precondition,
        a Boolean expression over fluents \( \fluents \)
    \item \( \saAdd{\snapAction}\subseteq \propositions \) are its positive effects
    \item \( \saDel{\snapAction}\subseteq \propositions \) are its negative effects
    \item \( \saNumEff{\snapAction} \) is a set of variable updates over the fluents \( \fluents \)
  \end{enumerate*}
    \item \( \initState \) is an initial state
    \item \( \goalCond\equiv\langle\propcondof{\goalCond},\numcondof{\goalCond}\rangle \),
        where \( \propcondof{\goalCond} \) is a propositional condition and \( \numcondof{\goalCond} \) is a Boolean expression,
        characterising the goal
  \end{enumerate*}.
  A full example of a temporal planning problem can be found in \Cref{ex:matchcellar_instance_1_1_2_56_5}.

\end{paperdefinition}

\begin{example}\label{ex:matchcellar_instance_1_1_2_56_5}
  The first instance of our running example can be seen in \Cref{fig:example_plan} 
  and is formally declared in \Cref{def:matchcellar_example_instance}.
  It has one fuse \( \fusev_0 \) and one match \( \matchv_0 \).
  It extends MatchCellar~\cite{unsolvableTemporalPlanning} with some numeric fluents initialised by \( \initState \), namely:
  \begin{enumerate*}[itemjoin={;\ }, itemjoin*={; and\ }]
    \item \( \sparef \) for the number of spare fuses
      decremented when one is used by \( \daStart{\mendact(\fusev_0, \matchv_0)} \equiv \daStart{\mendactsrt_{00}} \)
    \item \( \litf \) for the total number of lit matches,
      which is incremented when one is lit
      by \( \daStart{\lightact(\matchv_0)} \equiv \daStart{\lightactsrt_0} \)
    \item \( \durlitlbf(\matchv_0) \), \( \durlitubf(\matchv_0) \) and \( \durmendf(\fusev_0) \),
      which are lower and upper bounds for the time durations of the
      action that keeps a match lit \( \lightactsrt_0 \)
      and the action that mends a fuse \( \mendactsrt_{00} \)
  \end{enumerate*}.
  We write these as intervals \( [\durlitlbf(\matchv_0), \durlitubf(\matchv_0)] \) 
  and \( [\durmendf(\fusev_0), \durmendf(\fusev_0)] \) respectively.

  We denote actions, propositions, and fluents
  using schematic actions, predicates and functions.
  Typically, the first step to solving a planning
  problem is grounding~\cite{helmertConciseFinitedomainRepresentations2009}\label{par:grounding}, which instantiates
  a set of these schematic actions with type-correct objects,
  such as \( \lightact(x) \) with \( \matchv_0 \) for \( x \),
  which yields the start effect \( \litf \coloneq \litf + 1 \) of \( \daStart{\lightact(\matchv_0)} \).
  A naive grounding would enumerate all type-correct permutations,
  which is why good grounding is hard~\cite{correaGroundingPlanningTasks2023} 
  and recently lifted planners try to avoid grounding altogether~\cite{DBLP:conf/aips/CorreaPHF20,DBLP:conf/ijcai/CorreaG24}.

  \( \litf \) is a non-trivially bounded numeric fluent, 
  which appears unbounded,
  because its upper bound is implicit in the total number
  of actions that increment its value and can execute at most once each.
  Static analysis could find the bound.
  Compiling objects into counters could also work~\cite{DBLP:journals/aicom/FuentetajaR15,riddle:etal:16}.
  Lifted planners could also use this knowledge.

  In principle, truly unbounded fluents make PDDL~\cite{fox2003pddl2} more expressive than timed automata
  and planning undecidable~\cite{DBLP:conf/aips/Helmert02}.
  The bound for \( \sparef \) can be inferred using interval arithmetic and abstract interpretation.
  See \Cref{sec:numeric_fluents}.
\end{example}

{\emergencystretch=1.5em
Numeric values assigned to variables in planning are taken to be real-valued~\cite{fox2003pddl2},
which means that they cannot be represented in timed automata.
In fact, continuous planning semantics are based on ODEs and
correspond to reachability in hybrid automata~\cite{fox2003pddl2}.
We present temporal planning with integer fluents.
The benchmarks by \citet{unsolvableTemporalPlanning} include domains
with rational and constant duration constraints and provide versions with durations normalised to integers, 
as well as those that our unverified grounder must normalise.
Normalisation of constant duration constraints with respect to rationals is standard in timed automata~\cite{timedautomata}
but makes it hard to relate the original and normalised values.
Duration constraints in planning can also depend on 
the value of arithmetic operations over fluents
and vice versa,
in which case normalisation is not trivial or even possible~\cite{fox2003pddl2,DBLP:journals/jair/LongF03}.\par}

\begin{paperdefinition}{Effect application}\label{def:effect_application}
    The application of a set of snap actions \( \snapActs \) on a state \( \planningState \)
    is defined as \( \appSas{\snapActs}{\planningState} \equiv \planningState' \equiv \langle \propof{\planningState'}, \numof{\planningState'} \rangle \),
    where \( \propof{\planningState'} = (\propof{\planningState} - \bigcup_{\snapAct{} \in \snapActs} \saDel{\snapAct{}}) \cup \bigcup_{\snapAct{} \in \snapActs} \saAdd{\snapAct{}} \)
    and \( \numof{\planningState'} = \numof{\planningState} \) except for every written fluent \( v \),
    which is assigned \( \eval(\numof{\planningState}, e) \),
    where \( (v,e) \in \bigcup_{\snapAct{} \in \snapActs} \saNumEff{\snapAct{}} \).
    For \( \planningState' \) to be defined,
    no \( \eval(\numof{\planningState}, e) \) is \( \undefVal \),
    where \( (v,e) \in \bigcup_{\snapAct{} \in \snapActs} \saNumEff{\snapAct{}} \).

\end{paperdefinition}

{\emergencystretch=1.5em
\begin{paperdefinition}{Interference}\label{def:mutex_snap}
  Snap actions \( \snapAct{1} \) and \( \snapAct{2} \) are interfering (\( \mutex(\snapAct{1}, \allowbreak \snapAct{2}) \)) if and only if
  \begin{enumerate*}[itemjoin={;\ }, itemjoin*={; or\ }]
    \item \( \saPre{\snapAct{1}} \cap (\saAdd{\snapAct{2}} \cup \saDel{\snapAct{2}}) \neq \emptyset \)
    \item \( \saPre{\snapAct{2}} \cap (\saAdd{\snapAct{1}} \cup \saDel{\snapAct{1}}) \neq \emptyset \)
    \item \( \saAdd{\snapAct{1}} \cap \saDel{\snapAct{2}} \neq \emptyset \)
    \item \( \saAdd{\snapAct{2}} \cap \saDel{\snapAct{1}} \neq \emptyset \)
    \item \( \readsf{\snapAct{1}} \cap \writesf{\snapAct{2}} \neq \emptyset \)
    \item \( \readsf{\snapAct{2}} \cap \writesf{\snapAct{1}} \neq \emptyset \)
    \item \( \writesf{\snapAct{1}} \cap \writesf{\snapAct{2}} \neq \emptyset \)
  \end{enumerate*},
  where \( \readsf{\snapAct{i}} = \readsf{\saNumEff{\snapAct{i}}} \cup \readsf{\saNumPre{\snapAct{i}}} \)
  and \( \writesf{\snapAct{i}} = \writesf{\saNumEff{\snapAct{i}}} \) for \( i \in \{1,2\} \).

\end{paperdefinition}\par}

\begin{paperremark}
If two snap actions \( \snapAct{1} \) and \( \snapAct{2} \) interfere, 
then \( (\propState - \saDel{\snapAct{2}}) \cup \saAdd{\snapAct{1}} \) is not necessarily equivalent to \( (\propState \cup \saAdd{\snapAct{1}}) - \saDel{\snapAct{2}} \), 
and \( \saPre{\snapAct{2}} \subseteq \propState \) does not imply \( \saPre{\snapAct{2}} \subseteq \propState - \saDel{\snapAct{1}} \).
Neither is the sequential application of \( \saNumEff{\snapAct{1}} \) and \( \saNumEff{\snapAct{2}} \)
guaranteed to produce the same result as \( \saNumEff{\snapAct{2}} \) and then \( \saNumEff{\snapAct{1}} \).

Planning abstracts away from \emph{commonsense reasoning} about interfering actions
by prohibiting such interference in the semantics.

Our interference condition is more restrictive than that of \citet{fox2003pddl2}, 
which allows simultaneous writes of the same fluent \( f \) by
additive effects of the form \( f \pluseq e \) or \( f \minuseq e \).
Ours treats additive effects as assignments of the form \( f \coloneqq f + e \) 
or \( f \coloneqq f - e \) and interprets these as interfering.

In contrast, 
transitions in the timed automata formalism (\Cref{def:nta_semantics})
are applied sequentially,
which is only order-independent, 
if a set of such transitions do not interfere.

See \Cref{fig:invalid_actions,fig:matchcellar_interference} for an example of
mutual exclusion preventing unintuitive effects of actions.
\end{paperremark}

\begin{example}\label{ex:matchcellar_instance_1_3_5_2}
  Our next example is MatchCellar with:
  \begin{enumerate*}[label={}, afterlabel={}, itemjoin={,\ }, itemjoin*={, and\ }]
    \item one match \( \matchv_0 \)
    \item \( 3 \) fuses \( \fusev_i \) for \( 0 \le i < 3 \)
    \item a fixed \( 5 \)-second burn duration (\( \durlitlbf(\matchv_0) \), \( \durlitubf(\matchv_0) \))
    \item \( 2 \)-second mend duration \( \durmendf(\fusev_i) \) for \( 0 \le i < 3 \)
    \item no notion of lit matches or spare fuses
  \end{enumerate*}.
  This is equivalent to instance 2 from \citeauthor{unsolvableTemporalPlanning}'s benchmarks~\cite{unsolvableTemporalPlanning}.
  It is formally defined in \Cref{def:matchcellar_unsolvable_instance}.
  It demonstrates the complexity of 
  manually inspecting an encoding from planning to timed automata:
  \begin{enumerate*}[itemjoin={;\ }, itemjoin*={; and\ }, label=(\arabic*)]
    \item each of the \( 3 \) mend actions \( \mendact(\fusev_i, \matchv_0) \equiv \mendactsrt_{i0} \), 
      where \( 0 \le i < 3 \), 
      depends on the hand being free (written \( \handfreep \))
      \( \handfreep \in \saPre{\daStart{\mendactsrt_{i0}}} \), 
      occupies it \( \handfreep \in \saDel{\daStart{\mendactsrt_{i0}}} \), 
      and frees it again \( \handfreep \in \saAdd{\daEnd{\mendactsrt_{i0}}} \),
      which makes their starts interfere with starts and ends of other \( \mendactsrt_{j0} \), 
      where \( 0 \le j < 3 \),
      and not executable in the same instant (this needs \( 27 \) clock constraints to encode as we and \citet{giganteDecidabilityComplexity2022} do)
    \item the durative invariant that the match must be burning (\( \lightp(\matchv_0) \)) throughout 
      a mend action (\( \mendact(\fusev_i, \matchv_0) \))
      is enforced by a lock counter,
      by preventing deletions of the proposition (assignments to 0) while the counter is above 0
      (see \Cref{def:action_automata})
  \end{enumerate*}.

\end{example}

\begin{paperdefinition}{Plan}\label{def:plan}\label{def:htps}
    A plan is a list of tuples \( \langle \langle\durAct{1},t_1,d_1\rangle,\dots,\langle \durAct{n},t_n,d_n\rangle \rangle \), 
    where, for \( 1\leq i\leq n \), \( \durAct{i}\in \actions \) is an action, 
    \( t_i\in \rat_{\ge 0} \) and \( d_i\in \rat_{\ge 0} \) are rational numbers, 
    to which we refer as the starting time point and the duration, respectively.

    For a plan \( \plan \), we call a sorted sequence \( \happTimePoint_0,\dots,\happTimePoint_m \) 
    of the set of rational numbers \( \{t\mid\langle a,t,d\rangle\in\plan\} \cup \{t+d\mid\langle a,t,d\rangle\in\plan\} \) 
    the happening time points (htps) of the plan, 
    and we denote it by \( \happTimePoints(\plan) \).

\end{paperdefinition}

\begin{paperdefinition}{Valid state sequence}\label{def:valid_state_sequence}\label{def:zero_eps_sep}\label{def:induced_parallel_plan}
    For \( t\in\rat_{\ge 0} \) and a plan \( \plan \), 
    let the \emph{induced parallel plan} be \( \recSetActs_\happTimePoint\equiv\{\daStart{\durAction}\mid\langle\durAction,\happTimePoint,d\rangle\in\plan\}\cup\{\daEnd{\durAction}\mid\langle\durAction,\happTimePoint-d,d\rangle\in\plan\} \) 
    and \( \recSetInvs_\happTimePoint\equiv\{\durAction\mid\langle\durAction,\happTimePoint',d\rangle\in\plan\wedge \happTimePoint' <\happTimePoint\le\happTimePoint'+d\} \).
    Also, let \( \happTimePoint_0,\dots,\happTimePoint_m \) be the sorted happening time points of \( \plan \).

    For a sequence of states \( \planningState_0,\dots,\planningState_{m+1} \), 
    we say the sequence of states is valid with respect to a plan \( \plan \) if and only if, 
    for every happening time point \( \happTimePoint_i \) of \( \plan \), we have:
    \begin{enumerate*}[itemjoin={;\ }, itemjoin*={; and\ }]
        \item \( \pEntails{\propof{(\planningState_i)}}{\daInv{\durAction}} \)
            and \( \bEntails{\numof{(\planningState_i)}}{\daNumInv{\durAction}} \),
            for every \( \durAction\in\recSetInvs_{\happTimePoint_i} \)
        \item \( \pEntails{\propof{(\planningState_i)}}{\saPre{\snapAction}} \)
            and \( \bEntails{\numof{(\planningState_i)}}{\saNumPre{\snapAction}} \),
            for every \( \snapAction\in\recSetActs_{\happTimePoint_i} \)
        \item for every \( \happTimePoint_j \) in \( \happTimePoints(\plan) \), 
            \( \recSetActs_{\happTimePoint_i}\cup\recSetActs_{\happTimePoint_j} \) is pairwise non-interfering
            when \( \happTimePoint_i = \happTimePoint_j \) (\( 0 \)-separation) 
            or \( 0 < |\happTimePoint_i - \happTimePoint_j| < \epsilon \) (\( \epsilon \)-separation)
        \item \( \planningState_{i+1} = \appSas{\recSetActs_{\happTimePoint_i}}{\planningState_i} \)
    \end{enumerate*}.
    Here, \( 0 \)-separation means that no snap actions at the same time point can interfere 
    and \( \epsilon \)-separation means that the minimum separation between two interfering snap actions is at least some \( \epsilon \),
    which is fixed to define the semantics.
    If it is set to zero, this is equivalent to \( 0 \)-separation between actions.
    (See \Cref{fig:example_timeline,fig:matchcellar_pictogram} for example state sequences.)

\end{paperdefinition}

\begin{paperdefinition}{Valid plan}\label{def:valid_plan}\label{def:duration_sem}\label{def:plan_validity_rec}\label{def:no_self_overlap}
   Let, for a plan \( \plan \), 
   \( \noselfoverlap(\plan) \) denote that for all distinct \( i \) and \( j \),
    where \( \langle \planAction_i, t_i, d_i \rangle \) and \( \langle \planAction_j, t_j, d_j \rangle \)
    are the \( i \)-th and \( j \)-th elements of \( \plan \),
    there is either no \emph{overlap}, i.e.\ \( \lnot(t_i \le t_j \le t_i + d_i) \),
    in the closed intervals \( [t_i, t_i + d_i] \) and \( [t_j, t_j + d_j] \),
    or \( \planAction_i \ne \planAction_j \).

  \emergencystretch=1.5em %
  \( \plan \) is a valid plan for a problem \( \planningproblem \) if and only if:
  \begin{enumerate*}[itemjoin={;\ }, itemjoin*={; and\ }]
  \item there is a state sequence \( \planningState_1,\allowbreak\dots,\allowbreak\planningState_{m+1} \) such that 
    \( \initState,\allowbreak\planningState_1,\allowbreak\dots,\allowbreak\planningState_{m+1} \) is valid with respect to \( \plan \)
  \item \( \pEntails{\propof{(\planningState_{m+1})}}{\propcondof{\goalCond}} \)
        and \( \bEntails{\numof{(\planningState_{m+1})}}{\numcondof{\goalCond}} \)
  \item any tuple \( \langle \planAction_i, \happTimePoint_i, d_i \rangle \in \plan \) satisfies duration constraints 
    \( \lowerboundval\durAction \lowerboundop\durAction d_i \upperboundop\durAction \upperboundval\durAction \) of the durative action \( \durAction \equiv \planAction_i \), 
    where \( \daLower{\durAction} \equiv \langle \lowerboundval\durAction, \lowerboundop\durAction \rangle \),
    \( \daUpper{\durAction} \equiv \langle \upperboundval\durAction, \upperboundop\durAction \rangle \)
    and \( \{ \lowerboundop\durAction, \upperboundop\durAction \} \subseteq \{ <, \le \} \)
  \item \( \noselfoverlap(\plan) \) holds
  \end{enumerate*}.
  We write \( \dcsatf(\durAction, d) \) to indicate the satisfaction of duration constraints by a duration \( d \).

\end{paperdefinition}

\begin{paperremark}
\Cref{fig:example_plan} shows a valid plan for \Cref{ex:matchcellar_instance_1_1_2_56_5}.
For an informal illustration of mutual exclusivity in relation to valid state sequences, see \Cref{fig:timeline_pictogram,fig:matchcellar_pictogram,ex:invalid_instant,fig:matchcellar_run}.
\end{paperremark}

\begin{paperremark}
\noindent Note: we add the \emph{no self-overlap} condition to plans for validity because \citet{giganteDecidabilityComplexity2022} showed that, 
without it, planning is undecidable under \( 0 \)-separation
(which is what \citet{aaai2022} use in their separation constraint)
and EXPSPACE-complete under \( \epsilon \)-separation,
vs. PSPACE-complete without self-overlap.
\end{paperremark}

\begin{figure}[t]
  \centering
  \begin{tikzpicture}
  \def\gap{0.6}       %
  \def\sepgap{0.9}    %
  \renewcommand{\picheight}{1.05cm}
  \tikzset{state scale=0.72, diagonal radius=1.0}
  \pic[dark, spares=2, lit=0] (s0) at (0,0)                                {state={unlit}{free}{broken}};
  \draw[gray] ($(s0-box.north east)+(\sepgap/2,0.15)$) -- ($(s0-box.south east)+(\sepgap/2,-0.5)$);
  \pic[dark, spares=2, lit=0] (s1) at ($(s0-box.south east)+(\sepgap,0)$)  {state={unlit}{free}{broken}};
  \pic[spares=1, lit=1] (s2) at ($(s1-box.south east)+(\gap,0)$)     {state={lit}{busy}{broken}};
  \pic[dark, spares=1, lit=1] (s3) at ($(s2-box.south east)+(\gap,0)$)     {state={burnt}{free}{mended}};
  \draw[->] (s1-box.east) -- (s2-box.west);
  \draw[->] (s2-box.east) -- (s3-box.west);
  \node[below=1pt of s0-box.south] {\( \initState \)};
  \def\ytop{-0.45}  \def\ybot{-2.05}   %
  \coordinate (tl) at ($(s1-box.south west)+(0,\ytop)$);
  \coordinate (tr) at ($(s3-box.south east)+(0,\ytop)$);
  \coordinate (bl) at ($(s1-box.south west)+(0,\ybot)$);
  \coordinate (br) at ($(s3-box.south east)+(0,\ybot)$);
  \draw[->] (tl) -- (tr) node[below left, font=\scriptsize, inner sep=1.5pt] {time (s)};
  \draw[->] (bl) -- (br) node[below left, font=\scriptsize, inner sep=1.5pt, align=right] {\( \happTimePoints(\plan) \)\\[4pt]\( \recSetActs_{\happTimePoint} \)};
  \foreach \b/\tm/\hp/\sa/\sb in {s1/0/\happTimePoint_0/\daStart{\lightactsrt_0}/\daStart{\mendactsrt_{00}},
                                 s2/5/\happTimePoint_1/\daEnd{\lightactsrt_0}/\daEnd{\mendactsrt_{00}}} {
    \coordinate (m\b) at ($(\b-box.south)+(0,\ytop)$);
    \draw[black!60, dashed] ($(m\b)+(0,-0.08)$) -- ($(\b-box.south)+(0,\ybot+0.08)$);
    \draw (m\b) ++(0,0.08) -- ++(0,-0.16) node[above=0.16, font=\scriptsize] {\( \tm \)};
    \draw ($(\b-box.south)+(0,\ybot)$) ++(0,0.08) -- ++(0,-0.16)
      node[below=0.02, font=\scriptsize, align=center, inner sep=1pt] {\( \hp \)\\[-1pt]\( \sa \)\\[-2pt]\( \sb \)};
  }
  \draw ($(s3-box.south)+(0,\ytop)$) ++(0,0.08) -- ++(0,-0.16) node[above=0.16, font=\scriptsize] {\( 7 \)};
  \node[font=\scriptsize, anchor=west, inner sep=1pt] at ($(tl)!0.5!(bl)$) {\( \plan \)};
  \coordinate (ms6) at ($(ms2)!0.5!($(s3-box.south)+(0,\ytop)$)$);
  \draw (ms6) ++(0,0.08) -- ++(0,-0.16) node[above=0.16, font=\scriptsize] {\( 6 \)};
  \draw[black!60, dotted] ($(ms6)+(0,-0.08)$) -- ($(ms6)+(0,-0.35)$);
  \draw[draw=\nlmacol!60, dashed, pattern=north west lines, pattern color=\nlmacol!60]
    ($(ms2)+(0,-0.35)$) rectangle ($(ms6)+(0,-0.75)$);
  \draw[action box=\nlmacol] ($(ms1)+(0,-0.35)$) rectangle ($(ms2)+(0,-0.75)$)
    node[pos=.5, font=\scriptsize] {\( \lightactsrt_0 \)};
  \draw[action box=\nmfacol] ($(ms1)+(0,-0.9)$) rectangle ($(ms2)+(0,-1.3)$)
    node[pos=.5, font=\scriptsize] {\( \mendactsrt_{00} \)};
\end{tikzpicture}
  \caption[Initial state and a valid plan of the MatchCellar instance]{
    Left: \( \initState \): 
    \begin{enumerate*}[itemjoin={,\ }, itemjoin*={, and\ }]
      \item \( \matchv_0 \) is not lit
      \item \( \fusev_0 \) is not mended
      \item the hand is free \( \handfreep = \truetruth \)
      \item \( 2 \) spare fuses \( \sparef \)
      \item \( 0 \) lit matches \( \litf \)
    \end{enumerate*}.
    Right: A valid state sequence showing states at time points before applying effects.
    Bottom: A valid plan \( \plan \) progresses the state sequence from \( \initState \) 
    to a state satisfying \( \goalCond \) after \( 5 \)~s.
    \( \lightactsrt_0 \) (lighting a match) and \( \mendactsrt_{00} \) (mending a fuse) are both executed 
    for \( 5 \)~s at \( \happTimePoint_0 = 0 \).
    A timeline shows alternative valid durations for \( \lightactsrt_0 \) (e.g. the interval \([5, 6]\)).
  }
  \label{fig:example_plan}
\end{figure}

\paragraph{ITPs and Isabelle/HOL}
Isabelle/HOL~\cite{IsabelleHOLRef} is an ITP for the development and checking of natural deduction proofs in Higher-Order Logic (HOL).
Isabelle combines automated proof checking with an interactive proof development environment.

Definitions and statements about those definitions which are to be proven can be declared by the user in a formal language that is mainly based on Church's simple type theory~\cite{DBLP:journals/jsyml/Church40,DBLP:conf/colog/Paulson88}.
Isabelle/HOL's language, however, has a number of features that are either added to or built upon simple type theory, like locales~\cite{LocalesKammueler}, type classes, and facilities to define and reason about algebraic datatypes~\cite{datatypesIsabelle} and recursive functions~\cite{functionPackageIsabelle}.
Proofs for claims in Isabelle/HOL are either fully automatically constructed by the ITP or, if that fails, the user provides higher-level arguments for the proofs and the ITP tries to bridge the gaps in those arguments.
The process is continued until a full formal proof, down to the logical axioms, is constructed.

ITPs have been used in many applications, ranging from the verification of low-level software, like the seL4~\cite{klein2009sel4} OS kernel; 
algorithms for planning~\cite{verifiedSATPlan}, model checking~\cite{10.1007/978-3-030-29662-9_14}, and SAT solving~\cite{verifiedSATSolver}; 
all the way to formally verifying proofs in pure mathematics~%
\cite{avigadFormalMaths,DBLP:conf/ascm/Gonthier07,10.1007/978-3-642-39634-2_14,10.1145/3372885.3373830,DBLP:conf/stoc/BlanchardBDFFGH26,halesFormalProofKepler,buzzard2025lean}
---here they facilitate large collaborative efforts.

\section{Encoding}\label{sec:reduction}
As shown in \Cref{ex:matchcellar_instance_1_1_2_56_5,ex:matchcellar_instance_1_3_5_2}, 
even features of planning that can be modelled are hard to model in TA networks.
E.g.
\begin{enumerate*}[itemjoin={,\ }, itemjoin*={, and\ }]
  \item non-trivially bounded (or unbounded) fluents (\Cref{ex:matchcellar_instance_1_1_2_56_5,lst:pddl_custom})
  \item mutex constraints (\Cref{ex:matchcellar_instance_1_3_5_2})
  \item durative invariants over numeric fluents (which we only encode as an over-approximation)
  \item durations that depend on fluents (\Cref{lst:pddl_custom}; see also~\citet{DBLP:journals/jair/LongF03})
  \item simultaneous snap action execution (\Cref{def:valid_state_sequence} and urgent locations in \Cref{def:action_automata})
  \item no self-overlap (\Cref{def:no_self_overlap}) which is implicit in the structure of the automaton for each action (\Cref{def:action_automata})
  \item no actions running at the end (\Cref{def:int_vars})
\end{enumerate*}.

The correctness of these features is hard to assert by manual inspection---%
in particular the correctness of the inferred bounds in \Cref{sec:bound_inference}.
Therefore, we argue for placing trust in Isabelle,
as opposed to the code itself;
we do this for the concepts in this section by formalising \Cref{sec:proof} in Isabelle
with respect to \citeauthor{10.1007/978-3-030-29662-9_14}'s semantics~\cite{10.1007/978-3-030-29662-9_14,wimmerTimedAutomataAFP2016} 
and a modified version of \citeauthor{aaai2022}'s semantics~\cite{aaai2022} to add numeric fluents.

This section describes a practical encoding of temporal planning problems into a 
network of timed automata based on the work of \citet{temporalPlanningRefinementMC} and \citet{giganteDecidabilityComplexity2022}. 
\citet{temporalPlanningRefinementMC} use integer variables 
and \citet{bogomolovPlanningModelChecking2014} a global lock automaton
to allow one snap action per \( \epsilon \) interval.
This enforces mutex constraints (\Cref{def:mutex_snap}),
but prevents non-interfering snap actions from executing in the same instant, 
unlike ours and \citeauthor{giganteDecidabilityComplexity2022}'s~\cite{giganteDecidabilityComplexity2022},
which both use clock constraints.
Hence, some transitions must query a number of clocks,
which scales linearly with the number of actions (see \Cref{ex:matchcellar_instance_1_3_5_2,def:action_automata}).
Our approach thus checks and proves a more general notion of unsolvability than that of \citet{temporalPlanningRefinementMC},
and we provide a formal proof of correctness.

Our encoding (detailed below) has one automaton per action, and one main automaton.
The main automaton encodes the planning problem by initialising a propositional state
and checking the satisfaction of a goal,
and each action automaton encodes a unique durative action.
An action automaton checks that a snap action's preconditions have been satisfied, 
before transitioning to the corresponding urgent location and resetting a unique clock. 
Moreover, using a set of pre-compiled clock constraints,
an action automaton checks that no interfering snap actions have done the same in the same instant.
Like \citet{temporalPlanningRefinementMC}, we use a simple counter assigned to each propositional variable to ensure that no invariants are violated when an action is running. 
Effects are applied to a propositional state encoded using a set of binary variables.
When any action automaton is in an urgent location (\Cref{def:urgency_sem}), the instant must correspond to a happening
time point. 

\begin{paperdefinition}{Integer variables \( \vars \)}\label{def:int_vars}
    We obtain a set of variables \( \vars \) by 
    adding the following variables and bounds to the set of fluents \( \fluents \) in the planning problem \( \planningproblem \):
    \begin{enumerate*}[itemjoin={;\ }, itemjoin*={; and\ }]
        \item \label{ite:prop_var} \( \vp{p} \) with bounds \( [0, 1] \) (\emph{binary variable}) for every proposition \( p \in \propositions \), 
            which records \( p \)'s truth value encoded as \( 0 \) or \( 1 \)
        \item \label{ite:prop_lock} \( \lp{p} \) with bounds \( [0, N] \) (\emph{lock counter}) for every proposition \( p \in \propositions \), 
            which records the number of active actions
            that require \( p \) to be satisfied as invariant
        \item \( \actsactive \) with bounds \( [0, N] \) (\emph{active action counter}) records the number of active actions
        \item \( \pl \) with bounds \( [0, 2] \) (\emph{planning state}) records one of:
        \( 0 \) (not started),
        \( 1 \) (planning),
        \( 2 \) (done)
    \end{enumerate*}.
    These must be fresh with respect to \( \fluents \).
    Note: by construction, no transition assigning a value to \( \pl \) or \( \vp{p} \) for any \( p \in \propositions \) outside of their bounds is possible.
    For the bounds of \( \lp{p} \) and \( \actsactive \)
    we refer to \Cref{thm:main_theorem},
    and for those of \( \fluents \) we refer to \Cref{sec:bound_inference}.

\end{paperdefinition}

\begin{paperdefinition}{Action clocks}\label{def:enc_clocks}
    We assign two clocks per durative action \( \durAction \):
    \begin{enumerate*}[itemjoin={;\ }, itemjoin*={; and\ }]
        \item a \emph{start clock} \( \clk{\daStart{\durAction}} \) for the time since \( \daStart\durAction \)
        \item an \emph{end clock} \( \clk{\daEnd{\durAction}} \) for the time since \( \daEnd\durAction \)
    \end{enumerate*}
    has happened.
    When the kind (start/end) of a snap action \( \snapAct{1} \) is unspecified, we use \( \clk{\snapAct{1}} \) to denote such a clock.

\end{paperdefinition}

\begin{paperdefinition}{Main Automaton}\label{def:main_automaton}
\emergencystretch=3em %
The conditions and effects are:
\begin{enumerate*}[itemjoin={;\ }, itemjoin*={; and\ }]
    \item \( \initedge \equiv \langle \emptyset, \allowbreak \emptyset, \allowbreak \{ (\pl \coloneq 1) \} \cup \{ (\vp{p} \coloneq 1) \mid p \in \propof{\initState} \} \cup \{ (f \coloneq \numof{\initState}(f)) \mid f \in \fluents \}, \emptyset \rangle \)
    \item \( \goaledge \equiv \langle \{ (\actsactive = 0)? \} \cup \{ (\vp{p} = 1)? \mid p \in \propcondof{\goalCond} \} \cup \{ \numcondof{\goalCond} \}, \emptyset, \{ (\pl \coloneq 2) \}, \emptyset \rangle \)
    \item \( \loopedge \equiv \langle \emptyset, \emptyset, \emptyset, \emptyset \rangle \)
\end{enumerate*}.
An illustration of this automaton can be seen in \Cref{fig:action_automata}.

\end{paperdefinition}

\begin{figure}
    \centering
    \begin{tikzpicture}[
            every state/.style={
                draw,
                rectangle,
                minimum size=0pt,
                inner ysep=1pt,
                text height=1.8ex,
                text depth=0.8ex
            },
            >=stealth',
            shorten >=1pt
        ]
        \begin{scope}
            \node[state, initial] (init) at (0,0.6) {\initloc!};
            \node[state] (plan) at (1.8, 0.6) {\planningloc};
            \node[state, double] (goal) at (1.8, -0.6) {\goalloc};  %
        \end{scope}
        \begin{scope}[auto]
            \draw[->] (init) -- (plan) node[midway, above, anchor=south] {\( \initedge \)};
            \draw[->] (plan) -- (goal) node[midway, right, anchor=west] {\( \goaledge \)};
            \draw[->] (goal) edge [loop right] node {\( \loopedge \)} ();
        \end{scope}

        \begin{scope}[shift={(4.9,0)}]
            \node[state, initial] (inactive) at (0,0) {\olc};
            \node[state] (starting) at (2.0,0.6) {\slc!};
            \node[state] (running) at (4.0, 0) {\rlc};
            \node[state] (ending) at (2.0,-0.6) {\elc!};
        \end{scope}
        \begin{scope}[auto]
            \draw[->] (inactive) |- (starting) node[midway, above, anchor=south] {\( \sse_\durAction \)};
            \draw[->] (starting) -| (running) node[midway, above, anchor=south] {\( \see_\durAction \)};
            \draw[->] (running) |- (ending) node[midway, below, anchor=north] {\( \ese_\durAction \)};
            \draw[->] (ending) -| (inactive) node[midway, below, anchor=north] {\( \eee_\durAction \)};
            \draw[->] (starting) -- (ending) node[midway, anchor=west] {\( \inste_\durAction \)};
        \end{scope}
    \end{tikzpicture}
    \caption{The main automaton \( \autonet_M \) and an automaton \( \autonet_\durAction \) for some action \( \durAction \).}
    \label{fig:action_automata}
\end{figure}

\begin{paperdefinition}{Action Automata}\label{def:action_automata}
    \Cref{fig:action_automata} shows the structure of an automaton \( \autonet_\durAction \), which encodes an action \( \durAction \in \actions \).
    The structure of four locations to record four states of actions is inspired by the work of \citet{temporalPlanningRefinementMC}.
    First, we characterise a few abbreviations and functions
    \begin{enumerate*}[itemjoin={;\ }, itemjoin*={; and\ }]
    \item \( \mutexenc(\snapAct 1) \) checks clocks to ensure that no scheduled snap actions interfere:
    \( \{ (0 < \clk{\snapAct{2}})?, (\epsilon \le \clk{\snapAct{2}})? \mid \snapAct{2} \in \mathit{intf}(\snapAct{1}) \} \),
    where \( \mathit{intf}(\snapAct 1) \) is the set of all snap actions \( \snapAct 2 \) in the problem, 
    such that \( \mutex(\snapAct 1, \snapAct 2) \);
    when \(\epsilon = 0\), this is \( 0 \)-separation semantics (see \Cref{def:zero_eps_sep,def:mutex_snap} and \citet{giganteDecidabilityComplexity2022})
    
    \item \( \durenc(\durAction) \) is satisfied when \( \dcsatf(\durAction,\Clk(\clk{\daStart{\durAction}})) \) is 
    (see \Cref{def:duration_sem}): 
    i.e. 
    \( \{ (\lowerboundval\durAction \lowerboundop\durAction \clk{\daStart{\durAction}})?, (\clk{\daStart{\durAction}} \upperboundop\durAction \upperboundval\durAction)? \} \) 
    where \( \daLower{\durAction} \equiv \langle \lowerboundval\durAction, \lowerboundop\durAction \rangle \) 
    and \( \daUpper{\durAction} \equiv \langle \upperboundval\durAction, \upperboundop\durAction \rangle \)
    
    \item \( \effenc(\snapAct 1) \) encodes the effects of a snap action \( \snapAct 1 \) on propositional variables;
    it is defined as the union of
    \( \{ (\vp{p} \coloneq 0) \mid p \in \saDel{\snapAct{1}} \land p \notin \saAdd{\snapAct{1}} \} \)
    and
    \( \{ (\vp{p} \coloneq 1) \mid p \in \saAdd{\snapAct{1}} \} \)
    
    \item \( \preenc(\snapAct 1) \) encodes the preconditions as a Boolean condition:
    \( \{ (\vp{p} = 1)? \mid p \in \saPre{\snapAct{1}} \} \cup \{ \saNumPre{\snapAct{1}} \} \)
    
    \item \( \invenc(\snapAct 1) \) checks that no effect of \( \snapAct{1} \) invalidates an active invariant: 
    \( \{ (\lp{p} = 0)? \mid p \in \saDel{\snapAct{1}} - \saAdd{\snapAct{1}} \} \)
    \end{enumerate*}.
    The transitions' relevant conditions and effects are: 
    \begin{enumerate*}[itemjoin={;\ }, itemjoin*={; and\ }]
        \item \( \sse_\durAction \equiv \langle
            \preenc(\daStart{\durAction}) \cup \invenc(\daStart{\durAction}), \allowbreak
            \mutexenc(\daStart{\durAction}), \allowbreak
            \effenc(\daStart{\durAction}) \cup \saNumEff{(\daStart{\durAction})}
                \cup \{ ( \actsactive \coloneq \actsactive + 1 ) \}, \allowbreak
            \{ \clk{\daStart{\durAction}} \} \rangle \)
        \item \( \see_\durAction \equiv \langle
            \{ (\vp{p} = 1)? \mid p \in \daInv{\durAction} \} \cup \{ \daNumInv{\durAction} \}, \allowbreak
            \emptyset, \allowbreak
            \{ (\lp{p} \coloneq \lp{p} + 1) \mid p \in \daInv{\durAction} \}, \allowbreak
            \emptyset \rangle \)
        \item \( \ese_\durAction \equiv \langle
            \{ \daNumInv{\durAction} \}, \allowbreak
            \mutexenc(\daEnd{\durAction}) \cup \durenc(\durAction), \allowbreak
            \{ (\lp{p} \coloneq \lp{p} - 1) \mid p \in \daInv{\durAction} \}, \allowbreak
            \{ \clk{\daEnd{\durAction}} \} \rangle \)
        \item \( \eee_\durAction \equiv \langle
            \preenc(\daEnd{\durAction}) \cup \invenc(\daEnd{\durAction}), \allowbreak
            \emptyset, \allowbreak
            \effenc(\daEnd{\durAction}) \cup \saNumEff{(\daEnd{\durAction})}
                \cup \{ (\actsactive \coloneq \actsactive - 1) \}, \allowbreak
            \emptyset \rangle \)
        \item \( \inste_\durAction \equiv \langle
            \emptyset, \allowbreak
            \mutexenc(\daEnd{\durAction}) \cup \durenc(\durAction), \allowbreak
            \emptyset, \allowbreak
            \{ \clk{\daEnd{\durAction}} \} \rangle \)
    \end{enumerate*}.
    Moreover, every transition in \( \autonet_\durAction \) is conditioned on \( (\pl = 1) \), 
    which ensures that no automaton \( \autonet_\durAction \)
    starts a simulated execution of an action \( \durAction \in \actions \) 
    before all variables have been initialised by \( \initedge \), 
    which applies \( (\pl \coloneq 1) \).
    \( (\pl = 1) \) also prevents further changes
    once \( \goaledge \) applies \( \pl \coloneq 2 \).

\end{paperdefinition}

\begin{paperremark}
This encoding is not exact, 
but it is sound, 
which is needed to check unsolvability.
It only checks the numeric invariant satisfaction at the start and end of the action,
which means that runs simulating invalid plans could be admitted, 
but none that simulate a valid plan will be rejected.

On the other hand, 
propositional invariants \( p \) in \( \daInv{\durAction} \) are exactly modelled
by \( \lp{p} \) on the condition that they are conjunctions of propositions.
\end{paperremark}

\begin{paperremark}
The defining features of the encoding in \Cref{def:action_automata} for an action \( \durAction \) 
are the encoding of the conditions of \( \daStart{\durAction} \) and \( \daEnd{\durAction} \)
into a location and two transitions each.
The \emph{urgency} (\Cref{def:urgency_sem}) of (\( \slc_{\durAction}! \)) and (\( \elc_{\durAction}! \))
ensures that all transitions necessary to simulate the execution of a snap action are applied in the same instant. 
Encoding a snap action as two transitions
is necessary to capture the semantics of \( \daInv{\durAct 2} \) conditions in a valid state sequence (\Cref{def:valid_state_sequence}). 
See \Cref{ex:formalism,ex:valid_instant,fig:timeline_pictogram,fig:matchcellar_pictogram,fig:construct_run,fig:matchcellar_run} for examples.
Every update of the form \( (\vp{p} \coloneq 0) \) in an \( \eee_{\durAction} \) or \( \sse_{\durAction} \) edge is accompanied by a \( (\lp{p} = 0)? \) check.
This check can require \( \lp{p} \) for some \( p \in \propositions \) to be decremented by some \( \ese_{\durAct 2} \) for which \( p \in \daInv{\durAct 2} \).
In other cases \( \see_{\durAct 2} \) can only increment \( \lp{p} \) after the \( (\lp{p} = 0)? \) check has occurred in some \( \sse_{\durAct 3} \) or \( \eee_{\durAct 3} \).

Most importantly, 
by adding \( \inste_\durAction \)
and using our relaxed definition of a valid plan (\Cref{def:valid_plan}), 
which admits durations \( d = 0 \),
we use durative actions to encode instantaneous snap actions.
\( \inste_\durAction \) is applicable only if \( \durAction \in \actions \) has a lower duration bound of \( \daLower{\durAction} \equiv \langle 0, \le \rangle \), 
because \( \slc_\durAction \) is urgent (\Cref{def:urgency_sem}) and any transition to \( \slc_\durAction \) resets the corresponding \( \clk{\daStart{\durAction}} \) to \( 0 \).
\end{paperremark}

\begin{paperdefinition}{Network for a planning problem}\label{def:network_for_problem}\label{def:initial_config}\label{enc:goal}
  For a planning problem \( \planningproblem \), the list of automata \( \autolist{\autonet} \) is \( \langle \autonet_M, \autonet_{\durAct 1}, \dots, \autonet_{\durAct \numacts} \rangle \).

  \par\noindent\emph{Initial configuration.}
The initial configuration \( \bLvc{_I} \) is characterised as follows.
The initial locations are
\( \Loc_I \equiv \langle \initloc, \olc_{\durAct 1}, \ldots, \olc_{\durAct \numacts} \rangle \).
They are marked with arrows in \Cref{fig:action_automata}.
The variable assignment and clock valuation are initialised to satisfy \( \Var_I(z) = \dlo{z} \) for any variable \( z \in \vars \) 
and \( \Clk_I(x) = 0 \) for any clock \( x \in \clocks \).

  \par\noindent\emph{Acceptance condition.}
    See \Cref{fig:action_automata}.
    The acceptance condition is the reachability of a configuration in which \( \autonet_M\) is in the \( \goalloc \) location, 
    i.e.\ \( \autonet \vDash \mathrm{EF}(\mathit{loc}(\autonet_M) = \goalloc) \)~\cite{113766,DBLP:reference/mc/ClarkeHV18,DBLP:conf/birthday/BouyerLMOW17}.
    \begin{remark}
        \( \mathrm{EF} \phi \) means that some run eventually reaches a configuration satisfying \( \phi \),
        which in this case asserts that \( \Loc_i \equiv \langle \goalloc, \dots \rangle \),
        where \( \bLvc{_i} \) is the \( i \)-th configuration of some run. 
    \end{remark}

\end{paperdefinition}

\subsection{Numeric Fluent Bounds}\label{sec:bound_inference}\label{sec:numeric_fluents}

We encode fluents directly as a subset of variables of the timed automaton network.
Note that the planning semantics
use unbounded \( \rat \)-ranged variables,
and to guarantee termination and a finite reachable state space of discrete states,
the implementation restricts variables and operations to integers with finite bounds.
We reject planning problems that contain non-integer constants,
except in rational duration constraints,
which can be evaluated exactly independent of the state (\Cref{def:planning_state}).
This is sufficient to precisely encode propositional states,
and addition, subtraction and multiplication involving 
numeric fluents.
However, integer division is not exact, so we also reject planning problems 
involving integer division for now.
Our encoding can, however, be extended to admit problems involving only
exact integer division.

For unbounded numeric fluents in the planning problem, we infer approximate, sound bounds in the automaton.
We use widening with thresholds~\cite{halbwachsVerificationRealTimeSystems1997,Blanchet2002},
a technique from abstract interpretation~\cite{10.1145/512950.512973},
with a carefully deduced set of thresholds.
The inference starts with the interval of each fluent \( f \) set to 
\( [\numof{\initState}(f), \allowbreak \numof{\initState}(f)] \).
In each iteration, every snap action is applied to the current intervals:
its numeric precondition first restricts the intervals of the fluents it depends on.
Then its numeric effects are evaluated with interval arithmetic.
The current intervals are extended to cover the results.
If an interval grows, 
its new bound is widened to the nearest threshold beyond it,
or to \( \pm\infty \) if there is none.
The thresholds are collected from the initial values,
the constants in the actions' preconditions and effects,
and the constants and initial values of the expressions compared with a fluent in a precondition,
shifted by the increment that an effect of the same action applies to that fluent.
The iteration stops once no interval grows,
which happens after finitely many iterations limited by the set of thresholds (see \Cref{prop:infer_bounds_termination}).
See also \Cref{sec:threshold_inference} for the pseudocode and \Cref{lem:variable_bounds_respected} for how these bounds are used 
in the proof of \Cref{thm:main_theorem}.

\section{Correctness}\label{sec:proof}
\Contribref{contrib:proof} is a detailed proof that the existence of a plan is sufficient to demonstrate the existence of a valid run in the encoding of temporal planning as timed automata in \Cref{sec:reduction}.
Proving this is what is needed to show the encoding is useful for certifying unsolvability.
In addition to being formally verified, we believe our proof's level of detail, which is higher than all those by earlier authors, and its auxiliary concepts and lemmas will be of general interest for others working on encoding temporal planning into automata.

The main idea is that we construct a run which simulates the remainder of an induced parallel plan (see \Cref{def:induced_parallel_plan}).
Then we combine this with a run beginning at the initial configuration \( \bLvc{_I} \)
and append a final transition into an accepting state \( \bLvc{_G} \).
We first define conditions under which each part of such a run can be created.
\begin{paperdefinition}{Encoded states}\label{def:enc_states}\label{def:enc_active_time}\label{def:enc_active_state}
    \( \runat(\durAct{},\happTimePoint_i,\triangleleft) \),
    where \( \triangleleft \in \{\le, <\} \) encodes whether an action \( \durAct{} \) is running 
    at or after time point \( \happTimePoint_i \):
    \( \runat(\durAction, \allowbreak \happTimePoint_i, \allowbreak \triangleleft) \) is true if and only if there exist \( s \) and \( d \) such that \( \langle \durAct{}, s, d\rangle \in \plan \) and \( s \triangleleft \happTimePoint_i \) and \( \lnot (s + d \triangleleft \happTimePoint_i) \).

    \( \timesinceat(\snapAct{1},\happTimePoint_i,\triangleleft) \) is the time between \( \happTimePoint_i \) and 
    the last time a snap action \( \snapAct{1} \) has been applied, including \( \happTimePoint_i \):
    \( \timesinceat(\snapAct{1},\happTimePoint_i,\triangleleft) \) is \( \happTimePoint_i - s\), 
    where \( s \) is the greatest time point
    for which \( \snapAct{1} \in \recSetActs_{s} \) and \( s \triangleleft \happTimePoint_i \) hold. 
    If no such time point exists,
    then \( \timesinceat(\snapAct{1},\happTimePoint_i,\triangleleft) \) evaluates to
    a value \( \happTimePoint_i - \happTimePoint_0 + \initgap + \epsilon \) where \( \initgap \) is fixed and \( 0 < \initgap \).

\emergencystretch=3em %
    \( \encaft(\bLvc{_I}, \allowbreak i, \allowbreak \bLvc{_i}) \) denotes the encoded state of a planning problem.
    When a happening time point sequence is not empty,%
    \footnote{When the htp sequence is empty,
    we define \( \encaft(\bLvc{_I}, \allowbreak i, \allowbreak \bLvc{_i}) \) to encode
    the propositional and numeric conditions of the initial state
    after a small delay has elapsed.
    See \Cref{def:encaft_initial_state}.}
    \( \encaft(\bLvc{_I}, \allowbreak i, \allowbreak \bLvc{_i}) \) is satisfied if for the \( i \)-th happening time point \( \happTimePoint_i \in \happTimePoints(\plan) \), all of the following hold:
    \begin{enumerate*}[itemjoin={;\ }, itemjoin*={; and\ }]
        \item For any action \( \durAct{j} \in \actions \):
        \begin{enumerate*}[itemjoin={;\ }, itemjoin*={; and\ }, label=(\roman{enumi}.\alph*)]
            \item \( \Loc_i \) records the location \( l_j \) of \( \autonet_{\durAction} \) as \( \olc_{\durAction} \) if \( \lnot \runat(\durAction, \happTimePoint_i, \le) \)
            \item \( \Loc_i \) records the location \( l_j \) as \( \rlc_{\durAction} \) if \( \runat(\durAction, \happTimePoint_i, \le) \)
            \item \( \Clk_i (\clk{x}) \) evaluates to \( \timesinceat(x,\happTimePoint_i, \le) \), 
                for all \( x = \daStart{\durAct{j}} \) or \( x = \daEnd{\durAct{j}} \)
        \end{enumerate*}
        \item For any proposition \( p \in \propositions \):
        \begin{enumerate*}[itemjoin={;\ }, itemjoin*={; and\ }, label=(\roman{enumi}.\alph*)]
            \item \( \Var_i (\vp{p}) \) evaluates to \( 1 \) if \( p \in \propof{(\planningState_{i+1})} \) and \( 0 \) if \( p \notin \propof{(\planningState_{i+1})} \)
            \item \( \Var_i (\lp{p}) \) evaluates to \( |\{\durAct{} \in \actions \mid p \in \daInv{\durAct{}} \land \runat(\durAct{}, \happTimePoint_i, \le) \}| \)
        \end{enumerate*}
        \item For any fluent \( f \in \fluents \), \( \Var_i (f) \) evaluates to \( \numof{(\planningState_{i+1})}(f) \)
        \item \( \Var_i (\actsactive) \) evaluates to \( |\{\durAct{} \in \actions \mid \runat(\durAct{}, \happTimePoint_i, \le) \}| \)
        \item \( \Var_i (\pl) \) evaluates to \( 1 \)
    \end{enumerate*}.

\end{paperdefinition}

\begin{theorem}\label{thm:main_theorem}
    If a plan \( \plan \) is valid for a planning problem \( \planningproblem \) 
    and \( \autonet \) is the network for \( \planningproblem \) (\Cref{def:network_for_problem}), 
    then \( \autonet \vDash \mathrm{EF}(\mathit{loc}(\autonet_M) = \goalloc) \) holds.
\end{theorem}
\begin{proofsketch}[\Cref{appendix:proofs} has a detailed proof]
It is sufficient to show the existence of a run \( \cfgs \) beginning at \( \bLvc{_I} \) and ending in \(\bLvc{_G} \), 
where \( \Loc_G \) records \( \autonet_M.\goalloc \).
\( \cfgs \) is constructed in three stages.
First, \( \bLvc{_I} \) must transition to a configuration \( \bLvc{_0'} \).
The action transition \( \initedge \) can update
\( \Var_I' \) of \(\bLvc{_I'}\) to encode
\( \initState \) of \( \planningproblem \)
from which a delay \( \delta = \initgap + \epsilon \) (\Cref{def:enc_active_time})
will result in a configuration \( \bLvc{_0} \)
which satisfies mutex conditions on \( \Clk_0 \).
By applying transitions which encode all snap actions \( \snapAct{} \in \recSetActs_{\happTimePoint_0} \),
an initial run \( \langle \bLvc{_I}, \bLvc{_I'}, \bLvc{_0}, \ldots, \bLvc{_0'} \rangle \) is obtained, 
where \( \encaft(\bLvc{_I},0,\bLvc{_0'}) \) is satisfied.
Then we can prove by induction on the length of \( \happTimePoints(\plan) \)
that there is some \( \bLvc{_k} \)
such that \( run([\bLvc{_0'}, \allowbreak \dots, \allowbreak \bLvc{_k}]) \)
and \( \encaft(\bLvc{_I}, \allowbreak k, \allowbreak \bLvc{_k}) \) are satisfied,
where \( k \le m \).
Note that our construction ensures that all variable bounds are respected throughout the run
(see \Cref{lem:variable_bounds_respected}).
For the details of the inductive step,
which involves a specific order of transitions that must be applied,
we refer to \Cref{appendix:proofs,fig:construct_run,fig:matchcellar_run,ex:valid_instant,fig:matchcellar_pictogram}.
If \( \bLvc{_k} \) satisfies \( \encaft(\bLvc{_I},k,\bLvc{_k}) \),
where \( k = m \),
and \( \Loc_k \) records \( \planningloc \) for \( \autonet_M \),
then \( \goaledge \) can be applied to \( \bLvc{_k} \) (\Cref{lem:goal_trans_possible})
to reach \( \bLvc{_G} \),
where \( \Loc_G \equiv \langle \goalloc, \dots \rangle \).
\end{proofsketch}

\section{Verification and Implementation}
\label{sec:formalisation}

The effort in the encoding and its correctness proof is substantial given the complexity of the definitions involved.
Indeed, the total size of the implementation and the associated proofs is 43K lines of Isabelle/HOL scripts (of which 23K are hand-written), 
which is much larger than similar projects.
E.g.\ \citeauthor{verifiedSATPlan}'s certifying SAT-based planner~\cite{verifiedSATPlan} is 17.5K lines; 
the difference is primarily due to the relative complexity of the concepts and thus the proofs for temporal planning and timed automata, 
compared to classical planning and SAT formulas (see~\Cref{table:formalisations} for more comparisons).
We now describe our approach to the verification and implementation.

\paragraph{Implementation} 
We use stepwise refinement~\cite{refinementWirth} of datatypes and definitions
to produce provably correct executable code from abstract definitions.
In this methodology, one proves the correctness of an algorithm in a number of steps: 
starting with an abstract version of the algorithm, 
where the proofs focus on mathematical ideas, 
and adding and proving correctness of the implementation in further steps
that replace mathematical functions and relations with datatypes and executable code.

\paragraph{Proofs}
For the proof of \Cref{thm:main_theorem},
we define a list of transitions \( \mathit{ts} \),
from which we obtain a list of functions \( \mathit{fs} \),
such that each \( \mathit{fs}_i \) is associated with a precondition \( P_i \) and postcondition \( Q_i \),
which match those of the corresponding transition \( \mathit{ts}_i \in \mathit{ts} \).
We iteratively remove the first element of \( \mathit{fs} \), 
apply it to the final element in \( \cfgs \), and append the result,
beginning from \( \cfgs = \langle \bLvc{_I} \rangle \).
In the inductive step, 
we prove that each \( P_i \) implies \( Q_i \) after applying \( \mathit{fs}_i \) to \( \cfgs_i \),
and that each \( Q_i \) implies \( P_{i+1} \).
Each pre- and postcondition includes the bounds for all variables.
We also prove that the addition of numeric conditions and updates to
the network preserves the correctness of the run at each step to obtain the full result involving numeric fluents. %
The additional assumptions necessary are obtained from \Cref{lem:variable_bounds_respected}.

We define a reusable library of \emph{automation lemmas} to simplify the proofs of these steps.
Each postcondition \( Q_i \) is proven by proving that all assertions in a
list \( \langle Q_{i,0},Q_{i,1},\ldots \rangle \) hold.
Moreover, from each precondition \( P_i \) we can deduce the truth of a list of facts 
\( \langle P_{i,0}, P_{i,1}, \ldots \rangle \).
Automation lemmas guide Isabelle's automated proof methods towards proofs,
by pairing postconditions with only the relevant preconditions,
i.e.\ trimming the search space.
See \Cref{tab:loc_table} for the lines of code in each part 
and our supplementary material for a mapping of concepts to Isabelle definitions and lemmas.

\paragraph{Trust}
{\emergencystretch=1.5em %
Our verified unsolvability checker 
(i.e.\ our entire pipeline:
\begin{enumerate*}[label=(\alph*), itemjoin={;\ }, itemjoin*={; and\ }]
    \item unverified grounding (written based on Fast Downward~\cite{fastdownward} and Temporal Fast Downward~\cite{eyerichUsingContextEnhancedAdditive2012})
    \item verified encoding
    \item model checking and certificate extraction using TChecker~\cite{TicktacprojectTchecker2023}
    \item verified certificate checking~\cite{wimmerCertTimedAutomata}
\end{enumerate*})
places trust in:
\begin{enumerate*}[itemjoin={;\ }, itemjoin*={; and\ }]
    \item the unverified grounder, which we plan to verify
    \item Isabelle's kernel, which we trust to be sound with respect to HOL's semantics
    \item the code generator~\cite{DBLP:conf/flops/HaftmannN10}, 
        which we trust to preserve the semantics of the HOL programs it translates
    \item the MLton compiler~\cite{DBLP:conf/ml/Weeks06}, hardware and runtime, 
        which we trust to compile and execute the generated code correctly
    \item the PDDL parser, which we trust to produce correct abstract syntax trees (ASTs)
\end{enumerate*}.
The model checking stage, 
which involves a number of pretty printing/parsing stages and model checking
does not need to be trusted, 
because the certificate's correctness is verifiably checked with respect to the AST of the network in memory.\par}

Our system reports ``unsolvable'' only if the certificate checker accepts the certificate as sufficient to demonstrate 
that the grounded problem has no plan;
this is machine-checked by Isabelle's kernel,
based on the contrapositive of \Cref{thm:main_theorem}
together with the formally verified correctness of the certificate checker.

\section{Experimental Evaluation}
\label{sec:experiments}
\begin{table}[tb]
    \centering
\setlength{\tabcolsep}{1pt}
\resizebox{\columnwidth}{!}{%
\begin{tabular}{l*{6}{r@{\,}r}}
\hline
Domain (\#inst) & \multicolumn{2}{c}{ours} & \multicolumn{2}{c}{T-CTP} & \multicolumn{2}{c}{T-FTP} & \multicolumn{2}{c}{OPTIC} & \multicolumn{2}{c}{nuXmv} & \multicolumn{2}{c}{UPPAAL} \\
\hline
MatchCellar (9) & (5/1/0) & \textbf{3} & (2/5/0) & \textbf{2} & (7/0/0) & \textbf{2} & (0/6/0) & \textbf{3} & (7/0/0) & \textbf{2} & (6/0/0) & \textbf{3} \\
SYNC (8) & (5/0/0) & \textbf{3} & (3/0/0) & \textbf{5} & (0/8/0) & \textbf{0} & (2/6/0) & \textbf{0} & (2/0/0) & \textbf{6} & (4/0/0) & \textbf{4} \\
PAINTER (12) & (10/1/0) & \textbf{1} & (7/0/0) & \textbf{5} & (2/10/0) & \textbf{0} & (8/2/0) & \textbf{2} & (8/0/0) & \textbf{4} & (8/0/0) & \textbf{4} \\
MAJSP-1 (128) & (72/51/5) & \textbf{0} & (56/60/0) & \textbf{12} & (84/44/0) & \textbf{0} & (84/43/0) & \textbf{1} & (106/0/0) & \textbf{22} & (127/0/0) & \textbf{1} \\
MAJSP-2 (240) & (84/82/6) & \textbf{68} & (86/66/0) & \textbf{88} & (212/28/0) & \textbf{0} & (103/37/0) & \textbf{100} & (128/0/0) & \textbf{112} & (185/0/1) & \textbf{54} \\
\hline
Total (397) & (176/135/11) & \textbf{75} & (154/131/0) & \textbf{112} & (305/90/0) & \textbf{2} & (197/94/0) & \textbf{106} & (251/0/0) & \textbf{146} & (330/0/1) & \textbf{66} \\
\hline
\end{tabular}
}
    \caption{For each domain/method pair: in parentheses the number of
    timeouts, memouts, and unknowns (crash or no verdict), followed by the
    number of instances proved unsolvable (\textbf{bold}, right-aligned).
    Here, we have our pipeline (ours) and OPTIC run on PDDL;
    TAMER run on ANML~\cite{smith2008anml} with their CTP (T-CTP) and FTP (T-FTP) algorithms;
    and nuXmv and UPPAAL run on the respective encodings emitted by TAMER.}
    \label{tab:solved_instances}
\end{table}

The experimental evaluation compares the performance of our verified unsolvability checker 
with other unverified planners and model checkers.
We compare against two unsolvability variants of the planner TAMER~\cite{TAMERPlanner}, 
as well as its encodings to UPPAAL~\cite{bengtssonUppaalTool1996,DBLP:conf/ftrtft/BehrmannBDLPY02} 
and nuXmv~\cite{nuXmv,cimattiExtendingNuXmvTimed2019},
which are available online~\cite{panjkovicUnsolvabilityArtifact2022};
and OPTIC~\cite{OPTICplanner}.
We run the planners on the domains used by~\citet{unsolvableTemporalPlanning},
on a few identical servers (AMD EPYC 9554P, 64c/128t, 1511 GiB RAM),
with a per-run limit of 48 GiB RAM and a timeout of 1800 seconds. 
See \Cref{tab:solved_instances} and \Cref{sec:additional_benchmark_data}.

The benchmark domains each present their own difficulties:
\begin{enumerate*}[label=(\arabic*), itemjoin={\ }]
    \item MatchCellar~\cite{DBLP:journals/aim/ColesCOJLSY12} (see \Cref{def:matchcellar_unsolvable_instance}) seems to be hard for every solver, 
        because the unsolvability argument either needs to detect symmetries between mend actions
        \cite{emersonSymmetryModelChecking1996,foxDetectionExploitationSymmetry1999,DBLP:conf/formats/HendriksBLNV03},
        or relies on domain-specific information~\cite{DBLP:journals/aicom/FuentetajaR15,riddle:etal:16} (see \Cref{ex:matchcellar_instance_1_1_2_56_5}),
        neither of which any solver uses.
    \item In PAINTER~\cite{TAMERPlanner} a set of items must each receive a sequence of treatments in a fixed order,
        one treatment at a time.
        The goal is to execute a join action, 
        that needs two items to have started and not finished a treatment each.
        These four conditions cannot all hold simultaneously.
        This seems to be hard for every method:
        nuXmv needs more frames than in other domains to find an inductive invariant,
        OPTIC's heuristic does not detect the mutual exclusion and cannot deduce unsolvability immediately,
        and ours, UPPAAL and TAMER-CTP unfold state spaces that scale multiplicatively with the number of treatments.
    \item In SYNC~\cite{unsolvableTemporalPlanning}, 
        a robot must start a set of parallel jobs of different durations within a time frame \( c_1 \)
        and end them within a time frame \( c_2 \),
        which is impossible because the durations are not all within \( c_1 + c_2 \) of each other.
        The difficulty scales with the number of interleavings of action executions,
        and hence of discrete and symbolic clock states:
        methods that explore state spaces must explore the discrete state space in full,
        which better abstractions of symbolic states mitigate somewhat.
    \item In MAJSP-1 and -2~\cite{micheliTemporalPlanningTemporal2019}, 
        robots carry pallets from a depot to locations on a line where they are treated,
        while being limited in travel distance by their battery level.
        These domains have many potentially concurrent actions
        and hence large discrete state spaces.
        MAJSP-1 seems harder than MAJSP-2,
        because its robots have a higher battery level relative to the cost of movement
        and hence can perform more actions.
        The unsolvability condition closely depends on the numeric fluents for battery level and distance,
        which favours methods that reason about their relationship over those that explore state spaces.
\end{enumerate*}

OPTIC and nuXmv reason about battery levels symbolically,
with a static numeric check and an inductive invariant respectively.
OPTIC's preprocessor can correctly identify 101 instances of MAJSP as unsolvable,
by finding that no robot can move given the initial battery levels.
nuXmv searches for longer proofs for unsolvability in the form of an inductive invariant
over comparisons between real variables~\cite{DBLP:conf/vmcai/Bradley11}
without exploring the state space~\cite{nuXmv},
which makes it perform particularly well in SYNC and MAJSP.
However, it struggles to find the invariant \( 2 \cdot m < f \) 
(\( m \) matches, \( f \) fuses) that refutes MatchCellar (\Cref{ex:matchcellar_instance_1_3_5_2}).

ANML, which TAMER and its encodings use but OPTIC and our method do not,
supports effects between an action's start and end,
so one ANML action may need several PDDL actions to represent,
which can enlarge the state space that ours and OPTIC must explore.

\section{Discussion}

We presented a formally verified, executable, practical encoding of numeric temporal planning problems as networks of timed automata.
Using \Cref{thm:main_theorem} we have proven, once and for all, that an encoding from grounded planning problems can be formally verified and made executable,
in a way that is sound for checking and verifying the unsolvability of any planning problem which satisfies some assumptions (namely: no self-overlap of actions, no disjunctive invariants, integer fluents, trusted grounding; see \Cref{sec:formalisation}), and that in some cases it performs on par with unverified state-of-the-art planning software.
Combined with a certificate checker for timed automata~\cite{wimmerCertTimedAutomata}, our work presents a fully verified pipeline for showing unsolvability.

We used~\citeauthor{temporalPlanningRefinementMC}'s encoding~\cite{temporalPlanningRefinementMC} as our starting point.
Most notably, we extend it to work for problems with numeric fluents
(since numeric planning with unbounded fluents is undecidable~\cite{DBLP:conf/aips/Helmert02}, we decide unsolvability for problems for which we can infer verified bounds).
We also generalised its notion of unsolvability with ideas from~\citet{giganteDecidabilityComplexity2022}.
We note that there are others who used related bounding techniques in planning: \citet{DBLP:conf/ecai/ScalaHTR16} use interval arithmetic to compute a heuristic value for a state,
which is also sound for reachability, 
\citet{DBLP:conf/ecai/0002SB23} infer bounds of numeric variables
and use them in heuristics for optimal linear numeric planning, and 
\citet{DBLP:journals/jair/ColesCFL13} compute bounds statically.

Other authors have designed encodings of temporal planning into timed automata~\cite{temporalPlanningTimedAutomata2002,liPlanningModelChecking2012,temporalPlanningRefinementMC}
and explored planning using model checking~\cite{DBLP:conf/aaai/HoffmannHKSKM20}.
However, in all these cases the main focus was solvable planning problems.
For unsolvability, \citet{unsolvableTemporalPlanning} recently adapted the TAMER planning system~\cite{TAMERPlanner} for detecting unsolvability using model checking techniques as well as by integrating the nuXmv~\cite{nuXmv,cimattiExtendingNuXmvTimed2019} model checker.
Earlier approaches that could be used to detect unsolvability are the temporal version of Fast Downward~\cite{temporalFD}, a temporal version of the landmark heuristics~\cite{DBLP:conf/ecai/MarzalSO08}, and the OPTIC planner~\cite{OPTICplanner}.

ITPs have been used to verify related software, e.g.\ SAT-based planning~\cite{verifiedSATPlan}, SAT encodings~\cite{avigadSATEncodings}, classical and temporal plan validation~\cite{ictai2018,aaai2022}, and certification for combinatorial problems~\cite{gratLammich,modelcountingCert,certBinarisedNNs}.
We show that combining a verified encoding with a verified certificate checker, first devised by \citet{verifiedSATPlan} for SAT encodings of classical planning, scales to source and target formalisms substantially more complex than those previously handled with ITPs in planning.
The alternative is to verify an entire solver, as \citet{esparza2013fully} did for an executable model checker; this rules out crashes due to bugs, which can occur in our approach if the unverified model checker crashes, but verifying a solver that performs on par with unverified state-of-the-art solvers is most likely infeasible, due to the number of optimisations it uses.

As future work, we would like to verify an unsolvability checker for other richer planning formalisms, e.g.\ for hybrid planning, which involves verifying a certificate checker for hybrid automata~\cite{hybridautomata}.
Another direction is verifying a practical grounding algorithm, e.g.\ one based on Fast Downward's grounder~\cite{fastdownward}: until now all previous encodings assume grounded problems as input.

\section{Data Availability}
Detailed proofs and examples are in the appendix; 
the certifier,
its Isabelle/HOL proofs and source code,
and the benchmark data are in the supplementary material and will be published online.

We have used AI (Claude Opus 4.7 to 5.5 and Fable 5.0 to 5.1) to make minor changes and structural moves in the original propositional formalisation,
and extensively to implement and prove the correctness of the numeric extension.
The propositional part of the formalisation,
which consists of 27.4K lines of code, is still 83\% human, 7\% mixed, and 10\% AI.
The numeric extension (15.8K lines of code) consists of ${>}99\%$ AI and ${<}1\%$ human AI collaboration.
The total is 43.2K, of which 52\% can be attributed to humans.
Details can be found in the supplementary material and the appendix.
We have also used AI to lay out figures and tables in this manuscript.

\bibliographystyle{splncs04}
\bibliography{long_paper,paper}
\clearpage
\appendix
\emergencystretch=3em
\section{MatchCellar Example}\label{sec:pddl_matchcellar}

\begin{definition}[MatchCellar Template]\label{def:matchcellar_template}
  The template is parameterised by
  \begin{enumerate*}[itemjoin={,\ }, itemjoin*={, and\ }]
    \item an indexed set \( \matchnums \) of matches
    \item an indexed set \( \fusenums \) of fuses
    \item a duration for the actions involving matches and fuses,
      namely \( \durlitlb \) and \( \durlitub \) for the duration range a match burns
      and \( \durmend \) for the duration needed to mend a fuse
    \item a number of spare fuses \( \numspare \)
  \end{enumerate*}.
  This is in fact possible in standard PDDL.~See \Cref{lst:pddl_custom}.

  We identify objects using shorthand notation and indices.
  {\emergencystretch=2.5em
  We obtain a planning problem \( \planningproblem \equiv \langle \propositions, \fluents, \actions, \initState, \goalCond \rangle \), where
  \begin{enumerate*}[itemjoin={;\ }, itemjoin*={; and\ }, label=(\arabic*)]
    \item \( \propositions =
      \splitatcommas{
      \{ \handfreep, \unusedp(\matchv_i), \lightp(\matchv_i), \mendedp(\fusev_j) \}
      }
      \),
      for \( i \in \matchnums \) and \( j \in \fusenums \)   %
    \item \( \fluents =
      \splitatcommas{
      \{ \durlitlbf(\matchv_i), \durlitubf(\matchv_i), \durmendf(\fusev_j), \sparef, \litf \}
      }
      \)
    \item \( \actions = \splitatcommas{\{
      (\lightactsrt_{i} \equiv \lightact(\matchv_i)),
      (\mendactsrt_{ji} \equiv \mendact(\fusev_j,\matchv_i))
      \}} \),
      for \( i \in \matchnums \) and \( j \in \fusenums \),
      with conditions and effects
      \begin{enumerate*}[itemjoin={;\ }, itemjoin*={; and\ }, label={(\arabic{enumi}.\alph*)}]
        \item \(
          \splitatcommas{
          \lightactsrt_{i}
          \equiv
          \langle
              \langle
                  \{ \unusedp(\matchv_i) \},
                  \true,
                  \{ \lightp(\matchv_i) \},
                  \{ \unusedp(\matchv_i) \},
                  \{ \litf \coloneq {\litf + 1} \}
              \rangle,
              \langle
                  \emptyset,
                  \true,
                  \emptyset,
                  \{ \lightp(\matchv_i) \},
                  \emptyset
              \rangle,
              \emptyset,
              \true,
              \langle \durlitlbf(\matchv_i), \le \rangle,   %
              \langle \durlitubf(\matchv_i), \le \rangle
          \rangle
          }
          \)
        \item \(
          \splitatcommas{
          \mendactsrt_{ji}
          \equiv
          \langle
              \langle
                  \{ \handfreep \},
                  \sparef > 0,
                  \emptyset,
                  \{ \handfreep \},
                  \{ \sparef \coloneq {\sparef - 1} \}
              \rangle,
              \langle
                  \emptyset,
                  \true,
                  \{ \handfreep, \mendedp(\fusev_j) \},
                  \emptyset,
                  \emptyset
              \rangle,
              \{ \lightp(\matchv_i) \},
              \true,
              \langle \durmendf(\fusev_j), \le \rangle,
              \langle \durmendf(\fusev_j), \le \rangle
          \rangle
          }
          \)
      \end{enumerate*}
    \item the initial state \( \initState \equiv \langle \propof{\initState}, \numof{\initState} \rangle \) is characterised by
      \begin{enumerate*}[itemjoin={;\ }, itemjoin*={; and\ }, label={(\arabic{enumi}.\alph*)}]
        \item \( \numof{\initState}(\sparef) = \numspare \)
        \item \( \numof{\initState}(\litf) = 0 \)
        \item \( \numof{\initState}(\durlitlbf(\matchv_i)) = \durlitlb \) and \( \numof{\initState}(\durlitubf(\matchv_i)) = \durlitub \) for \( i \in \matchnums \)
        \item \( \numof{\initState}(\durmendf(\fusev_j)) = \durmend \) for \( j \in \fusenums \)
        \item \( \propof{\initState} \equiv \splitatcommas{\{ \handfreep, \unusedp(\matchv_i) \mid i \in \matchnums \}} \)
      \end{enumerate*}
    \item the goal is
      \( \goalCond \equiv \langle \{ \mendedp(\fusev_j) \mid j \in \fusenums \}, \true \rangle \)
  \end{enumerate*}.\par}
\end{definition}

\begin{definition}[MatchCellar Example Instance]\label{def:matchcellar_example_instance}
  This can be constructed based on \Cref{def:matchcellar_template}
  by setting
  \begin{enumerate*}[itemjoin={,\ }, itemjoin*={, and\ }]
    \item both the set of identifiers for matches \( \matchnums \)
      and that for fuses \( \fusenums \) to \( \{ 0 \} \)
    \item the number of spare fuses \( \numspare \) to \( 2 \)
    \item the burn duration bounds \( \durlitlb \) to \( 5 \) and \( \durlitub \) to \( 6 \),
    and the mend duration \( \durmend \) to \( 5 \)
  \end{enumerate*}.
\end{definition}

\begin{definition}[MatchCellar Unsolvable Instance]\label{def:matchcellar_unsolvable_instance}
  This example can be derived from \Cref{def:matchcellar_template}.
  Disregard the number of spare fuses (or set it to a value greater than \( 3 \)) 
  and the number of lit matches,
  as these are irrelevant for the planning instance.

  Set the remaining parameters to:
  \( 1 \) match,
  \( 3 \) fuses,
  a burn duration of \( 5 \) (\( \durlitlb = \durlitub = 5 \)), 
  and a mend duration of \( 2 \).
\end{definition}

\begin{lstlisting}[language=pddl, caption={MatchCellar domain (\Cref{def:matchcellar_template})}, label={lst:pddl_custom}]
(define (domain matchcellar)
    (:requirements :typing :durative-actions :numeric-fluents)
    (:types
        match fuse
    )
    (:predicates
        (handfree)
        (unused ?match - match)
        (mended ?fuse - fuse)
        (light ?match - match)
    )
    (:functions
        (dur-lit-lb ?x - match)
        (dur-lit-ub ?x - match)
        (dur-mend ?x - fuse)
        (spare-fuses)
        (matches-lit)
    )

    (:durative-action LIGHT_MATCH
        :parameters (?match - match)
        :duration (and
            (>= ?duration (dur-lit-lb ?match))
            (<= ?duration (dur-lit-ub ?match)))
        :condition (and
            (at start (unused ?match)))
        :effect (and
            (at start (not (unused ?match)))
            (at start (light ?match))
            (at start (assign (matches-lit) (+ (matches-lit) 1)))
            (at end (not (light ?match))))
    )

    (:durative-action MEND_FUSE
        :parameters (?fuse - fuse ?match - match)
        :duration (= ?duration (dur-mend ?fuse))
        :condition (and
            (at start (handfree))
            (over all (light ?match))
            (at start (> (spare-fuses) 0)))
        :effect (and
            (at start (not (handfree)))
            (at start (decrease (spare-fuses) 1))
            (at end (mended ?fuse))
            (at end (handfree)))
    )
)
\end{lstlisting}

\begin{lstlisting}[language=pddl, caption={Instance of \Cref{ex:matchcellar_instance_1_1_2_56_5}}, label={lst:pddl_custom_instance}]
(define (problem instance)
 (:domain matchcellar)
 (:objects
    match0 - match
    fuse0 - fuse
)
 (:init
  (handfree)
  (unused match0)
  (= (dur-lit-lb match0) 5)
  (= (dur-lit-ub match0) 6)
  (= (dur-mend fuse0) 5)
  (= (spare-fuses) 2)
  (= (matches-lit) 0)
)
 (:goal
  (and
     (mended fuse0)
))
 (:metric minimize (matches-lit))
)
\end{lstlisting}

\begin{lstlisting}[language=pddl, caption={Instance of \Cref{ex:matchcellar_instance_1_3_5_2}}, label={lst:pddl_matchcellar}]
(define (problem instance)
 (:domain matchcellar)
 (:objects
    match0 - match
    fuse0 fuse1 fuse2 - fuse
)
 (:init
  (handfree)
  (unused match0)
  (= (dur-lit-lb match0) 5)
  (= (dur-lit-ub match0) 5)
  (= (dur-mend fuse0) 2)
  (= (dur-mend fuse1) 2)
  (= (dur-mend fuse2) 2)
  (= (spare-fuses) 3)
  (= (matches-lit) 0)
)
 (:goal
  (and
     (mended fuse0)
     (mended fuse1)
     (mended fuse2)
))
 (:metric minimize (matches-lit))
)
\end{lstlisting}

\begin{figure}[tb]
    \centering
    \begin{tikzpicture}[font=\scriptsize]
        \draw[->] (-1.3,0) -- (8.3,0) node[below left, inner sep=1.5pt] {time (s)};
        \foreach \x/\tm/\hp in {0/0/\happTimePoint_0, 7/5/\happTimePoint_1} {
            \draw (\x,0.06) -- (\x,-0.06) node[above=0.1] {\( \tm \)};
            \draw[black!60, dashed] (\x,-0.06) -- (\x,-0.68);
            \node[below, inner sep=1pt] at (\x,-0.68) {\( \hp \)};
        }
        \draw[action box=\nlmacol] (0,-0.1) rectangle (7,-0.36)
            node[pos=.5] {\( \lightactsrt_0 \)};
        \draw[action box=\nmfacol] (0,-0.4) rectangle (7,-0.66)
            node[pos=.5] {\( \mendactsrt_{00} \)};
    \end{tikzpicture}

    \smallskip
    \begin{tikzpicture}[font=\footnotesize]
        \def\xA{8.4} \def\yA{1.42} \def\xB{-1.4} \def\yB{-1.47}
        \path[use as bounding box, draw, dotted] (\xA,\yA) rectangle (\xB,\yB);
        \begin{scope}[
                every state/.style={draw, rectangle, inner sep=2pt, minimum height=0},
                >={Stealth}, shorten >=1pt, node distance=2, on grid
            ]
            \node[state] (inactive) at (0,0) {\olc};
            \node[state] (starting) at (3.5,0.45) {\slc!};
            \node[state] (running) at (7, 0) {\rlc};
            \node[state] (ending) at (3.5,-0.45) {\elc!};
        \end{scope}
        \node[anchor=south west, inner sep=1pt] (lbl) at (-1.3,0.47)
            {\( \sse_{\lightactsrt_0}(t_0, 1) \)};
        \node[above=0pt of lbl.north west, anchor=south west, inner sep=1pt, font=\scriptsize]
            {\( \langle \vp{\lightp(\matchv_0)} \coloneq 1 \rangle \)};
        \node[anchor=south east, inner sep=1pt] at (8.3,0.47) {\( \see_{\lightactsrt_0}(t_0, 3) \)};
        \node[anchor=north east, inner sep=1pt] at (8.3,-0.47) {\( \ese_{\lightactsrt_0}(t_1, 1) \)};
        \node[anchor=north west, inner sep=1pt] (lbl) at (-1.3,-0.47)
            {\( \eee_{\lightactsrt_0}(t_1, 3) \)};
        \node[below=0pt of lbl.south west, anchor=north west, inner sep=1pt, font=\scriptsize]
            {\( \langle (\lp{\lightp(\matchv_0)} = 0)?, \vp{\lightp(\matchv_0)} \coloneq 0 \rangle \)};
        \begin{scope}[>={Stealth}, shorten >=1pt, auto]
            \draw[->, semithick, dashed, draw=\nlmacol] (inactive) |- (starting);
            \draw[->, semithick, dashed, draw=\nlmacol] (starting) -| (running);
            \draw[->, semithick, draw=\nlmacol] (running) |- (ending);
            \draw[->, semithick, draw=\nlmacol] (ending) -| (inactive);
            \draw[->] (starting) -- (ending)
                node[midway, anchor=west, inner sep=2pt] {\( \inste_{\lightactsrt_0} \)};
        \end{scope}
    \end{tikzpicture}

    \smallskip
    \begin{tikzpicture}[font=\footnotesize]
        \def\xA{8.4} \def\yA{1.42} \def\xB{-1.4} \def\yB{-1.47}
        \path[use as bounding box, draw, dotted] (\xA,\yA) rectangle (\xB,\yB);
        \begin{scope}[
                every state/.style={draw, rectangle, inner sep=2pt, minimum height=0},
                >={Stealth}, shorten >=1pt, node distance=2, on grid
            ]
            \node[state] (inactive) at (0,0) {\olc};
            \node[state] (starting) at (3.5,0.45) {\slc!};
            \node[state] (running) at (7, 0) {\rlc};
            \node[state] (ending) at (3.5,-0.45) {\elc!};
        \end{scope}
        \node[anchor=south west, inner sep=1pt] (lbl) at (-1.3,0.47)
            {\( \sse_{\mendactsrt_{00}}(t_0, 2) \)};
        \node[above=0pt of lbl.north west, anchor=south west, inner sep=1pt, font=\scriptsize]
            {\( \langle (\vp{\handfreep} = 1)?, \vp{\handfreep} \coloneq 0 \rangle \)};
        \node[anchor=south east, inner sep=1pt] (lbl) at (8.3,0.47)
            {\( \see_{\mendactsrt_{00}}(t_0, 4) \)};
        \node[above=0pt of lbl.north east, anchor=south east, inner sep=1pt, font=\scriptsize]
            {\( \langle (\vp{\lightp(\matchv_0)} = 1)?, \lp{\lightp(\matchv_0)} \pluseq 1 \rangle \)};
        \node[anchor=north east, inner sep=1pt] (lbl) at (8.3,-0.47)
            {\( \ese_{\mendactsrt_{00}}(t_1, 2) \)};
        \node[below=0pt of lbl.south east, anchor=north east, inner sep=1pt, font=\scriptsize]
            {\( \langle \lp{\lightp(\matchv_0)} \minuseq 1 \rangle \)};
        \node[anchor=north west, inner sep=1pt] (lbl) at (-1.3,-0.47)
            {\( \eee_{\mendactsrt_{00}}(t_1, 4) \)};
        \node[below=0pt of lbl.south west, anchor=north west, inner sep=1pt, font=\scriptsize]
            {\( \langle \vp{\handfreep} \coloneq 1 \rangle \)};
        \begin{scope}[>={Stealth}, shorten >=1pt, auto]
            \draw[->, semithick, dashed, draw=\nmfacol] (inactive) |- (starting);
            \draw[->, semithick, dashed, draw=\nmfacol] (starting) -| (running);
            \draw[->, semithick, draw=\nmfacol] (running) |- (ending);
            \draw[->, semithick, draw=\nmfacol] (ending) -| (inactive);
            \draw[->] (starting) -- (ending)
                node[midway, anchor=west, inner sep=2pt] {\( \inste_{\mendactsrt_{00}} \)};
        \end{scope}
    \end{tikzpicture}

    \caption{Automata for the actions of the plan in \Cref{fig:example_plan}.
    Dashed edges are applied to simulate \( \happTimePoint_0 \) (\rule[.5ex]{0.2em}{.4pt}\hspace{0.2em}\rule[.5ex]{0.2em}{.4pt}\hspace{0.2em}\rule[.5ex]{0.2em}{.4pt}\hspace{0.2em}\rule[.5ex]{0.2em}{.4pt}) and solid ones to simulate \( \happTimePoint_1 \) (\rule[.5ex]{1.4em}{.4pt}).
    The order is indicated with respect to each time point \(t_i\) as \( (t_i, n) \).}
    \label{fig:matchcellar_run}
\end{figure}

\begin{figure}[tb]
    \centering
    \begin{tikzpicture}[x=0.95cm]
        \draw[->] (0,0) -- (3.3,0) node[below left, font=\scriptsize, inner sep=1.5pt] {time (s)};
        \foreach \x in {0, 2} {
            \draw (\x,0.08) -- (\x,-0.08) node[above=0.16, font=\scriptsize] {\( \x \)};
        }
        \draw[action box=\nmfacol] (0,-0.15) rectangle (2,-0.5)
            node[pos=.5, font=\scriptsize] {\( \mendactsrt_{00} \)};
        \draw[action box=\nmfbcol] (0,-0.6) rectangle (2,-0.95)
            node[pos=.5, font=\scriptsize] {\( \mendactsrt_{10} \)};
        \node[red, font=\large] at (0,-0.55) {\( \times \)};
        \node[red, font=\scriptsize, anchor=north, inner sep=1pt] at (0,-1.0) {\( \handfreep \)};
        \begin{scope}[x=1cm, y=1cm]
            \renewcommand{\picheight}{1.05cm}
            \tikzset{state scale=0.72, diagonal radius=1.0}
            \pic (st) at (3.4,-1.95) {state={lit}{busy}{broken}};
            \coordinate (hand) at ($(st-box.south west)+(0.72*2.257,0.72*0.663)$);
            \begin{scope}[transparency group, opacity=0.35]
                \renewcommand{\picheight}{0.756cm}
                \pic at ($(hand)+(0.16,0.12)$) {hand={busy}};
            \end{scope}
            \draw[->, >={Stealth}, draw=\nmfacol, shorten >=3pt] (2,-0.33) to[bend right=12] ($(hand)+(0.02,0.38)$);
            \draw[->, >={Stealth}, draw=\nmfbcol, shorten >=2pt] (2,-0.78) to[bend right=10] ($(hand)+(-0.25,0.05)$);
        \end{scope}
        \node[font=\scriptsize] at (2.9,-2.25) {(a)};
    \end{tikzpicture}
    \hfill
    \begin{tikzpicture}[x=0.95cm]
        \def\e{0.4} %
        \draw[->] (0,0) -- (5.6,0) node[below left, font=\scriptsize, inner sep=1.5pt] {time (s)};
        \foreach \x/\lab in {0/0, 2/2, 2+\e/2+\delta, 4+\e/4+\delta} {
            \draw (\x,0.08) -- (\x,-0.08);
        }
        \foreach \x/\lab in {0/0, 2/2, 4+\e/4+\delta} {
            \node[above=0.16, font=\scriptsize] at (\x,0) {\( \lab \)};
        }
        \draw[black!60, dotted] (2+\e,0.08) -- (2+\e,0.42);
        \node[above, font=\scriptsize, inner sep=1pt] at (2+\e+0.15,0.42) {\( 2+\delta \)};
        \draw[action box=\nmfacol] (0,-0.15) rectangle (2,-0.5)
            node[pos=.5, font=\scriptsize] {\( \mendactsrt_{00} \)};
        \draw[action box=\nmfbcol] (2+\e,-0.6) rectangle (4+\e,-0.95)
            node[pos=.5, font=\scriptsize] {\( \mendactsrt_{10} \)};
        \draw[->, >={Stealth}, shorten >=1pt, shorten <=1pt] (2,-0.5) to[bend right=40] (2+\e,-0.6);
        \node[font=\scriptsize, anchor=east, inner sep=1pt] at (1.95,-0.95) {\( \handfreep \)};
        \node[font=\scriptsize] at (2.8,-2.25) {(b)};
    \end{tikzpicture}
    \caption{Interfering actions of \Cref{ex:matchcellar_instance_1_3_5_2}.
    (a) The start snap actions of two mend actions both require and delete \( \handfreep \),
    so they interfere and cannot be scheduled at the same time point.
    (b) The second mend can start \( \delta \) after the first has ended and added \( \handfreep \) again.}
    \label{fig:matchcellar_interference}
\end{figure}
\begin{figure}[t]
    \centering
    \stackMath
    \begin{tikzpicture}[x=1.7cm, >={Stealth}]
        \def\da{0.3}    %
        \def\db{0.3}    %
        \def\htp{-2.05} %
        \draw[->] (-0.6,0) -- ++(6.4,0);
        \draw[->] (-0.6,\htp) -- ++(6.4,0);
        \node[above, anchor=south east, font=\scriptsize] at (-0.65,0) {time (s)};
        \node[below, anchor=north east, font=\scriptsize] at (-0.65,0) {\( \plan \)};
        \node[below, anchor=north east, font=\scriptsize] at (-0.65,\htp) {\( \recSetActs_{\happTimePoint} \)};
        \draw[action box=\nlmacol] (0,-0.15) rectangle (5,-0.47)
            node[pos=.5, font=\scriptsize] {\( \lightactsrt_{0} \)};
        \draw[action box=\nmfacol] (\da,-0.57) rectangle (\da+2,-0.89)
            node[pos=.5, font=\scriptsize] {\( \mendactsrt_{00} \)};
        \draw[invalid action box=\nmfbcol] (\da+2,-0.99) rectangle (\da+4,-1.31)
            node[pos=.5, font=\scriptsize] {\( \mendactsrt_{10} \)};
        \draw[draw=\nmfbcol!60, solid]
            (\da+2,-0.99) edge ++(2,0)
            (\da+2,-0.99) edge ++(0,-0.32)
            (\da+2,-1.31) edge ++(2,0);
        \node[fill=white, inner sep=1pt, font=\scriptsize] at (\da+2,-1.15) {\forbiddenmark};
        \draw[action box=\nmfbcol] (\da+2+\db,-1.41) rectangle (\da+4+\db,-1.73)
            node[pos=.5, font=\scriptsize] {\( \mendactsrt_{10} \)};
        \newcommand{\htpmark}[6]{%
            \draw[draw=gray, dashed] (#1,#2) -- (#1,\htp-0.2);
            \node[anchor=south #3, font=\footnotesize, inner sep=2pt] at (#1,#2) {\( #4 \)};
            \node[anchor=north #6, font=\scriptsize, inner sep=2pt] at (#1,\htp-0.2) {\( #5 \)};
            \draw (#1,0) -- ++(0,2mm);
            \draw (#1,\htp) -- ++(0,-2mm);
        }
        \htpmark{0}{0.05}{east}{t_0 = 0}{\daStart{\lightactsrt_{0}}}{east}
        \htpmark{\da}{0.05}{west}{t_1 = \delta_1}{\daStart{\mendactsrt_{00}}}{west}
        \htpmark{\da+2}{0.05}{east}{t_2 = \delta_1 + 2}{\Shortstack{{\daEnd{\mendactsrt_{00}}} {\raisebox{-1pt}{\text{\usym{1F6AB}}}\text{:}\daStart{\mendactsrt_{10}}}}}{east}
        \htpmark{\da+2+\db}{0.4}{west}{t_3 = \delta_1 + 2 + \delta_2}{\daStart{\mendactsrt_{10}}}{west}
        \htpmark{\da+4+\db}{0.05}{east}{t_4 = \delta_1 + 4 + \delta_2}{\daEnd{\mendactsrt_{10}}}{east}
        \htpmark{5}{0.05}{west}{t_5 = 5}{\daEnd{\lightactsrt_{0}}}{west}
        \node[font=\scriptsize, inner sep=1pt, anchor=south] at (\da/2,0.02) {\( \delta_1 \)};
        \node[font=\scriptsize, inner sep=1pt, anchor=south] at (\da+2+\db/2,0.02) {\( \delta_2 \)};
        \node[font=\scriptsize, inner sep=1pt, anchor=south] at ({(\da+4+\db+5)/2},0.02) {\( \delta_3 \)};
    \end{tikzpicture}

    \medskip
    \begin{tikzpicture}[>={Stealth}]
        \renewcommand{\picheight}{1.05cm}
        \tikzset{state scale=0.72, hand at={1.4,0.7}, fuse sep=1.05, column top=0.6}
        \def\cs{4.5}    %
        \def\rowy{-3.7} %
        \pic[dark] (s0) at (0,0)          {state={unlit}{free}{broken,broken}};
        \pic (s1) at (\cs,0)              {state={lit}{free}{broken,broken}};
        \pic (s2) at (2*\cs,0)            {state={lit}{busy}{broken,broken}};
        \pic (s3) at (2*\cs,\rowy)        {state={lit}{free}{mended,broken}};
        \pic (s4) at (\cs,\rowy)          {state={lit}{busy}{mended,broken}};
        \pic (s5) at (0,\rowy)            {state={lit}{free}{mended,mended}};
        \node[anchor=south, font=\footnotesize] at (s0-box.north) {\( t_0 \)};
        \node[anchor=south, font=\footnotesize] at (s1-box.north) {\( t_1 \)};
        \node[anchor=south, font=\footnotesize] at (s2-box.north) {\( t_2 \)};
        \node[anchor=north, font=\footnotesize] at (s3-box.south) {\( t_3 \)};
        \node[anchor=north, font=\footnotesize] at (s4-box.south) {\( t_4 \)};
        \node[anchor=north, font=\footnotesize] at (s5-box.south) {\( t_5 \)};
        \newcommand{\splitarrow}[5]{%
            \coordinate (sa) at ($(#1)!0.6!(#2)$);
            \draw[->, semithick, draw=#3] (#1) -- (sa);
            \draw[->] (sa) -- (#2) node[midway, below, font=\scriptsize] {\( #5 \)};
            \node[above, font=\scriptsize, inner sep=2pt] at ($(#1)!0.5!(#2)$) {\( #4 \)};
        }
        \splitarrow{s0-box.east}{s1-box.west}{\nlmacol}{\daStart{\lightactsrt_{0}}}{\delta_1}
        \splitarrow{s1-box.east}{s2-box.west}{\nmfacol}{\daStart{\mendactsrt_{00}}}{2}
        \splitarrow{s3-box.west}{s4-box.east}{\nmfbcol}{\daStart{\mendactsrt_{10}}}{2}
        \splitarrow{s4-box.west}{s5-box.east}{\nmfbcol}{\daEnd{\mendactsrt_{10}}}{\delta_3}
        \coordinate (v0) at ($(s2-box.south west)!0.78!(s2-box.south east)$);
        \coordinate (v2) at ($(s3-box.north west)!0.78!(s3-box.north east)$);
        \coordinate (v1) at ($(v0)!0.55!(v2)$);
        \draw[->, semithick, draw=\nmfacol] (v0) -- (v1)
            node[midway, right, font=\scriptsize] {\( \daEnd{\mendactsrt_{00}} \)};
        \draw[->] (v1) -- (v2) node[midway, right, font=\scriptsize] {\( \delta_2 \)};
        \coordinate (ra) at ($(s2-box.south west)+(0.12,0)$);  \coordinate (rb) at ($(s4-box.north east)+(-0.72,0)$);
        \coordinate (ba) at ($(s2-box.south west)+(0.60,0)$);  \coordinate (bb) at ($(s4-box.north east)+(-0.24,0)$);
        \coordinate (rm) at ($(ra)!0.55!(rb)$);     \coordinate (bm) at ($(ba)!0.55!(bb)$);
        \draw[->, semithick, draw=\nmfbcol] (ra) -- (rm)
            node[pos=0.7, sloped, above, font=\scriptsize] {\( \daStart{\mendactsrt_{10}} \)}
            node[pos=0.25, fill=white, inner sep=0.5pt] {\forbiddenmark};
        \draw[->] (rm) -- (rb);
        \draw[->, semithick, draw=\nmfacol] (ba) -- (bm)
            node[pos=0.5, sloped, below, font=\scriptsize] {\( \daEnd{\mendactsrt_{00}} \)};
        \draw[->] (bm) -- (bb) node[pos=0.5, sloped, below, font=\scriptsize] {\( \delta_2 + 2 \)};
    \end{tikzpicture}
    \caption{Timeline (top) and partial state space (bottom) of \Cref{ex:matchcellar_instance_1_3_5_2}
    with transitions representing executions of durative actions and time delays \( \delta_1 \), \( \delta_2 \), \( \delta_3 \).
    Each state is shown at its time point before the snap actions of that time point are applied.
    Starting \( \mendactsrt_{10} \) at \( t_2 \), when \( \mendactsrt_{00} \) ends, is forbidden
    according to \Cref{def:mutex_snap} (hatched box in the timeline; in the state space, the pair of transitions from \( t_2 \) to \( t_4 \)):
    \( \daEnd{\mendactsrt_{00}} \) adds \( \handfreep \), which \( \daStart{\mendactsrt_{10}} \) requires,
    so the second mend can only start after a delay \( \delta_2 \) (\( t_3 \)).}
    \label{fig:matchcellar_pictogram}
\end{figure}

\clearpage

\section{Propositional Example}

\begin{example}[Rooms, Robots, Doors, Blocks (Intuition)]\label{ex:example_intuition}
We provide an example of a temporal planning problem concerning
rooms (denoted as \( {\roomv} \)),
robots (denoted as \( {\robv} \)),
doors (denoted as \( {\doorv} \)), and 
blocks (denoted as \( {\blockv} \)).

We refer to propositions in text by treating a descriptive word as a function
and objects as arguments in a function call.
For example, \emph{``robot number 1 is holding block number 1''} is \( \holdingp(\robv_1,\blockv_1)\).

Similarly, we treat actions as functions and use objects as arguments. 
For example, the action of \emph{``robot number 1 picking up the block from within room 1''} is:
\(\splitatcommas{\puact(\robv_1,\blockv,\roomv_1)
\equiv
\langle
    \langle
        \{ \idlep(\robv_1), \inp(\blockv, \roomv_1),
           \batteryf(\robv_1) - \pucostf(\blockv) \ge 0 \},
        \emptyset,
        \{ \idlep(\robv_1), \inp(\blockv, \roomv_1) \}
    \rangle,
    \{ \inp(\robv_1, \roomv_1) \},
    \langle 
        \emptyset,
        \{ \idlep(\robv_1), \holdingp(\robv_1, \blockv) \},
        \emptyset
    \rangle,
    \langle 4, \le \rangle,
    \langle 4, \le \rangle%
\rangle}\).

This resembles PDDL's syntax with commas and parentheses instead of spaces to separate arguments.
\end{example}

\begin{figure}
    \begin{center}
        \begin{tikzpicture}
            \def\xsep{5} 

            \def\xstep{1.5}
            \def\ystep{1.5}

            \def\xmult{3}
            \def\ymult{-1.5}

            \tikzset{room width=2*\xstep, room height=\ystep}

            \foreach \x/\y/\lbtx/\ds in {
                0/0/$t_1$/doorclosed,
                1/0/$t_2$/doorclosed} {
                \pic at (\x*\xmult*\xstep,\y*\ymult*\ystep) {room={\ds}{\lbtx}};
            }

            \begin{scope}[shift={(0*\xmult*\xstep,0*\ymult*\ystep)}]
                \draw pic[global scale=0.375] at (0.125*\xstep,0.125*\ystep) {block};
                \draw pic[global scale=0.375] at (0.25*\xstep,0.65*\ystep) {robot1};
                \draw pic[global scale=0.375] at (0.55*\xstep,0.3*\ystep) {robot2};
            \end{scope}
            
            \begin{scope}[shift={(1*\xmult*\xstep,0*\ymult*\ystep)}]
                \draw pic[global scale=0.375] at (0.25*\xstep,0.65*\ystep) {r1b1r};
                \draw pic[global scale=0.375] at (0.55*\xstep,0.3*\ystep) {r2b1r};
            \end{scope}

            \node[draw=none] (o1) at (2*\xstep, 0.5*\ystep) {};
            \node[draw=none] (i2) at (\xmult*\xstep, 0.5*\ystep) {};

            \draw ($(o1.east) + (0,0.2)$) edge[->]
                    node[midway,above=0.1]{\( \puactsrt_{11} \)}
                    node[midway,fill=white,inner sep=1pt]{\forbiddenmark}
                ($(i2.west) + (0,0.2)$);

            \draw ($(o1.east) + (0,-0.2)$) edge[->]
                    node[midway,below=0.1]{\( \puactsrt_{21} \)}
                    node[midway,fill=white,inner sep=1pt]{\forbiddenmark}
                ($(i2.west) + (0,-0.2)$);
                
        \end{tikzpicture}
    \end{center}
    \caption{Interfering actions in \Cref{ex:formalism}.
    Two robots picking up the same block.
    Both starting snap actions are conditioned on \( \inp(\blockv,\roomv_1) \) 
    and delete \( \inp(\blockv,\roomv_1) \).}
    \label{fig:invalid_actions}
\end{figure}

\begin{example}[Rooms, Robots, Doors, Blocks (Formalism)]\label{ex:formalism}
We identify objects using shorthand notation and indices, 
and thereby identify some propositions.
The set of relevant propositions \( \propositions \) in our problem is:
\(
\splitatcommas{
\{ \idlep(\robv_i), \closedp(\doorv), \openp(\doorv), 
\inp(\robv_i, \roomv_j), \inp(\blockv, \roomv_j), \holdingp(\robv_i, \blockv),
\connectp(\doorv,\roomv_i, \roomv_j) \} 
}
\),
for \( i \in \{1,2\} \) and \( j \in \{1,2\} \).

The set of numeric fluents in our problem is:
\(
\splitatcommas{
\{ \batteryf(\robv_i), \mvcostf(\roomv_j, \roomv_k), \pucostf(\blockv), \opcostf(\doorv) \}
}
\),
for \( i \in \{1,2\} \), \( j \in \{1,2\} \) and \( k \in \{1,2\} \),
where \( \mvcostf(\roomv_j, \roomv_k) = 2 \).

The set of actions \( \actions \) is:
\( \splitatcommas{\{
(\puactsrt_{ij} \equiv \puact(\robv_i,\blockv,\roomv_j)),
(\pdactsrt_{ij} \equiv \pdact(\robv_i,\blockv,\roomv_j)),
(\opactsrt_{ij} \equiv \opact(\robv_i,\doorv,\roomv_j)),
(\clactsrt_{ij} \equiv \clact(\robv_i,\doorv,\roomv_j)),
(\mvactsrt_{ijk} \equiv \mvact(\robv_i,\doorv,\roomv_j,\roomv_k))
\}} \),
for \( i \in \{1,2\} \), \( j \in \{1,2\} \) and \( k \in \{1,2\} \).

The precise conditions and effects are:\\
\(
\splitatcommas{
\{
(\opactsrt_{ij}
\mathord{\equiv}
\langle
    \langle
        \{ \idlep(\robv_i), \closedp(\doorv),
           \batteryf(\robv_i) - \opcostf(\doorv) \ge 0 \},
        \emptyset,
        \{ \idlep(\robv_i), \closedp(\doorv) \}
    \rangle,
    \{ \inp(\robv_i,\roomv_j) \},
    \langle 
        \emptyset,
        \{ \idlep(\robv_i), \openp(\doorv) \},
        \emptyset
    \rangle,
    \langle 3, \le \rangle,
    \langle 3, \le \rangle 
\rangle),
(\pdactsrt_{ij}
\mathord{\equiv}
\langle
    \langle
        \{ \idlep(\robv_i), \holdingp(\robv_i, \blockv) \},
        \emptyset,
        \{ \idlep(\robv_i), \holdingp(\robv_i, \blockv)\}
    \rangle,
    \{ \inp(\robv_i, \roomv_j) \},
    \langle 
        \emptyset,
        \{ \idlep(\robv_i),  \inp(\blockv, \roomv_j) \},
        \emptyset
    \rangle,
    \langle 4, \le \rangle,
    \langle 4, \le \rangle 
\rangle),
(\clactsrt_{ij}
\mathord{\equiv}
\langle
    \langle
        \{ \idlep(\robv_i), \openp(\doorv) \},
        \emptyset,
        \{ \idlep(\robv_i), \openp(\doorv) \}
    \rangle,
    \{ \inp(\robv_i,\roomv_j) \},
    \langle 
        \emptyset,
        \{ \idlep(\robv_i), \closedp(\doorv) \},
        \emptyset
    \rangle,
    \langle 3, \le \rangle,
    \langle 3, \le \rangle 
\rangle),
(\mvactsrt_{ijk}
\mathord{\equiv}
\langle
    \langle
        \{ \idlep(\robv_i), \inp(\robv_i,\roomv_j), \connectp(\doorv,\roomv_j,\roomv_k),
           \batteryf(\robv_i) - \mvcostf(\roomv_j,\roomv_k) \ge 0 \},
        \emptyset,
        \{ \idlep(\robv_i), \inp(\robv_i,\roomv_j) \}
    \rangle,
    \{ \openp(\doorv) \},
    \langle 
        \emptyset,
        \{ \idlep(\robv_i), \inp(\robv_i,\roomv_k) \},
        \emptyset
    \rangle,
    \langle 2, \le \rangle,
    \langle 5, \le \rangle 
\rangle)
}
\)
\end{example}

\begin{figure}
\centering
\begin{tikzpicture}
    \begin{scope}[every state/.style={circle,thick,draw}, node distance=2]
        \node[state] (0) at (0,0) {$l_0$};
        \node[state] (1) [right=6.5cm of 0] {$l_1!$};
    \end{scope}
    \path [->] (0) edge node[above] {$ \langle (y = 2)? \land \top, \{(x \ge 1)?\}, \{ y := y - 1 \}, \{ x \} \rangle$} (1);
    \draw (0) edge[->,out=300,in=330,loop] node[right] {$\delta = 1.5$} ();
    \draw (1) edge[->,out=240,in=210,loop] node[left=2mm] {Not allowed: $\delta$} node[sloped, midway] {$||$} ();
\end{tikzpicture}
\caption{A timed automaton with named locations.
The action transition checks the condition \((y = 2)? \land \top \) and the guard \( \{(x \ge 1)?\}\), applies the update \( \{ y := y - 1 \} \), and resets \( \{ x \} \).}
\label{fig:TA_example}
\end{figure}

\begin{figure}[t]
    \providecommand{\timelinescale}{1}
    \begin{center}
        \begin{subfigure}{\linewidth}
            \begin{center}
            \scalebox{\timelinescale}{\begin{tikzpicture}
                \def\xsep{5} 

                \def\xstep{1.5}
                \def\ystep{1.5}

                \def\xmult{3}
                \def\ymult{-1.5}

                \tikzset{room width=2*\xstep, room height=\ystep}

                \foreach \x/\y/\lbtx/\ds in {
                    0/0/$t_2$/doorclosed,
                    1/0/$t_3$/doorhalfopen,
                    0/1/$t_3$/doorhalfopen,
                    1/1/$t_6$/dooropen,
                    0/2/$t_4$/dooropen,
                    1/2/$t_7$/doorhalfopen} {
                    \pic at (\x*\xmult*\xstep,\y*\ymult*\ystep) {room={\ds}{\lbtx}};
                }

                \begin{scope}[shift={(0*\xmult*\xstep,0*\ymult*\ystep)}]
                    \draw pic[global scale=0.375] at (0.3*\xstep,0.65*\ystep) {r1b1r};
                    \draw pic[global scale=0.375] at (0.65*\xstep,0.3*\ystep) {robot2};
                \end{scope}
                
                \begin{scope}[shift={(0*\xmult*\xstep,1*\ymult*\ystep)}]
                    \draw pic[global scale=0.375] at (0.3*\xstep,0.65*\ystep) {r1b1r};
                    \draw pic[global scale=0.375] at (0.65*\xstep,0.3*\ystep) {robot2};
                \end{scope}
                
                \begin{scope}[shift={(1*\xmult*\xstep,0*\ymult*\ystep)}]
                    \draw pic[global scale=0.375] at (0.3*\xstep,0.65*\ystep) {r1b1r};
                    \draw pic[global scale=0.375] at (0.65*\xstep,0.3*\ystep) {robot2};
                \end{scope}

                \begin{scope}[shift={(1*\xmult*\xstep,1*\ymult*\ystep)}]
                    \draw pic[global scale=0.375] at (1.65*\xstep,0.3*\ystep) {robot2};
                \end{scope}
                
                \begin{scope}[shift={(0*\xmult*\xstep,2*\ymult*\ystep)}]
                    \draw pic[global scale=0.375] at (0.5*\xstep,0.5*\ystep) {robot2};
                \end{scope}
                
                \begin{scope}[shift={(1*\xmult*\xstep,2*\ymult*\ystep)}]
                    \draw pic[global scale=0.375] at (1.5*\xstep,0.65*\ystep) {r1b1r};
                    \draw pic[global scale=0.375] at (1.7*\xstep,0.25*\ystep) {robot2};
                \end{scope}

                \foreach \ox/\oy/\ix/\iy/\oxoff/\oyoff/\ixoff/\iyoff/\nmout/\nmin in {
                    0/0/1/0/2/0.5/0/0.5/out1/in2,
                    0/0/0/1/0.5/0/0.5/1/out12/in31,
                    0/1/1/1/2/0.5/0/0.5/out3/in4,
                    0/1/0/2/0.5/0/0.5/1/out32/in5,
                    0/1/0/2/1.5/0/1.5/1/out33/in52,
                    1/0/1/1/0.5/0/0.5/1/out2/in42,
                    1/0/1/1/1.5/0/1.5/1/out22/in43,
                    1/1/1/2/0.5/0/0.5/1/out4/in6,
                    0/2/1/1/2/1/0/0/out5/in44}
                     {
                \node[draw=none] (\nmout) at (\ox*\xmult*\xstep + \oxoff*\xstep, \oy*\ymult*\ystep + \oyoff*\ystep) {};
                \node[draw=none] (\nmin) at (\ix*\xmult*\xstep + \ixoff*\xstep, \iy*\ymult*\ystep + \iyoff*\ystep) {};
                }

                \draw (out1.east) 
                    edge[->, draw=\nopacol] 
                    node[midway,sloped,above]{\( \nopa \)} 
                    (in2.west);
                    
                \draw (out12.south)  
                    edge[->, draw=\nopacol]
                    node[midway,left]{\( \nopa \)}
                    (in31.north);

                \draw (out2.south) edge[->,draw=\nmvraiicol]
                    node[midway,left]{\( \nmvraii \)}
                    node[midway,fill=white,inner sep=1pt]{\forbiddenmark}
                    (in42.north);
                    
                \draw (out22.south) edge[->,draw=\nmvraicol]
                    node[midway,right]{\( \nmvrai \)}
                    (in43.north);

                \draw (out32.south) edge[->]
                    node[midway,left]{\( \delta_4 \)}
                    (in5.north);
                    
                \draw (out33.south) edge[->,draw=\nmvraicol]
                    node[midway,right]{\( \nmvrai \)}
                    (in52.north);
                    
                \draw[draw=none] ($(out5) + (-0.05,0.1)$) --
                    ($(in44) + (-0.05,0.1)$)
                    
                    coordinate[near end] (mid1) {};
                \draw ($(out5) + (-0.05,0.1)$) edge[->,draw=\nmvraiicol]
                    node[midway,sloped,above]{\( \nmvraii \)}
                    (mid1);
                    
                \draw (mid1) edge[->]
                    node[midway,sloped,above]{\( \delta_6 \)}
                    ($(in44) + (-0.05,0.1)$);
                   
                \draw ($(out5) + (0.05,-0.1)$)  edge[->,draw=\nmvraicol]
                    node[midway,sloped,below]{\( \nmvrai \)}
                    ($(in44) + (0.05,-0.1)$);

                \draw[draw=none] (out4.south)edge[->,draw=\nclacol]
                    node[midway,left]{\( \ncla \)}
                    (in6.north);

            \end{tikzpicture}}
                \caption{Partial state space with transitions representing executions of durative actions \( \durAct{} \) and time delays \( \delta \). 
                A forbidden transition according to \Cref{def:mutex_snap} representing the start of \( \nmvraii \) is shown at \( t_3 \) in the top right. }
                \label{fig:timeline_pictogram}
            \end{center}
        \end{subfigure}
        \begin{subfigure}{\linewidth}
            \stackMath
            \begin{center}
                \scalebox{\timelinescale}{\begin{tikzpicture}
                    \def\htp{-2.125} 
                    \draw[->] (-0.75,0) -- ++(7.75,0);
                    \draw[->] (-0.75,\htp) -- ++(7.75,0);
                    
                    \def\labelloc{-0.5} 
                    \draw (\labelloc,0) node[above, anchor=south east] {time};
                    \draw (\labelloc,0) node[below, anchor=north east] {$\plan$};
                    \draw (\labelloc,\htp) node[below, anchor=north east] {$\recSetActs_{\plan}$};
    
                    \foreach \x/\w/\y/\he/\tx/\clr [count=\i, evaluate=\i as \n using int(\i+2)] in {
                        0/1.5/-0.125/0.375/$\nopa$/\nopacol,
                        1.5/3/-0.625/0.375/$\nmvrai$/\nmvraicol, 
                        4.5/1.5/-0.125/0.375/$\ncla$/\nclacol,
                        2.5/1/-1.125/0.375/$\nmvraii$/\nmvraiicol
                    } {
                        \draw[action box=\clr] (\x, \y) rectangle ++(\w,-\he) node[pos=.5] {\tx};
                    }

                    \draw[invalid action box=\nmvraiicol]
                        (1.5,-1.625) rectangle
                        ++(2.5,-0.375) node[pos=.5] {$\nmvraii$};

                    \draw (1.5,-1.625) ++($(0,-0.1875)$) node[fill=white,inner sep=1pt] {\forbiddenmark};
                    
                    \draw[draw=\nmvraiicol!60,solid] 
                        (1.5,-1.625) edge ++(2.5,0)  
                        (1.5,-1.625) edge ++(0,-0.375) 
                        (1.5,-1.625) ++(0,-0.375) edge ++(2.5,0);
                    
                    \foreach \x/\w/\ts/\te/\hps/\hpe in {
                        0/1.5/$t_2$/${}$/$\daStart{\nopa}$/$\Shortstack{
                            {\daEnd{\nopa}}
                            {\daStart{\nmvrai}}
                            {\raisebox{-1pt}{\text{\usym{1F6AB}}}\text{:}\daStart{\nmvraii}}}$, 
                        1.5/3/$t_3$/${}$/${}$/$\stackanchor{\daEnd{\nmvrai}}{\daStart{\ncla}}$, 
                        4.5/1.5/$t_6$/$t_7$/${}$/$\daEnd{\ncla}$,
                        2.5/1/$t_4$/$t_5$/
                        $\daStart{\nmvraii}$/
                        $\daEnd{\nmvraii}$
                    } {
                        \draw[draw=gray,dashed] (\x,0) 
                            node[anchor=south, above=1mm] {\ts}
                            -- ++(0,\htp) 
                            node[anchor=north, below=1mm] {\hps};
                        \draw[draw=gray,dashed] (\x+\w,0) 
                            node[anchor=south, above=1mm] {\te}
                            -- ++(0,\htp) 
                            node[anchor=north, below=1mm] {\hpe};
                        \draw (\x,0) -- ++(0,2mm);
                        \draw (\x+\w,0) -- ++(0,2mm);
                        \draw (\x,\htp) -- ++(0,-2mm);
                        \draw (\x+\w,\htp) -- ++(0,-2mm);
                    }
                \end{tikzpicture}}
                \caption{Timeline illustrating valid and invalid schedules for actions.}
                \label{fig:example_timeline}
            \end{center}
        \end{subfigure}
        \caption{Parts of candidate plans for \Cref{ex:formalism}. }
        \label{fig:plan_timeline}
    \end{center}
\end{figure}
\begin{example}[Rooms, Robots, Doors, Blocks (Encoding)]\label{ex:valid_instant}
Continuing our example, we assign an index \( i \) to some actions to uniquely identify them:\\
\( \splitatcommas{
(\nopa \equiv \opactsrt_{21} \equiv \opact(\robv_2,\doorv,\roomv_1)), 
(\nmvrai \equiv \mvactsrt_{112} \equiv \mvact(\robv_1,\roomv_1,\roomv_2,\doorv)), 
(\nmvraii \equiv \mvactsrt_{212} \equiv \mvact(\robv_2,\roomv_1,\roomv_2,\doorv)), 
(\ncla \equiv \clactsrt_{22} \equiv \clact(\robv_2,\doorv,\roomv_2))
} \).

The top right state in \Cref{fig:timeline_pictogram} characterises the instant in which \( \nopa \) ends.
Because \( \daEnd{\nopa} \) has not been applied, neither \( \openp(\doorv) \) nor \( \closedp(\doorv) \) is satisfied. \( \nmvraii \) cannot start executing when \( \nopa \) ends, 
because \( \daEnd{\nopa} \) adds \( \idlep(\robv_2) \) while \( \daStart{\nmvraii} \) deletes \( \idlep(\robv_2) \), 
which is invalid according to \Cref{def:mutex_snap}.
This invalidates any immediate application of \( \daStart{\nmvraii} \) to 
the top right state labelled \( t_3 \) to eventually obtain the state at \( t_6 \).
A valid alternative is shown in the left of \Cref{fig:timeline_pictogram}.
A delay of \( \delta_4 \equiv t_4 - t_3\) must separate \( \daEnd{\nopa} \) from \( \daStart{\nmvraii} \),
which can be seen in the transition from \( t_3 \) to \( t_4 \) in the left.
Similarly, a delay of \( \delta_6 \equiv t_6 - t_5 \) must separate the end of \( \nmvraii \) from the start of \( \ncla \). 
On the other hand, 
since \( \nmvrai \) checks \( \openp(\doorv) \) in \( \daInv{\nmvrai} \)
but neither in \( \daStart{\nmvrai} \) nor \( \daEnd{\nmvrai} \)'s precondition,
\( \nmvrai \)'s execution can be started at \( t_3 \) and ended at \( t_6 \).
\end{example}

\begin{figure}[t]
    \centering
    \begin{tikzpicture}
        \def\xA{6} \def\yA{1.35} \def\xB{-1} \def\yB{-2.125}
        \path[use as bounding box, draw, dotted] (\xA,\yA) rectangle (\xB,\yB);
        \begin{scope}[
                every state/.style={
                    draw,
                    rectangle
                },
                >={Stealth},
                shorten >=1pt,
                node distance=2,
                on grid
            ]
            \node[state] (inactive) at (0,0) {\olc};
            \node[state] (starting) at (2.5,0.75) {\slc!};
            \node[state] (running) at (5, 0) {\rlc};
            \node[state] (ending) at (2.5,-0.75) {\elc!};
            
            \node[anchor=north west] (lbl) at (-0.9,-0.75) 
                {\( \eee_{\nopa}( t_2, 3 ) \)};

            \node[below=1mm of lbl.south west, anchor=north west]
                {\( \langle \vp{\openp(\doorv)} \coloneq 1 \rangle \)};
                
            \node[anchor=north east] at (5.9,-0.75) {\( \ese_{\nopa} ( t_2, 1 ) \)};
        \end{scope}
        \begin{scope}[>={Stealth},shorten >=1pt,auto]
            \draw[->] (inactive) |- (starting) 
                node[midway, above, anchor=south] {\( \sse_{\nopa} \)};
            \draw[->] (starting) -| (running) 
                node[midway, above, anchor=south] {\( \see_{\nopa} \)};
            \draw[->, dashed, draw=\nopacol] (running) |- (ending);
            \draw[->, dashed, draw=\nopacol] (ending) -| (inactive);
            \draw[->] (starting) -- (ending) 
                node[midway, anchor=west] {\( \inste_{\nopa} \)};
        \end{scope}
    \end{tikzpicture}
    
    \begin{tikzpicture}
        \def\xA{6} \def\yA{2.125} \def\xB{-1} \def\yB{-2.125}
        \path[use as bounding box, draw, dotted] (\xA,\yA) rectangle (\xB,\yB);
        \begin{scope}[
                every state/.style={
                    draw,
                    rectangle
                },
                >={Stealth},
                shorten >=1pt,
                node distance=2,
                on grid
            ]
            \node[state] (inactive) at (0,0) {\olc};
            \node[state] (starting) at (2.5,0.75) {\slc!};
            \node[state] (running) at (5, 0) {\rlc};
            \node[state] (ending) at (2.5,-0.75) {\elc!};
        \end{scope}
            
        \node[anchor=south west] (lbl) at (-0.9,0.75) 
            {\( \sse_{\nmvrai} ( t_2, 2 ) \)};
        \node[above=1mm of lbl.north west, anchor=south west] 
            {\( \langle (\vp{\openp(\doorv)} = 1)?, \lp{\openp(\doorv)} \pluseq 1 \rangle \)};
        
        \node[anchor=south east] at (5.9,0.75) {\( \see_{\nmvrai} ( t_2, 4 )\)};
        \node[anchor=north west] at (-0.9,-0.75) {\( \eee_{\nmvrai}  ( t_6, 3 ) \)};
        
        \node[anchor=north east] (lbl) at (5.9,-0.75) 
            {\( \ese_{\nmvrai} ( t_6, 1 ) \)};

        \node[below=1mm of lbl.south east, anchor=north east]
                {\( \langle \lp{\openp(\doorv)} \minuseq 1 \rangle \)};
        
        \begin{scope}[>={Stealth},shorten >=1pt,auto]
            \draw[->, dashed, draw=\nmvraicol] (inactive) |- (starting);
            \draw[->, dashed, draw=\nmvraicol] (starting) -| (running);
            \draw[->, draw=\nmvraicol] (running) |- (ending);
            \draw[->, draw=\nmvraicol] (ending) -| (inactive);
            \draw[->] (starting) -- (ending) 
                node[midway, anchor=west] {\( \inste_{\nmvrai} \)};
        \end{scope}
    \end{tikzpicture}
    
    \begin{tikzpicture}
        \def\xA{6} \def\yA{2.125} \def\xB{-1} \def\yB{-1.35}
        \path[use as bounding box, draw, dotted] (\xA,\yA) rectangle (\xB,\yB);
        \begin{scope}[
                every state/.style={
                    draw,
                    rectangle
                },
                >={Stealth},
                shorten >=1pt,
                node distance=2,
                on grid
            ]
            \node[state] (inactive) at (0,0) {\olc};
            \node[state] (starting) at (2.5,0.75) {\slc!};
            \node[state] (running) at (5, 0) {\rlc};
            \node[state] (ending) at (2.5,-0.75) {\elc!};
        \end{scope}
        
        \node[anchor=south west] (abc) at (-0.9,0.75) 
            {\( \sse_{\ncla}( t_6, 2 ) \)};
        \node[above=1mm of abc.north west, anchor=south west]
            {\( \langle (\lp{\openp(\doorv)} = 0)?, \vp{\openp(\doorv)} \coloneq 0 \rangle \)};
        
        \node[anchor=south east] at (5.9,0.75) {\( \see_{\ncla} ( t_6, 4 ) \)};
        
        \begin{scope}[>={Stealth},shorten >=1pt,auto]
            \draw[->, draw=\nclacol] (inactive) |- (starting);
            \draw[->, draw=\nclacol] (starting) -| (running);
            \draw[->] (running) |- (ending) 
                node[midway, below, anchor=north] {\( \ese_{\ncla} \)};
            \draw[->] (ending) -| (inactive) 
                node[midway, below, anchor=north] {\( \eee_{\ncla} \)};
            \draw[->] (starting) -- (ending) 
                node[midway, anchor=west] {\( \inste_{\ncla} \)};
        \end{scope}
    \end{tikzpicture}
 
    \caption{Examples of automata for actions from \Cref{fig:plan_timeline}.
    Dashed edges are applied to simulate \( t_2 \) (\rule[.5ex]{0.2em}{.4pt}\hspace{0.2em}\rule[.5ex]{0.2em}{.4pt}\hspace{0.2em}\rule[.5ex]{0.2em}{.4pt}\hspace{0.2em}\rule[.5ex]{0.2em}{.4pt}) and solid ones to simulate \( t_6 \) (\rule[.5ex]{1.4em}{.4pt}).
    The order is indicated with respect to each time point.}
    \label{fig:construct_run}
\end{figure}
\begin{example}[Rooms, Robots, Doors, Blocks (Transition Order)]\label{ex:invalid_instant}
    In \Cref{fig:construct_run} we see an example of a few transitions' conditions
    and effects in automata for actions from \Cref{fig:plan_timeline}.
    We can see that some transitions to simulate snap action executions
    at a time point must occur in a certain order.
    \( \eee_{\nopa} \) sets \( \vp{\openp(\doorv)} \) to \( 1 \) at \( t_2 \)  
    and \( \vp{\openp(\doorv)} \) must equal \( 1 \) before \( \see_{\nmvrai} \) can be taken.
    Since \( \see_{\nmvrai} \) has incremented \( \lp{\openp(\doorv)} \) by \( 1 \)
    it must be decremented by \( \ese_{\nmvrai} \) before  \( \sse_{\ncla} \) can be taken at \( t_6 \).
\end{example}

\clearpage

\section{Proof of \Cref{thm:main_theorem}}
\label{appendix:proofs}
We now construct a run \( \langle \bLvc{_{i;0}}, \allowbreak \dots, \allowbreak \bLvc{_{i;l_i}} \rangle \) for every happening time point \( \happTimePoint_i \), which satisfies \( \encbef(\bLvc{_I},\allowbreak i,\allowbreak\bLvc{_{i;0}}) \) and \( \encaft(\bLvc{_I},\allowbreak i,\allowbreak\bLvc{_{i;l_i}}) \), where \( l_i \) indicates the length of the run.

Furthermore, if \( 0 < i \), then applying a delay of 
\( \delta_i = \happTimePoint_i - \happTimePoint_{i-1} \) to \( \bLvc{_{i-1;l_{i-1}}} \) creates \( \bLvc{_{i-1;l_{i-1}}'} \),
which is equivalent to \( \bLvc{_{i;0}} \) according to the conditions in \Cref{def:enc_states}.
E.g.\ \( \bLvc{_{i-1;l_{i-1}}} \longrightarrow_{\delta_i} \bLvc{_{i;0}} \).

We first give a more precise characterisation of encoding the state of a world just before a 
set of snap actions at a time point \( \happTimePoint_i \)
is applied.

\begin{paperdefinition}{Encoded states before}\label{def:enc_states_bef}
    From \( \encaft(\bLvc{_I}, \allowbreak i, \allowbreak \bLvc{}) \) (\Cref{def:enc_states})
    we derive \( \encbef(\bLvc{_I}, \allowbreak i, \allowbreak \bLvc{}) \)
    by substituting \( \planningState_i \) for \( \planningState_{i+1} \),
    \(\runat(\durAction, \happTimePoint_i, <) \) for \( \runat(\durAction, \happTimePoint_i, \le) \),
    and \(\timesinceat(\snapAct{}, \happTimePoint_i, <) \) for \( \timesinceat(\snapAct{}, \happTimePoint_i, \le) \).

\end{paperdefinition}

\begin{paperdefinition}{Encoded initial state}\label{def:encaft_initial_state}
    When the htp sequence is empty,
    \( \encaft(\bLvc{_I}, i, \bLvc{_i}) \) holds for any \( i \) if the following hold:
    \begin{enumerate*}[itemjoin={;\ }, itemjoin*={.\ }]
        \item For any \( \durAct{j} \in \actions \):
        \begin{enumerate*}[itemjoin={;\ }, itemjoin*={; and\ }, label=(\roman{enumi}.\alph*)]
            \item \( \Loc_i\) records \( l_j \) as \( \olc_{\durAction} \)
            \item \( 0 < \Clk_i (\clk{\daStart{\durAct{j}}}) \) and \( 0 < \Clk_i (\clk{\daEnd{\durAct{j}}}) \)
        \end{enumerate*}
        \item For any \( p \in \propositions \):
        \begin{enumerate*}[itemjoin={;\ }, itemjoin*={; and\ }, label=(\roman{enumi}.\alph*)]
            \item \( \Var_i (\vp{p}) \) evaluates to \( 1 \) if \( p \in \propof{\initState} \)
                and \( 0 \) if \( p \notin \propof{\initState} \)
            \item \( \Var_i (\lp{p}) \) evaluates to \( 0 \)
        \end{enumerate*}
        \item For any fluent \( f \in \fluents \):
        \begin{enumerate*}[itemjoin={;\ }, itemjoin*={; and\ }, label=(\roman{enumi}.\alph*)]
            \item \( \Var_i (f) \) evaluates to \( \numof{\initState}(f) \)
        \end{enumerate*}
    \end{enumerate*}.

\end{paperdefinition}

\begin{lemma}\label{lem:variable_bounds_respected}
    Take \( \infbox \) to be the environment returned by \fn{infer-bounds} (\Cref{lst:infer_bounds}).
    Take a numeric state \( \numof{\planningState} \) whose assignment \( \numof{\planningState}(f) \in \infbox(f) \)
    within the bounds \( \infbox(f) \) for any fluent \( f \in \fluents \).
    Assume \( \numof{\planningState} \) satisfies the preconditions \( \saNumPre{h} \) of a snap action \( h \).

    \begin{enumerate}[label={Claim~\arabic*.}, ref=\arabic*, align=left, leftmargin=*]
        \emergencystretch=2em
        \item\label{claim:bounds_single_effect}
        Take any numeric effect \( (f \coloneq e) \in \saNumEff{h} \).
        Evaluating \( e \) against \( \numof{\planningState} \)
        results in a value \( f' \) that is within \( \infbox(f) \).
        \item\label{claim:bounds_all_effects}
        Evaluating all \( (f_i \coloneq e_i) \in \saNumEff{h} \) in
        any order against \( \numof{(\planningState'_i)} \), where \( \numof{(\planningState'_0)} = \allowbreak \numof{\planningState} \)
        and \( \numof{(\planningState'_{i+1})}(g) = \allowbreak \numof{(\planningState'_i)}(g) \) where \( g \neq f_i \)
        and \( \numof{(\planningState'_{i+1})}(f_i) = \allowbreak \eval(\numof{(\planningState'_i)}, e_i) \) otherwise
        maintains that \( \numof{(\planningState'_{i+1})}(f) \in \infbox(f) \) for any \( f \in \fluents \)
        where \( 0 \le i < |\saNumEff{h}| \).
        \item\label{claim:bounds_init}
        The initial numeric state \( \numof{\initState} \) is also subsumed by the box.
    \end{enumerate}
\end{lemma}
\begin{proof}[Sketch]
    Take the final iteration of \fn{infer-bounds}.
    We know that \( \mathit{env} \) has reached a pre-fixed point.
    We know that \( \mathit{env} \sqsupseteq \mathit{env}'' \sqsupseteq \mathit{env}' \) because of widening and the loop's condition.
    \fn{guarded-asteps} only updates fluents when the conditions of \( h \) are satisfied.
    \( \mathit{env}' \sqsupseteq \mathit{env} \), because \fn{guarded-asteps} is monotone.
    Hence \( \mathit{env}'' = \mathit{env}' = \mathit{env} \).

    Also, for any single action whose numeric preconditions are satisfied by \( \mathit{env} \),
    evaluating all numeric effects with respect to \( \mathit{env} \) and applying them at once,
    cannot have made the accumulator wider than \( \mathit{env} \) in \fn{guarded-asteps},
    since the union of intervals in the accumulator is monotonic.
    In fact, applying any single one of them could not have made the accumulator wider.
    Thus follows Claim~\ref{claim:bounds_single_effect}.

    We also know that these numeric effects do not interfere.
    Hence evaluating any one \( e_i \) against some \( \numof{(\planningState'_i)} \) results in
    \( \eval(\numof{(\planningState'_0)}, e_i) = \eval(\numof{\planningState}, e_i) \),
    which maintains \( \numof{(\planningState'_{i+1})}(f) \in \infbox(f) \) for all \( f \in \fluents \).
    Thus follows Claim~\ref{claim:bounds_all_effects}.

    Claim~\ref{claim:bounds_init} follows from monotonicity of the repeat-until loop and the fact that \( \mathit{env} \) starts with all assignments in \( [\numof{\initState}(f), \numof{\initState}(f)] \).
\end{proof}

\begin{lemma}\label{lem:plan_trans_possible}
    Given a configuration \( \encaft(\bLvc{_I},\allowbreak 0,\allowbreak\bLvc{_0}) \),
    it is possible to prove the existence of a configuration \( \bLvc{_k} \) 
    and a run \( \langle \bLvc{_0}, \allowbreak \dots, \allowbreak \bLvc{_k} \rangle \),
    such that \( run([\bLvc{_0}, \allowbreak \dots, \allowbreak \bLvc{_k}]) \)
    and \( \encaft(\bLvc{_I},\allowbreak k,\allowbreak\bLvc{_k}) \) are satisfied,
    where \( k \le m \).
\end{lemma}
The proof of \Cref{lem:plan_trans_possible} proceeds by induction on \( k \).
The base case holds since \( run([\bLvc{_0}]) \) holds.
From \( \encaft(\bLvc{_I},\allowbreak k,\allowbreak\bLvc{_k}) \) we can deduce \( \encbef(\bLvc{_I},\allowbreak  k+1, \langle \Loc_k, \Var_k, \Clk_k \oplus \delta_{k+1} \rangle ) \),
where \( \delta_{k+1} = t_{k+1} - t_k \), \( \{ t_k, t_{k+1} \} \subseteq \happTimePoints(\plan) \), 
and there exists no \( s \) in \( \happTimePoints(\plan) \), such that \( t_k < s < t_{k+1} \).
This is because no snap action in the induced parallel plan can change the discrete state of the world between 
\( t_k \) and \( t_{k+1} \). 
From \( \encbef(\bLvc{_I},\allowbreak  k+1, \langle \Loc_k, \Var_k, \Clk_k \oplus \delta_{k+1} \rangle ) \)
all automata \( \autonet{_\durAction} \), where \( \langle t_{k+1}, \daStart{\durAction} \rangle \in \recSetActs_{\plan} \) and/or \( \langle t_{k+1}, \daEnd{\durAction} \rangle \in \recSetActs_{\plan} \), 
can apply the transitions encoding the conditions and effects of the start \( \daStart{\durAction} \)
and/or end \( \daEnd{\durAction} \) of \( \durAction \) at \( t_{k+1} \).
Moreover, each action (and hence each automaton)
can at most add \( 1 \) to the total number of active actions,
encoded by \( \Var(\actsactive) \) and one to each of \( \Var(\lp{p}) \),
which means that both must stay beneath \( N \).
When an update decrements it, the \( 1 \) must have been contributed,
so \( \Var(\actsactive) > 0 \) and \( \Var(\lp{p}) > 0 \) hold.
When an update increments it, the action does not yet contribute \( 1 \),
which means that \( \Var(\actsactive) < N  \) and \( \Var(\lp{p}) < N \) hold.

From a configuration \( \bLvc{_I} \), which encodes the initial state \( \initState \) of the planning problem at \( t_0 \), it is possible to reach a configuration \( \bLvc{_k} \), 
which encodes the state of the world after the final happening time point \( t_k \) of the plan.
To prove this, we can use \Cref{lem:plan_trans_possible}.
To use \Cref{lem:plan_trans_possible}, we need a \( \bLvc{_0} \), 
which encodes the state of the planning problem after \(t_0\). 
This \( \bLvc{_0} \) can be reached from \( \bLvc{_I} \), 
which we prove in \Cref{lem:init_trans_possible}.

\begin{lemma}\label{lem:init_trans_possible}
    There exists a finite run \( \langle \bLvc{_I}, \allowbreak \dots, \allowbreak \bLvc{_0} \rangle \) from the initial configuration
    \( \bLvc{_I} \) to some configuration \( \bLvc{_0} \), which satisfies \( \encaft(\bLvc{_I},\allowbreak 0,\allowbreak\bLvc{_0}) \).
\end{lemma}
This holds, because \( \initedge \) encodes the updates necessary to initialise the discrete components of the configuration
of the network \( \Loc_I' \) and \( \Var_I' \) to encode the initial state \( \initState \) of the planning problem.
Another \( \longrightarrow_{\delta_I} \) transition can initialise \( \Clk_I' \), where \( \delta_I > 0 \).
If there are no entries \( \langle t_0, \snapAct{} \rangle \) in the induced parallel plan \( \recSetActs_{\plan} \), then \( \encbef(\bLvc{_I},\allowbreak  0, \bLvc{_I'}) \iff \encaft(\bLvc{_I},\allowbreak  0, \bLvc{_I'})\).
If there is at least one entry, then all transitions to check the conditions of any \( \snapAct{} \), 
where \( \langle t_0, \snapAct{} \rangle \in \recSetActs_{\plan} \), must be applicable in some order.
After applying all transitions, some \( \bLvc{_0} \) can be reached, where \(\encaft(\bLvc{_I},\allowbreak  0, \bLvc{_0})\).
Again \( \Var(\actsactive) \) and \( \Var(\lp{p}) \) respect the bounds, using a similar argument to \Cref{lem:plan_trans_possible}. 

A configuration \( \bLvc{_G} \) satisfying the goal condition can be reached from \( \bLvc{_k} \) in one transition, 
if \( \bLvc{_k} \) encodes the state of the planning problem after the final happening time point \( t_k \in \happTimePoints(\plan) \).

\begin{lemma}\label{lem:goal_trans_possible}\emergencystretch=6em
    We assume that a configuration \( \bLvc{_k} \) satisfies \( \encaft(\bLvc{_I},\allowbreak k,\allowbreak\bLvc{_k}) \),
    where \( t_k \) is the final happening time point in \( \happTimePoints(\plan) \).
    We also assume that \( \Loc_k \) records \( \planningloc \) for \( \autonet{_M} \).
    Then there exists a valid run \( \langle \bLvc{_k}, \allowbreak \dots, \allowbreak \bLvc{_G} \rangle \) ending in a configuration 
    in which \( \Loc_G \equiv \langle \goalloc, \dots \rangle \).
\end{lemma} 
It is sufficient to prove that \( \goaledge \) can be applied to \( \bLvc{_k} \), 
because applying \( \goaledge \) to any \( \bLvc{} \) results in a \( \bLvc{'} \), 
where \( \Loc' \equiv \langle \goalloc, \dots \rangle \). 
The conditions of \( \goaledge \) are satisfied, 
because \( \encaft(\bLvc{_I},\allowbreak k,\allowbreak\bLvc{_k}) \) is a stronger assumption.

\subsection{Proof of \Cref{lem:plan_trans_possible}}\label{sec:plan_trans_possible_proof}
Before continuing on to the proof of \Cref{lem:plan_trans_possible}, we refer to \Cref{ex:valid_instant,fig:matchcellar_pictogram,fig:construct_run,fig:matchcellar_run},
to informally illustrate our argument for the order of applied transitions.
\begin{proof}
The proof is by induction on \( k \).
    \emph{\textbf{Base case}: \( k = 0 \).}
    
    Since a single configuration is a run, we obtain the fact: \( run([\bLvc{_0}]) \).
    From the assumption \( \encaft(\bLvc{_I},\allowbreak  0, \bLvc{_0}) \) 
    and \Cref{def:enc_states} we trivially deduce \( \encaft(\bLvc{_I},\allowbreak  0, \bLvc{_0}) \), 
    which completes this case.
    
    \noindent\emph{\textbf{Inductive case}: \( k = i + 1 \) for an arbitrary \( i \).}
    
    We can assume the existence of some \( i \) and \( \bLvc{_i} \), 
    which satisfy \( \encaft(\bLvc{_I},\allowbreak  i, \bLvc{_i}) \)
    and \( run([\bLvc{_I}, \allowbreak \dots, \allowbreak \bLvc{_i}]) \) from the inductive hypothesis.
    
    We can prove that \( \bLvc{_i'} \) exists, 
    where \( \encbef(\bLvc{_I},\allowbreak  i+1, \bLvc{_i'}) \) and 
    \( \bLvc{_i} \longrightarrow_{\delta_{i+1}} \bLvc{_i'} \) 
    where \( \delta_{n} = \happTimePoint_{n} - \happTimePoint_{n-1} \).
    This is because \( t_{i+1} \) is the immediate successor time point of \( t_i \)
    according to \Cref{def:htps},
    and, therefore, only the durations recorded by \( \Clk{_i} \) 
    must be incremented to satisfy \( \encbef \).
    The properties encoded by \( \Var \) and \( \Loc \)
    refer to propositions and states of actions,
    which can only change at happening time points, 
    of which there are none between \( t_i \) and \( t_{i+1} \).

    Hence, we know of the existence of some \( \bLvc{_i'} \) that
    satisfies \( \encbef(\bLvc{_I},\allowbreak  i+1, \bLvc{_i'}) \).
    From \Cref{def:valid_plan}, we can deduce that this permits the network to simulate
    the execution of all snap actions at \( t_{i+1} \) in at least one order.

    An action can either be scheduled to \emph{start} at \( t_{i+1} \), \emph{end}, or both.
    Due to the no self-overlap condition (\Cref{def:no_self_overlap}), 
    if an action is scheduled to start and end,
    it must be scheduled with a duration of \( 0 \), i.e. be \emph{instantaneous}.

    All automata \( \autonet{_\durAction} \) for all \emph{ending} actions \( \durAction \) must transition from
    \( \rlc_\durAction \) to \( \elc_\durAction \). 
    This can decrement some \( \lp{p} \) for \( p \in \daInv{\durAction} \) to \( 0 \).
    This then permits \( \eee_{\durAction_2} \) and \( \sse_{\durAction_2} \) transitions to be taken
    without violating some  \( (\lp{p} = 0)? \) conditions, 
    if \( \daStart{\durAction_2} \) and \( \daEnd{\durAction_2} \) are scheduled at \( t_{i+1} \) in \( \recSetActs_{\plan} \).

    Automata for \emph{instantaneous} actions \( \durAction \) 
    do not take their \( \see_\durAction \) and \( \ese_\durAction \) transitions at all,
    skipping them by taking their \( \inste_\durAction \) transition.
    \( \see_\durAction \)'s guard, which checks \( \Var(\vp{p}) = 1 \) for \( p \in \daInv{\durAction} \),
    is not guaranteed to be satisfied, 
    because \( \daInv{\durAct{}} \) conditions are not active in the instants an action starts or ends.
    
    All \emph{instantaneous} actions \( \durAction \) must take \( \sse_\durAction \), 
    \( \inste_\durAction \), and \( \eee_\durAction \) transitions.
    
    Then all \emph{starting} snap actions' automata must take their \( \sse_\durAction \)
    and all \emph{ending} snap actions' automata must take their \( \eee_\durAction \) transitions.
    This applies their effects.

    Finally, any automaton \( \autonet_\durAction \) for a \emph{starting} snap action \( \daStart{\durAction} \)
    can take its \( \see_\durAction \) transition to lock the invariants \( \daInv{\durAction} \)
    by incrementing \( \lp{p} \) for all \( p \in \daInv{\durAction} \).
    
    Transitions in this order are possible 
    and lead to a configuration satisfying \( \encaft(\bLvc{_I},\allowbreak i+1,\allowbreak\bLvc{_{i+1}}) \) because:
    \begin{enumerate*}[itemjoin={;\ }, itemjoin*={; and\ }]
        \item no mutex constraints (see \Cref{def:mutex_snap}) in transitions' guards are violated, 
            due to  non-interference of snap actions in a valid plan (see \Cref{def:zero_eps_sep})
        \item the precondition on \( \vp{p} \) where \( p \in \saPre{\snapAct1} \) 
            of every snap action \( \snapAct1 \) scheduled at \( \happTimePoint_i \) 
            is satisfied by every intermediate \( \Var' \) in this part of the run, 
            because any precondition \( \saPre{\snapAct1} \) of a scheduled snap action \( \snapAct1 \)
            is already satisfied by \( \planningState_i \) (see \Cref{def:valid_state_sequence})
            and mutex constraints ensure that neither 
            \( \saAdd{\snapAct2} \) nor \( \saDel{\snapAct2} \) include any \( p \) in \( \saPre{\snapAct1} \) 
            if \( \snapAct2 \) and \( \snapAct1 \) are scheduled at \( \happTimePoint_i \)
        \item no \( \clk{\daStart{\durAction}} \) or \( \clk{\daEnd{\durAction}} \) clock that is reset by any transition 
            interferes with any other transition's guard, 
            because no mutually exclusive 
            snap actions can be scheduled at \( \happTimePoint_i \) (see \Cref{def:valid_plan,def:zero_eps_sep})
        \item duration constraints are satisfied (\( \dcsatf(\durAction,\Clk(\clk{\daStart{\durAction}})) \)) for all \emph{ending} snap actions
        \item every update stays within the bounds, 
            either by construction or because each active action \( \durAction \) (and automaton \( \autonet_\durAction \)) 
            can only contribute \( 1 \) to the value \( \Var(z) \)
            bounded by \( [0, \numacts] \)---%
            when the transition increments, 
            the \( \durAction \) (and \( \autonet_\durAction \)) cannot have been contributing to that value
            and when the transition decrements, they must have been
    \end{enumerate*}.
    This is illustrated in \Cref{fig:construct_run,fig:matchcellar_run}.
    
    This completes this part of the run and lets us obtain \( \bLvc{_{i+1}} \),
    where \( \encaft(\bLvc{_I},\allowbreak i+1,\allowbreak\bLvc{_{i+1}}) \) 
    and \( run([\bLvc{_i'}, \allowbreak \dots, \allowbreak \bLvc{_{i+1}}]) \).
    Furthermore, with \( \bLvc{_i} \longrightarrow_{\delta_{i+1}} \bLvc{_i'} \)
    we deduce \( run([\bLvc{_i}, \allowbreak \dots, \allowbreak \bLvc{_{i+1}}]) \).
    From the induction hypothesis we obtain the fact \( run([\bLvc{_0}, \allowbreak \dots, \allowbreak \bLvc{_{i}}]) \).
    By combining these two runs we obtain the fact \( run([\bLvc{_{i}}, \allowbreak \dots, \allowbreak \bLvc{_{i+1}}]) \),
    which completes the proof with \( \encaft(\bLvc{_I},\allowbreak i+1,\allowbreak\bLvc{_{i+1}}) \).
\end{proof}

\subsection{Proof of \Cref{lem:init_trans_possible}}
\begin{proof}[by cases on the size of \( \happTimePoints(\plan) \)]$ $\newline
    The transition \( \bLvc{_I} \longrightarrow_{\initedge} \bLvct{_I'}{_I'}{_I} \)
    is possible for some \( \bLvct{_I'}{_I'}{_I} \). 
    \( \Loc_I' \) and \( \Var_I' \) satisfy the conditions in
    \( \encbef(\bLvc{_I},\allowbreak 0,\bLvct{_I'}{_I'}{_I}) \) in \Cref{def:enc_states_bef},
    because no action is active before \( \happTimePoint_0 \) and
    all \( \Var(\vp{p}) \) for \( p \in \propositions \) encode the propositional state \( \initState \). 

    Since no snap actions have been scheduled before \( \happTimePoint_0 \),
    for all \( \Clk_I (\clk{\snapAct{}}) \) to satisfy \( \timesinceat(\snapAct{}, \happTimePoint_0, <) \), 
    the values need to be incremented by
    \( \delta_0 = \initgap + \epsilon \) (\Cref{def:enc_active_time}). 
    This is achievable using a \( \bLvct{_I'}{_I'}{_I} \longrightarrow_{\delta_0} \bLvc{_I'} \) transition.

    Therefore, there exists a configuration \( \bLvc{'}\)
    such that \( run([\bLvc{_I}, \allowbreak \dots, \allowbreak \bLvc{_I'}]) \) and
    \( \encbef(\bLvc{_I},\allowbreak 0,\allowbreak\bLvc{_I'})\).
    
    We proceed with the proof by splitting the cases.
    
    \noindent\emph{\textbf{Case 1}. \( \happTimePoints(\plan) \) is empty.}
    In this case, \( \encbef(\bLvc{_I},\allowbreak 0,\allowbreak\bLvc{_I'}) \)
    is equivalent to \( \encaft(\bLvc{_I},\allowbreak 0,\allowbreak\bLvc{_I'}) \),
    due to the definition of both in case \( \happTimePoints(\plan) \) is empty,
    which concludes the proof. 
    See \Cref{def:enc_states,def:enc_states_bef}.

    \noindent\emph{\textbf{Case 2}. \( \happTimePoints(\plan) \) is not empty.}
    In this case, \( \encbef(\bLvc{_I},\allowbreak 0,\allowbreak\bLvc{_I'}) \)
    permits the construction of a run \( \langle \bLvc{_I'}, \allowbreak \dots, \allowbreak \bLvc{_0} \rangle \) 
    such that \( \encaft(\bLvc{_I},\allowbreak 0,\allowbreak\bLvc{_0}) \) and
    \( run([\bLvc{_I'}, \allowbreak \dots, \allowbreak \bLvc{_0}]) \) are satisfied
    using reasoning like in \Cref{lem:plan_trans_possible}.
    This run can be combined with \(\langle \bLvc{_I}, \allowbreak \dots, \allowbreak \bLvc{_I'} \rangle\)
    to construct \(\langle \bLvc{_I}, \allowbreak \dots, \allowbreak \bLvc{_0} \rangle\).
\end{proof}

\subsection{Proof of \Cref{lem:goal_trans_possible}}
\begin{proof}
    \( \bLvc{_k} \mathop{\longrightarrow_{\goaledge}} \bLvc{_G} \) is possible if \( \encaft(\bLvc{_I},\allowbreak k,\allowbreak\bLvc{_k}) \) holds.

    Every starting snap action \( \daStart{\durAction} \) scheduled before or at \( t_k \) can be paired with an ending snap action \( \daEnd{\durAction} \) before or at \( t_k \).
    
    Therefore, the number of \( \sse \) and \( \eee \) edges traversed is equal.
    Hence, \( \Var_k(\actsactive) = 0 \) and \( \Var_k(\lp{p}) = 0 \) for all \( p \).

    Moreover, \( \Var_k(\vp{p}) \) equals \( 1 \) for all \( p \) in the
    final propositional state \( \propof{(\planningState_{k+1})} \)
    and \( \Var_k(f) \) equals \( \numof{(\planningState_{k+1})}(f) \) for all \( f \in \fluents \),
    where \( \planningState_{k+1} \) is the final state of the valid state sequence (see \Cref{def:enc_states,def:valid_state_sequence}).
    Since \( \propcondof{\goalCond} \subseteq \propof{(\planningState_{k+1})} \)
    and \( \bEntails{\numof{(\planningState_{k+1})}}{\numcondof{\goalCond}} \) (see \Cref{def:valid_plan}), we can deduce
    \( \Var_k(\vp{p}) = 1 \) for all \( p \in \propcondof{\goalCond} \)
    and that \( \Var_k \) satisfies \( \numcondof{\goalCond} \).
    Hence, the conditions of \( \goaledge \) are satisfied.
\end{proof}
\subsection{Proof of \Cref{thm:main_theorem}}
\begin{proof}
    It is sufficient to prove the existence of a run \( \cfgs \) and configuration \( \bLvc{_G} \) of the automata network, such that:
    \begin{enumerate*}[itemjoin={;\ }, itemjoin*={; and\ }]
        \item \( \cfgs \) is a valid run: \( \run(\cfgs) \)
        \item \( \cfgs \) contains \( \bLvc{_G} \): \( \cfgs \equiv \langle \bLvc{_I}, \allowbreak \dots, \allowbreak \bLvc{_G}, \dots \rangle \)
        \item \( \Loc_G \) records the main automaton's location as \( \goalloc \): i.e. \( \Loc_G \equiv \langle \goalloc, \dots \rangle \).
    \end{enumerate*}
    
    The proof is completed by fixing \( k = m \) in \Cref{lem:plan_trans_possible} and combining the run obtained from this fact with those from \Cref{lem:init_trans_possible,lem:goal_trans_possible}.
\end{proof}
\clearpage
\section{Details of Formalisation}
\noindent
\begin{table}[h]
\caption{Size of formalisations (percentages rounded to sum to 100\,\%)}\label{table:formalisations}\label{tab:loc_table}
{\footnotesize\setlength{\tabcolsep}{3pt}
\begin{tabularx}{\linewidth}{| >{\raggedright\arraybackslash\hyphenpenalty=10000\exhyphenpenalty=10000}X | r@{\ }r | r@{\ }r | r@{\ }r | r |}
    \hline
    \multicolumn{8}{|c|}{Size of formalisation} \\
    \hline
    \multicolumn{1}{|c|}{\multirow{2}{*}{Component}}
    & \multicolumn{7}{c|}{Lines of code (thousands) by attribution} \\
    \cline{2-8}
    \multicolumn{1}{|c|}{}
    & \multicolumn{2}{c|}{Human}
    & \multicolumn{2}{c|}{Human + AI}
    & \multicolumn{2}{c|}{AI}
    & \multicolumn{1}{c|}{Total} \\
    \hline
    \multicolumn{8}{|c|}{\emph{Propositional}} \\
    \hline
    Abstract temporal planning semantics
    & 6.9 & (98\,\%) & $<$0.1 & ($<$1\,\%) & 0.1 & (2\,\%) & 7.0 \\
    Reduction in \Cref{sec:reduction}
    & 1.1 & (97\,\%) & $<$0.1 & (3\,\%) & $<$0.1 & ($<$1\,\%) & 1.1 \\
    Technical setup for proofs
    & 4.1 & (95\,\%) & 0.2 & (5\,\%) & $<$0.1 & ($<$1\,\%) & 4.3 \\
    Proofs of \Cref{thm:main_theorem} including lemmas, conditions and rules
    & 3.7 & (80\,\%) & 0.9 & (19\,\%) & $<$0.1 & (1\,\%) & 4.6 \\
    Reduction from PDDL AST to abstract planning
    & 2.8 & (47\,\%) & 0.7 & (12\,\%) & 2.4 & (41\,\%) & 6.0 \\
    Refinement to executable code
    & 1.0 & (90\,\%) & $<$0.1 & (4\,\%) & 0.1 & (6\,\%) & 1.1 \\
    Verified certificate check and code export
    & 1.2 & (95\,\%) & $<$0.1 & (3\,\%) & $<$0.1 & (2\,\%) & 1.3 \\
    Miscellaneous lemmas and definitions
    & 1.8 & (94\,\%) & \multicolumn{2}{c|}{---} & 0.1 & (6\,\%) & 1.9 \\
    \hline
    Total (propositional)
    & 22.6 & (83\,\%) & 2.0 & (7\,\%) & 2.8 & (10\,\%) & 27.4 \\
    \hline
    \multicolumn{8}{|c|}{\emph{Numeric extension}} \\
    \hline
    Abstract semantics
    & \multicolumn{2}{c|}{---} & \multicolumn{2}{c|}{---} & 0.8 & (100\,\%) & 0.8 \\
    Reduction in \Cref{sec:reduction}
    & \multicolumn{2}{c|}{---} & \multicolumn{2}{c|}{---} & 0.5 & (100\,\%) & 0.5 \\
    Technical setup for proofs
    & \multicolumn{2}{c|}{---} & \multicolumn{2}{c|}{---} & 2.7 & (100\,\%) & 2.7 \\
    Proofs of \Cref{thm:main_theorem}
    & \multicolumn{2}{c|}{---} & \multicolumn{2}{c|}{---} & 5.8 & (100\,\%) & 5.8 \\
    Reduction from PDDL AST
    & \multicolumn{2}{c|}{---} & $<$0.1 & (1\,\%) & 0.4 & (99\,\%) & 0.4 \\
    Refinement to executable code
    & \multicolumn{2}{c|}{---} & \multicolumn{2}{c|}{---} & 1.2 & (100\,\%) & 1.2 \\
    Numeric bounds: inference
    & \multicolumn{2}{c|}{---} & $<$0.1 & (2\,\%) & 1.8 & (98\,\%) & 1.9 \\
    Numeric bounds: carrying
    & \multicolumn{2}{c|}{---} & \multicolumn{2}{c|}{---} & 1.9 & (100\,\%) & 1.9 \\
    Certificate check and code export
    & \multicolumn{2}{c|}{---} & $<$0.1 & (6\,\%) & 0.6 & (94\,\%) & 0.6 \\
    Miscellaneous lemmas and definitions
    & \multicolumn{2}{c|}{---} & \multicolumn{2}{c|}{---} & $<$0.1 & (100\,\%) & $<$0.1 \\
    \hline
    Total (numeric extension)
    & \multicolumn{2}{c|}{---} & 0.1 & (1\,\%) & 15.7 & (99\,\%) & 15.8 \\
    \hline
    Total (ours)
    & 22.6 & (52\,\%) & 2.0 & (5\,\%) & 18.5 & (43\,\%) & \textbf{43.2} \\
    \hline
    \multicolumn{8}{|c|}{\emph{Formal PDDL semantics (with others)}} \\
    \hline
    Basics: syntax, happening semantics
    & 5.7 & (100\,\%) & \multicolumn{2}{c|}{---} & \multicolumn{2}{c|}{---} & 5.7 \\
    State-sequence semantics
    & \multicolumn{2}{c|}{---} & 1.8 & (100\,\%) & \multicolumn{2}{c|}{---} & 1.8 \\
    \hline
    Total
    & 5.7 & (76\,\%) & 1.8 & (24\,\%) & \multicolumn{2}{c|}{---} & 7.5 \\
    \hline
    \multicolumn{8}{|c|}{\emph{Grounding (with others)}} \\
    \hline
    Total
    & 1.6 & (59\,\%) & 0.8 & (28\,\%) & 0.4 & (13\,\%) & 2.7 \\
    \hline
    \multicolumn{8}{|c|}{\emph{Third-party}} \\
    \hline
    Interval domain from HOL-IMP
    & \multicolumn{2}{c|}{---} & \multicolumn{2}{c|}{---} & \multicolumn{2}{c|}{---} & 0.4 \\
    Two lemmas copied from Munta
    & \multicolumn{2}{c|}{---} & \multicolumn{2}{c|}{---} & \multicolumn{2}{c|}{---} & $<$0.1 \\
    \hline
    Total
    & \multicolumn{2}{c|}{---} & \multicolumn{2}{c|}{---} & \multicolumn{2}{c|}{---} & 0.4 \\
    \hline
    Total
    & 30.3 & (56\,\%) & 4.6 & (9\,\%) & 18.8 & (35\,\%) & \textbf{53.9} \\
    \hline
\end{tabularx}
}
\end{table}

\begin{table}[h]
\noindent
\begin{tabularx}{\linewidth}{| X | l |}
    \hline
    \multicolumn{2}{|c|}{Size of related formalisations} \\
    \hline
    \multicolumn{1}{|c|}{Formalisation} 
    & \multicolumn{1}{c|}{Lines of code} \\
    \hline
    Verified SAT-based planner~\cite{verifiedSATPlan}
    & $\sim17500$ \\
    Formal semantics and a verified plan validator for temporal PDDL~\cite{aaai2022}
    & $\sim6500$ \\
    Formal semantics and a verified plan validator for classical PDDL~\cite{ictai2018}
    & $\sim3000$ \\
    \hline
    Formal semantics of Markov Decision Processes (MDPs), a verified solver for MDPs and certificate checker for linear programming solutions~\cite{ApproxPolIterIsabelle}
    & $\sim35000$ \\
    Verified MDP algorithms~\cite{MDP-Algorithms-AFP}
    & $\sim10000$ \\
    Formalised MDP theory~\cite{MDP-Rewards-AFP}
    & $\sim5000$ \\
    \hline
\end{tabularx}
\end{table}

\subsubsection{Locales}
\begin{figure}
    \centering
    \begin{tikzpicture}[
        inheritance/.style={
            -{Triangle[open]}
        },
        sublocale/.style={
            -{Stealth}
        }
    ]
        \begin{scope}[on grid]
            \begin{scope}[node distance=0.6cm and 0.75cm]
                \begin{scope}[node distance=0.65cm and 0.8cm]
                    \node (tp) {\tp};
                    \node[below=of tp] (tpp) {\tpp};
        
                    \node[right=of tpp] (tppa) {\tppa};
                    \node[right=of tp] (tpa) {\tpa};
                    
                    \draw (tp) edge[inheritance] (tpp);
                    \draw (tpa) edge[inheritance] (tppa);
                    
                    \draw (tpa) edge[sublocale] (tp);
                    \draw (tppa) edge[sublocale] (tpp);
                \end{scope}

                \begin{scope}
                    \node[above left=of tp] (tpli) {\tpli};
                    \node[below left=of tpp] (tppli) {\tppli};
            
                    \node[above right=of tpa] (tplia) {\tpali};
                    \node[below right=of tppa] (tpplia) {\tppali};
                    
                    \draw (tpli) edge[inheritance] (tppli);
                    \draw (tplia) edge[inheritance] (tpplia);
                    
                    \draw (tplia) edge[sublocale] (tpli);
                    \draw (tpplia) edge[sublocale] (tppli);
                    
                    \draw (tppli) edge[sublocale] (tpp);
                    \draw (tpplia) edge[sublocale] (tppa);
                    
                    \draw (tpli) edge[sublocale] (tp);
                    \draw (tplia) edge[sublocale] (tpa);
                \end{scope}
        
                \begin{scope}
                    \node[above left=of tpli] (tplin) {\tplin};
                    \node[below left=of tppli] (tpplin) {\tpplin};
            
                    \node[above right=of tplia] (tplina) {\tpalin};
                    \node[below right=of tpplia] (tpplina) {\tppalin};
                    
                    \draw (tplin) edge[inheritance] (tpplin);
                    \draw (tplina) edge[inheritance] (tpplina);
                    
                    \draw (tplina) edge[sublocale] (tplin);
                    \draw (tpplina) edge[sublocale] (tpplin);
    
                    \begin{scope}
                        \draw (tpplin) edge[inheritance] (tppli);
                        \draw (tpplina) edge[inheritance] (tpplia);
                        \draw (tplin) edge[inheritance] (tpli);
                        \draw (tplina) edge[inheritance] (tplia);
                    \end{scope}

                \end{scope}
            \end{scope}

            \begin{scope}[on grid, node distance=1cm and 1.375cm]
                \node[right=of tpplina] (gpp) {\gpdpp};
                \node[right=of tplina] (gp) {\gpdp};

                \node[right=of gpp] (gppdef) {\gpdppdefs};

                \node[above right=of gpp] (pdp) {\pdpp};

                \draw (gpp) edge[inheritance] (pdp);
                \draw (gp) edge[inheritance] (gpp);
                \draw (gpp) edge[sublocale] (tpplina);
                \draw (gp) edge[sublocale] (tplina);

                \draw (gppdef) edge[sublocale] (gpp);
                \draw (gppdef) edge[inheritance] (pdp);
            \end{scope}

            \begin{scope}
                \node[left=of tpplin] (net) {\timaut};

                \draw (tpplin) edge[sublocale] (net);
            \end{scope}
        \end{scope}
        
    \end{tikzpicture}
    
    \caption[Locale extensions add and sublocales replace definitions and assumptions]{
    Locale extensions (\protect\smash{\protect\begin{tikzpicture}[baseline=0pt] \protect\draw (0,-.5ex) edge[-{Triangle[open]}](0,1.75ex);\protect\end{tikzpicture}}) add
    and sublocales (\protect\smash{\protect\begin{tikzpicture}[baseline=0pt] \protect\draw (0,-.5ex) edge[-Stealth] (0,1.75ex);\protect\end{tikzpicture}}) replace
    definitions and assumptions: 
    timed automata (\timaut),
    planning problem (\tpp),
    plan (\tp),
    abstract syntax (\( S \)),
    unique naming functions (\( N \)),
    executable definitions (\( D \)),
    groundedness (\( g \)),
    annotated actions as snap actions (\( a \)),
    lists as sets (\( l \)),
    machine integers as clock constraints (\( i \)).
    }
    \label{fig:locale_def_replacement}
\end{figure}

The first aspect of formally verifying our encoding is the approach to formalising the definitions involved in Isabelle/HOL's logic.
Since the definitions are numerous and intricately connected, a naive approach could lead to unneeded verification effort and a badly structured model of the encoding.
We thus use Isabelle/HOL locales~\cite{LocalesBallarin} to structure our library of definitions.
A locale can be seen as a formal mathematical context that introduces constants and assumptions, in which one develops definitions and theorem statements that implicitly recognise these constants and assumptions.
The most fundamental locale (see~\Cref{lst:planning_problem}) is that defining constants and assumptions on a planning problem (\( \planningproblem \)).
Locales can help structure a library of definitions by representing two main types of relationships between mathematical concepts.
First is \emph{extension}, where one concept can be seen as another concept, but with more variables and assumptions added.
E.g.\ a valid plan (\( \plan \)) (see~\Cref{lst:valid_temporal_plan_locale}) presumes a valid planning problem, extended with a sequence of actions and assumptions on those actions, such that they all come from the planning problem.
The other type of relationship is that of a \emph{sublocale}: a mathematical concept satisfies the assumptions of another, under a suitable mapping of one concept's variables to another's.
E.g.\ \( \tpp \) is a sublocale of a network of timed automata (\( \timaut \)), after encoding, i.e.\ our encoding always gives a valid automaton.
\Cref{fig:locale_def_replacement} shows the relationships between the main concepts formalised as locales in our project.

As shown in \Cref{fig:locale_def_replacement}, the core concept of a planning problem \( \planningproblem \), which uses abstract mathematical concepts like sets and dense linear orders, is where much of the abstract theory about temporal planning is proved.
It is refined into more concrete variants; e.g.\ the \( \tppa \) locale replaces snap actions in \( \tpp \) with
annotated actions,
to make it easier to identify actions,
and hence transitions that need to be taken to simulate a snap action's execution,
in the proof.
The \( \tppali \) locale fixes types to represent sets as lists,
to make it easier to order transitions,
and clock constraints as integers (as opposed to dense linear orders) for simpler proofs and executability.
Within the \( \tppalin \) locale,
we provide an executable definition of the encoding in \Cref{sec:reduction}.
Another refinement step is to deal with parsed problems: \( \gpdpp \) is defined over abstract syntax trees of grounded PDDL.
\( \gpdppdefs \) refines \( \gpdpp \), and is the most refined variant of a planning problem.
It has executable datatypes such as native strings, red-black trees, and arbitrary-length integers, and it has the most efficient encoding that we run.

We reuse from the work of \citet{aaai2022} the abstract syntax (and a semantics) of PDDL in \( \pdpp \), a minimal trusted Standard ML parser for that abstract syntax, and executable code to check well-formedness.
We also use the Collections framework by \citet{collectionsFramework} to refine to efficient data structures.
Moreover, our code that interfaces with the formalisation of timed automata by \citet{wimmerCertTimedAutomata} uses 
the Imperative Refinement Framework by \citet{LammichRefinement} for efficient imperative implementations, since \citeauthor{wimmerCertTimedAutomata}'s code~\cite{wimmerCertTimedAutomata} is developed within that framework.

\subsubsection{Important locales}
\begin{lstlisting}[language=isabelle, caption={Locale for a temporal planning problem}, label={lst:planning_problem}]
locale temp_planning_problem = action_defs at_start at_$\texttt{end}$ over_all lower upper pre adds dels
  for at_start::"'action  $\Rightarrow$ 'snap_action"
  and at_$\texttt{end}$::  "'action  $\Rightarrow$ 'snap_action"
  and over_all::"'action  $\Rightarrow$ 'proposition set"
  and lower::"'action  $\rightharpoonup$ ('time::time) lower_bound"
  and upper::"'action  $\rightharpoonup$ 'time upper_bound"
  and pre::"'snap_action $\Rightarrow$ 'proposition set"
  and adds::"'snap_action $\Rightarrow$ 'proposition set"
  and dels::"'snap_action $\Rightarrow$ 'proposition set" +
fixes init::"'proposition set"
  and goal::"'proposition set"
  and $\epsilon$::"'time::time" 
assumes eps_ran: "0 $\le$ $\epsilon$"
\end{lstlisting}

\begin{lstlisting}[language=isabelle, caption={Locale for a valid temporal plan}, label={lst:valid_temporal_plan_locale}]

$\textbf{type-synonym}$ ('i, 'a, 't) temp_plan = "'i $\rightharpoonup$ ('a $\times$ 't $\times$ 't)"
locale temp_plan_defs = temp_planning_problem  at_start at_$\texttt{end}$ over_all lower upper pre adds dels init goal $\epsilon$
  for at_start::"'action  $\Rightarrow$ 'snap_action"
  and at_$\texttt{end}$::  "'action  $\Rightarrow$ 'snap_action"
  and over_all::"'action  $\Rightarrow$ 'proposition set"
  and lower::
    "'action  $\rightharpoonup$ ('time::time) lower_bound"
  and upper::
    "'action  $\rightharpoonup$ 'time upper_bound"
  and pre::
    "'snap_action $\Rightarrow$ 'proposition set"
  and adds::
    "'snap_action $\Rightarrow$ 'proposition set"
  and dels::
    "'snap_action $\Rightarrow$ 'proposition set"
  and init::"'proposition set"
  and goal::"'proposition set"
and $\epsilon$::"'time" +
fixes $\pi$::"('i, 'action, 'time) temp_plan"

locale valid_temp_plan = temp_plan_defs +
  assumes vp: valid_plan
\end{lstlisting}

\subsubsection{Mathematical Induction}

To obtain a goal of the form \( esc(i, j, q_{i;j}) \rightarrow esc(i, j+1, q_{i;j+1}) \), we reason about lists. Instead of structural induction over the constructors of lists, we induct over indices of elements, i.e. natural numbers.

For the formalisation we separate the preconditions of these transitions from their effects. We define lists of effects as lists of functions.

\Cref{lst:list_funs} defines two functions. \lstinline|seq_apply| applies a list of functions beginning at an initial element. \lstinline|ext_seq| takes a function that creates a list from an initial element and a list. It applies the function to the final element of the list and appends the result. \lstinline|((ext_seq $\circ$ seq_apply) fs)| can be composed with itself using arbitrary \lstinline|fs| arguments to build longer lists.

\begin{lstlisting}[language=isabelle, caption={Lists and Functions}, label={lst:list_funs}]
definition ext_seq f xs $\equiv$
    xs @ f ( last xs)
definition seq_apply fs x $\equiv$
    map ( $\lambda$i. ( fold (id) ( take i fs) x))
        [1..<length fs + 1]
\end{lstlisting}

We also use standard functions, such as \lstinline|fold|, \lstinline|take| and \lstinline|nth|.
\lstinline|take| returns the first \( n \) elements of a list.
\lstinline|nth|, denoted as \lstinline|!|, obtains the \( n \)-th element by index.

Consider a list of functions \( fs = \langle f_0, \dots, f_{n-1} \rangle \).
Consider the list \( xs = \langle x_1, \dots, x_n \rangle \) of states obtained through: \lstinline|seq_apply $fs$ $x_0$|.
It is apparent \lstinline[language=isabelle]|$x_i$ = (last (seq_apply (take $i$ $fs$) $x_0$))| for \( i < n \)
and \lstinline[language=isabelle]|$x_{i+1}$ = last (seq_apply (take ($i + 1$) $fs$) $x_0$)|.

The \( (n+1) \)-th element is always obtained by applying the \( n \)-th function to the \( n \)-th element.
We define the induction-like inference rule in \Cref{lst:pre_post_induct}, which allows us to prove properties of individual states and the beginning and ends of lists obtained from \lstinline|seq_apply|.
\( R \) and \( S \) are properties independent of the index.
When \lstinline|x| is replaced with \lstinline|(last ys')| for some \lstinline|(ys' = (ext_seq $\circ$ seq_apply fs) ys)|, then we need to know some \lstinline|R (last ys')|, which can be proven with the same lemma.
\lstinline|S| must be instantiated with \lstinline|R| and \lstinline|R| must be instantiated with some \lstinline|R'|.
\lstinline|R (last (y # seq_apply fs ys))| for some \lstinline|y| is proven under the assumption of \lstinline|(R' y)|.
Thus, we compose proofs.

\begin{lstlisting}[language=isabelle, caption={Final state lemma}, label={lst:pre_post_induct}]
lemma seq_apply_pre_post:
  assumes PQ: ...
  and QP: ...
  and Rx: R x
  and len0: $\forall$x. 0 = length fs
    $\longrightarrow$ ( R x $\longrightarrow$ S x)
  and RP0: $\forall$x. 0 < length fs
    $\longrightarrow$ ( R x $\longrightarrow$ P 0 x)
  and QSl: $\forall$x. 0 < length fs
    $\longrightarrow$ ( Q ( length fs - 1) x $\longrightarrow$ S x)
    shows
    S ( last ( x # seq_apply fs x))
\end{lstlisting}

For the induction-like proof of \Cref{lem:plan_trans_possible}, we need to assert that the entire list is a run.
This is proven using \Cref{lst:list_pred_lemma}, by instantiating \lstinline|LP| with the \( \mathit{run} \) predicate.
When a state satisfies \lstinline|(P i)| and its successor satisfies \lstinline|(Q i)|, then they form a run consisting of one step.
A run ending with a configuration \lstinline|(last xs)| can be prepended to a run beginning at \lstinline|(last xs)|.

\begin{lstlisting}[language=isabelle, caption={List predicate lemma}, label={lst:list_pred_lemma}]
lemma ext_seq_comp_seq_apply_list_prop:
  assumes PQ: ...
  and QP: ...
  and R (last xs)
  and RP0: ...
  and LPi: $\forall$i s. i < length fs
    $\longrightarrow$ P i s $\longrightarrow$ Q i ( ( fs ! i) s)
    $\longrightarrow$ LP [s, ( fs ! i) s]
  and base: $\forall$x. LP [x]
  and step: $\forall$xs ys. LP xs
    $\longrightarrow$ LP ( last xs # ys)
    $\longrightarrow$ LP ( xs @ ys)
  and LPxs: LP xs
  shows LP ( ( ext_seq $\circ$ seq_apply) fs xs)
\end{lstlisting}

With rules analogous to \Cref{lst:list_pred_lemma,lst:pre_post_induct}, we prove \Cref{thm:main_theorem,lem:goal_trans_possible,lem:init_trans_possible,lem:plan_trans_possible}.

\clearpage
\section{Threshold Inference}\label{sec:threshold_inference}
\raggedbottom %

\begin{pseudocodebox}\begin{algorithmic}
\State \textbf{type} \( \fn{IntOrInf} = \integers \mid \infty \mid
    -\infty \)
\State \textbf{type} \( \fn{Interval} = \{ \fn{lo} : \fn{IntOrInf},\
    \fn{hi} : \fn{IntOrInf} \} \)
\State \textbf{type} \( \fn{Env} = \Mapt{\vars}{\fn{Interval}} \)
\end{algorithmic}\end{pseudocodebox}
\begin{pseudocodebox}\begin{algorithmic}
\Function{\fn{add-int-or-inf}}{\( a : \fn{IntOrInf},\
    b : \fn{IntOrInf} \)} \Comment{overloads \( (+) \)}
  \If{\( a = \infty \mathbin{\kw{OR}} b = \infty \)} \Return \( \infty \)
  \ElsIf{\( a = -\infty \mathbin{\kw{OR}} b = -\infty \)} \Return \( -\infty \)
  \Else
    \State \Return \( a + b \)
  \EndIf
\EndFunction
\end{algorithmic}\end{pseudocodebox}
\begin{pseudocodebox}\begin{algorithmic}
\Function{\fn{add-intv}}{\( e_1 : \fn{Interval},\
    e_2 : \fn{Interval} \)} \Comment{overloads \( (+) \)}
  \State \Return \( [e_1.\fn{lo} + e_2.\fn{lo},\
      e_1.\fn{hi} + e_2.\fn{hi}] \)
\EndFunction
\end{algorithmic}\end{pseudocodebox}
\begin{pseudocodebox}\begin{algorithmic}
\Function{\fn{negate-intv}}{\( e_1 : \fn{Interval} \)}
    \Comment{overloads (unary \( - \))}
  \State \Return \( [-e_1.\fn{hi},\ -e_1.\fn{lo}] \)
\EndFunction
\end{algorithmic}\end{pseudocodebox}
\noindent\begin{minipage}{\linewidth}
\begin{pseudocodebox}\begin{algorithmic}
\Function{\fn{sub-intv}}{\( e_1 : \fn{Interval},\
    e_2 : \fn{Interval} \)} \Comment{overloads \( (-) \)}
  \State \Return \( \fn{add-intv}(e_1, -e_2) \)
\EndFunction
\end{algorithmic}\end{pseudocodebox}
\captionof{listing}{Data structures, types, and operations.}\label{lst:intv_ops}
\end{minipage}
\par\medskip

\noindent\begin{minipage}{\linewidth}
\begin{pseudocodebox}\begin{algorithmic}
\Function{\fn{approx-eval-intv}}{\( \mathit{env} : \fn{Env},\
    e : \fn{Expr} \)}
  \If{\( e = \fn{Const}(c) \)} \Return \( [c, c] \)
  \ElsIf{\( e = \fn{Var}(v) \)} \Return \( \mathit{env}[v] \)
  \ElsIf{\( e = \fn{Add}(e_1, e_2) \)} \Return \( \fn{approx-eval-intv}(\mathit{env}, e_1) +
        \fn{approx-eval-intv}(\mathit{env}, e_2) \)
  \ElsIf{\( e = \fn{Sub}(e_1, e_2) \)} \Return \( \fn{approx-eval-intv}(\mathit{env}, e_1) -
        \fn{approx-eval-intv}(\mathit{env}, e_2) \)
  \ElsIf{\( e = \fn{Mul}(e_1, e_2) \)} \Return \( [-\infty, +\infty] \)
  \ElsIf{\( e = \fn{Div}(e_1, e_2) \)} \Return \( [-\infty, +\infty] \)
  \EndIf
\EndFunction
\end{algorithmic}\end{pseudocodebox}
\captionof{listing}{Approximate interval evaluation.}\label{lst:approx_eval}
\end{minipage}
\par\medskip

\begin{pseudocodebox}\begin{algorithmic}
\State \textbf{type} \( \fn{CompOp} = \{ <, \le, =, \ge, > \} \)
\State \textbf{type} \( \fn{Expr} = \fn{Const}(c : \integers)
    \mid \fn{Var}(v : \vars) \)
\Statex \qquad \( \mid \fn{Add}(e_1 : \fn{Expr}, e_2 : \fn{Expr})
    \mid \fn{Sub}(e_1 : \fn{Expr}, e_2 : \fn{Expr}) \)
\Statex \qquad \( \mid \fn{Mul}(e_1 : \fn{Expr}, e_2 : \fn{Expr})
    \mid \fn{Div}(e_1 : \fn{Expr}, e_2 : \fn{Expr}) \)
\State \textbf{type} \( \fn{BExp} = \fn{Comp}(e_1 : \fn{Expr},
    \mathit{op} : \fn{CompOp}, e_2 : \fn{Expr}) \)
\Statex \qquad \( \mid \fn{And}(b_1 : \fn{BExp}, b_2 : \fn{BExp})
    \mid \fn{True} \)
\State \textbf{type} \( \fn{Effect} = \{ l : \vars,\
    r : \fn{Expr} \} \)
\State \textbf{type} \( \fn{snap-action} = \{ \saNumPre{} : \fn{BExp},\
    \saNumEff{} : \Sett{\fn{Effect}} \} \)
\end{algorithmic}\end{pseudocodebox}
\begin{pseudocodebox}\begin{algorithmic}
\Function{\fn{increments}}{\( f : \vars,\ \mathit{exp} :
    \fn{Expr} \)}
  \If{\( \mathit{exp} = \fn{Add}(f, \mathit{exp}') \)} \Return \( \{ \mathit{exp}' \} \)
  \ElsIf{\( \mathit{exp} = \fn{Sub}(f, \mathit{exp}') \)} \Return \( \{ \fn{Sub}(\fn{Const}(0), \mathit{exp}') \} \)
  \ElsIf{\( \mathit{exp} = \fn{Add}(\mathit{exp}', f) \)} \Return \( \{ \mathit{exp}' \} \)
  \Else
    \State \Return \( \emptyset \)
  \EndIf
\EndFunction
\end{algorithmic}\end{pseudocodebox}
\noindent\begin{minipage}{\linewidth}
\begin{pseudocodebox}\begin{algorithmic}
\Function{\fn{comps}}{\( b : \fn{BExp} \)}
  \If{\( b = \fn{True} \)} \Return \( \emptyset \)
  \ElsIf{\( b = \fn{And}(b_1, b_2) \)} \Return \( \fn{comps}(b_1) \cup \fn{comps}(b_2) \)
  \ElsIf{\( b = \fn{Comp}(e_1, \mathit{op}, e_2) \)} \Return \( \{ (e_1, \mathit{op}, e_2) \} \)
  \EndIf
\EndFunction
\end{algorithmic}\end{pseudocodebox}
\captionof{listing}{Expressions and guards.}\label{lst:seeding_types}
\end{minipage}
\par\medskip

\begin{pseudocodebox}\begin{algorithmic}
\Function{\fn{shifted-consts}}{\( \mathit{env} :
    \Mapt{\vars}{\integers},\ e : \fn{Expr},\
    \mathit{offset} : \integers \)}
  \State \( \mathit{ks} \gets \emptyset \)
  \ForAll{\( k \in \fn{consts}(e) \cup
      \{ \fn{eval}(\mathit{env}, e) \} \)}
    \State \( \mathit{ks} \gets \mathit{ks} \cup
        \{ k + \mathit{offset} \} \)
  \EndFor
  \State \Return \( \mathit{ks} \)
\EndFunction
\end{algorithmic}\end{pseudocodebox}
\begin{pseudocodebox}\begin{algorithmic}
\Function{\fn{snap-fluent-thrs}}{\( \mathit{env} :
    \Mapt{\vars}{\integers},\ f : \vars,\
    a : \fn{snap-action} \)}
  \State \( \mathit{thrs} \gets \emptyset \)
  \ForAll{\( \mathit{eff} \in \saNumEff{a} \)}
    \ForAll{\( \mathit{increment} \in
        \fn{increments}(f, \mathit{eff}.r) \)}
      \State \( \mathit{offset} \gets
          \fn{eval}(\mathit{env}, \mathit{increment}) \)
      \ForAll{\( (e_1, \mathit{op}, e_2) \in
          \fn{comps}(\saNumPre{a}) \)}
        \If{\( f = e_1 \)}
          \State \( \mathit{thrs} \gets \mathit{thrs} \cup
              \fn{shifted-consts}(\mathit{env}, e_2,
              \mathit{offset}) \)
        \ElsIf{\( f = e_2 \)}
          \State \( \mathit{thrs} \gets \mathit{thrs} \cup
              \fn{shifted-consts}(\mathit{env}, e_1,
              \mathit{offset}) \)
        \EndIf
      \EndFor
    \EndFor
  \EndFor
  \State \Return \( \mathit{thrs} \)
\EndFunction
\end{algorithmic}\end{pseudocodebox}
\begin{pseudocodebox}\begin{algorithmic}
\Function{\fn{snap-thrs}}{\( \mathit{env} :
    \Mapt{\vars}{\integers},\ \mathit{fluents} : \Sett{\vars},\
    a : \fn{snap-action} \)}
  \State \( \mathit{cs} \gets \fn{consts}(\saNumPre{a})
      \cup \fn{consts}(\saNumEff{a}) \)
  \State \( \mathit{thrs} \gets \emptyset \)
  \ForAll{\( f \in \mathit{fluents} \)}
    \State \( \mathit{thrs} \gets \mathit{thrs} \cup
        \fn{snap-fluent-thrs}(\mathit{env}, f, a) \)
  \EndFor
  \State \Return \( \mathit{thrs} \cup \mathit{cs} \)
\EndFunction
\end{algorithmic}\end{pseudocodebox}
\noindent\begin{minipage}{\linewidth}
\begin{pseudocodebox}\begin{algorithmic}
\Function{\fn{init-thresholds}}{\( \mathit{fluents} :
    \Sett{\vars},\ \mathit{inits} : \Mapt{\vars}{\integers},\
    S : \Sett{\fn{snap-action}} \)}
  \State \( \textit{snaps-thrs} \gets \emptyset \)
  \ForAll{\( a \in S \)}
    \State \( \textit{snaps-thrs} \gets \textit{snaps-thrs} \cup
        \fn{snap-thrs}(\mathit{inits}, \mathit{fluents},
        a) \)
  \EndFor
  \State \( \textit{init-thrs} \gets \emptyset \)
  \ForAll{\( f \in \mathit{fluents} \)}
    \State \( \textit{init-thrs} \gets \textit{init-thrs} \cup
        \{ \mathit{inits}[f] \} \)
  \EndFor
  \State \Return \( \textit{snaps-thrs} \cup \textit{init-thrs} \)
\EndFunction
\end{algorithmic}\end{pseudocodebox}
\captionof{listing}{Threshold initialisation.}\label{lst:thr_init}
\end{minipage}
\par\medskip

\begin{pseudocodebox}\begin{algorithmic}
\Function{\fn{widen}}{\( \mathit{thrs} : \Listt{\integers},\
    \mathit{cur} : \fn{Interval},\ \mathit{nxt} : \fn{Interval} \)}
  \State \( \mathit{nl} \gets \mathit{cur}.\fn{lo} \)
  \State \( \mathit{nu} \gets \mathit{cur}.\fn{hi} \)
  \State \( \mathit{thrs} \gets \fn{sort}(\fn{concat}(\mathit{thrs},
      [-\infty, \infty])) \)
  \If{\( \mathit{nxt}.\fn{lo} < \mathit{cur}.\fn{lo} \)}
    \State \( \mathit{nl} \gets -\infty \)
    \For{\( i \gets 0;\ \mathit{thrs}[i] \le \mathit{nxt}.\fn{lo}
        \mathbin{\kw{AND}} i < \fn{length}(\mathit{thrs});\ i{+}{+} \)}
      \State \( \mathit{nl} \gets \mathit{thrs}[i] \)
    \EndFor
  \EndIf
  \If{\( \mathit{nxt}.\fn{hi} > \mathit{cur}.\fn{hi} \)}
    \State \( \mathit{nu} \gets \infty \)
    \For{\( i \gets \fn{length}(\mathit{thrs}) - 1;\
        \mathit{thrs}[i] \ge \mathit{nxt}.\fn{hi} \mathbin{\kw{AND}} i \ge 0;\
        i{-}{-} \)}
      \State \( \mathit{nu} \gets \mathit{thrs}[i] \)
    \EndFor
  \EndIf
  \State \Return \( [\mathit{nl}, \mathit{nu}] \)
\EndFunction
\end{algorithmic}\end{pseudocodebox}
\noindent\begin{minipage}{\linewidth}
\begin{pseudocodebox}\begin{algorithmic}
\Function{\fn{widen-all}}{\( \mathit{fluents} : \Sett{\vars},\
    \mathit{thrs} : \Listt{\integers},\ \mathit{env} : \fn{Env},\
    \mathit{env}' : \fn{Env} \)}
  \State \( \mathit{acc} \gets \mathit{env} \)
  \ForAll{\( f \in \mathit{fluents} \)}
    \State \( \mathit{acc}[f] \gets \fn{widen}(\mathit{thrs},
        \mathit{env}[f], \mathit{env}'[f]) \)
  \EndFor
  \State \Return \( \mathit{acc} \)
\EndFunction
\end{algorithmic}\end{pseudocodebox}
\captionof{listing}{Widening.}\label{lst:widening}
\end{minipage}
\par\medskip

\begin{pseudocodebox}\begin{algorithmic}
\Function{\fn{astep-updates}}{\( \mathit{effects} :
    \Sett{\fn{Effect}},\ \mathit{env} : \fn{Env} \)}
  \State \( \mathit{env}' \gets \mathit{env} \)
  \ForAll{\( \mathit{eff} \in \mathit{effects} \)}
    \State \( \mathit{env}'[\mathit{eff}.l] \gets
        \fn{approx-eval-intv}(\mathit{env}, \mathit{eff}.r) \)
  \EndFor
  \State \Return \( \mathit{env}' \)
\EndFunction
\end{algorithmic}\end{pseudocodebox}
\noindent\begin{minipage}{\linewidth}
\begin{pseudocodebox}\begin{algorithmic}
\Function{\fn{unguarded-astep}}{\( S : \Sett{\fn{snap-action}},\
    \mathit{env} : \fn{Env} \)}
  \State \( \mathit{env}' \gets \mathit{env} \)
  \ForAll{\( a \in S \)}
    \State \( \mathit{env}' \gets \mathit{env}' \sqcup
        \fn{astep-updates}(\saNumEff{a}, \mathit{env}) \)
  \EndFor
  \State \Return \( \mathit{env}' \)
\EndFunction
\end{algorithmic}\end{pseudocodebox}
\captionof{listing}{Updates for snap actions (astep updates) and unguarded updates for all snap actions.}\label{lst:asteps}
\end{minipage}
\par\medskip

\begin{pseudocodebox}\begin{algorithmic}
\Function{\fn{comp-intv}}{\( \mathit{op} : \fn{CompOp},\
    \mathit{intv} : \fn{Interval} \)}
  \If{\( \mathit{intv}.\fn{hi} < \mathit{intv}.\fn{lo} \)} \Return \( \mathit{intv} \)
  \ElsIf{\( \mathit{op} = {\le} \)} \Return \( [-\infty, \mathit{intv}.\fn{hi}] \)
  \ElsIf{\( \mathit{op} = {<} \)} \Return \( [-\infty, \mathit{intv}.\fn{hi} - 1] \)
  \ElsIf{\( \mathit{op} = {\ge} \)} \Return \( [\mathit{intv}.\fn{lo}, \infty] \)
  \ElsIf{\( \mathit{op} = {>} \)} \Return \( [\mathit{intv}.\fn{lo} + 1, \infty] \)
  \EndIf
  \State \Return \( \mathit{intv} \)
\EndFunction
\end{algorithmic}\end{pseudocodebox}
\begin{pseudocodebox}\begin{algorithmic}
\Function{\fn{flip-op}}{\( \mathit{op} : \fn{CompOp} \)}
  \If{\( \mathit{op} = {\le} \)} \Return \( {\ge} \)
  \ElsIf{\( \mathit{op} = {<} \)} \Return \( {>} \)
  \ElsIf{\( \mathit{op} = {\ge} \)} \Return \( {\le} \)
  \ElsIf{\( \mathit{op} = {>} \)} \Return \( {<} \)
  \Else
    \State \Return \( {=} \)
  \EndIf
\EndFunction
\end{algorithmic}\end{pseudocodebox}
\begin{pseudocodebox}\begin{algorithmic}
\Function{\fn{refine-comp}}{\( (e_1, \mathit{op}, e_2) :
    (\fn{Expr}, \fn{CompOp}, \fn{Expr}),\ \mathit{env} : \fn{Env} \)}
  \State \( e_2' \gets \fn{approx-eval-intv}(\mathit{env}, e_2) \)
  \State \( e_1' \gets \fn{approx-eval-intv}(\mathit{env}, e_1) \)
  \If{\( e_1 = \fn{Var}(f) \)}
    \State \( \mathit{env}[f] \gets \mathit{env}[f] \sqcap
        \fn{comp-intv}(\mathit{op}, e_2') \)
  \EndIf
  \If{\( e_2 = \fn{Var}(f) \)}
    \State \( \mathit{env}[f] \gets \mathit{env}[f] \sqcap
        \fn{comp-intv}(\fn{flip-op}(\mathit{op}), e_1') \)
  \EndIf
  \State \Return \( \mathit{env} \)
\EndFunction
\end{algorithmic}\end{pseudocodebox}
\noindent\begin{minipage}{\linewidth}
\begin{pseudocodebox}\begin{algorithmic}
\Function{\fn{refine-guards}}{\( \mathit{env} : \fn{Env},\
    G : \Sett{\fn{BExp}} \)}
  \State \( \mathit{env}' \gets \mathit{env} \)
  \ForAll{\( g \in G \)}
    \ForAll{\( c \in \fn{comps}(g) \)}
      \State \( \mathit{env}' \gets
          \fn{refine-comp}(c, \mathit{env}') \)
    \EndFor
  \EndFor
  \State \Return \( \mathit{env}' \)
\EndFunction
\end{algorithmic}\end{pseudocodebox}
\captionof{listing}{Refining guards.}\label{lst:refine_guards}
\end{minipage}
\par\medskip

\begin{pseudocodebox}\begin{algorithmic}
\Function{\fn{interval-inter-bexp}}{\( (e_1, \mathit{op}, e_2) :
    (\fn{Expr}, \fn{CompOp}, \fn{Expr}),\ \mathit{env} : \fn{Env} \)}
  \State \( e_1' \gets \fn{approx-eval-intv}(\mathit{env}, e_1) \)
  \State \( e_2' \gets \fn{approx-eval-intv}(\mathit{env}, e_2) \)
  \State \( i_1 \gets e_1' \sqcap \fn{comp-intv}(\mathit{op}, e_2') \)
  \State \( i_2 \gets e_2' \sqcap
      \fn{comp-intv}(\fn{flip-op}(\mathit{op}), e_1') \)
  \If{\( i_1.\fn{hi} < i_1.\fn{lo} \)} \Return \fn{false}
  \ElsIf{\( i_2.\fn{hi} < i_2.\fn{lo} \)} \Return \fn{false}
  \EndIf
  \State \Return \fn{true}
\EndFunction
\end{algorithmic}\end{pseudocodebox}
\begin{pseudocodebox}\begin{algorithmic}
\Function{\fn{guarded-astep-updates}}{\( a : \fn{snap-action},\
    \mathit{env} : \fn{Env} \)}
  \State \( \mathit{env}' \gets \mathit{env} \)
  \ForAll{\( c \in \fn{comps}(\saNumPre{a}) \)}
    \If{\( \mathop{\kw{NOT}}\, \fn{interval-inter-bexp}(c, \mathit{env}) \)} \Return \( \mathit{env} \)
    \EndIf
    \State \( \mathit{env}' \gets
        \fn{refine-comp}(c, \mathit{env}') \)
  \EndFor
  \State \Return \( \fn{astep-updates}(\saNumEff{a}, \mathit{env}') \)
\EndFunction
\end{algorithmic}\end{pseudocodebox}
\noindent\begin{minipage}{\linewidth}
\begin{pseudocodebox}\begin{algorithmic}
\Function{\fn{guarded-asteps}}{\( \mathit{env} : \fn{Env},\
    S : \Sett{\fn{snap-action}} \)}
  \State \( \mathit{acc} \gets \mathit{env} \)
  \ForAll{\( a \in S \)}
    \State \( \mathit{acc} \gets \mathit{acc} \sqcup
        \fn{guarded-astep-updates}(a, \mathit{env}) \)
  \EndFor
  \State \Return \( \mathit{acc} \)
\EndFunction
\end{algorithmic}\end{pseudocodebox}
\captionof{listing}{Guarded steps.}\label{lst:guarded_steps}
\end{minipage}
\par\medskip

\noindent\begin{minipage}{\linewidth}
\begin{pseudocodebox}\begin{algorithmic}
\Function{\fn{infer-bounds}}{\( \mathit{fluents} : \Sett{\vars},\
    \mathit{inits} : \Mapt{\vars}{\integers},\
    S : \Sett{\fn{snap-action}} \)}
  \State \( \mathit{env} \gets \lambda v.\, [0, -1] \)
      \Comment{empty interval}
  \ForAll{\( f \in \mathit{fluents} \)}
    \State \( \mathit{env}[f] \gets [\mathit{inits}[f],
        \mathit{inits}[f]] \)
  \EndFor
  \State \( \mathit{thrs} \gets \fn{init-thresholds}(\mathit{fluents},
      \mathit{inits}, S) \)
  \State \( \mathit{env}'' \gets \mathit{env} \)
  \Repeat
    \State \( \mathit{env} \gets \mathit{env}'' \)
    \State \( \mathit{env}' \gets \fn{guarded-asteps}(\mathit{env},
        S) \)
    \State \( \mathit{env}'' \gets \fn{widen-all}(\mathit{fluents},
        \mathit{thrs}, \mathit{env}, \mathit{env}') \)
  \Until{\( \mathit{env}'' \sqsubseteq \mathit{env} \)}
  \State \Return \( \mathit{env} \)
\EndFunction
\end{algorithmic}\end{pseudocodebox}
\captionof{listing}{Main loop.}\label{lst:infer_bounds}
\end{minipage}
\par\medskip

\begin{proposition}\label{prop:infer_bounds_termination}
We argue that termination of \fn{infer-bounds} (\Cref{lst:infer_bounds}) is guaranteed,
because in each iteration
which does not satisfy \( \mathit{env}'' \sqsubseteq \mathit{env} \),
the finite set of thresholds for a variable \( v \in \fluents \)
not within \(v\)'s bounds must shrink or
one of a finite number of bound variables is irreversibly set to \( +\infty \) or \( -\infty \).
When the final environment contains an infinite bound,
we reject the planning instance with no verdict.
\end{proposition}

\clearpage
\flushbottom %
\section{Additional Data}\label{sec:additional_benchmark_data}

\pgfplotsset{compat=1.18}

\definecolor{clrours}{HTML}{0072B2}
\definecolor{clrctp}{HTML}{D55E00}
\definecolor{clrftp}{HTML}{009E73}
\definecolor{clruppaal}{HTML}{CC79A7}
\definecolor{clrnuxmv}{HTML}{E69F00}
\definecolor{clroptic}{HTML}{56B4E9}

\pgfplotsset{
    expplot/.style={
        width=\linewidth,
        height=4.6cm,
        ymajorgrids,
        grid style={line width=0.4pt, draw=gray!25},
        axis line style={gray!60, line width=0.4pt},
        typeset ticklabels with strut,
        ytick pos=left, ytick style={draw=none},
        yticklabel style={anchor=north east, inner ysep=0.5pt},
        tick label style={font=\scriptsize},
        label style={font=\scriptsize},
        ylabel style={at={(axis description cs:0,0.5)}, anchor=south, yshift=\expylabelsep},
        xlabel style={at={(axis description cs:0.5,0)}, anchor=north, yshift=-\expxlabelsep,
                      inner ysep=0pt},
        legend style={font=\tiny, draw=gray!40, fill=white, fill opacity=0.85,
                      text opacity=1, inner sep=2pt},
        line width=1pt,
    }
}

\newlength{\expsubwidth}\setlength{\expsubwidth}{0.49\linewidth}
\newlength{\sidelegendwidth}\setlength{\sidelegendwidth}{1.2cm}
\newlength{\expaxisheight}\setlength{\expaxisheight}{3.2cm}
\newlength{\expylabelsep}\setlength{\expylabelsep}{0.7cm}
\newlength{\expxlabelsep}\setlength{\expxlabelsep}{0.42cm}
\pgfplotsset{
    expsub/.style={
        scale only axis,
        width=\expsubwidth-\sidelegendwidth-\expylabelsep-0.5cm,
        height=\expaxisheight,
        after end axis/.code={
            \node[anchor=north west, font=\scriptsize, fill=white, fill opacity=0.35,
                  text opacity=1, inner sep=1.5pt] at (rel axis cs:0,1) {(\thesubfigure)};
        },
        expextra,
    },
    expextra/.style={},
    expyticks/.style={
        ytick style={draw=gray!60, line width=0.4pt}, ytick align=outside,
        yticklabel style={xshift=4pt},
    },
}
\newcommand{\sidelegend}[3]{%
    \foreach [count=\n] \st/\lbl in {#3} {\xdef\sidelegendn{\n}}%
    \draw[gray!60, line width=0.4pt] (#1.north east) -- ++(\sidelegendwidth,0) |- (#1.south east);
    \path let \p1 = (#1.north east), \p2 = (#1.south east) in
        \pgfextra{\pgfmathsetmacro{\sidelegendslot}{(\y1-\y2)/\sidelegendn}%
                  \global\let\sidelegendslot\sidelegendslot};
    \begin{scope}[shift={(#1.north east)}, font=\tiny]
        \foreach [count=\i] \st/\lbl in {#3} {
            \tikzset{sidelegend sample/.style/.expand once=\st}
            \pgfmathsetmacro{\sltop}{-(\i-1)*\sidelegendslot - (\sidelegendslot-12)/2}
            \csname sidelegend@#2\endcsname
            \node[anchor=north west, inner sep=0pt] at (3pt, \sltop pt - 6pt) {\lbl};
            \ifnum\i<\sidelegendn
                \draw[gray!40, line width=0.4pt]
                    (0, {-\i*\sidelegendslot pt}) -- (\sidelegendwidth, {-\i*\sidelegendslot pt});
            \fi
        }
    \end{scope}%
}
\makeatletter
\def\sidelegend@line{\draw[sidelegend sample, line width=1pt]
    (3pt, \sltop pt - 2.5pt) -- ++(0.7cm, 0);}
\def\sidelegend@area{\draw[sidelegend sample]
    (3pt, \sltop pt - 0.5pt) rectangle ++(0.7cm, -4pt);}
\makeatother

\begin{figure}[t]
    \centering
    \begin{subfigure}[b]{\expsubwidth}
    \setlength{\sidelegendwidth}{1.6cm}
    \phantomcaption
    \begin{tikzpicture}
        \begin{axis}[
            expplot,
            expsub,
            name=cactus,
            ymode=log,
            log basis y=10,
            yticklabel={\pgfmathparse{int(round(\tick))}$10^{\ifnum\pgfmathresult<0 \text{-}\the\numexpr-\pgfmathresult\relax\else\pgfmathresult\fi}$},
            xlabel={instances proved unsolvable (cumulative)},
            ylabel={wall time per instance (s)},
            xmin=0, xmax=210,
            ymin=0.01, ymax=2000,
        ]
            \addplot[clrours, solid]
                table {data/cactus_ours.dat};
            \addplot[clrctp, dashed]
                table {data/cactus_tamer_ctp.dat};
            \addplot[clrftp, dotted]
                table {data/cactus_tamer_ftp.dat};
            \addplot[clruppaal, dashdotted]
                table {data/cactus_uppaal.dat};
            \addplot[clrnuxmv, densely dashed]
                table {data/cactus_nuxmv.dat};
            \addplot[clroptic, densely dotted]
                table {data/cactus_optic.dat};
        \end{axis}
        \sidelegend{cactus}{line}{
            {clrours, solid}/{ours (75)},
            {clrctp, dashed}/{T-CTP (112)},
            {clrftp, dotted}/{T-FTP (2)},
            {clruppaal, dashdotted}/{UPPAAL (66)},
            {clrnuxmv, densely dashed}/{nuXmv (146)},
            {clroptic, densely dotted}/{OPTIC (106)}%
        }
    \end{tikzpicture}
    \label{fig:cactus}
    \end{subfigure}%
    \hfill
    \begin{subfigure}[b]{\expsubwidth}
    \setlength{\sidelegendwidth}{1.35cm}
    \pgfplotsset{expextra/.style={expyticks}}
    \phantomcaption
    \begin{tikzpicture}
\begin{axis}[
  expplot,
  expsub, name=stages,
  xmode=log, xlabel={total time per instance (s)},
  ylabel={share of total time (\%)},
  ymin=0, ymax=100, xmin=0.168, xmax=598.522,
  grid=both, grid style={gray!25},
]
\addplot[draw=clrnuxmv!60!black, line width=0.3pt, fill=clrnuxmv, postaction={pattern=north east lines, pattern color=clrnuxmv!40!black}, forget plot] coordinates {(0.168,100.000) (0.171,100.000) (0.186,100.000) (0.263,100.000) (0.279,100.000) (0.399,100.000) (0.424,100.000) (0.427,100.000) (0.452,100.000) (0.590,100.000) (0.666,100.000) (0.697,100.000) (0.759,100.000) (0.799,100.000) (1.062,100.000) (1.217,100.000) (1.238,100.000) (1.316,100.000) (1.403,100.000) (1.522,100.000) (1.839,100.000) (1.938,100.000) (2.061,100.000) (2.489,100.000) (2.551,100.000) (2.821,100.000) (3.312,100.000) (4.009,100.000) (4.146,100.000) (4.614,100.000) (4.730,100.000) (5.104,100.000) (5.457,100.000) (6.239,100.000) (7.152,100.000) (8.271,100.000) (8.857,100.000) (8.897,100.000) (9.766,100.000) (10.773,100.000) (11.800,100.000) (12.571,100.000) (14.288,100.000) (14.348,100.000) (17.306,100.000) (18.336,100.000) (19.970,100.000) (21.923,100.000) (25.680,100.000) (27.453,100.000) (31.368,100.000) (33.125,100.000) (33.768,100.000) (43.536,100.000) (45.504,100.000) (48.107,100.000) (49.054,100.000) (50.500,100.000) (57.381,100.000) (58.562,100.000) (72.509,100.000) (72.550,100.000) (73.848,100.000) (81.196,100.000) (96.069,100.000) (112.394,100.000) (114.301,100.000) (132.321,100.000) (164.251,100.000) (194.850,100.000) (199.544,100.000) (286.330,100.000) (425.446,100.000) (546.142,100.000) (598.522,100.000)} \closedcycle;
\addplot[draw=clruppaal!60!black, line width=0.3pt, fill=clruppaal, postaction={pattern=crosshatch dots, pattern color=clruppaal!40!black}, forget plot] coordinates {(0.168,89.286) (0.171,88.304) (0.186,90.860) (0.263,82.129) (0.279,88.172) (0.399,86.216) (0.424,87.028) (0.427,88.056) (0.452,85.619) (0.590,79.322) (0.666,78.979) (0.697,86.657) (0.759,86.957) (0.799,85.857) (1.062,89.077) (1.217,76.007) (1.238,87.480) (1.316,77.128) (1.403,88.667) (1.522,80.289) (1.839,76.128) (1.938,90.557) (2.061,72.926) (2.489,90.237) (2.551,77.891) (2.821,79.511) (3.312,91.848) (4.009,81.591) (4.146,92.981) (4.614,79.801) (4.730,79.704) (5.104,80.564) (5.457,93.073) (6.239,93.877) (7.152,69.911) (8.271,84.476) (8.857,68.635) (8.897,94.358) (9.766,84.917) (10.773,83.208) (11.800,94.661) (12.571,85.451) (14.288,86.247) (14.348,66.706) (17.306,64.619) (18.336,67.092) (19.970,86.815) (21.923,87.552) (25.680,61.842) (27.453,60.984) (31.368,62.902) (33.125,88.562) (33.768,80.949) (43.536,61.799) (45.504,90.058) (48.107,64.194) (49.054,89.565) (50.500,63.491) (57.381,60.804) (58.562,85.718) (72.509,62.475) (72.550,91.661) (73.848,60.989) (81.196,62.682) (96.069,92.130) (112.394,62.201) (114.301,63.037) (132.321,64.687) (164.251,63.217) (194.850,65.207) (199.544,64.696) (286.330,92.727) (425.446,81.070) (546.142,70.231) (598.522,53.200)} \closedcycle;
\addplot[draw=clrftp!60!black, line width=0.3pt, fill=clrftp, forget plot] coordinates {(0.168,35.714) (0.171,62.573) (0.186,29.032) (0.263,60.076) (0.279,61.290) (0.399,69.173) (0.424,70.755) (0.427,17.564) (0.452,69.027) (0.590,35.254) (0.666,34.384) (0.697,73.458) (0.759,74.572) (0.799,73.217) (1.062,79.096) (1.217,39.523) (1.238,77.625) (1.316,42.097) (1.403,79.544) (1.522,15.243) (1.839,42.958) (1.938,83.282) (2.061,21.689) (2.489,83.809) (2.551,50.412) (2.821,52.960) (3.312,86.715) (4.009,58.818) (4.146,88.495) (4.614,58.973) (4.730,58.499) (5.104,60.580) (5.457,88.895) (6.239,90.287) (7.152,9.522) (8.271,69.556) (8.857,10.399) (8.897,91.222) (9.766,70.889) (10.773,68.783) (11.800,91.822) (12.571,72.476) (14.288,74.216) (14.348,13.640) (17.306,14.180) (18.336,8.442) (19.970,76.099) (21.923,77.690) (25.680,17.601) (27.453,18.373) (31.368,20.553) (33.125,80.211) (33.768,22.287) (43.536,23.980) (45.504,83.250) (48.107,26.549) (49.054,82.670) (50.500,26.750) (57.381,25.655) (58.562,13.779) (72.509,29.242) (72.550,86.210) (73.848,28.334) (81.196,7.136) (96.069,87.435) (112.394,34.512) (114.301,34.829) (132.321,37.216) (164.251,39.415) (194.850,5.829) (199.544,42.490) (286.330,48.629) (425.446,33.043) (546.142,11.371) (598.522,7.532)} \closedcycle;
\addplot[draw=clrctp!60!black, line width=0.3pt, fill=clrctp, postaction={pattern=north west lines, pattern color=clrctp!40!black}, forget plot] coordinates {(0.168,26.190) (0.171,51.462) (0.186,20.430) (0.263,41.825) (0.279,40.502) (0.399,33.333) (0.424,32.075) (0.427,10.304) (0.452,31.416) (0.590,21.695) (0.666,19.369) (0.697,22.238) (0.759,20.685) (0.799,21.402) (1.062,16.855) (1.217,12.325) (1.238,15.267) (1.316,12.462) (1.403,14.255) (1.522,3.351) (1.839,9.734) (1.938,12.023) (2.061,7.181) (2.489,9.442) (2.551,8.898) (2.821,7.586) (3.312,8.424) (4.009,6.436) (4.146,7.477) (4.614,6.567) (4.730,5.624) (5.104,5.270) (5.457,5.974) (6.239,6.235) (7.152,2.125) (8.271,4.413) (8.857,1.829) (8.897,5.474) (9.766,4.004) (10.773,4.168) (11.800,4.475) (12.571,3.890) (14.288,3.380) (14.348,1.701) (17.306,1.329) (18.336,0.960) (19.970,3.065) (21.923,2.983) (25.680,1.164) (27.453,1.326) (31.368,1.211) (33.125,2.165) (33.768,0.163) (43.536,1.169) (45.504,1.908) (48.107,0.817) (49.054,1.763) (50.500,1.135) (57.381,1.122) (58.562,0.155) (72.509,0.832) (72.550,1.691) (73.848,1.091) (81.196,0.281) (96.069,1.392) (112.394,0.875) (114.301,0.866) (132.321,0.904) (164.251,0.652) (194.850,0.120) (199.544,0.680) (286.330,0.546) (425.446,0.026) (546.142,0.034) (598.522,0.100)} \closedcycle;
\addplot[draw=clrours!60!black, line width=0.3pt, fill=clrours, postaction={pattern=horizontal lines, pattern color=clrours!40!black}, forget plot] coordinates {(0.168,8.929) (0.171,25.146) (0.186,4.839) (0.263,17.110) (0.279,15.412) (0.399,11.278) (0.424,10.849) (0.427,2.108) (0.452,9.735) (0.590,7.797) (0.666,6.456) (0.697,6.600) (0.759,6.324) (0.799,5.757) (1.062,4.614) (1.217,3.698) (1.238,3.716) (1.316,3.419) (1.403,3.564) (1.522,1.051) (1.839,2.447) (1.938,2.632) (2.061,2.183) (2.489,1.969) (2.551,1.764) (2.821,1.702) (3.312,1.600) (4.009,1.172) (4.146,1.230) (4.614,0.997) (4.730,0.951) (5.104,0.921) (5.457,0.935) (6.239,0.849) (7.152,0.629) (8.271,0.580) (8.857,0.497) (8.897,0.607) (9.766,0.481) (10.773,0.418) (11.800,0.466) (12.571,0.390) (14.288,0.357) (14.348,0.335) (17.306,0.272) (18.336,0.256) (19.970,0.250) (21.923,0.233) (25.680,0.175) (27.453,0.168) (31.368,0.150) (33.125,0.157) (33.768,0.050) (43.536,0.106) (45.504,0.123) (48.107,0.104) (49.054,0.108) (50.500,0.095) (57.381,0.080) (58.562,0.055) (72.509,0.069) (72.550,0.076) (73.848,0.065) (81.196,0.062) (96.069,0.056) (112.394,0.045) (114.301,0.041) (132.321,0.039) (164.251,0.031) (194.850,0.023) (199.544,0.026) (286.330,0.019) (425.446,0.003) (546.142,0.008) (598.522,0.008)} \closedcycle;
\end{axis}
\sidelegend{stages}{area}{
  {draw=clrnuxmv!60!black, line width=0.3pt, fill=clrnuxmv, postaction={pattern=north east lines, pattern color=clrnuxmv!40!black}}/{checking},
  {draw=clruppaal!60!black, line width=0.3pt, fill=clruppaal, postaction={pattern=crosshatch dots, pattern color=clruppaal!40!black}}/{parsing},
  {draw=clrftp!60!black, line width=0.3pt, fill=clrftp}/{TChecker},
  {draw=clrctp!60!black, line width=0.3pt, fill=clrctp, postaction={pattern=north west lines, pattern color=clrctp!40!black}}/{encoding},
  {draw=clrours!60!black, line width=0.3pt, fill=clrours, postaction={pattern=horizontal lines, pattern color=clrours!40!black}}/{grounding}%
}
\end{tikzpicture}%
    \label{fig:stage_times}
    \end{subfigure}
    \caption{(a) Cumulative number of instances (of 397) proved unsolvable within a given wall time (timeout 1800\,s);
    total in legend.
    (b) Per-stage times of our pipeline for the 75
    certified-unsolvable instances, 
    normalised per instance.}
\end{figure}

\Cref{tab:solved_instances} and \Cref{fig:cactus} show how the different planners compare to one another.
The data shows that our pipeline performs on par with the other unverified planners (3 vs. 2, 3) on MatchCellar,
is worse than TAMER-CTP, nuXmv, and UPPAAL on SYNC (3 vs. 5, 6, 4) and PAINTER (1 vs. 5, 4, 4);
OPTIC also outperforms ours on PAINTER with 2 solved instances, 
but does not solve any SYNC instance.
Our pipeline does not solve any instance of MAJSP-1, while TAMER-CTP and nuXmv solve 12 and 22 respectively;
OPTIC and UPPAAL solve 1 each.
In MAJSP-2, our pipeline outperforms UPPAAL (68 vs. 54),
while being outperformed by TAMER-CTP and nuXmv (88 and 112 respectively);
OPTIC detects 100 (out of 240) as unsolvable while preprocessing and grounding the problem,
while timing out on 103 and running out of memory on 37 of the remaining instances.

The dominant failure mode of our pipeline in MatchCellar, SYNC, and PAINTER is time,
while failures in MAJSP-1 and MAJSP-2 are roughly evenly split between time and memory.
See \Cref{fig:stage_times}.
We point out another thing: in our pipeline, 
the execution time and memory needed are wholly dominated by TChecker,
certificate parsing and verified certificate checking,
i.e.\ our verified encoding usually needs negligible resources
compared to the other stages.

In 11 MAJSP instances, 
our pipeline failed because TChecker was unable to infer clock bounds.
UPPAAL segfaulted in one MAJSP-2 instance.

TChecker can prove one more SYNC instance unsolvable 
when using an LU-abstracted zone subsumption
algorithm~\cite{HERBRETEAU201667} to ensure termination (see \Cref{tab:alu_times}).
This brings the performance of our approach on SYNC on par with UPPAAL and closer
to TAMER and nuXmv.
However, certificates in the form of LU-abstracted zone graphs are not 
currently supported by our certificate checker.

\begin{table}[tb]
    \centering
    \setlength{\tabcolsep}{3pt}
    \resizebox{\columnwidth}{!}{%
    \begin{tabular}{lrrrrrrrrr}
\toprule
 & \multicolumn{3}{c}{tck-reach (s)} & \multicolumn{3}{c}{end-to-end (s)} & \multicolumn{3}{c}{certificate (MB)} \\
\cmidrule(lr){2-4}\cmidrule(lr){5-7}\cmidrule(lr){8-10}
Instance & covreach & aLU & $\times$ & covreach & aLU & $\times$ & covreach & aLU & $\times$ \\
\midrule
MatchCellar 02 & 0.016 & 0.015 & 1.1 & 0.29 & 0.27 & 1.1 & $<$1 & $<$1 & -- \\
MatchCellar 03 & 0.031 & 0.025 & 1.2 & 0.55 & 0.43 & 1.3 & $<$1 & $<$1 & -- \\
MatchCellar 04 & 140 & 14.4 & 9.7 & 427 & 92 & 4.6 & 835 & 242 & 3.5 \\
SYNC 1\_2 & 0.016 & 0.017 & 0.9 & 0.28 & 0.27 & 1.0 & $<$1 & $<$1 & -- \\
SYNC 1\_3 & 0.181 & 0.091 & 2.0 & 1.66 & 0.94 & 1.8 & 2 & 1 & 2.0 \\
SYNC 1\_4 & 7.47 & 2.26 & 3.3 & 34.1 & 10.5 & 3.3 & 49 & 17 & 2.9 \\
\midrule
SYNC 1\_5 & timeout & 242 & -- & timeout & 359 & -- & -- & 240 & -- \\
\bottomrule
\end{tabular}
    }
    \caption{The table shows the times for TChecker's reachability check and the end-to-end pipeline, and the size of the emitted
    certificates. ``covreach'' here refers to TChecker not exploring states covered in the zone graph twice 
    and aLU refers to detecting coverage using LU-abstracted zone graphs.
    Note that the end-to-end figure stops when the certificate is rejected for the aLU algorithm, 
    which might overrepresent the time savings, but not by much, given the distribution of time taken for each stage 
    (see \Cref{fig:stage_times}) and the clearly smaller state space to check.}
    \label{tab:alu_times}
\end{table}

\end{document}